\documentclass[10pt,a4paper,oneside,english]{article}   
\usepackage[T1]{fontenc}
\usepackage[utf8]{inputenc}
\usepackage[english]{babel}
\usepackage{latexsym}
\usepackage{midpage}
\usepackage{newlfont}
\usepackage{color}
\usepackage{amsmath} 
\usepackage{amsfonts} 
\usepackage{amsthm} 
\usepackage[hidelinks]{hyperref} 
\usepackage{graphicx} 
\usepackage{caption}
\usepackage{array}
\usepackage{subcaption} 
\usepackage{mathdots}
\theoremstyle{plain} 
\usepackage{rotating}
\usepackage{csquotes}
\newtheorem{theorem}{Theorem}[section] 
\newtheorem{corollary}[theorem]{Corollary} 
\newtheorem{lemma}[theorem]{Lemma} 
\newtheorem{proposition}[theorem]{Proposition} 
\theoremstyle{definition} 
\newtheorem{definition}{Definition}[section] 
\theoremstyle{remark} 
\newtheorem{remark}{Remark} 
\newtheorem{example}{Example}[section] 

\usepackage{nicematrix}
\numberwithin{equation}{section}
\usepackage{longtable}
\usepackage{setspace}
\usepackage{booktabs}
\usepackage[margin=1in,left=1.04in, right=1.04in,includefoot]{geometry}
\usepackage{rotating}
\usepackage{listings}
\usepackage{enumitem}

\newcommand{\E}{\mathbb{E}} 
\newcommand{\Q}{\mathbb{Q}} 
\newcommand{\R}{\mathbb{R}} 
\newcommand{\F}{\mathcal{F}} 
\newcommand{\G}{\mathcal{G}} 

\author{Silvia Lavagnini \\ \small{Department of Mathematics, University of Oslo}\\ \small{silval@math.uio.no}}
\title{Pricing Asian Options with Correlators
\footnotetext{The author would like to thank Fred Espen Benth, Vegard Antun and Salvador Ortiz-Latorre for discussions.}}

\newlength{\bibitemsep}
\newlength{\bibparskip}
\let\oldthebibliography\thebibliography
\renewcommand\thebibliography[1]{%
	\oldthebibliography{#1}%
	\setlength{\parskip}{\bibitemsep}%
	\setlength{\itemsep}{\bibparskip}%
}

\begin{document}
\maketitle

\begin{abstract}
We derive a series expansion by Hermite polynomials for the price of an arithmetic Asian option. This series requires the computation of moments and correlators of the underlying price process, but for a polynomial jump-diffusion, these are given in closed form, hence no numerical simulation is required to evaluate the series. This allows, for example, for the explicit computation of Greeks. The weight function defining the Hermite polynomials is a Gaussian density with scale $b$. 
We find that the rate of convergence for the series depends on $b$, for which we prove a lower bound to guarantee convergence. Numerical examples show that the series expansion is accurate but unstable for initial values of the underlying process far from zero, mainly due to rounding errors. 

\end{abstract}

\paragraph*{Keywords} Asian option; Option pricing; Greeks; Orthogonal polynomials; Generalized Hermite polynomials; Polynomial jump-diffusion process; Correlators.

\section{Introduction}
\label{intro}
Asian options are path-dependent options whose payoff is based on the (discrete or continuous) average underlying price. This kind of derivatives has application in fields like currency, interest rate, energy and insurance markets, among others. Within the energy markets, for example, Asian options were traded a decade ago at Nord Pool, the Nordic commodity market for electricity \cite{weron2007}. However, because of their path-dependent nature, their valuation is not straightforward, and, in particular, no closed pricing formula is available in general. Possible approaches to evaluate Asian options are, e.g., via Monte Carlo simulations \cite{kemna1990, lapeyre2001} or Fourier transform \cite{fusai2008}. Other authors have derived exact representations for the pricing functional, for example as a triple integral to be evaluated numerically \cite{yor1992}, or by the Laplace transform \cite{geman1993}. Alternatively, one can approximate the unknown distribution of the average price \cite{fusai2002, li2016, turnbull1991}. Recently, new pricing approaches for Asian options have been considered in relation to polynomial processes and orthogonal polynomials \cite{dufresne2000, filipovic2020, willems2019}.

Our approach is partly similar to \cite{willems2019}, in the sense that we derive a series representation for the option price functional with orthogonal polynomials and we work with polynomial processes. However, we extend their result in at least two directions. First of all, in \cite{willems2019} the underlying spot price is considered to follow a geometric Brownian motion, hence, in particular, the jump behavior of spot prices is not taken into consideration. Then, by the time-reversal property of Brownian motions (\cite{carmona1997}), they derive a stochastic differential equation (SDE) whose solution process has the same distribution of the average price process (scaled by the terminal time). Specifically, the SDE defines a polynomial diffusion, so that the moments of the average price process can be computed in closed form by the moment formula for polynomial processes. Our approach is different because we model directly the underlying spot price with a polynomial jump-diffusion. Hence, on the one hand, we allow for discontinuities in the spot price paths, and, on the other hand, the moment formula for polynomial processes can be used in this case for computing the moments of the underlying spot price, and not of the average price, as in \cite{willems2019}. However, 
we still need to compute moments of the average price process. We do this by the multinomial theorem and the correlator formula for polynomial processes derived in \cite{benth2019}. 

We fix a stochastic basis $(\Omega, \F, \{\F_t\}_{t\ge 0},\Q)$, with $\Q$ a risk-neural measure. We want to price the fixed-strike call-style Asian option defined by
\begin{equation}
	\label{payoff}
	\Pi_K(t):=e^{-r(T-t)}\E\left[ \varphi_K(X(T))\left.\right|\F_t\right] \qquad \mbox{ with } \qquad \varphi_K(x) := \max(x-K, 0),
\end{equation}  
where $K>0$ is the strike price, $r\ge0$ the risk-free interest rate and $X$ is the discrete average of a stochastic process $Y$ over the period $(t,T]$, namely
\begin{equation}
	\label{X}
	X(T) = \frac{1}{m+1}\sum_{j=0}^m Y(s_j) \qquad \mbox{ for } t < s_0<s_1 <\dots< s_m = T \mbox{ and } m\ge 0.
\end{equation}
We point out that one can similarly consider $X$ to be the continuous average $X(T)=\int_{t}^{T}Y(s)ds$. In this case, if $t<s_0 <s_1<\dots < s_m = T$ is a discrete sampling with time steps $\Delta_{j}:=s_j-{s_{j-1}}$ small enough, $j=1,\dots,m$, then the integral $\int_{t}^{T}Y(s)ds$ can be reasonably well approximated with the sum $\sum_{j=1}^{m}Y(s_j)\Delta_{j}$. This coincides with equation \eqref{X} for $\Delta_j=\frac{1}{m+1}$, $j=1,\dots,m$.

We consider $Y$ to be a polynomial process in the sense introduced \cite{filipovic2020}. The idea is to derive the series representation of the payoff function $\varphi_K$ in terms of Hermite polynomials. More precisely, we shall introduce the so-called \emph{generalized Hermite polynomials} that form a basis for the space 
\begin{equation*}
	L^2\left(\R, \omega_{a,b}(x)dx\right) \qquad  \mbox{ with } \qquad \omega_{a,b}(x)= \exp\left(-\frac{(x-a)^2}{2b^2}\right), \;a, b\in \R, b>0. 
\end{equation*}
After evaluating the series at $x = X(T)$ in equation \eqref{X}, we obtain an infinite sum of polynomial functions in $X(T)$. The price of the Asian option is then given by the discounted expected value of this infinite sum. By the multinomial theorem, we rewrite the terms of the sum as a linear combination of correlator-type terms in the sense of \cite{benth2019}, that is, terms of the form
\begin{equation}
	\label{corr}
	\E\left[\left.Y(s_0)^{k_0}Y(s_1)^{k_1}\cdots Y(s_{m})^{k_{m}}\right|\F_t\right],
\end{equation}
which we compute by the closed formula for correlators \cite[Theorem 4.5]{benth2019}.

This procedure gives an exact formula for pricing discrete Asian options. However, for numerical purposes the infinite summation must be truncated to a certain $N>0$, leading to an approximation of the price. We study the behaviour of the approximation error with respect to the three parameters involved, namely $N$, $a$ and $b$, and we confirm  our findings with numerical examples. We also compare the results with a Monte-Carlo-simulation approach. This shows that the Hermite series can reach much higher accuracies than Monte Carlo. However, numerical instabilities are observed, mainly due to the intrinsic exploding nature of polynomial functions of high order. In particular, these are more likely to happen when the initial point of the underlying spot price process $Y$ is far out from $0$.

The rest of the paper is organized as follows. In Section \ref{hermitesection} we introduce the family of generalized Hermite polynomials and we derive the series expansion for a call-payoff function, also studying the approximation error as a function of the truncation number. In Section \ref{pricingsection} we derive the option price approximation, first for a European-style option  and then for an Asian option with discrete sampling. In Section \ref{polpross} we briefly introduce polynomial processes and recall the moment and correlator formulas, and we derive explicit representations for two of the Greeks of the option.  Finally in Section \ref{numerics} we show some numerical examples and in Section \ref{conclusions} we summarize the findings. Appendix \ref{proof} contains the proofs of the principal results and Appendix \ref{appA} some definitions for understanding the correlator formula.

\section{Payoff representation with Hermite polynomials}
\label{hermitesection}
We shall construct in this section a polynomial approximation for the payoff function $\varphi_K$ in equation \eqref{payoff}. Let $\mathrm{Pol}_n(\R)$ be the space of all polynomials on $\R$ with degree less than or equal to $n$, and let $q_0(x), q_1(x), \dots$ be orthogonal polynomial functions with values in $\R$, such that the family $\{q_0(x), q_1(x), \dots, q_n(x)\}$ forms a basis for $\mathrm{Pol}_n(\R)$. We then introduce the vector valued function
\begin{equation*}
	Q_n:\R \longrightarrow \R^{n+1}, \quad Q_n(x) = (q_0(x), q_1(x), \dots, q_n(x))^{\top}.
\end{equation*}
Similarly, we consider the family of monomials $\{1,x,\cdots,x^n\}$ that also forms a basis for $\mathrm{Pol}_n(\R)$, and introduce the vector valued function 
\begin{equation*}
	H_n:\R \longrightarrow \R^{n+1}, \quad H_n(x) = (1,x,x^2,\dots,x^n)^{\top}.
\end{equation*}
The reason for considering two basis vectors is because the basis of monomials $H_n$ is practical and allows to obtain explicit formulas. However, when it comes to applications, such as polynomial approximation, one needs an orthogonal or orthonormal basis. From classical linear algebra, there exists a matrix
\begin{equation}
	\label{M} 
M_n\in \R^{(n+1)\times (n+1)} \quad \mbox{ such that } \quad M_nH_n(x) = Q_n(x) \quad \mbox{ and } \quad M_n^{-1}Q_n(x)=H_n(x).
\end{equation}
By equation \eqref{M} we can exploit both the readability of $H_n$ and the orthogonality of $Q_n$.

\subsection{Generalized Hermite polynomials}
We restrict our attention to the so-called \emph{probabilistic Hermite polynomials} (which we shall refer to simply as \emph{Hermite polynomials}) defined by
\begin{equation*}
	q_n(x) := (-1)^ne^{\frac{x^2}{2}}\frac{d^n}{dx^n}e^{-\frac{x^2}{2}}, \quad n \ge 0.
\end{equation*}
The family $\{q_n\}_{n\ge0}$ forms an orthogonal basis for the Hilbert space $L^2(\R, w(x)dx)$  with weight function $w(x):=e^{-\frac{x^2}{2}}$. Moreover, the norm of $q_n$ in $L^2(\R, w(x)dx)$ is given by
\begin{equation}
	\label{norm}
	\left \| q_n\right \| ^2_{L^2(\R, w(x)dx)} = \int_{-\infty}^{\infty}q_n^2(x)w(x)dx = \sqrt{2\pi}n!\,.
\end{equation}
It is easy to check that $\varphi_K \in L^2(\R, w(x)dx)$. However, the weight function $w$ is centred in $x=0$, which means that an approximation with Hermite polynomials will have the main focus in a neighbourhood of $x=0$ and will potentially not be good for points far from it. Since we want to approximate the payoff function $\varphi_K$ whose most interesting point is $x=K$, we thus need a weight function possibly centred in $x=K$. Alternatively, in view of option pricing where $\varphi_K$ is evaluated on a random variable $X$, one might want to focus around the mean of $X$. We should then consider a weight function that allows to shift the focus of the approximation to the area of greatest interest.

To keep it general, for $a, b\in \R$, $b>0$, we introduce a family of weight functions and the corresponding orthogonal polynomials by
\begin{equation*}
	w_{a,b}(x):= e^{-\frac{(x-a)^2}{2b^2}}\qquad \mbox{ and } \qquad q_n^{a,b}(x) := (-1)^ne^{\frac{(x-a)^2}{2b^2}}\frac{d^n}{dx^n}e^{-\frac{(x-a)^2}{2b^2}}.
\end{equation*}
The family $\{q_n^{a,b}\}_{n\ge0}$ forms an orthogonal basis for the Hilbert space $L^2(\R, w_{a,b}(x)dx)$ equipped with the norm $$\| f\|^2_{L^2(\R, w_{a,b}(x)dx)}:= \int_{\R}f(x)^2 w_{a,b}(x)dx\quad  \mbox{ for } \quad f \in L^2(\R, w_{a,b}(x)dx).$$
The norm of $q_n^{a,b}$ in $L^2(\R, w_{a,b}(x)dx)$ is given in the following lemma.
\begin{lemma}
	\label{prop}
	For every $n\ge 0$, the norm of $q_n^{a,b}$ in $L^2(\R, w_{a,b}(x)dx)$ is
	$\left \| q_n^{a,b}\right \| ^2_{L^2(\R, w_{a,b}(x)dx)} = \frac{\sqrt{2\pi}n!}{b^{2n-1}}.$
\end{lemma}

We shall from now on refer to $a$ as the \emph{drift} and to $b$ as the \emph{scale}, while to $\{q_n^{a,b}\}_{n\ge0}$ as \emph{generalized Hermite polynomials} (GHPs). We also introduce the notation $L^2_{a,b}:=L^2(\R, w_{a,b}(x)dx)$, where $L^2_{0,1}=L^2(\R, w(x)dx)$. We point out that, while we shall use the terminology \enquote{weight function}, \enquote{approximating series}, etc., we deal in practice with a family of weight functions, a family of approximating series, etc., depending on the choice of the parameters $a$ and $b$. 


\subsection{Series expansion for the call-payoff function}
We introduce $\varphi_K^{a,b}$ as the series representation of $\varphi_K$ in terms of the GHPs $\{q_n^{a,b}\}_{n\ge0}$, namely
\begin{equation}
	\varphi_K^{a,b}(x):=\sum_{n=0}^{\infty} \frac{\left\langle \varphi_K, q_n^{a,b}\right\rangle_{L^2_{a,b}}}{\left \| q_n^{a,b}\right \| ^2_{L^2_{a,b}}} \; q_n^{a,b}(x)=  \sum_{n=0}^{\infty} \frac{b^{2n-1}}{\sqrt{2\pi}n!} \int_{-\infty}^{\infty}\varphi_K(y)q_n^{a,b}(y)w_{a,b}(y)dy \; q_n^{a,b}(x),\label{phiK}
\end{equation}
which we shall compute explicitly. From now on, we denote with $\phi$ and $\Phi$, respectively, the probability density function and the cumulative distribution function of a standard Gaussian random variable. 

\begin{proposition}
\label{propphi}
	The series $\varphi^{a,b}_K$ can be written in terms of the Hermite polynomials $\{q_n\}_{n\ge 0}$ by
	\begin{equation*}
		\varphi^{a,b}_K(x) = \sum_{n=0}^{\infty}
		\beta_n^{a,b} q_n\left(\frac{x-a}{b}\right) \quad \mbox{ with } \beta_n^{a,b} :=\begin{cases}
			b\,\phi\left(\frac{K-a}{b}\right)+\left(a-K\right)\left(1-\Phi\left(\frac{K-a}{b}\right)\right) & \mbox{ for } n = 0\\
			b\left(1-\Phi\left(\frac{K-a}{b}\right)\right) & \mbox{ for } n = 1\\
			\frac{b}{n!} \phi\left(\frac{K-a}{b}\right)q_{n-2}\left(\frac{K-a}{b}\right)& \mbox{ for } n\ge 2
		\end{cases}.
	\end{equation*}
\end{proposition}


\begin{example}
	\label{exmusigma}
	Let X be a random variable with mean and variance denoted respectively with $\mu$ and $\sigma^2$. We then consider the drift $a=\mu$ and the scale $b=\sigma$. From Proposition \ref{propphi} we get
	\begin{equation*}
		\varphi^{\mu,\sigma}_K(x) = \sigma\,\phi\left(\frac{K-\mu}{\sigma}\right) + \left(1-\Phi\left(\frac{K-\mu}{\sigma}\right)\right)\left(x-K\right)  + \sum_{n=2}^{\infty}
		\beta_n^{\mu,\sigma} q_n\left(\frac{x-\mu}{\sigma}\right).
	\end{equation*}
	Moreover, if $X$ follows a Gaussian distribution, then by computing the expectation of $\varphi^{\mu,\sigma}_K(X)$ we get
	\begin{align*}
		\E\left[\varphi^{\mu,\sigma}_K(X)\right] 
		&= \sigma\,\phi\left(\frac{K-\mu}{\sigma}\right) + \left(1-\Phi\left(\frac{K-\mu}{\sigma}\right)\right)\left(\mu-K\right)  +\E\left[\sum_{n=2}^{\infty}
		\beta_n^{\mu,\sigma} q_n\left(\frac{X-\mu}{\sigma}\right)\right]\\
		& = \E \left[\varphi_K(X)\right] + \E\left[\sum_{n=2}^{\infty}
		\beta_n^{\mu,\sigma} q_n\left(\frac{X-\mu}{\sigma}\right)\right].
	\end{align*}
	This means that if the drift is $a=\mu$ and the scale is $b=\sigma$, then the weight function $\omega_{\mu,\sigma}$ coincides with the density function of the random variable $X$. Hence, calculating the expectation $\E \left[\varphi_K(X)\right]$ by computing $\E\left[\varphi^{\mu,\sigma}_K(X)\right]$ might add uncertainty to the result, unless the coefficients are non significant.
\end{example}

\begin{example}
	For X as in Example \ref{exmusigma}, we let the drift be $a=K$ and the scale be $b=\sigma$. We get
	\begin{equation*}
		\varphi^{K,\sigma}_K(x) = \frac{\sigma}{\sqrt{2\pi}} + \frac{x-K}{2}  + \sum_{n=2}^{\infty}
		\beta_n^{K,\sigma} q_n\left(\frac{x-K}{\sigma}\right) \qquad \mbox{ with } \beta_n^{K,\sigma} = \frac{\sigma}{\sqrt{2\pi}n!}q_{n-2}(0).
	\end{equation*}
	In particular, if $n$ is an odd number then $\beta_n^{K,\sigma}=0$ because the Hermite polynomials of odd orders have no constant term. More precisely, for every $n\ge2$, we introduce $k\ge1$ as the integer such that either $n = 2k$ or $n = 2k+1$. Then, the coefficients $\beta_n^{K,\sigma}$ are given by
	\begin{equation*}
		\beta_n^{K,\sigma} = \begin{cases}
			\frac{(-1)^{k-1}\sigma}{\sqrt{2\pi}k!(2k-1)2^k}  & \mbox{ for } n = 2k\\
			0 & \mbox{ for } n = 2k+1\\
		\end{cases}.
	\end{equation*}
	This is obtained by observing that $q_{n}(0) = (-1)^{\frac{n}{2}}\frac{n!}{2^{n/2}k!}$
	for $n = 2k$ even, and $q_{n}(0) = 0$ for $n$ odd, 
	where $!!$ denotes the double factorial. Then $\beta_n^{K,\sigma}=0$ for $n$ odd, while for $n$ even we write that 
	\begin{equation*}
		\beta_n^{K,\sigma} = \frac{\sigma}{\sqrt{2\pi}n!}q_{n-2}(0)=\frac{(-1)^{\frac{n-2}{2}}(n-3)!!\,\sigma}{\sqrt{2\pi}n!} = \frac{(-1)^{\frac{2k-2}{2}}(2k-3)!!\,\sigma}{\sqrt{2\pi}(2k)!} = \frac{(-1)^{k-1}(2k-1)!!\,\sigma}{\sqrt{2\pi}(2k)!(2k-1)}.
	\end{equation*}
	In particular, $(2k-1)!!= \frac{(2k-1)!}{2^{k-1}(k-1)!}$, so that, after simplification, we obtain the formula above.
\end{example}

\subsection{Error analysis}
For computational purposes, the summation in Proposition \ref{propphi} must be truncated to a certain $N$ big enough so that the resulting series well approximates the original payoff function $\varphi_K$. This leads to
\begin{equation} 
	\label{phiNK}
	\varphi_{K,N}^{a,b}(x) := \sum_{n=0}^{N}
	\beta_n^{a,b} q_n\!\left(\frac{x-a}{b}\right)=
	 \boldsymbol{\beta}^{a,b \,\top}_N Q_N\!\left(\frac{x-a}{b}\right) = \boldsymbol{\beta}^{a,b \,\top}_NM_NH_N\!\left(\frac{x-a}{b}\right),
\end{equation}
where $\boldsymbol{\beta}^{a,b}_N := (\beta^{a,b}_0, \beta^{a,b}_1, \dots, \beta^{a,b}_N)^{\top}$ and $M_N$ is the matrix for the change of basis with respect to $H_N$ in equation \eqref{M}. We point out that $\varphi_{K,N}^{a,b}$ can be computed for any choice of the orthogonal basis $\{q_n\}_{n\ge 0}$. Then equation \eqref{phiNK} holds with the obvious modifications for  $M_N$ and $\boldsymbol{\beta}^{a,b}_N$.

In Figure \ref{hermite} we observe the behaviour of $\varphi_{K,N}^{a,b}$ for different values of $N\in \{5, 15, 30, 100\}$ and $b\in \{0.5, 1.0, 2.0, 3.0\}$. In particular, we fix the drift to $a=K$ so that the approximations are centred around the strike price value $K=5.0$. The area where the Hermite series well approximates the payoff function gets wider when increasing the scale $b$. Similarly, increasing the truncation number $N$ gives better performances, but this is more evident for bigger values of $b$. Moreover, including higher order polynomials in the series adds oscillations to the approximation. We point out that the value of the drift $a$ is kept constant: the only effect of changing the drift is a shift in the focus of the approximation, meaning that, since the GHPs are centred around $a$, then moving $a$ from $K$ implies a move of the centre of the approximation, which is not particularly interesting for this experiment.

\begin{figure}[tp]
	\setlength{\tabcolsep}{2pt}
	\resizebox{1\textwidth}{!}{
		\begin{tabular}{@{}>{\centering\arraybackslash}m{0.04\textwidth}@{}>{\centering\arraybackslash}m{0.24\textwidth}@{}>{\centering\arraybackslash}m{0.24\textwidth}@{}>{\centering\arraybackslash}m{0.24\textwidth}@{}>{\centering\arraybackslash}m{0.24\textwidth}@{}}
			&$\boldsymbol{N = 5}$ & $\boldsymbol{N = 15}$& $\boldsymbol{N = 30}$ & $\boldsymbol{N = 100}$\\
			\begin{turn}{90}$\boldsymbol{b =0.5}$\end{turn}
			&\includegraphics[width=0.23\textwidth]{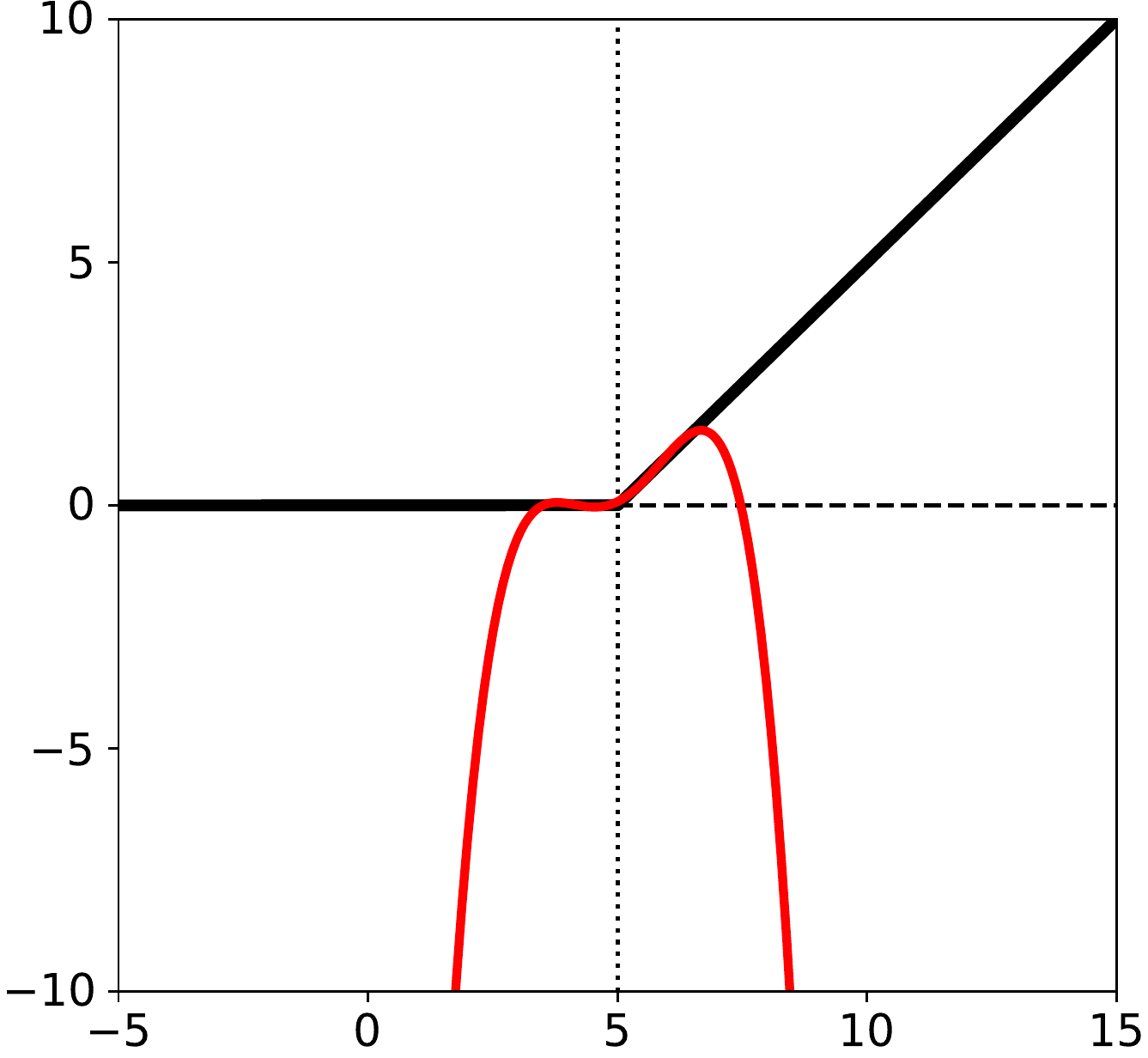} &	\includegraphics[width=0.23\textwidth]{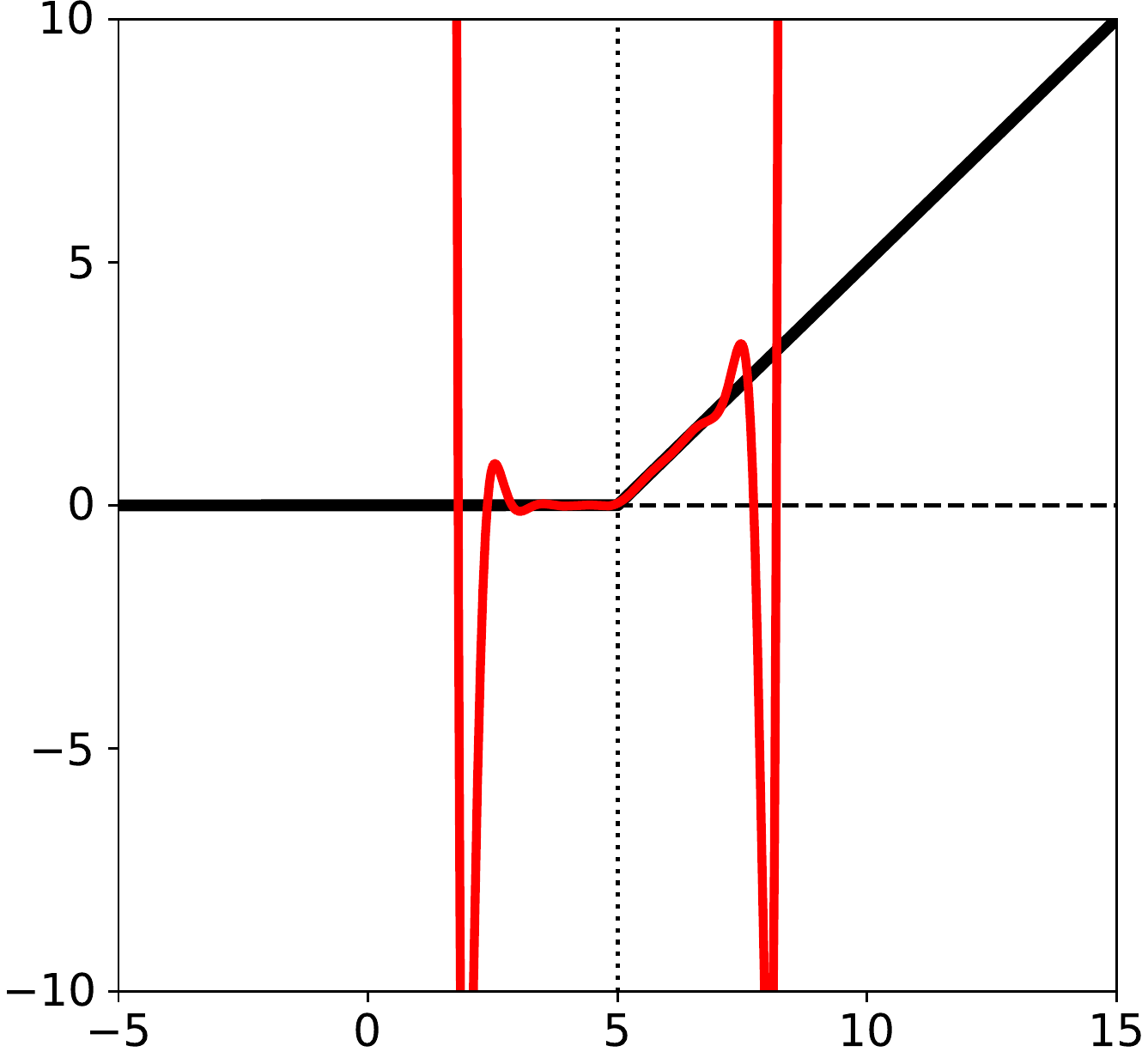} &
			\includegraphics[width=0.23\textwidth]{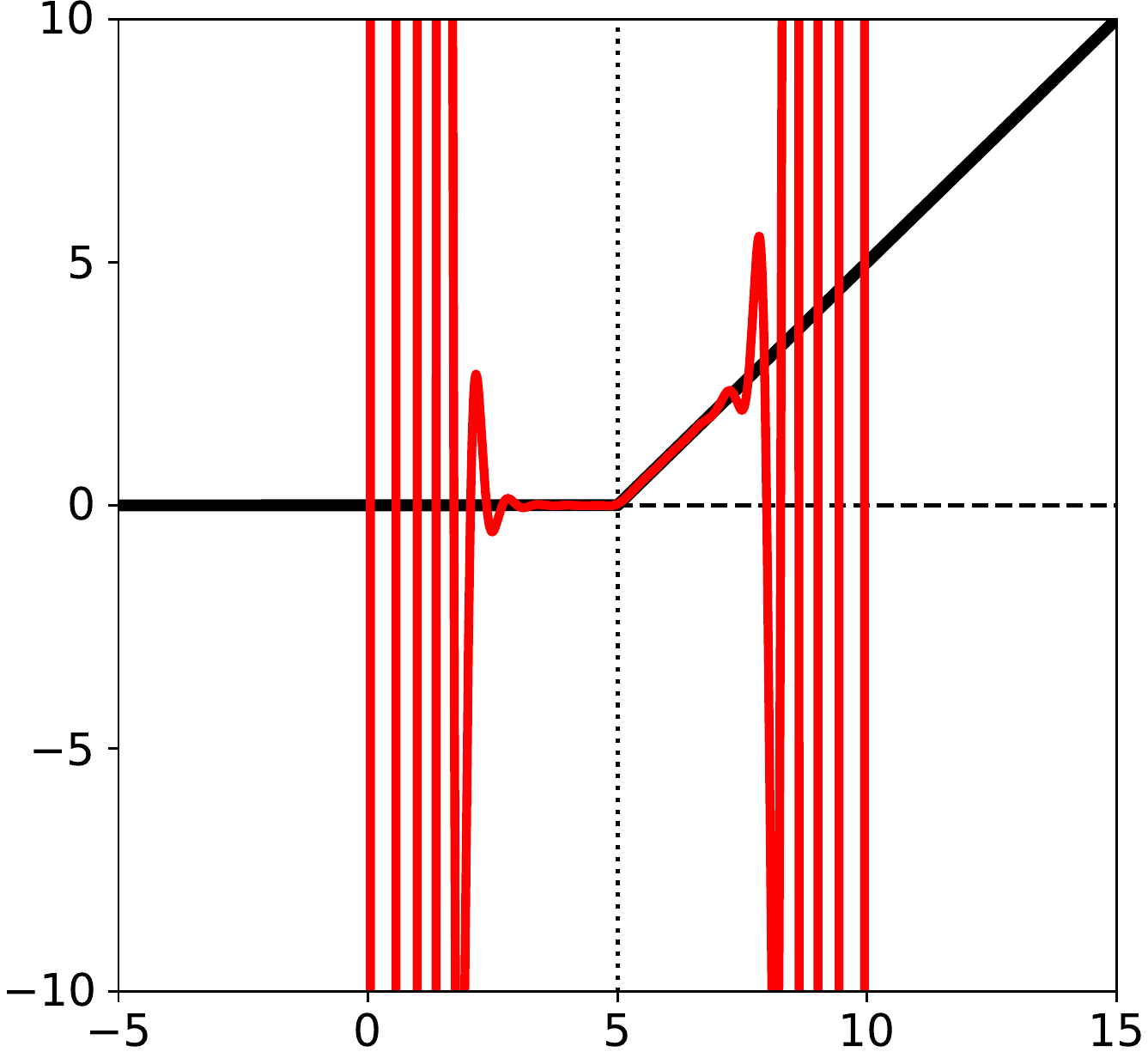}&
			\includegraphics[width=0.23\textwidth]{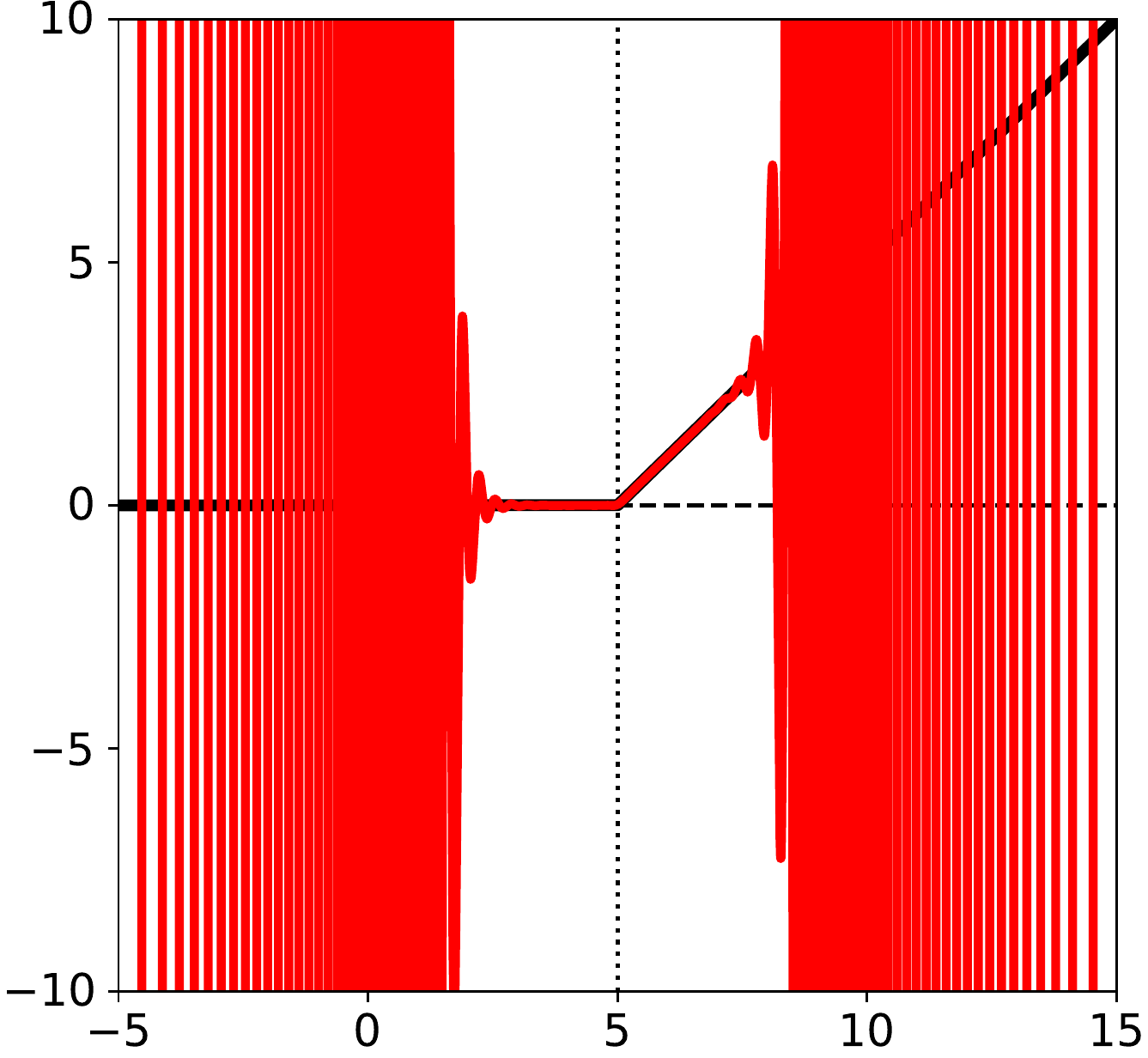}\\
			\begin{turn}{90}$\boldsymbol{b =1.0}$\end{turn}&\includegraphics[width=0.23\textwidth]{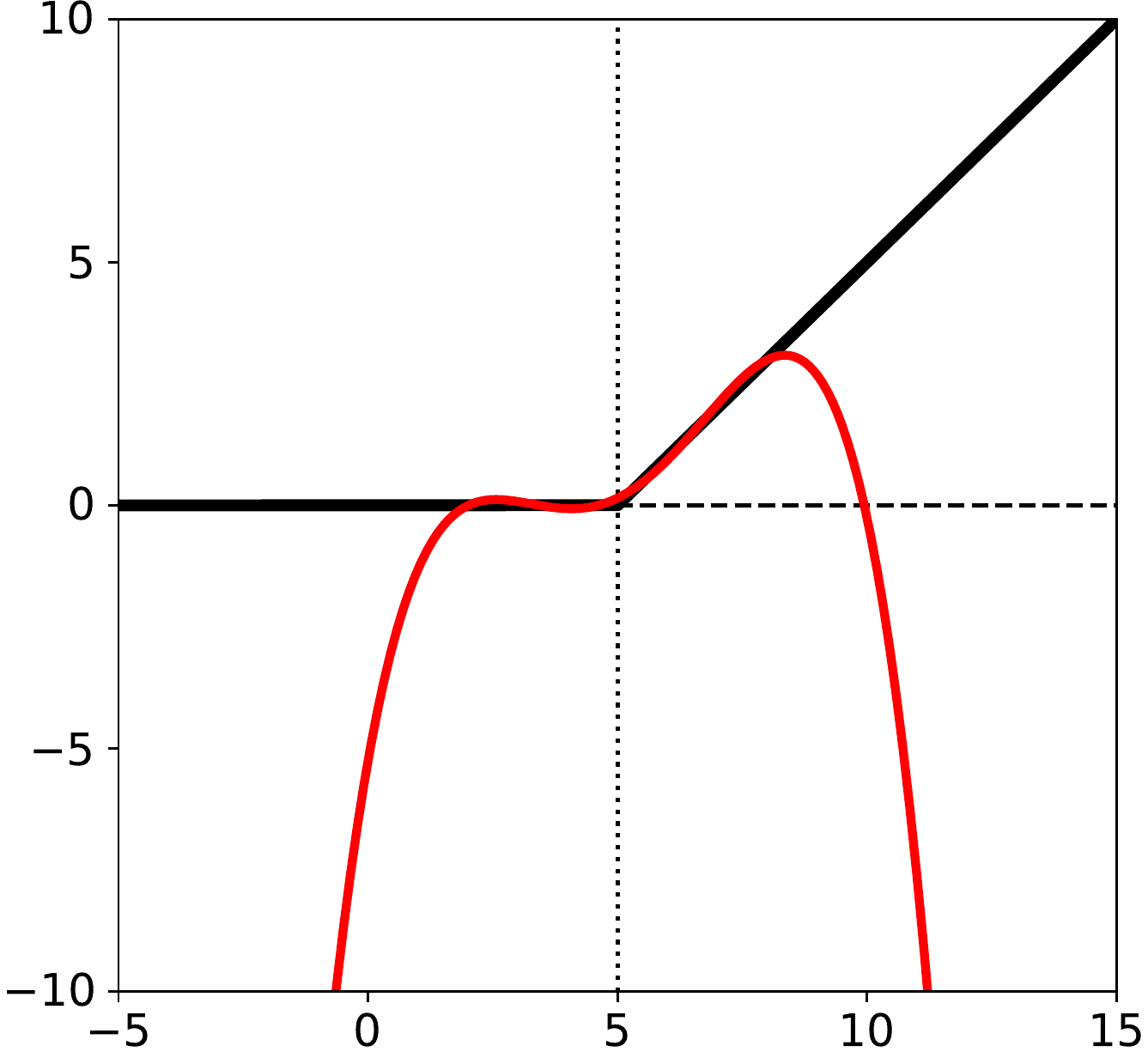} &	\includegraphics[width=0.23\textwidth]{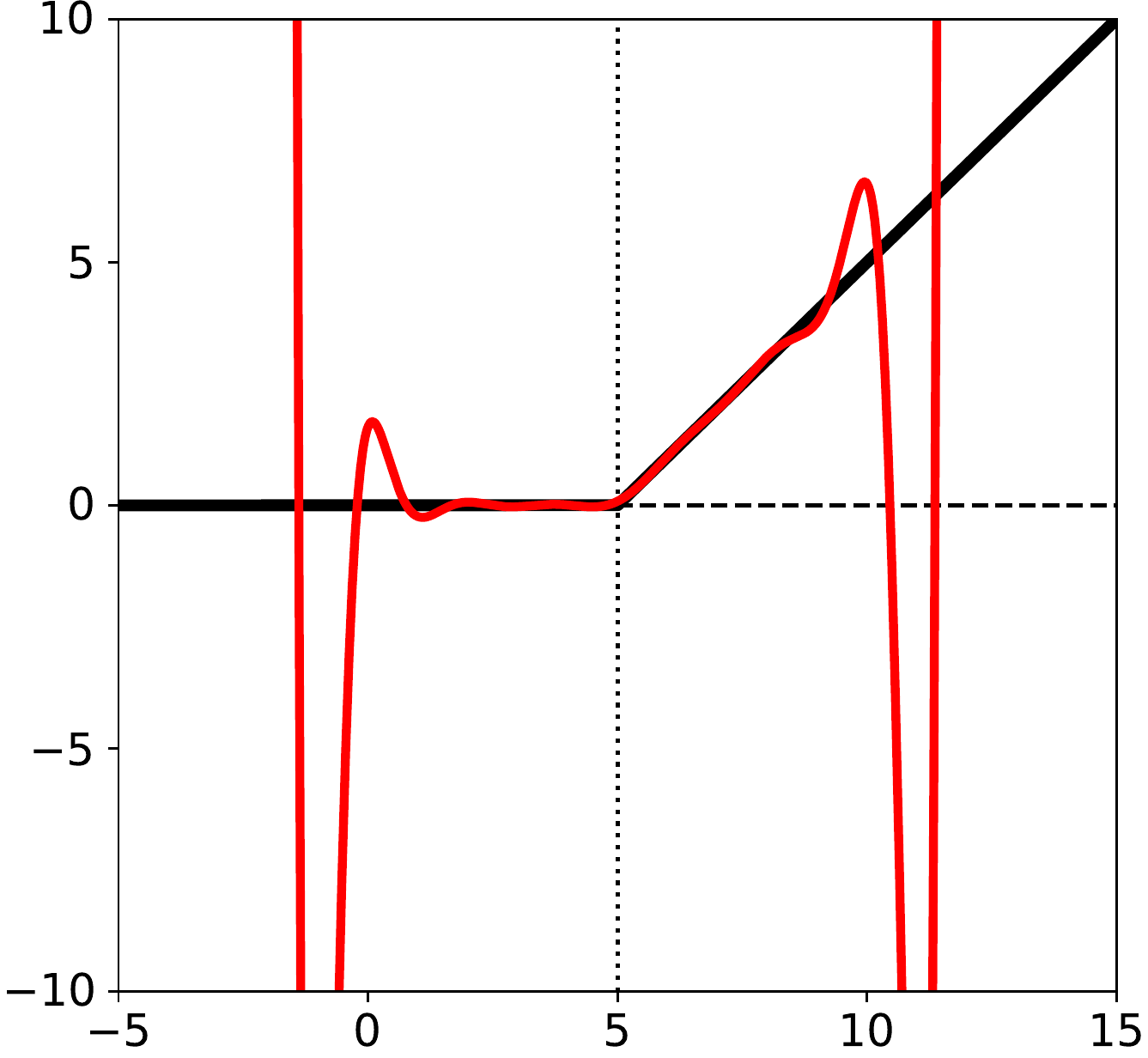} &
			\includegraphics[width=0.23\textwidth]{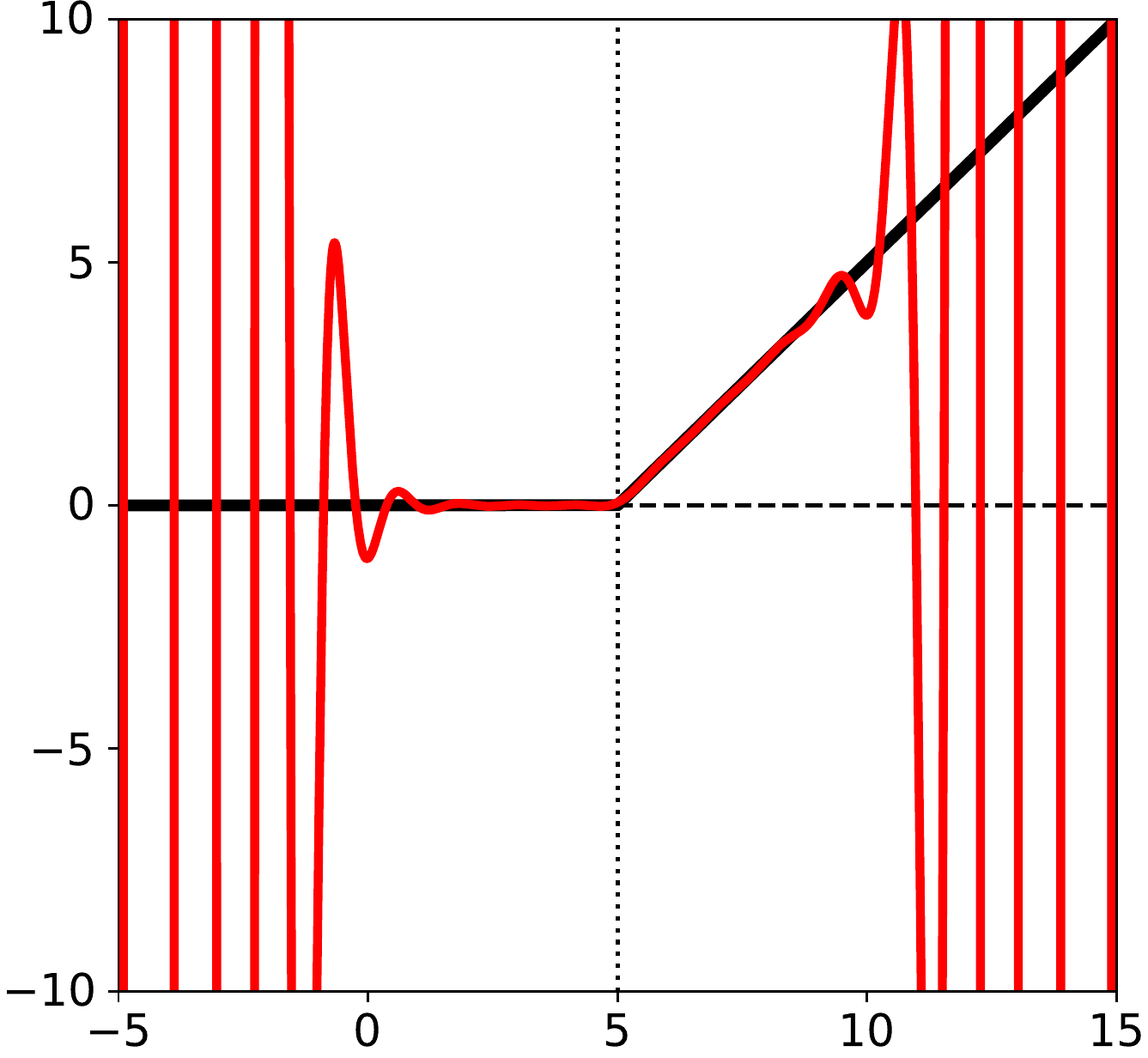}&
			\includegraphics[width=0.23\textwidth]{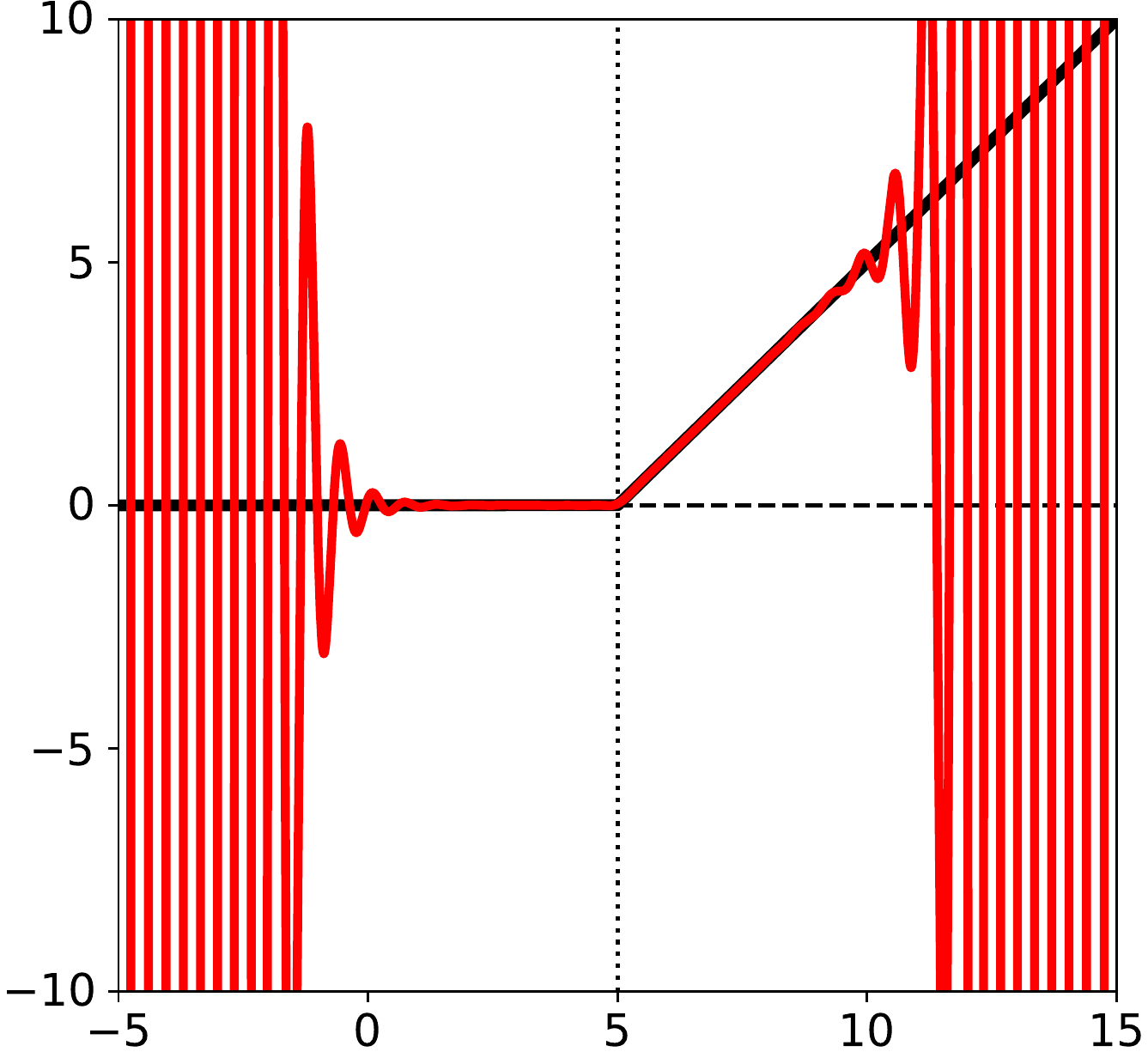}\\
			\begin{turn}{90}$\boldsymbol{b =2.0}$\end{turn}&\includegraphics[width=0.23\textwidth]{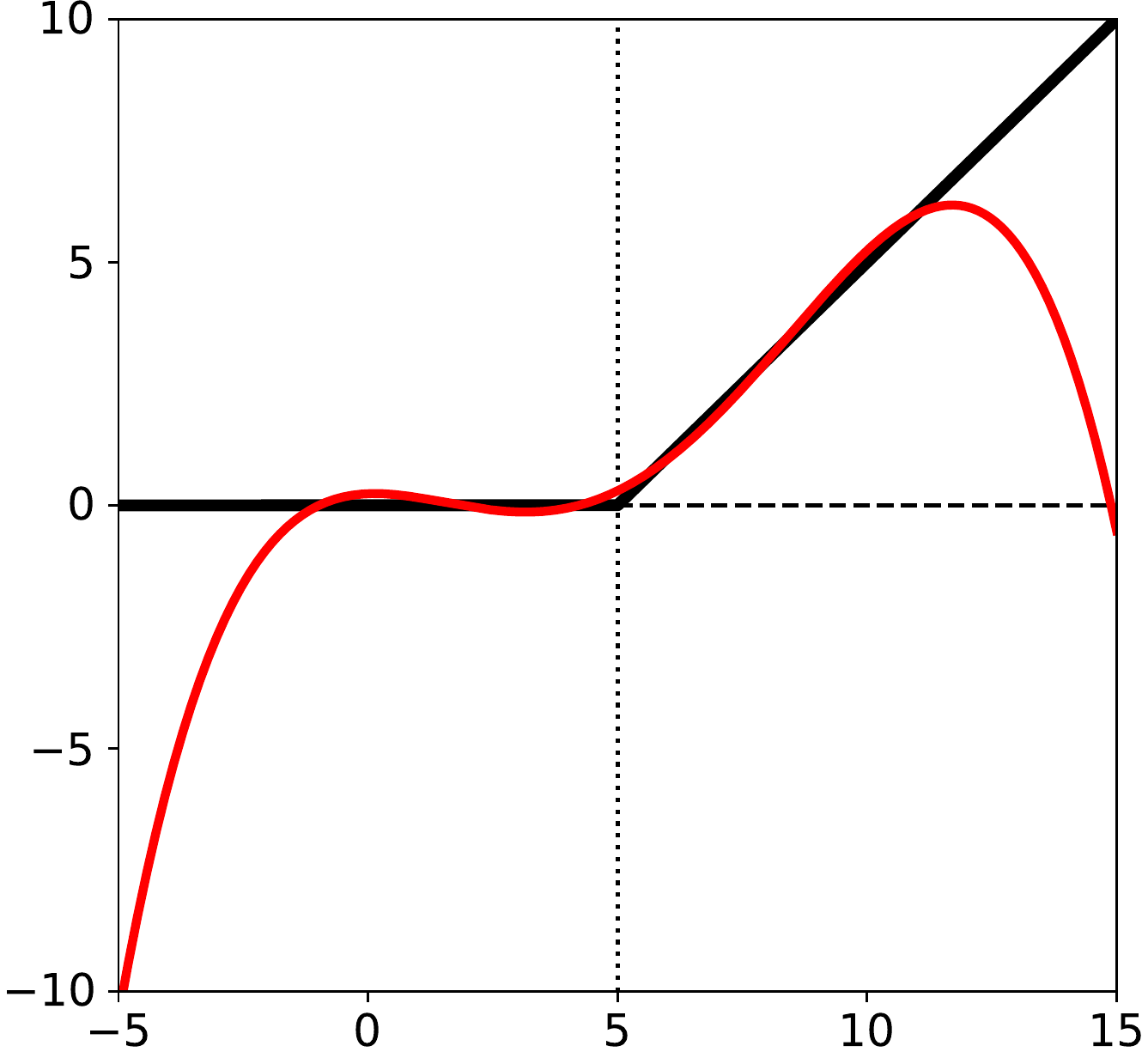} &	\includegraphics[width=0.23\textwidth]{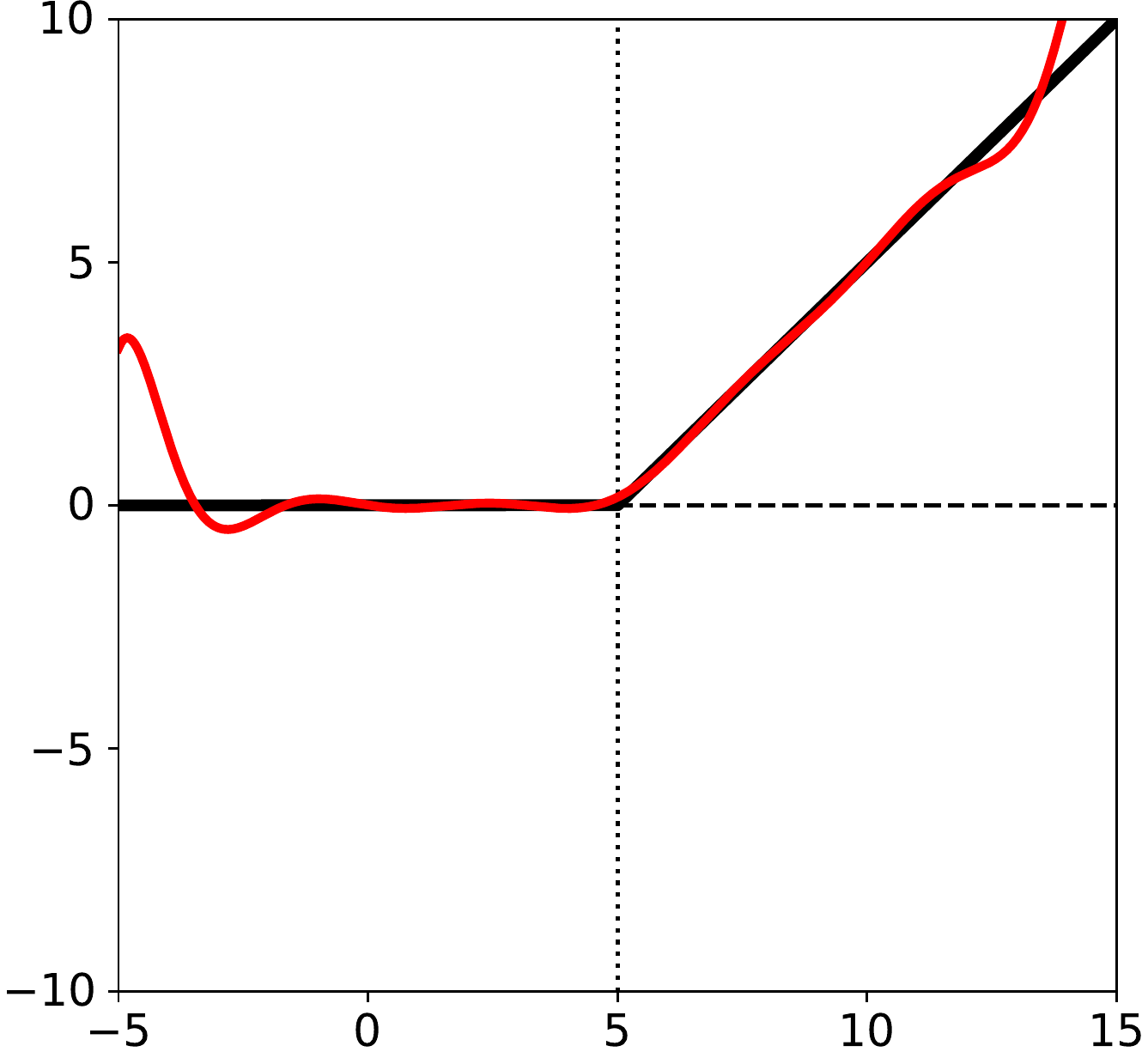} &
			\includegraphics[width=0.23\textwidth]{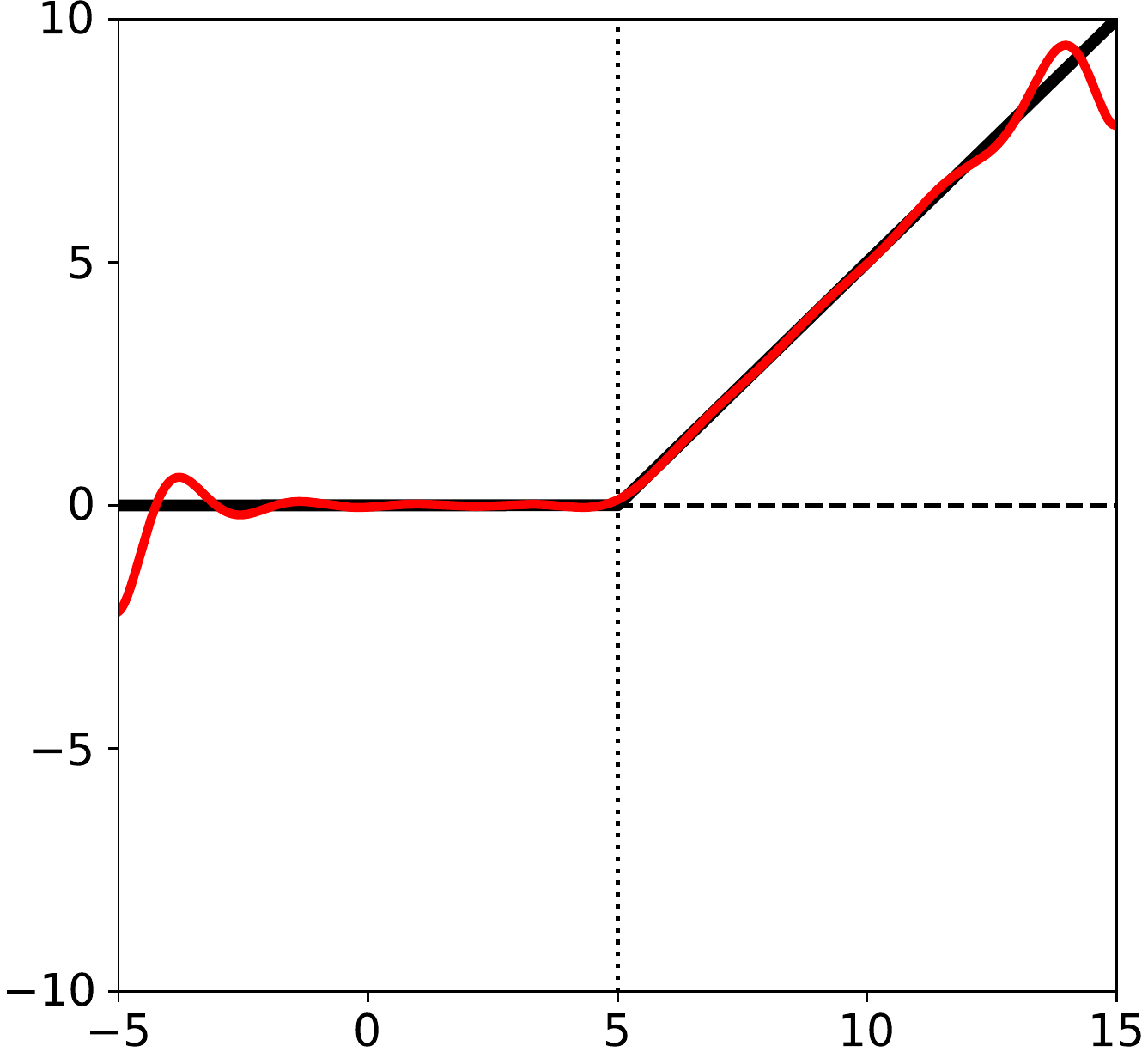}&
			\includegraphics[width=0.23\textwidth]{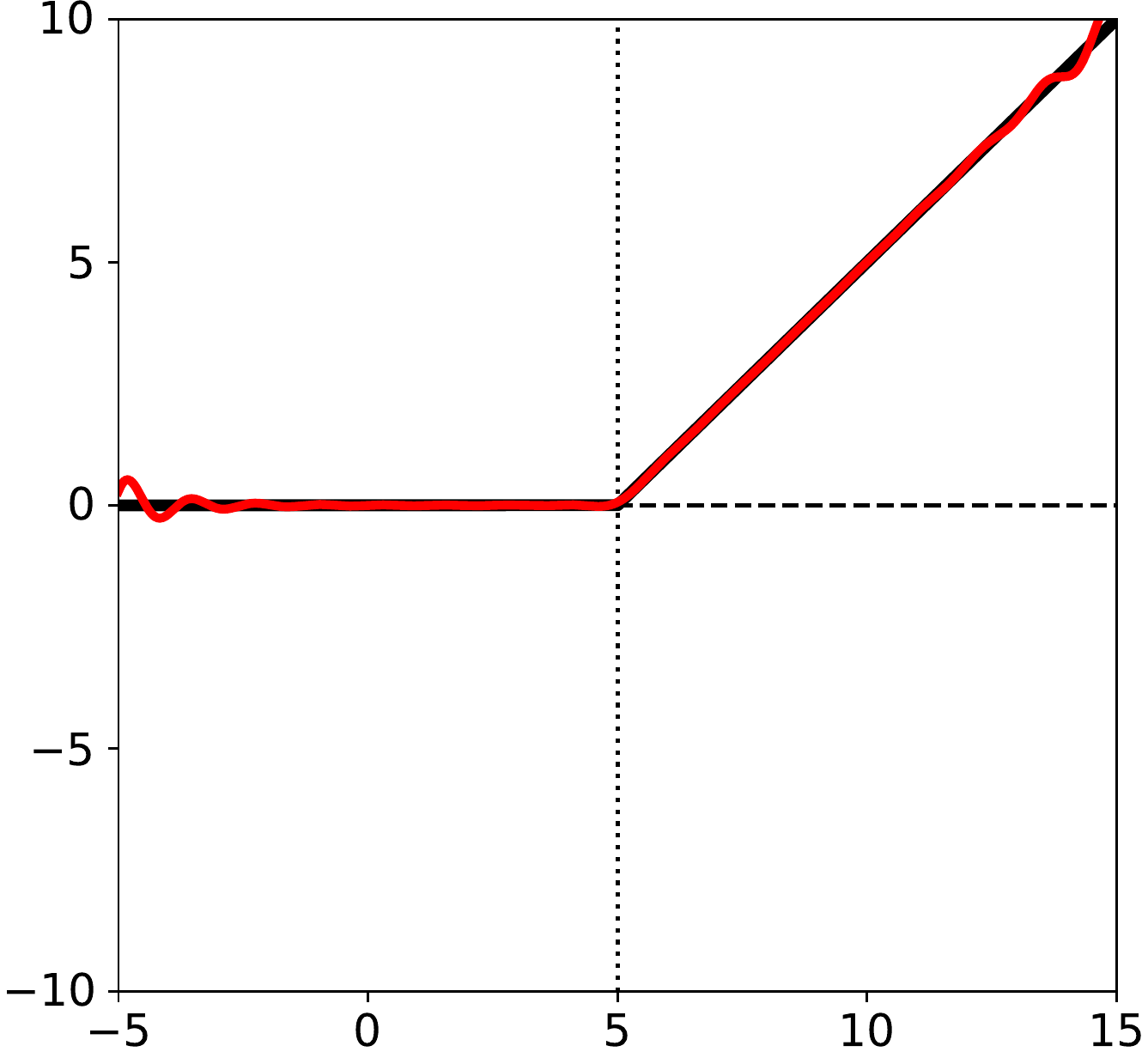}\\
			\begin{turn}{90}$\boldsymbol{b =3.0}$\end{turn}&\includegraphics[width=0.23\textwidth]{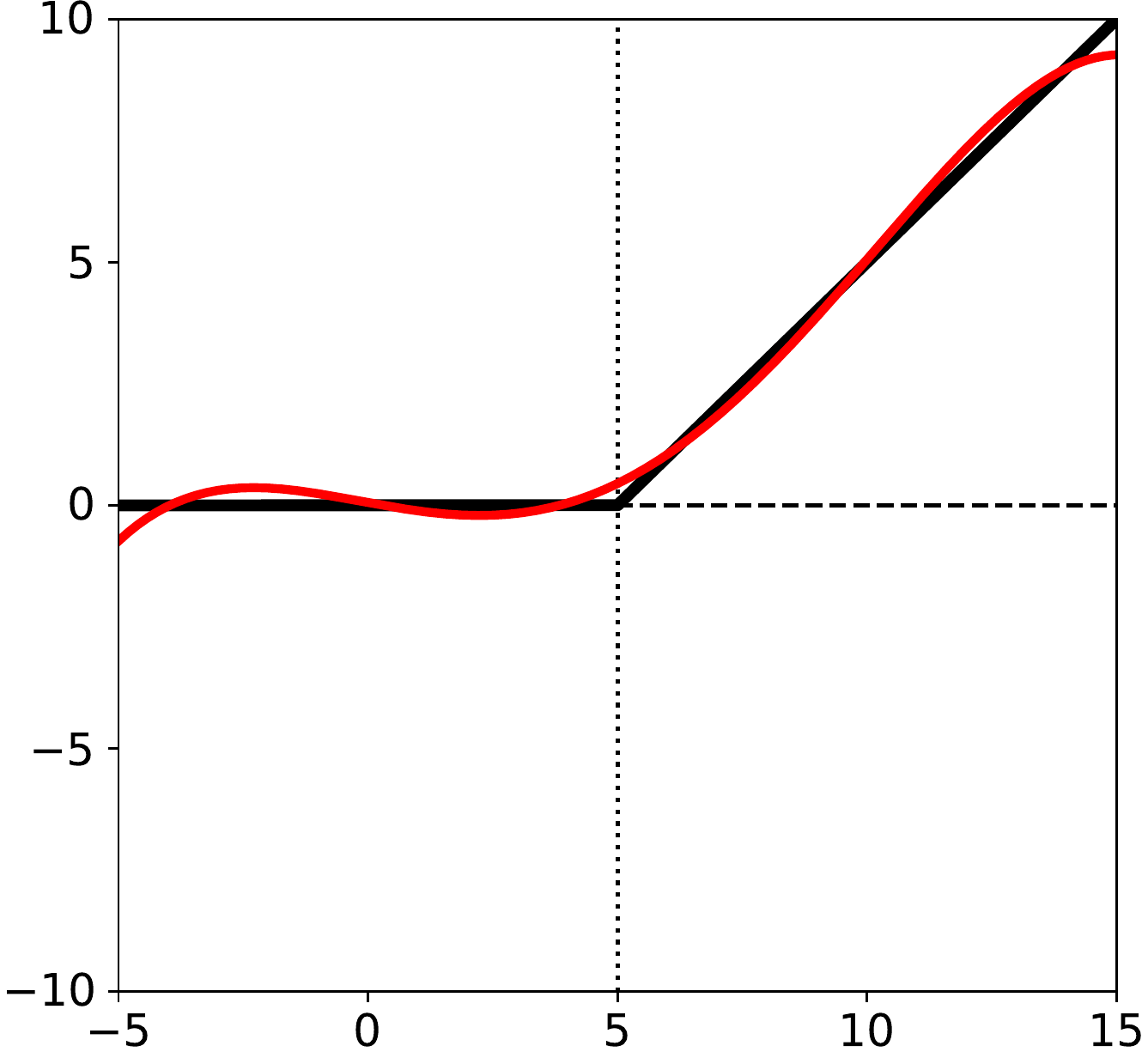} &	\includegraphics[width=0.23\textwidth]{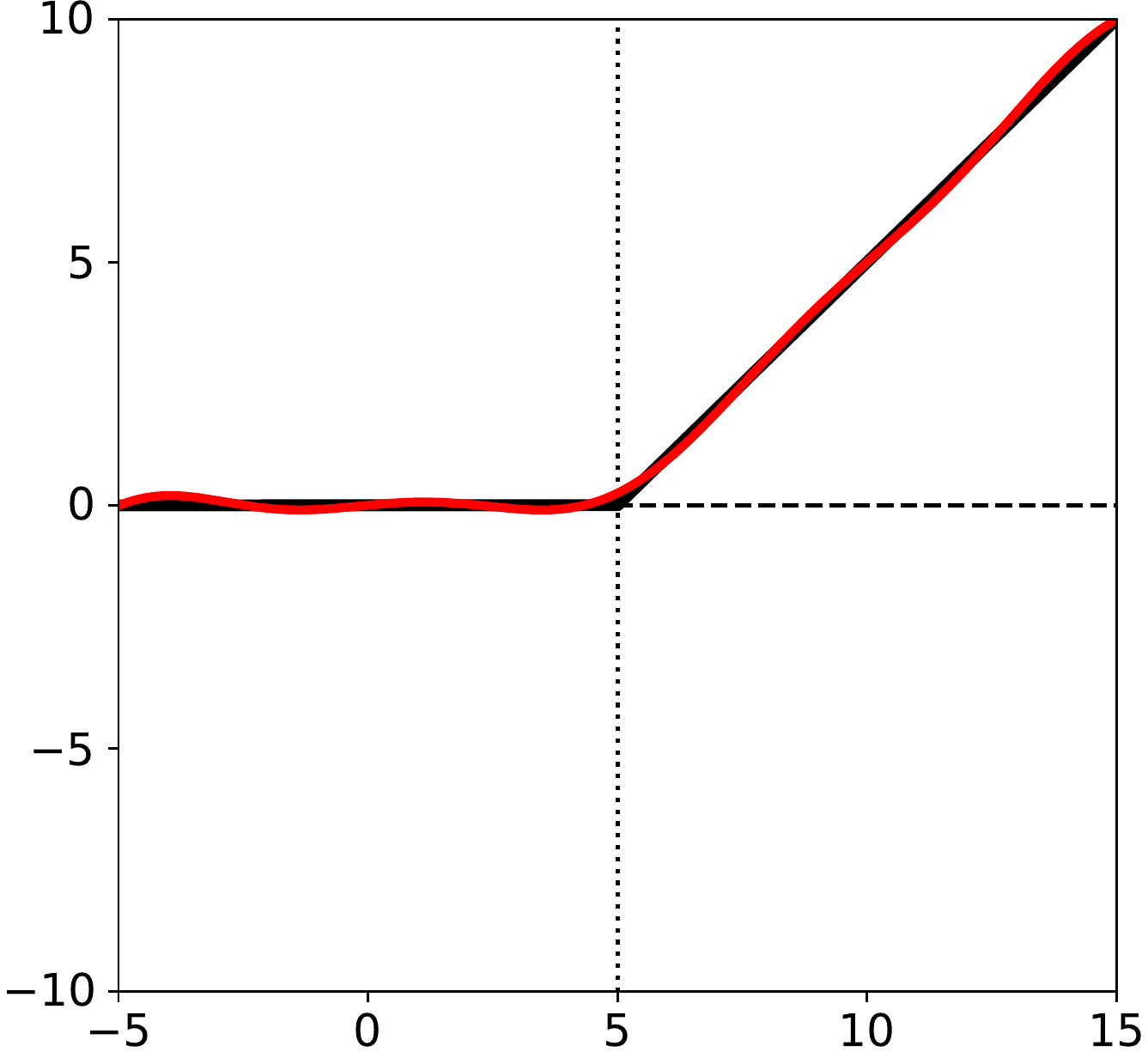} &
			\includegraphics[width=0.23\textwidth]{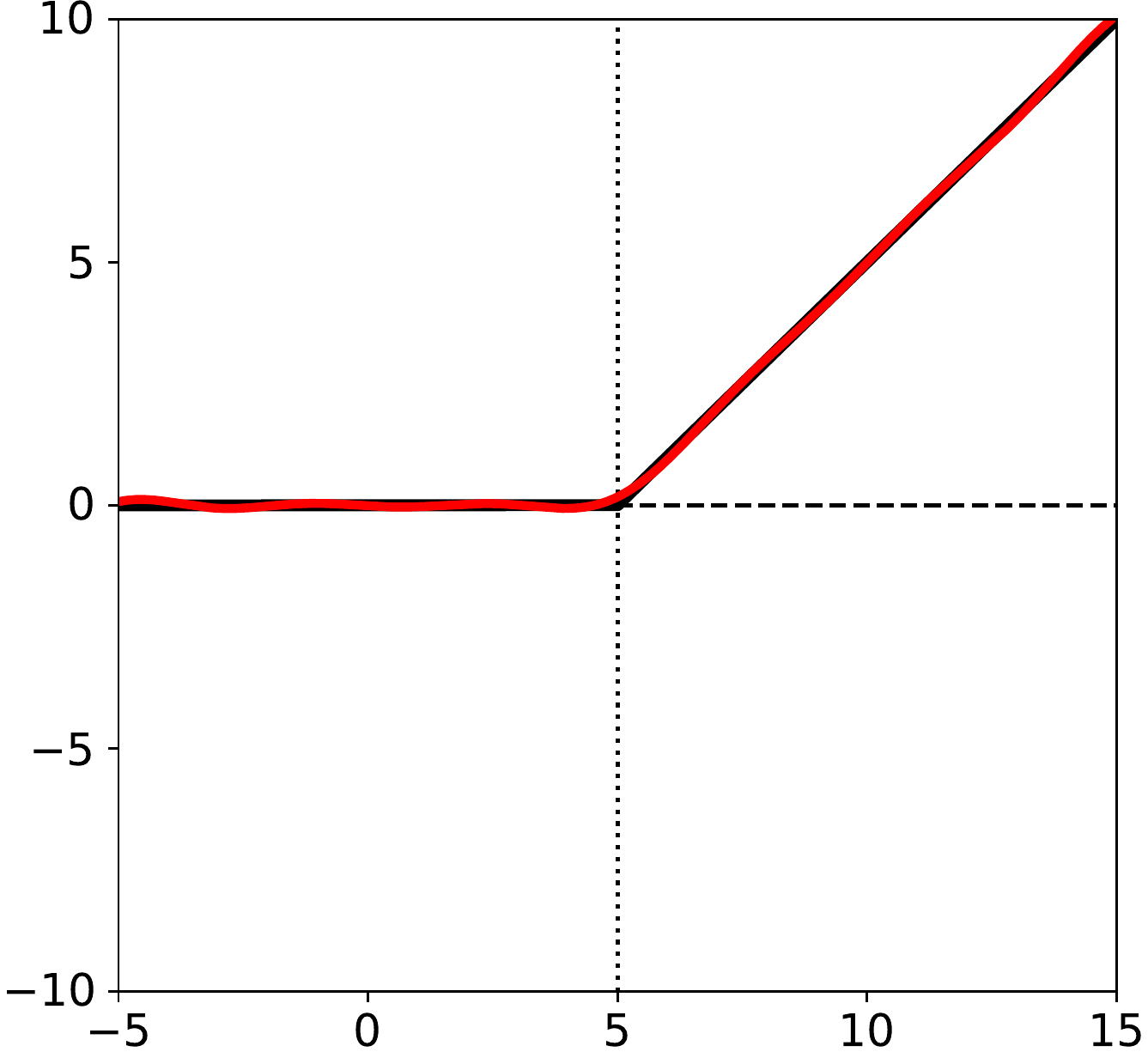}&
			\includegraphics[width=0.23\textwidth]{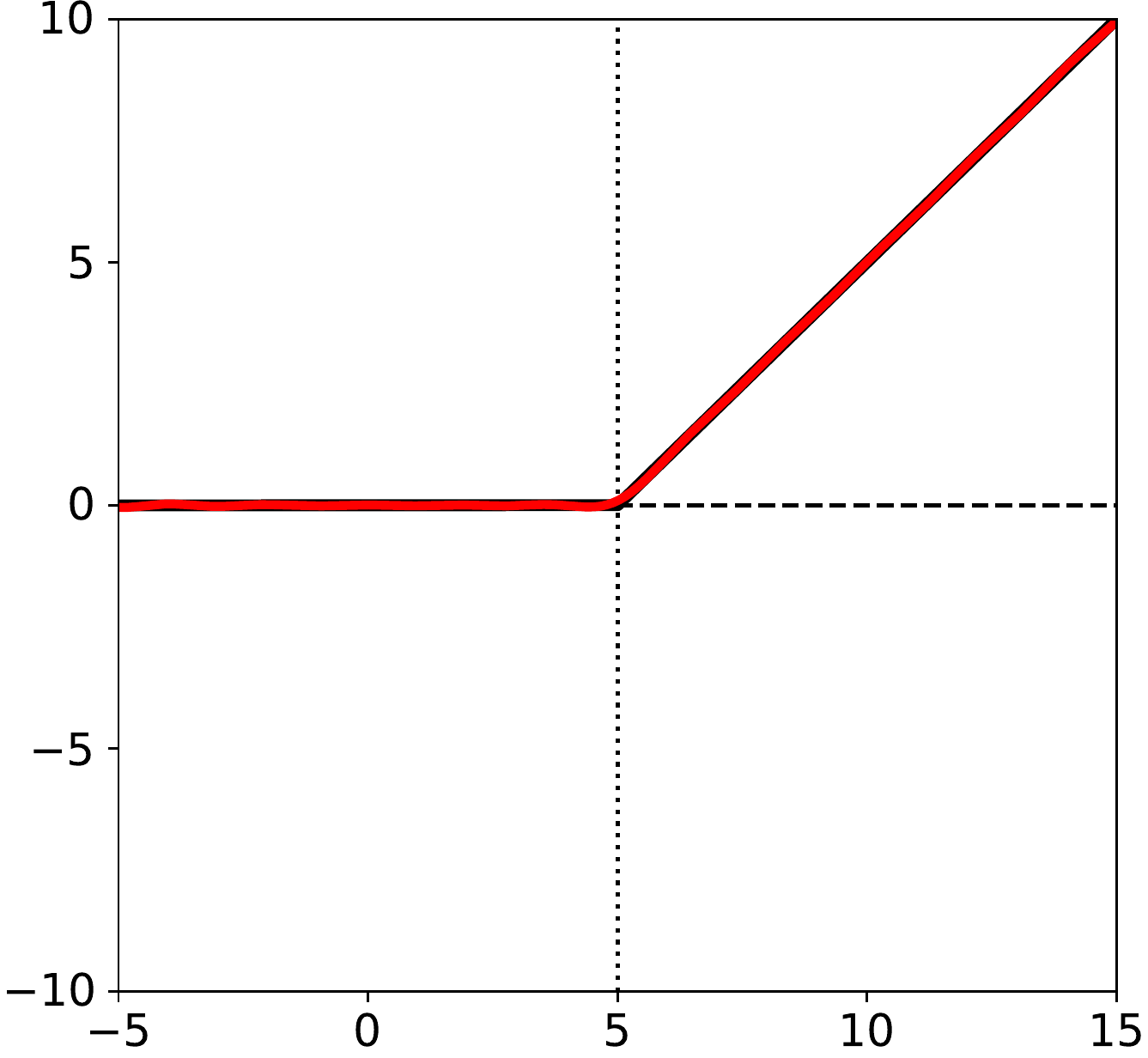}\\
		\end{tabular}
	}
	\caption{Approximation of the call payoff function $\varphi_K$ with strike price $K=5.0$ (black line) by generalized Hermite polynomials of different order $N$ and scale $b$ (red line). The drift is fixed to $a=K$.\label{hermite}}
\end{figure}

We stress the fact that the Hermite series $\varphi_{K,N}^{a,b}$ converges to $\varphi_K$ in the norm of $L^2_{a,b}$, so that we cannot expect convergence in the supremum norm, as it can also be observed in Figure \ref{hermite}. We now give a semi-explicit formula for the $L^2_{a,b}$-norm of the approximation error.
\begin{proposition}
	\label{errorphi}
	By the Parseval identity, the norm in $L^2_{a,b}$ of the approximation error is
	\begin{equation}
		\label{errorphieq}
		\left \|\varphi_K-\varphi_{K,N}^{a,b}\right \|_{L^2_{a,b}} = b\,\phi\left(\frac{K-a}{b}\right)\sqrt{\sum_{n=N+1}^{\infty}\left(\frac{1}{n!}q_{n-2}\left(\frac{K-a}{b}\right)\right)^2}.
	\end{equation}
\end{proposition}

From Proposition \ref{errorphi} we do not get an intuition on the behaviour of the approximation error as a function of $N$. We however observe that $$h^{a,b}_K:=b\,\phi\left(\frac{K-a}{b}\right) = \frac{b}{\sqrt{2\pi}}\exp\left(-\frac{(K-a)^2}{2b^2}\right)$$ does not depend on $N$ but it does depend on $b$. More precisely, if ignoring the dependence of the squared root in equation \eqref{errorphieq} on $b$ (as this is not straightforward) the approximation error is an increasing function in $b$. This means that, despite Figure \ref{hermite} shows an improving in the approximation for larger values of the scale, the $L^2_{a,b}$-norm of the approximation error might grow with the scale $b$. 

To analyse this further, in Figure \ref{approx_error} we plot the $L^2_{a,b}$-norm of the approximation error for different values of $N$, $b$ and $a$, with $K=5.0$. In the first row, the norm is a function of $N$ for three different cases, namely $a=5.0$, $a=7.0$ and $a=10.0$.  Here $b \in \{0.5, 1.0, 2.0, 3.0, 6.0, 10.0\}$, however the lines are not distinguishable. In the second row we report a zoom of the three previous plots, where we focus on $10 \le N \le 30$. Here we distinguish the six different lines and observe in particular that the approximation error is smaller for smaller values of $b$. Finally, in the last row we plot the error and the coefficient $h^{a,b}_K$ as a function of $b$. Here we fix $N=20$ and consider many values for the scale in the interval $0.5 \le b \le 10.0$. To calculate the infinite sum in equation \eqref{errorphieq}, we truncate it at $n=160$, since for bigger values the factorial cannot be converted to a floating-point number.

For the plots in the first column, since $a=K = 5.0$, the coefficient $h^{a,b}_K =\frac{b}{\sqrt{2\pi}}$ is proportional to $b$ and the squared root in equation \eqref{errorphieq} does not depend on $b$. We see indeed in the last plot of the first column that 
the approximation error is proportional to $h^{a,b}_K$ (proportional to $b$ in fact). In the second and third plots of the last row, since  $a \ne K$, the behaviour of the approximation error diverges from the one of $h^{a,b}_K$. In particular, it is not monotone in $b$. However, since the main interest for the polynomial approximation is around $K$, we shall mostly deal with configurations that we can approximatively consider increasing functions of the scale parameter. 

\begin{figure}[tp]
	\setlength{\tabcolsep}{2pt}
	\resizebox{1\textwidth}{!}{
		\begin{tabular}{@{}>{\centering\arraybackslash}m{0.33\textwidth}@{}>{\centering\arraybackslash}m{0.33\textwidth}@{}>{\centering\arraybackslash}m{0.33\textwidth}@{}}
			$\boldsymbol{a = 5.0}$&$\boldsymbol{a = 7.0}$&$\boldsymbol{a = 10.0}$\\
			\includegraphics[width=0.32\textwidth]{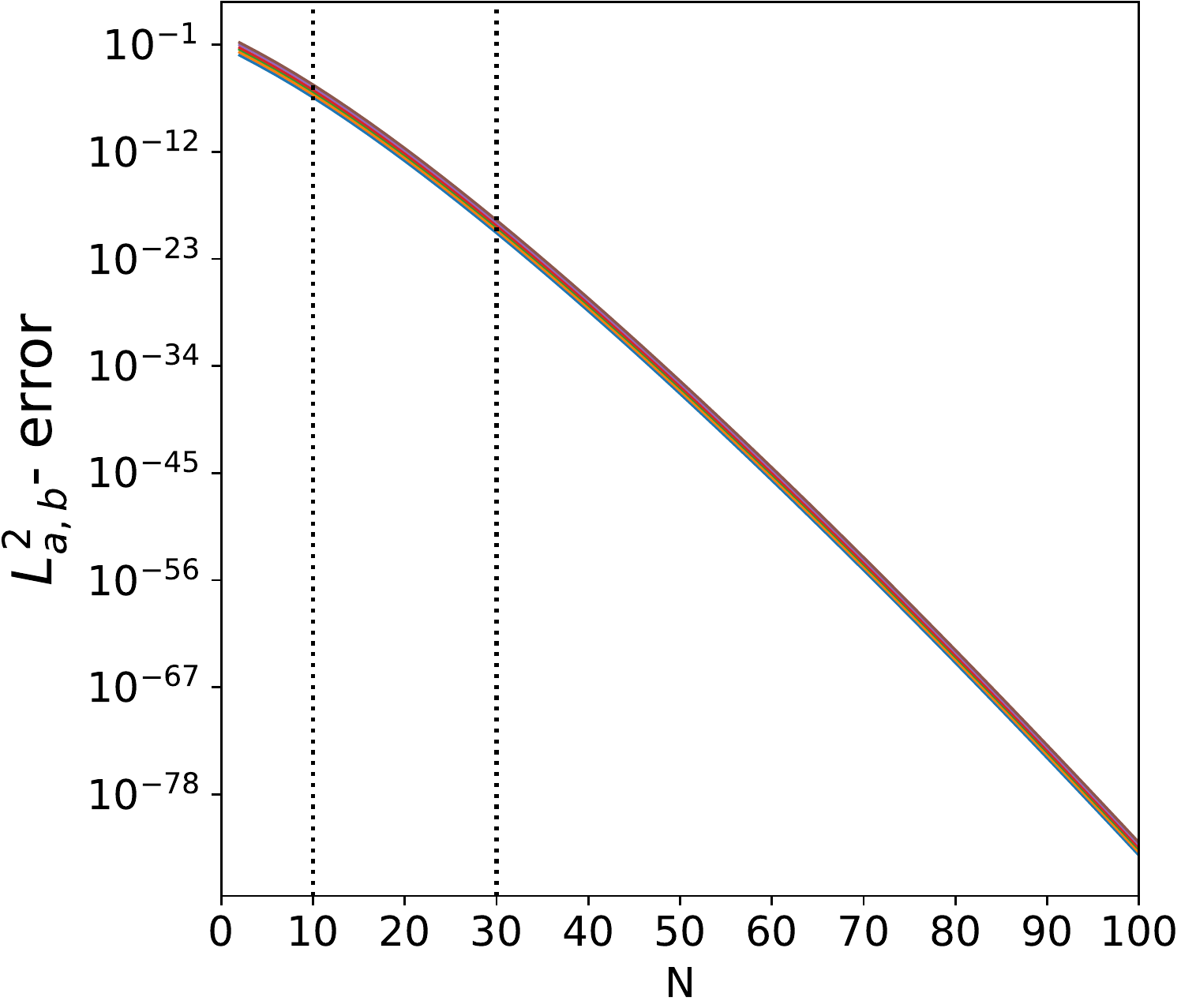}&
			\includegraphics[width=0.32\textwidth]{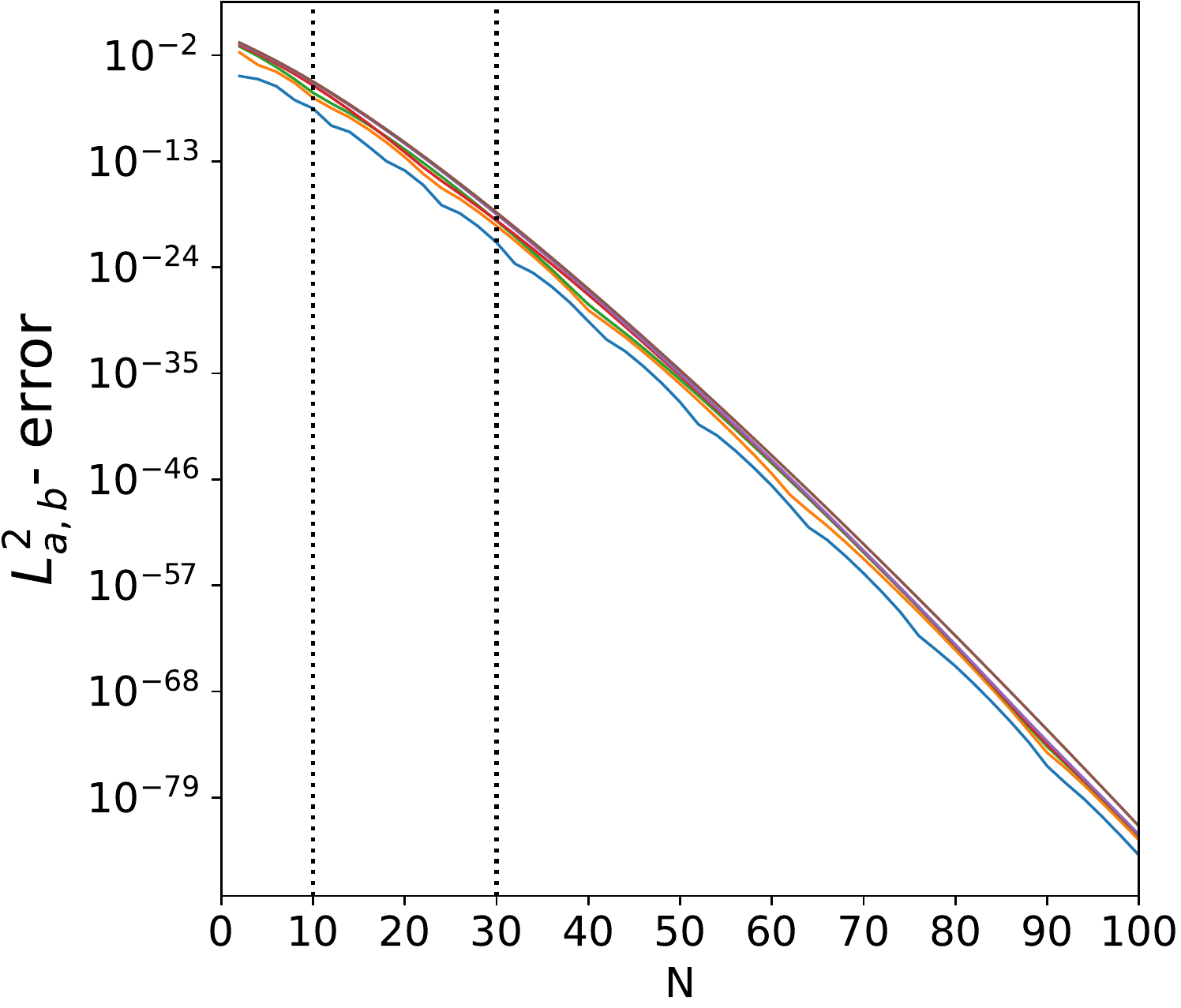}&
			\includegraphics[width=0.32\textwidth]{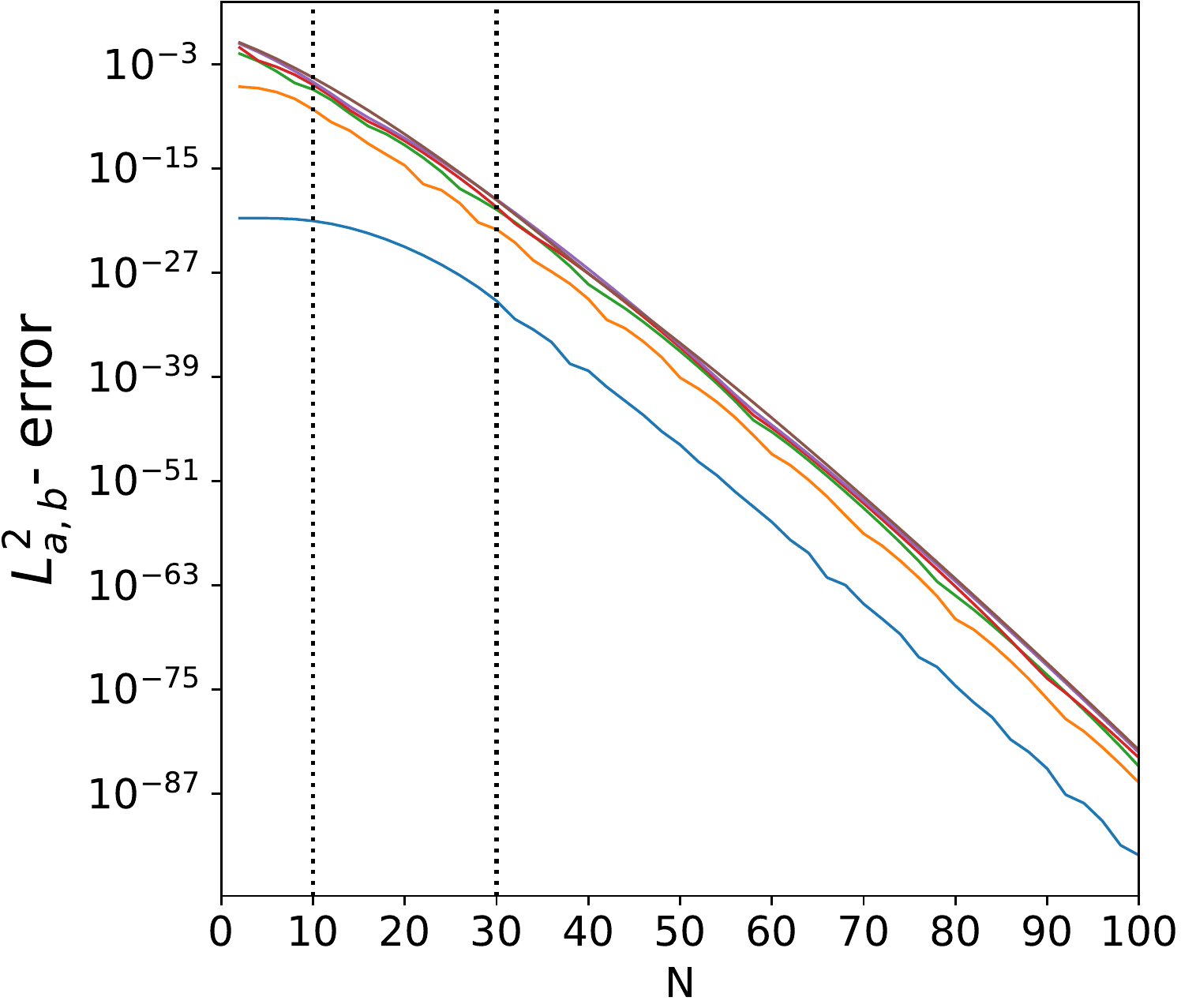}\\
			\includegraphics[width=0.32\textwidth]{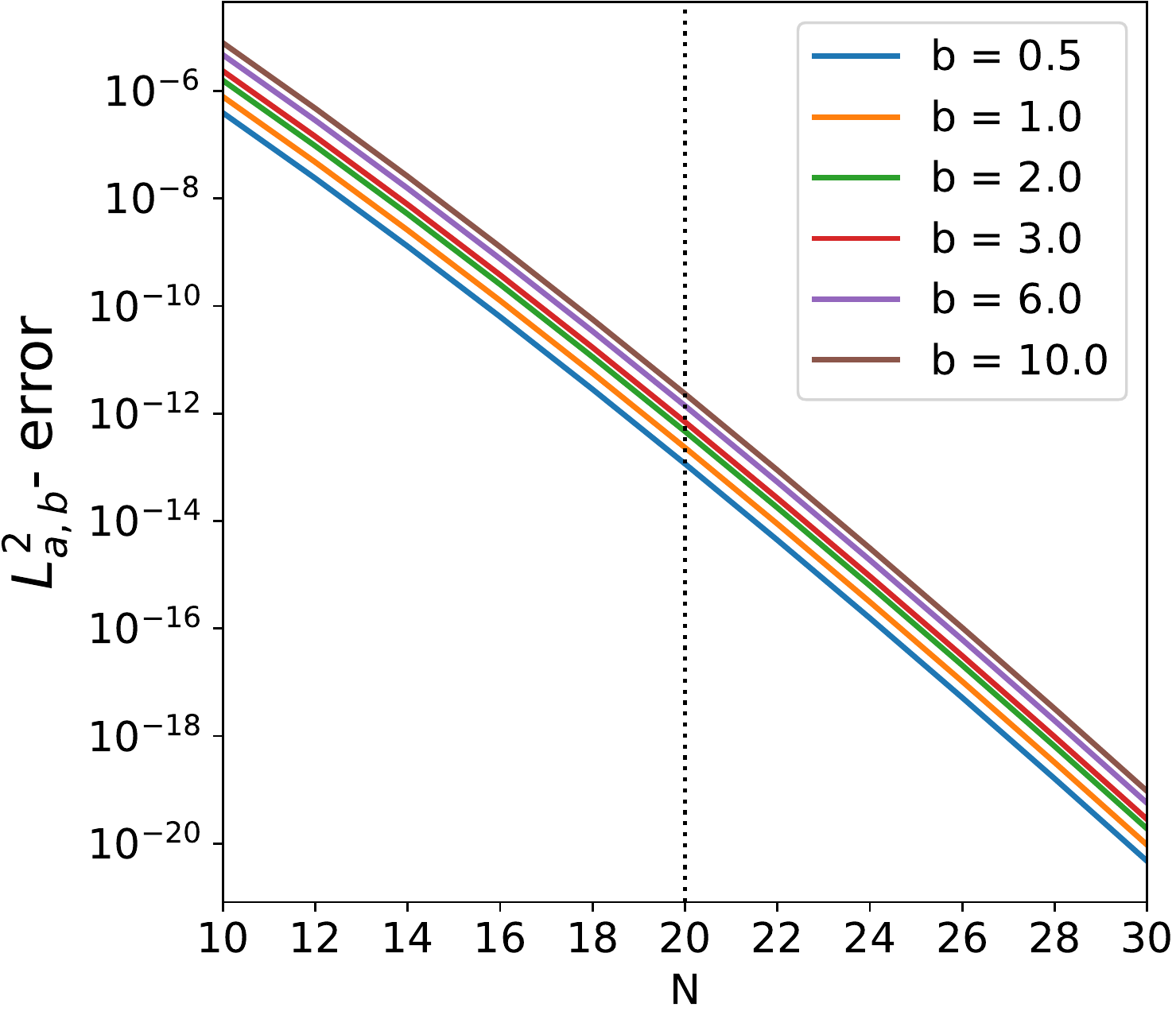}&
			\includegraphics[width=0.32\textwidth]{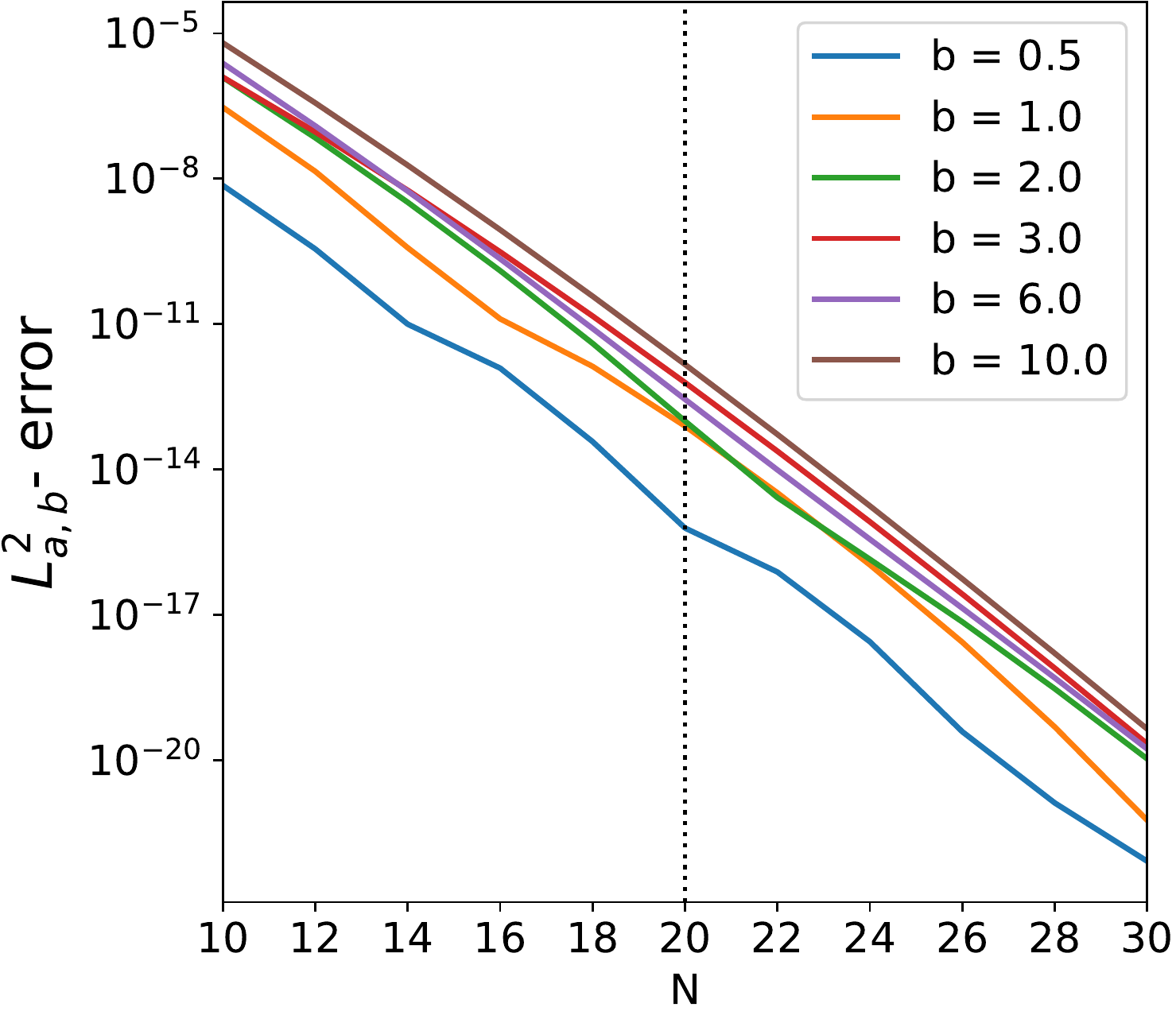}&
			\includegraphics[width=0.32\textwidth]{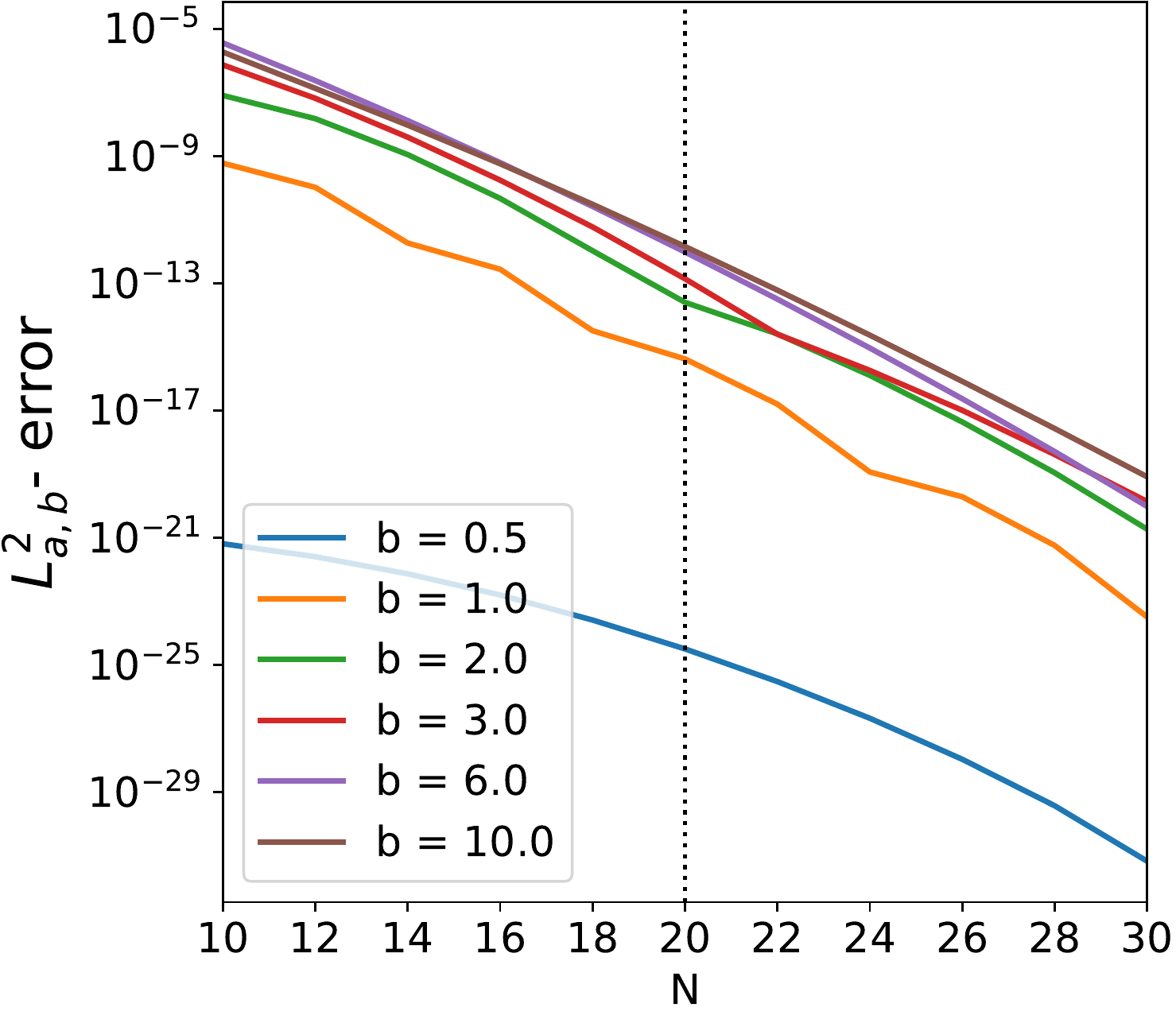}\\
			\includegraphics[width=0.32\textwidth]{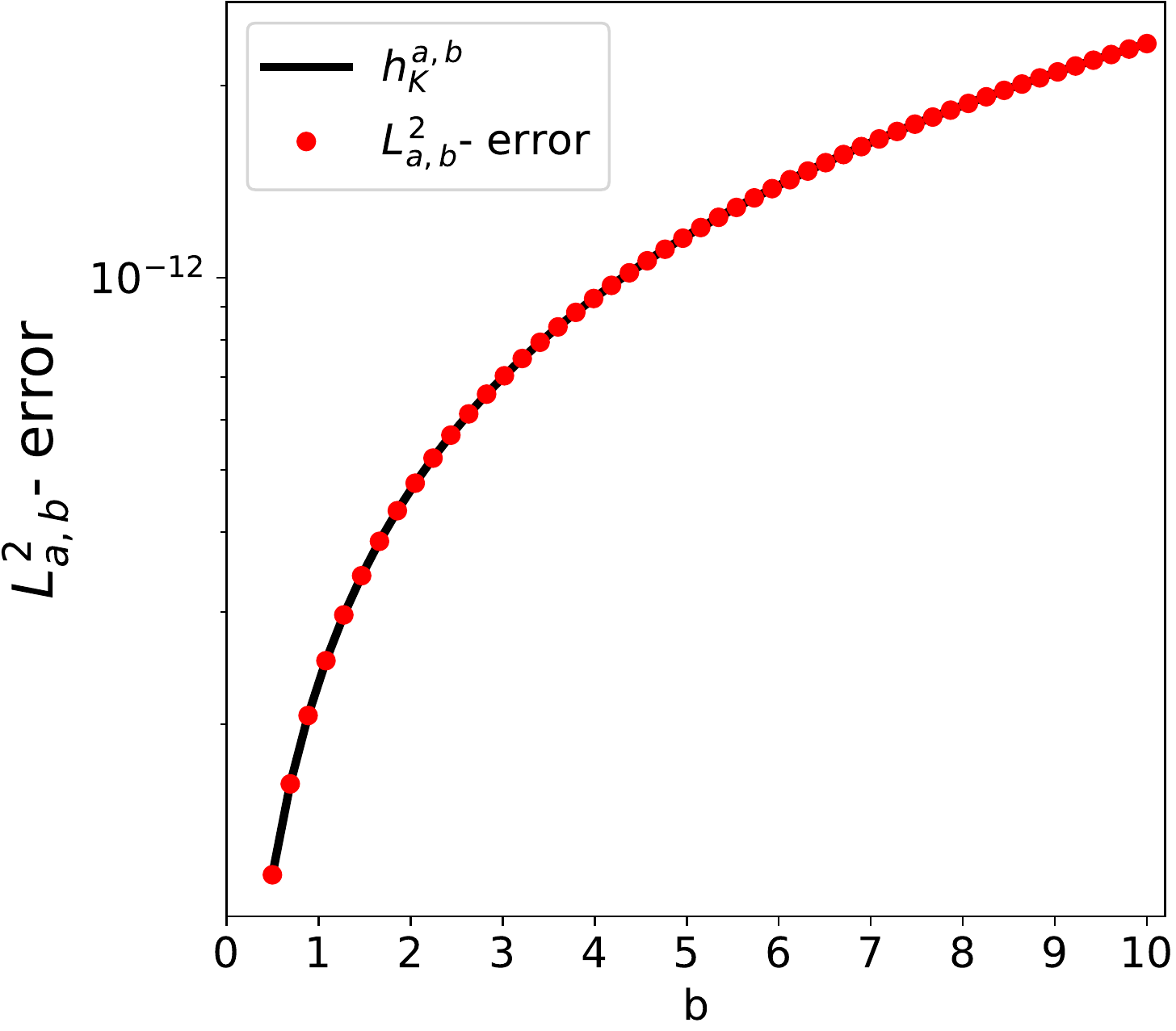}&
			\includegraphics[width=0.32\textwidth]{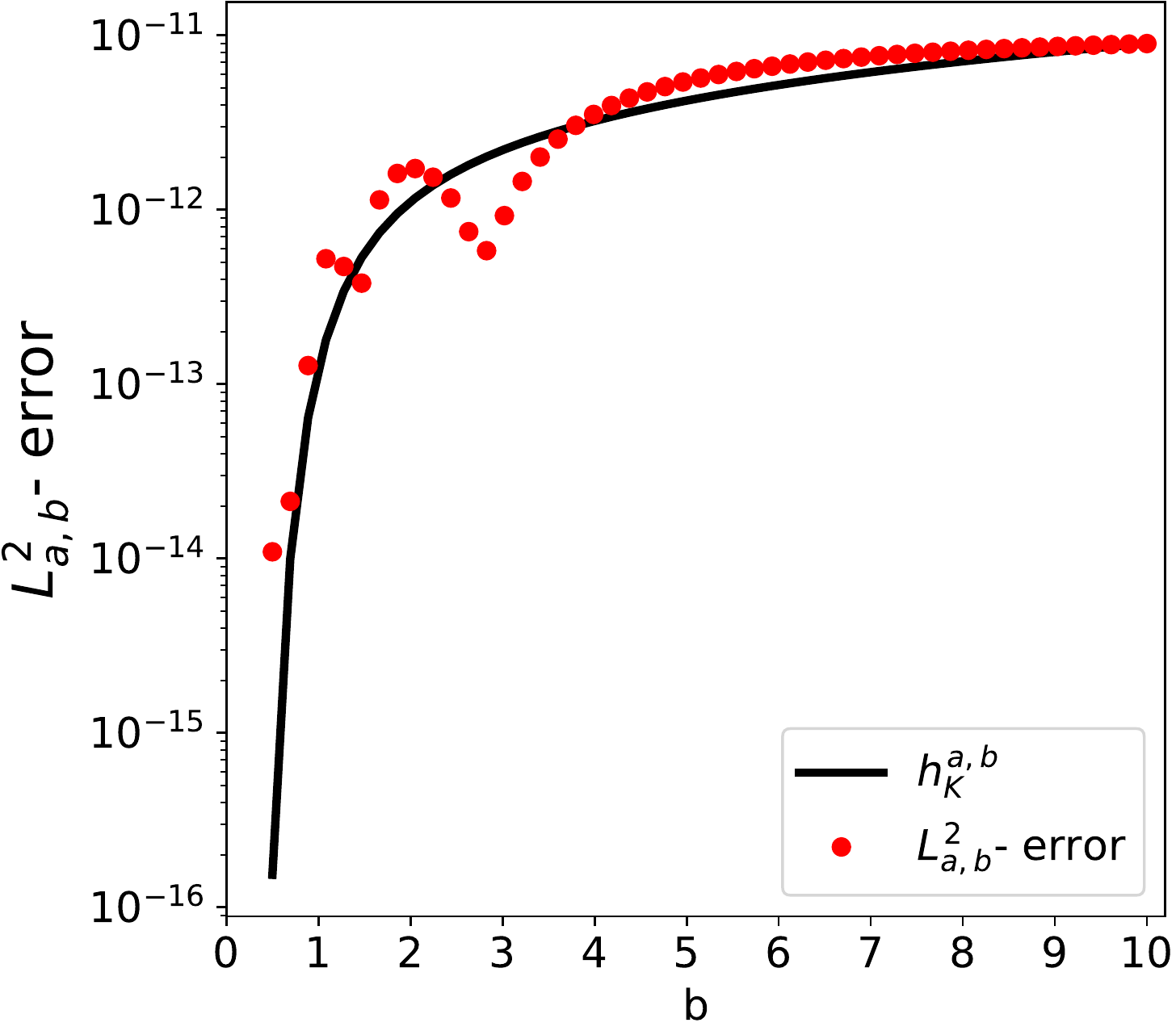}&
			\includegraphics[width=0.32\textwidth]{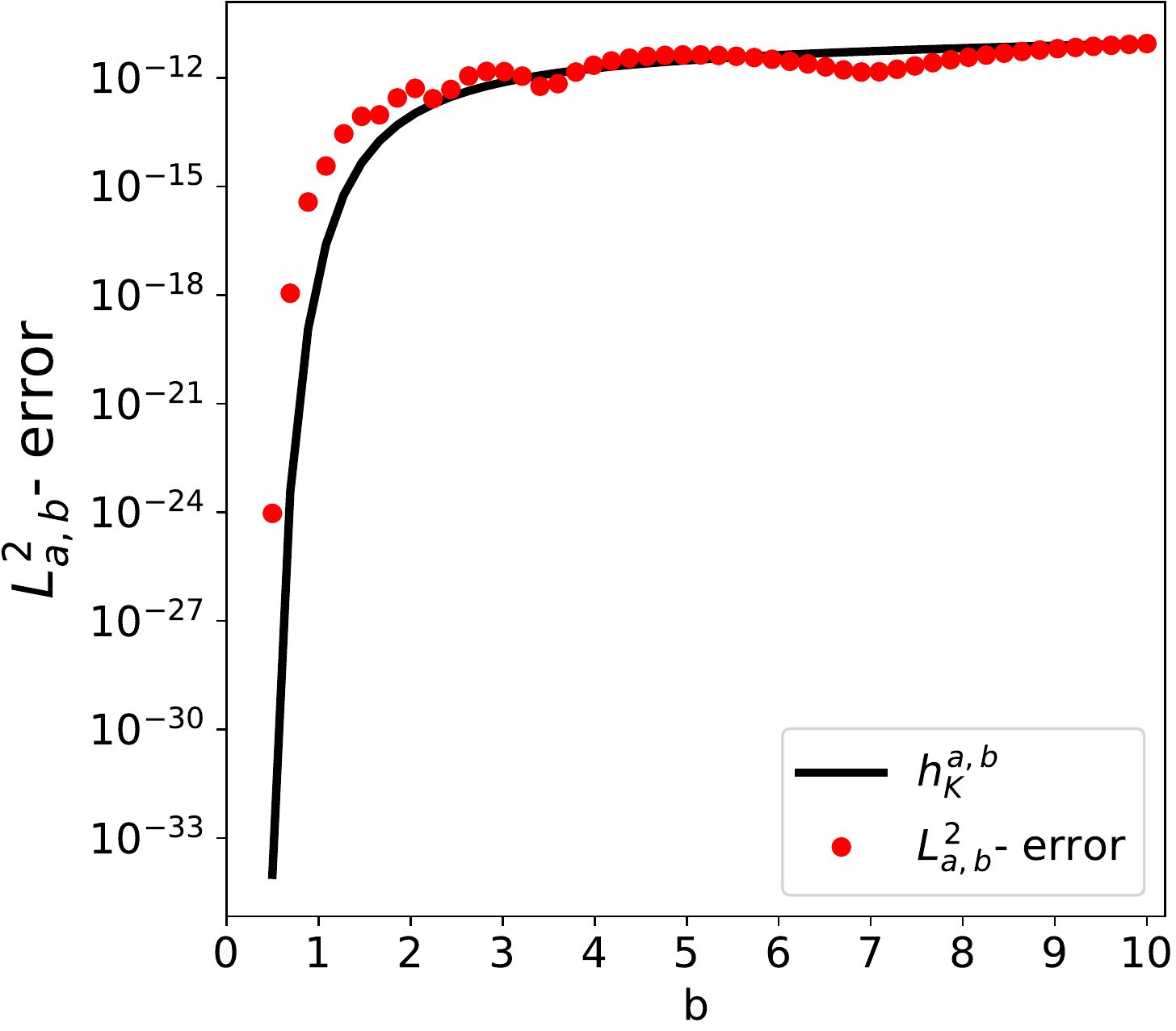}\\
		\end{tabular}
	}
	\caption{$L^2_{a,b}$-norm of the approximation error of $\varphi_{K,N}^{a,b}$ with $K=5.0$. In the first and second row, the error is a function of the truncation number $N$. In the third row, the error is a function of $b$. \label{approx_error}}
\end{figure}

\section{Pricing options with correlators}
\label{pricingsection}
We focus in this section on the pricing of call options. 
Given an $\F_t$-adapted stochastic process $X$ and the payoff function $\varphi_K$ in equation \eqref{payoff}, we want to compute the conditional expectation $\Pi_K(t) = \E\left[ \varphi_K(X(T))\left.\right|\F_t\right]$ with respect to the risk-neutral measure $\Q$. 
Starting from the Hermite series constructed in Section \ref{hermitesection}, we have a family of approximations depending on $a$, $b$ and $N$, namely
\begin{equation*}
	\Pi^{a,b}_{K,N}(t):= \E\left[\left. \varphi_{K,N}^{a,b}(X(T))\right|\F_t\right] \approx  
	\E\left[\left. \varphi_K(X(T))\right|\F_t\right]= \Pi_K(t).
\end{equation*}
These give an approximation for the price of a European-style call option with underlying process $X$\footnote{Notice that at this point the process $X$ is a generic $\F_t$-adapted stochastic process, and not the discrete average defined in equation \eqref{X}. Thus, in practice, Theorem \ref{price} gives an approximation for the price of a European-style call option with underlying process $X$. In Section \ref{asiansection}, this result will then be extended to Asian options (hence for $X$ as in equation \eqref{X}) by the multinomial theorem.}.
\begin{theorem}
	\label{price}
	The approximation by Hermite polynomials of the price of a call option is given by
	\begin{equation*}
		\Pi^{a,b}_{K,N}(t) = 
		 \sum_{k=0}^N\hat{\boldsymbol{\beta}}_{N,k+1}^{a,b}\frac{1}{b^k}\sum_{i=0}^k \binom{k}{i} (-a)^{k-i}\, \E\left[\left.X(T)^{i}\right|\F_t\right],
	\end{equation*}
	where $\hat{\boldsymbol{\beta}}_N^{a,b} := \boldsymbol{\beta}^{a,b \,\top}_N M_N$ with components $ \hat{\boldsymbol{\beta}}_{N,k+1}^{a,b}= \boldsymbol{\beta}^{a, b\,\top}_N\left(M_N\right)_{:,(k+1)}$ for $k=0,\dots, N$.
\end{theorem}

We notice in Theorem \ref{price} that the vector $\boldsymbol{\beta}^{a,b}_N$ and the matrix $M_N$ depend on the choice of the orthogonal basis, the GHPs in our case, while the moments $\E\left[\left.X(T)^{i}\right|\F_t\right]$ only depend on 
the distribution of the random variable $X(T)$. Then, once the orthogonal basis is chosen, the approximation of the expected payoff $\Pi_{K,N}^{a,b}(t)$ is fully determined by the conditional moments of $X(T)$. 

\subsection{Asian options}
\label{asiansection}
We now consider Asian-style options as introduced at the beginning of Section \ref{intro}. For $m\ge 0$, $X$ is the discrete average of an $\F_t$-adapted stochastic process $Y$ over the period $(t,T]$, namely
\begin{equation}
	\label{X2}
	X(T) = \frac{1}{m+1}\sum_{j=0}^m Y(s_j) \qquad \mbox{ for } \qquad t < s_0<s_1 <\dots< s_m = T.
\end{equation}
From Theorem \ref{price}, 
we need the conditional moments $\E\left[\left.X(T)^{i}\right|\F_t\right]$, for $i=1, \dots, N$, $N$ being the truncation number for the Hermite series \eqref{phiNK}. These can be rewritten in terms of correlator-type expectations (as defined in equation \eqref{corr}) by means of the multinomial theorem.
\begin{proposition}
	\label{nth_power2}
	For every $1\le i \le N$, the conditional moments of the process $X$ in equation \eqref{X2} can be rewritten as a linear combination of correlator terms for the process $Y$, namely
	\begin{equation*}
		\E\left[\left.X(T)^{i}\right|\F_t\right]
		\!=\! \frac{1}{(m+1)^{i}} \sum_{|\boldsymbol{k}|=i} \frac{i!}{k_0!k_1!\cdots k_{m}!}  \E\left[\left.Y(s_0)^{k_0}Y(s_1)^{k_1}\cdots Y(s_{m})^{k_{m}}\right|\F_t\right]
	\end{equation*}	
	where the summation is over the multi-indexes $\boldsymbol{k}=(k_0, \dots, k_m)$ with $|\boldsymbol{k}| = k_0+k_1+\dots+k_{m}$.
\end{proposition}
This allows to state the pricing formula for Asian options.
\begin{theorem}
	\label{price2}
	For every $N\ge 0$, the price of a discretely sampled arithmetic Asian option can be approximated with generalized Hermite polynomials by 
	\begin{equation*}
		\Pi_{K,N}^{a,b}(t)=\sum_{k=0}^N\sum_{i=0}^k\sum_{|\boldsymbol{k}|=i}  \binom{k}{i} \frac{\hat{\boldsymbol{\beta}}_{N,k+1}^{a,b}(-a)^{k-i}}{(m+1)^{i}b^k}\,  \frac{i!}{k_0!k_1!\cdots k_{m}!}  \E\left[\left.Y(s_0)^{k_0}Y(s_1)^{k_1}\cdots Y(s_{m})^{k_{m}}\right|\F_t\right],
	\end{equation*}
	where $\hat{\boldsymbol{\beta}}_N^{a,b} := \boldsymbol{\beta}^{a,b \,\top}_N M_N$ with components $ \hat{\boldsymbol{\beta}}_{N,k+1}^{a,b}= \boldsymbol{\beta}^{a, b\,\top}_N\left(M_N\right)_{:,(k+1)}$, $k=0,\dots, N$.
\end{theorem}

From Theorem \ref{price} and \ref{price2}, we have explicit approximation formulas for the price of European and Asian options with call-type payoff function. However, these formulas require the computation of the conditional moments of the underlying stochastic process in the case of Theorem \ref{price}, and the computation of conditional correlators in the case of Theorem \ref{price2}. For a jump-diffusion polynomial process, both conditional moments and correlators admit a closed formulation.

\subsection{Error analysis and scale criterion}
Let $\psi_{X(T)}$ be the density function of the process $X$ at time $T$. We then estimate the error in approximating $\Pi_K(t)$ with generalized Hermite polynomials.
\begin{theorem}
	\label{error}
	If the  density function $\psi_{X(T)}$ satisfies the condition
	\begin{equation}
		\label{psicond}
		\int_{\R}    \psi_{X(T)}^2(x)\omega_{a,b}^{-1}(x)dx < \infty,
	\end{equation}
	then the absolute error in approximating the option price $\Pi_K(t)$ with $\Pi_{K,N}^{a,b}(t)$ is bounded by the $L^2_{a,b}$-norm of the error in approximating the payoff function $\varphi_K$ with $\varphi_{K,N}^{a,b}$, namely
	\begin{equation*}
		\left| \Pi_K(t)-\Pi_{K,N}^{a,b}(t)\right|\le C_{a,b} \left \|\varphi_K-\varphi_{K,N}^{a,b}\right \|_{L^2_{a,b}},
	\end{equation*}
	where $C_{a,b}:=\left( \int_{\R}    \psi_{X(T)}^2(x)\omega_{a,b}^{-1}(x)dx\right)^{\frac{1}{2}}$.
\end{theorem}

\begin{remark}
	To have convergence of the price approximation, by Theorem \ref{error} we need $\psi_{X(T)}$ to satisfy condition \eqref{psicond}. It is obvious that this condition is quite restrictive, since
	it basically asks the tails of $\psi_{X(T)}$ to vanish faster than the tails of a Gaussian density function, which is not true in general for Lévy processes, characterized by heavy tails. However, we point out that condition \eqref{psicond} is only a sufficient condition for proving the error bound in Theorem \ref{error}, and not a necessary condition for convergence. We shall come back to this later.
\end{remark}
If $X(T)$ follows a Gaussian distribution with mean $\mu$ and variance $\sigma^2$, by direct computation one obtains that
\begin{equation}
	\label{Cab}
	C_{a,b}^2 = \int_{\R}    \psi_{X(T)}^2(x)\omega_{a,b}^{-1}(x)dx = \frac{b}{\sqrt{2\pi\sigma^2}}\frac{e^{\frac{(a-\mu)^2}{(2b^2-\sigma^2)}}}{\sqrt{2b^2-\sigma^2}},
\end{equation}
which leads to a more explicit formulation for condition \eqref{psicond}.
\begin{proposition}
	\label{condb}
	If the random variable $X(T)$ follows a Gaussian distribution with mean $\mu$ and variance $\sigma^2$, then condition \eqref{psicond} is equivalent to
	\begin{equation}
		\label{psicondG}
		b > \frac{\sigma}{\sqrt{2}}=: \underline{b}_{\sigma}
	\end{equation}
	where $b$ is the scale for the GHPs, and for $b = \underline{b}_{\sigma}$ we expect instabilities in the approximation.
\end{proposition}

\begin{remark}
	We point out that condition \eqref{psicond} extends \cite[Proposition 3.1]{willems2019}, since it basically coincides with asking that the likelihood ratio function $\eta$ defined by $\psi_{X(T)}(x)= \eta(x)\omega_{a,b}(x)$, is such that $\eta \in L^2_{a,b}$, but for any generic density function $\psi_{X(T)}$ (and not only for a log-normal density function). Indeed, if considering $\psi_{X(T)}$ and $\omega_{a,b}$ to be the density functions of two log-normal distributions with mean $\mu$ and variance $\sigma^2$, respectively, with mean $a$ and variance $b^2$, then $C_{a,b}$ takes a similar form as in equation \eqref{Cab}. In particular, the squared root $\sqrt{2b^2-\sigma^2}$ still appears. Then, with abuse of notation, by setting $\sigma^2 = \sigma^2T$, we recover \cite[Proposition 3.1]{willems2019} which now coincides with condition \eqref{psicondG}\footnote{To be more precise, the setting in \cite{willems2019} defines a likelihood ratio function $\ell$ in terms a the weight $\omega$ and a density $g$, where the latter one is the density function of the average price process defined in a continuos manner starting from a log-normally distributed underlying spot price. Then $g$ does not define a log-normal distribution. However, its tails are dominated by the tails of a log-normal density function, as proved in \cite{willems2019}.}.
\end{remark}

\begin{corollary}
	\label{cor}
	In the same setting of Proposition \ref{condb}, if $a=\mu$ then $C_{a,b}$ is a monotone decreasing function of the scale $b$ with limit $\frac{1}{\sqrt[4]{4\pi\sigma^2}}$.
\end{corollary}

\begin{remark}
	Theorem \ref{error} shows that the absolute error in approximating the option price $\Pi_K(t)$ with generalized Hermite polynomials is bounded by the product of $C_{a,b}$ with the $L^2_{a,b}$-norm of the error in approximating the payoff function $\varphi_K$. In particular, due to Proposition \ref{errorphi}, this last term is an increasing function of $b$ (if $a$ is in a neighbour of $K$), while, according to Corollary \ref{cor}, $C_{a,b}$ is a decreasing function of $b$, at least in the Gaussian case. 
\end{remark}

\section{Polynomial processes and correlator formula}
\label{polpross}
Following \cite{filipovic2020}, we consider a jump-diffusion operator on $\R$ of the form
\begin{equation*}
	\G f(x) = b(x)f'(x) + \frac{1}{2}\sigma^2(x)f''(x)+\int_{\R}\left( f(x+z)-f(x)-f'(x)z \right)\ell(x,dz),
\end{equation*} 
for some measurable maps $b:\R \to \R$ and $\sigma:\R \to \R$, and a transition kernel $\ell:\R\times \R \to \R$, such that the conditions in \cite[Lemma 1]{filipovic2020} are satisfied, namely
\begin{equation}
\label{poldefcond_FL}
b\in \mathrm{Pol}_1(\R),\qquad
\sigma^2 +\int_{\R}z^2\ell(\cdot,dz)\in \mathrm{Pol}_2(\R) \quad \mbox{and}\quad
\int_{\R}z^m\ell(\cdot,dz)\in \mathrm{Pol}_m(\R)  \mbox{ for all } m\ge 3.
\end{equation}
The operator $\G$ is then called \emph{polynomial} in the sense of \cite[Definition 1]{filipovic2020}, and the process $Y$ having $\G$ as extended generator is a \emph{polynomial jump-diffusion process}.

As a consequence of condition \eqref{poldefcond_FL}, a polynomial generator $\G$ can be expressed in matrix form. This leads to the so-called \emph{generator matrix}, which strictly depends on the polynomial basis of choice. By considering the vector valued function $H_n$ introduced in Section \ref{hermitesection}, for every $n\ge1$, the generator matrix associated with $\G$ is the matrix $G_n \in \R^{(n+1)\times(n+1)}$ satisfying $\G H_n(x) = G_nH_n(x).$ For $Y$ polynomial process and $p \in  \mathrm{Pol}_n(\R)$, the conditional expectation $\E\left[p(Y(T))\left.\right|\F_t\right]$ is then a polynomial function in $Y(t)$, $0 \le t\le T$, and is given in closed form in  \cite[Theorem 2.5]{filipovic2020} and in the following theorem..
\begin{theorem}
	\label{propmoments}
	For $Y$ polynomial process and $p\in \mathrm{Pol}_n(\R)$, the following moment formula holds
	\begin{equation*}
		\E\left[p(Y(T))\left.\right|\F_t\right] =  \vec{p}_n^{\top}e^{G_n(T-t)}H_n(Y(t)), \qquad 0\le t\le T,
	\end{equation*}
	with $\vec{p}_n\in\R^{n+1}$ vector of coefficients for $p$ with respect to $H_n$ and $G_n \in \R^{(n+1)\times(n+1)}$ generator matrix.
\end{theorem}

\begin{corollary}
	\label{prop2}
	For every $n\ge 0$, we then get that $\E\left[\left.Y(T)^{n}\right|\F_t\right] = \boldsymbol{e}_{n+1,n+1}^{\top}\, e^{G_n(T-t)}H_n(Y(t))$.
\end{corollary}

In \cite[Theorem 4.5]{benth2019}, the moment formula is extended to correlators, namely conditional expectations of products of polynomial functions in the polynomial process $Y$ evaluated at different time points. We shall not discuss here the meaning of the symbols appearing in the formula. However the main details can be found in Appendix \ref{appA}. For more interested readers, we refer to \cite{benth2019}.

\begin{theorem}
	\label{theorem}
	For $m\ge1$, we consider $m+1$ polynomial functions $p_k\in \mathrm{Pol}_{n_k}(\R)$, $k=0,\dots,m$, in the polynomial process $Y$, evaluated at different time points, $t < s_0 < s_1 < \dots < s_m$, and with vector of coefficients $\vec{p}_{k,n}\in \R^{n+1}$, for $n= \max\{n_0,\dots,n_m\}$. Then, there exist $m+1$ matrices $\tilde{G}_n^{(r)}\in \R^{(n+1)^{r+1}\times(n+1)^{r+1}}$, $r=0, \dots, m$, 
	such that the following expectation formula holds:
	\begin{multline*}
		\E\left[p_m\left(Y(s_0)\right)p_{m-1}\left(Y(s_1)\right)\cdot \dots \cdot p_0\left(Y(s_m)\right)\left.\right|\F_t\right] \\
		= \vec{p}_{m,n}^{\top}\left\{ vec^{-1} \circ e^{\tilde{G}_n^{(m)}(s_0-t)}\circ vec\left( H_n(Y(t))^{\top}\otimes^{m} H_n(Y(t))\right)\right\} \cdot \\  \cdot \prod_{k=1}^{m}e^{\tilde{G}_n^{(m-k)\top}(s_k-s_{k-1})}\left\{I_{n+1}\otimes^{m-k}\vec{p}_{m-k,n} \right\}
	\end{multline*}
where $\prod_{k=1}^{m}$ is the product obtained starting with the matrix corresponding to $k = 1$ and multiplying on the right by the following matrices until the matrix corresponding to $k = m$. In particular, $\tilde{G}_n^{(r)} = D_{n+1}^{(r)}G_{n(r+1)}E_{n+1}^{(r)}$
and $e^{\tilde{G}_n^{(r)}t} = D_{n+1}^{(r)}e^{G_{n(r+1)}t}E_{n+1}^{(r)}$, with $\tilde{G}_n^{(0)}= G_n^{(0)}= G_n$.
\end{theorem}

\begin{corollary}
	\label{corrformula}
	For every $n\ge 0$ and $0 \le k_0, k_1, \dots, k_m\le n$, the following formula holds
	\begin{align*}
		\label{generalP}
		&\E\left[\left.Y(s_0)^{k_0}Y(s_1)^{k_1}\cdots Y(s_{m})^{k_{m}}\right|\F_t\right]\\ 
		&= \boldsymbol{e}_{n+1,k_0+1}^{\top}\left\{ vec^{-1} \circ e^{\tilde{G}_n^{(m)}(s_0-t)}\circ vec\left( H_n(Y(t))^{\top}\otimes^{m} H_n(Y(t))\right)\right\} \cdot \\ &\qquad  \qquad \qquad \qquad \qquad \qquad\qquad  \cdot \prod_{j=1}^{m}e^{\tilde{G}_n^{(m-j)\top}(s_j-s_{j-1})}\left\{I_{n+1}\otimes^{m-j}\boldsymbol{e}_{n+1, k_{j}+1} \right\}.
	\end{align*}
\end{corollary}

\subsection{Greeks for path-dependent options}
\label{greeks}
The pricing formulas of Theorem \ref{price} and \ref{price2}, together with the moment and correlator formulas of Theorem \ref{propmoments} and \ref{theorem}, allows for sensitivity analysis and risk management. Indeed, thanks to the compact and closed formulation, it is possible to differentiate the price functional with respect to various parameters and obtain the so-called \emph{Greeks}. In \cite[Section 6]{benth2019} the authors derive an expression for the Delta and Theta for correlators, namely for the partial derivatives
\begin{equation*}
\Delta^{k_0,\dots, k_m}_{(s_0,\dots, s_m;t)}:= \frac{\partial C^{k_0,\dots, k_m}(s_0,\dots, s_m;t)}{\partial Y(t)} \quad \mbox{and} \quad \Theta j^{k_0,\dots, k_m}_{(s_0,\dots, s_m;t)}:= \frac{\partial C^{k_0,\dots, k_m}(s_0,\dots, s_m;t)}{\partial s_j},\; 0\le j \le m,
\end{equation*}
with $C^{k_0,\dots, k_m}(s_0,\dots, s_m;t) := \E\left[\left.Y(s_0)^{k_0} Y(s_1)^{k_1} \cdots Y(s_m)^{k_m}\right|\F_t\right]$. Starting from their results, we can obtain Delta and Theta for the discretely sampled arithmetic Asian option studied in this article.

\begin{proposition}
For every $N\ge 0$, the Delta of a discretely sampled arithmetic Asian option can be approximated with generalized Hermite polynomials by 
\begin{equation*}
	\frac{\partial\Pi_{K,N}^{a,b}(t)}{\partial Y(t)}=\sum_{k=0}^N\sum_{i=0}^k\sum_{|\boldsymbol{k}|=i}  \binom{k}{i} \frac{\hat{\boldsymbol{\beta}}_{N,k+1}^{a,b}(-a)^{k-i}}{(m+1)^{i}b^k}\,  \frac{i!}{k_0!k_1!\cdots k_{m}!} \; \Delta^{k_0,\dots, k_m}_{(s_0,\dots, s_m;t)}
\end{equation*}
with $\Delta^{k_0,\dots, k_m}_{(s_0,\dots, s_m;t)}$ given in \cite[Proposition 6.1]{benth2019}.
\end{proposition}

\begin{proposition}
For every $N\ge 0$, the Theta of a discretely sampled arithmetic Asian option can be approximated with generalized Hermite polynomials by 
\begin{equation*}
	\frac{\partial\Pi_{K,N}^{a,b}(t)}{\partial s_j}=\sum_{k=0}^N\sum_{i=0}^k\sum_{|\boldsymbol{k}|=i}  \binom{k}{i} \frac{\hat{\boldsymbol{\beta}}_{N,k+1}^{a,b}(-a)^{k-i}}{(m+1)^{i}b^k}\,  \frac{i!}{k_0!k_1!\cdots k_{m}!} \; \Theta j^{k_0,\dots, k_m}_{(s_0,\dots, s_m;t)} \quad \mbox{for } 0\le j \le m,
\end{equation*}
with $\Theta j^{k_0,\dots, k_m}_{(s_0,\dots, s_m;t)}$ given in \cite[Proposition 6.2]{benth2019}.
\end{proposition}

\section{Numerical examples}
\label{numerics}
We shall implement numerically the pricing formulas of Theorem \ref{price} and \ref{price2}. In particular, we shall start by testing the pricing formula with moments (Theorem \ref{price}) for a Brownian motion, a Gaussian Ornstein--Uhlenbeck process and a jump-diffusion process, all of these being polynomial processes as described in Section \ref{polpross}. We shall then test the pricing formula with correlators (Theorem \ref{price2}) for the  Gaussian Ornstein--Uhlenbeck process and the jump-diffusion process.

\subsection{Brownian motion}
We consider $X=B$, where $B$ is a Brownian motion. Then $\left.X(T)\right|\F_t\sim \mathcal{N}(0,T-t)$ 
and the price functional $\Pi_K(t)$ is given in closed form  by
\begin{equation}
	\label{closed_price}
	\Pi_K(t) =\sigma_X(T;t)\;\phi\!\left(\frac{K-\mu_X(T;t)}{\sigma_X(T;t)}\right)-(K-\mu_X(T;t))\left( 1-\Phi\!\left(\frac{K-\mu_X(T;t)}{\sigma_X(T;t)}\right) \right)
\end{equation}
with $\sigma_X(T;t)=\sqrt{T-t}$ and $\mu_X(T;t)=0$, 
so that we can benchmark the price approximation. To do that, we introduce the quantity 
\begin{equation*}
	\gamma_{a,b}^N:=-\log\left(\frac{\left|\Pi_K(t) -\Pi_{K,N}^{a,b}(t) \right|}{\Pi_K(t)} \right) 
\end{equation*}
which measures the accuracy of $\Pi_{K,N}^{a,b}(t)$, namely, the accuracy is of order $10^{-\gamma_{a,b}^N}$. We also compare $\gamma_{a,b}^N$ with the accuracy of a Monte-Carlo-simulation (MC) approach calculated in the same manner. We report $\gamma_{a,b}^N$ in Figure \ref{BMplot} and \ref{BMplot2} as a function of $N$ for different values of $K$ and $b$. Here we draw with a red horizontal dashed line the MC accuracy, and with a red vertical bar the value of $N$ for which the Hermite series reaches the same accuracy as the MC method. This latter one is obtained by averaging over $10^2$ outcomes, each of them coming from $2 \cdot 10^{4}$ simulations.

\begin{figure}[!tp]
	\setlength{\tabcolsep}{2pt}
	\resizebox{1\textwidth}{!}{
		\begin{tabular}{@{}>{\centering\arraybackslash}m{0.04\textwidth}@{}>{\centering\arraybackslash}m{0.24\textwidth}@{}>{\centering\arraybackslash}m{0.24\textwidth}@{}>{\centering\arraybackslash}m{0.24\textwidth}@{}>{\centering\arraybackslash}m{0.24\textwidth}@{}}
			& $\boldsymbol{K = 0.0}$&$\boldsymbol{K = 0.2}$ & $\boldsymbol{K = 0.6}$& $\boldsymbol{K = 1.0}$ \\
			\begin{turn}{90}$\boldsymbol{b =0.5}$\end{turn}&
			\includegraphics[width=0.23\textwidth]{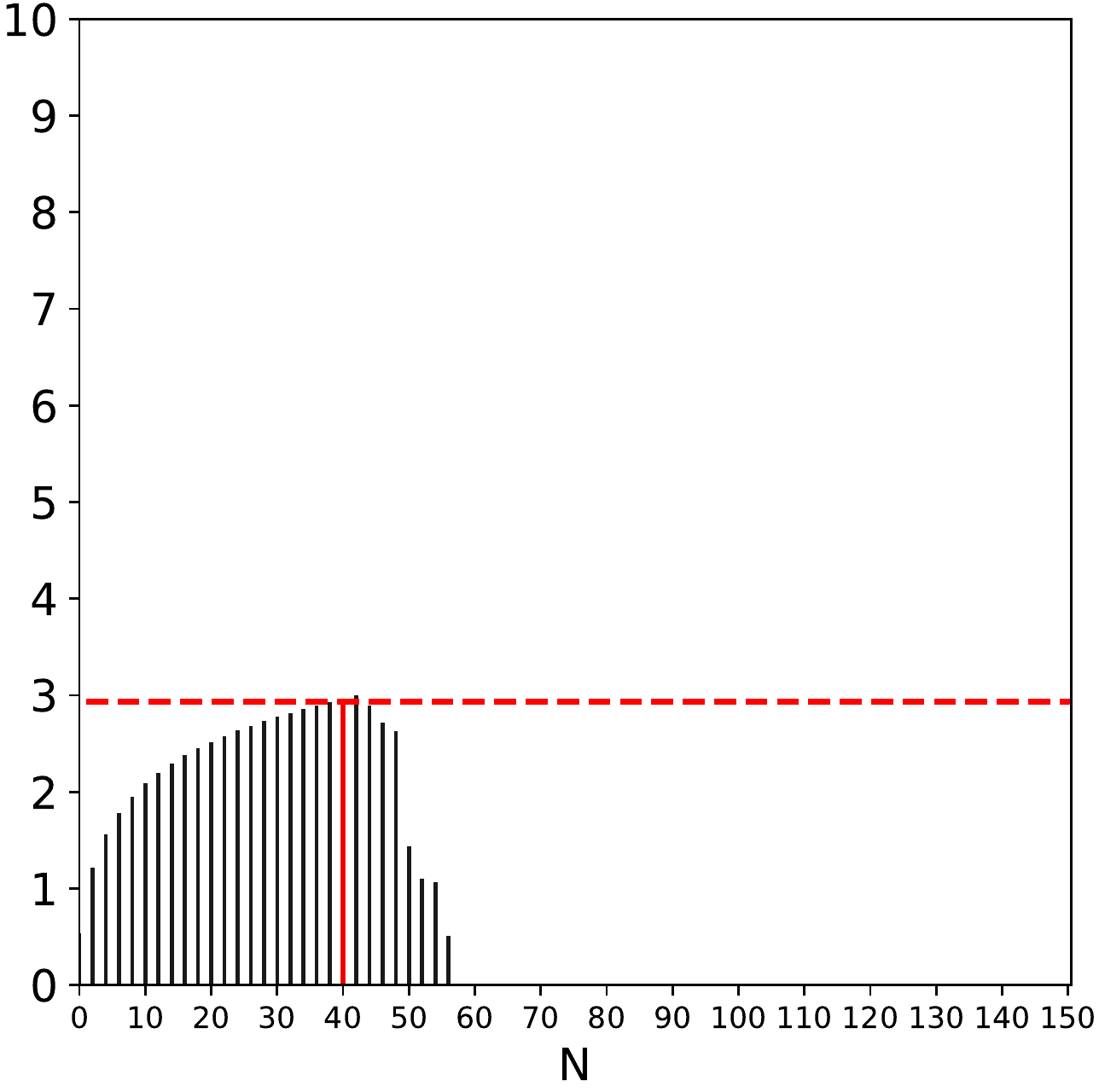}&
			\includegraphics[width=0.23\textwidth]{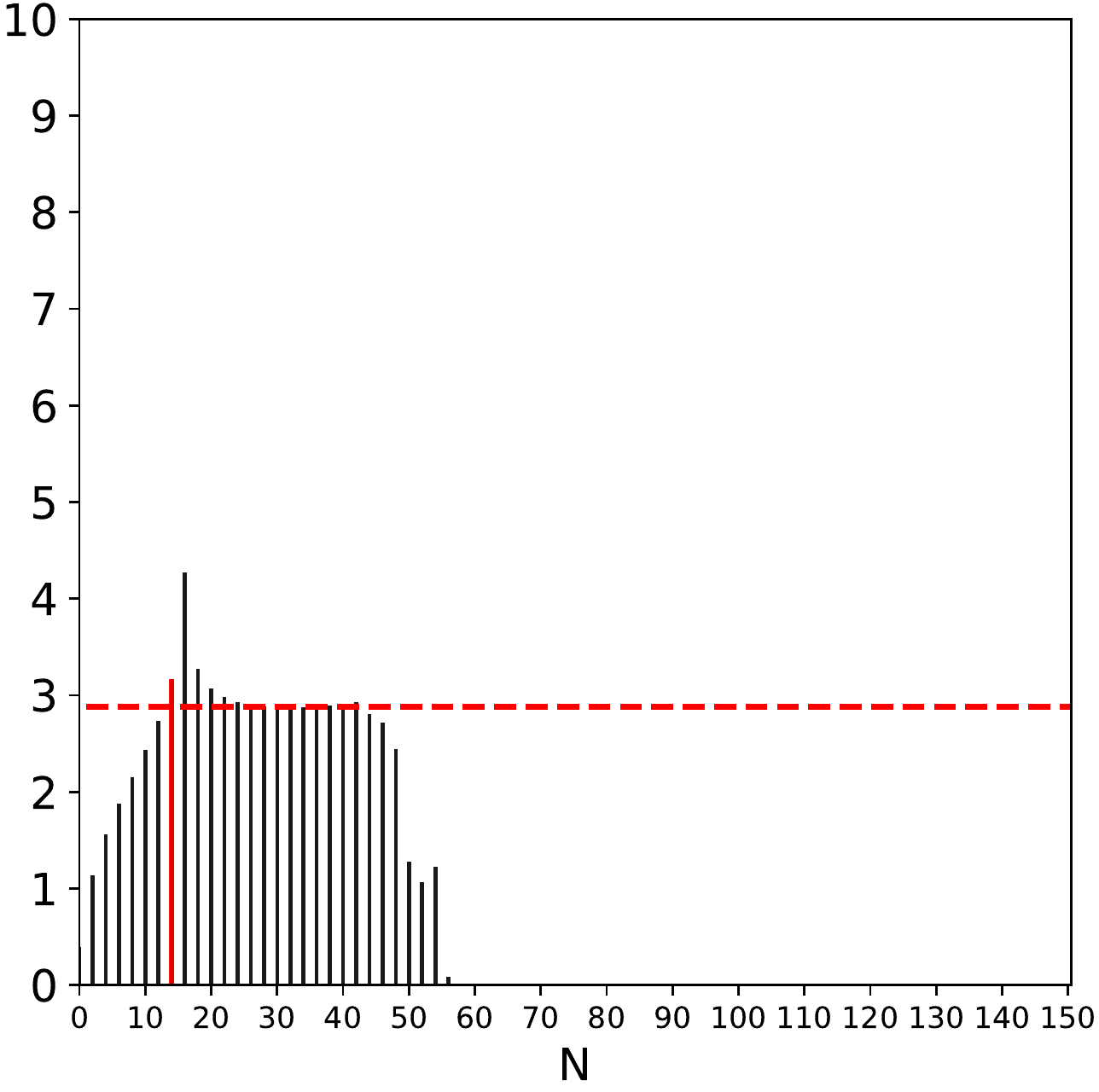} &	\includegraphics[width=0.23\textwidth]{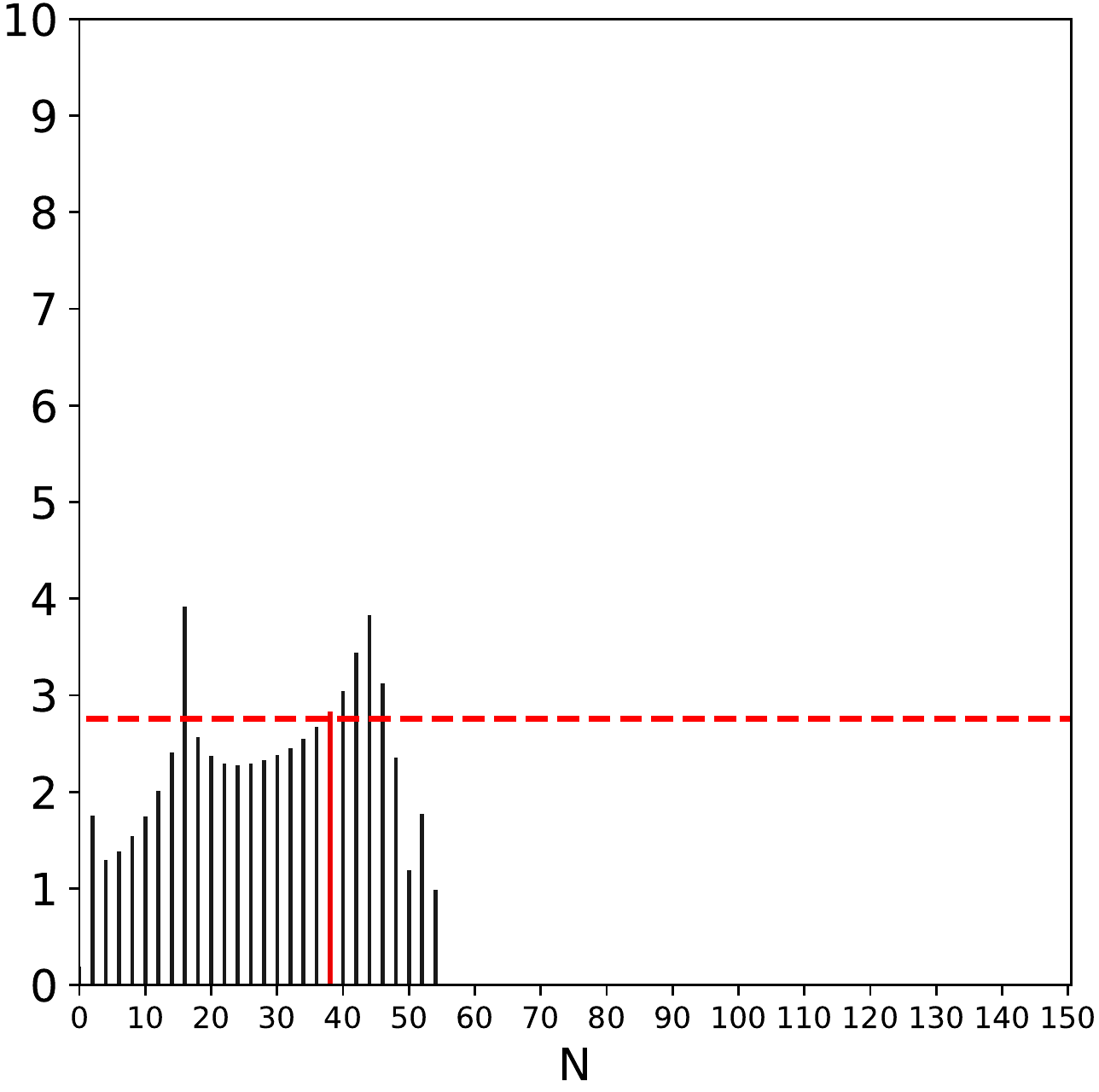} &
			\includegraphics[width=0.23\textwidth]{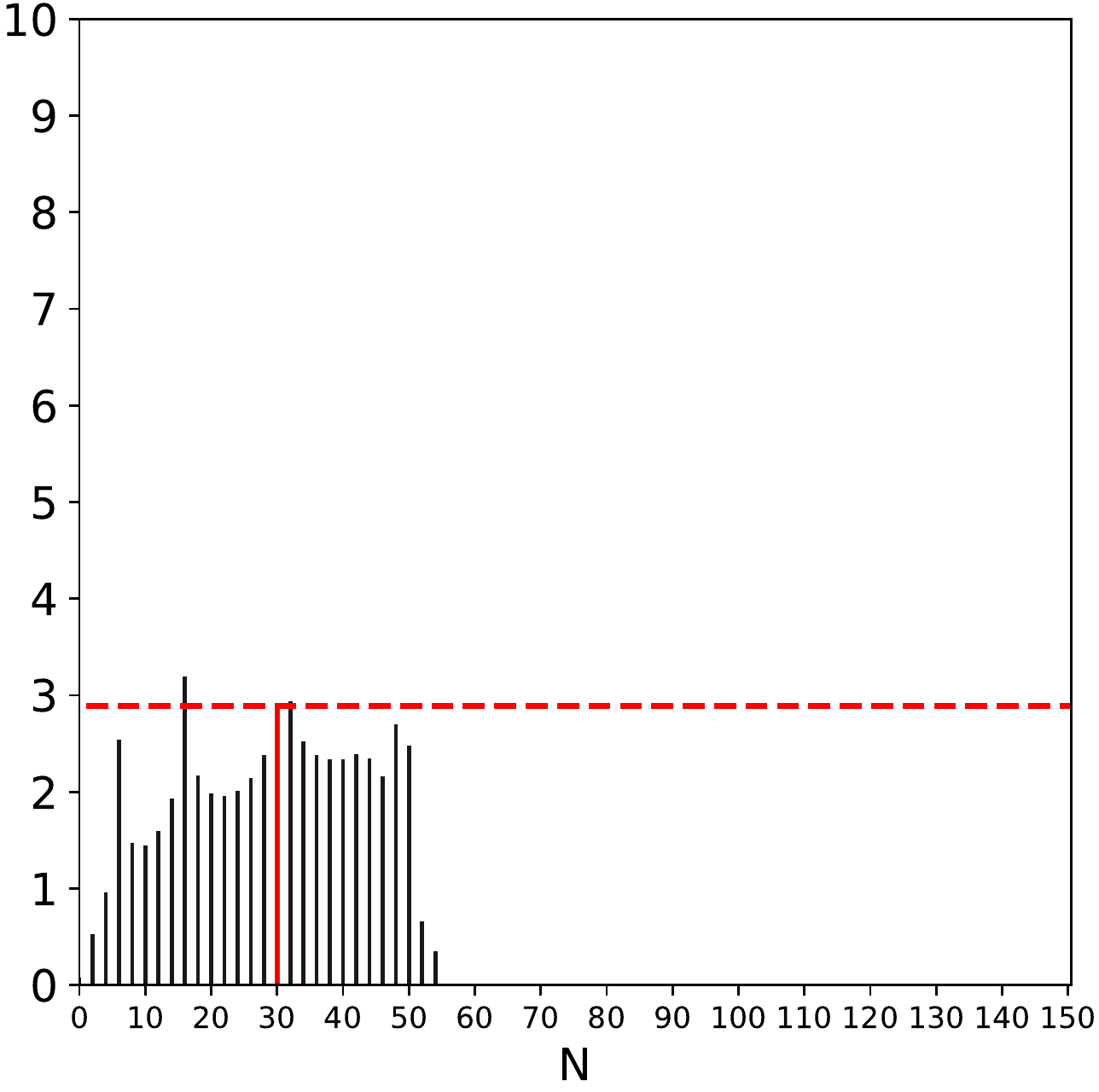}\\
			\begin{turn}{90}$\boldsymbol{b =0.6}$\end{turn}&
			\includegraphics[width=0.23\textwidth]{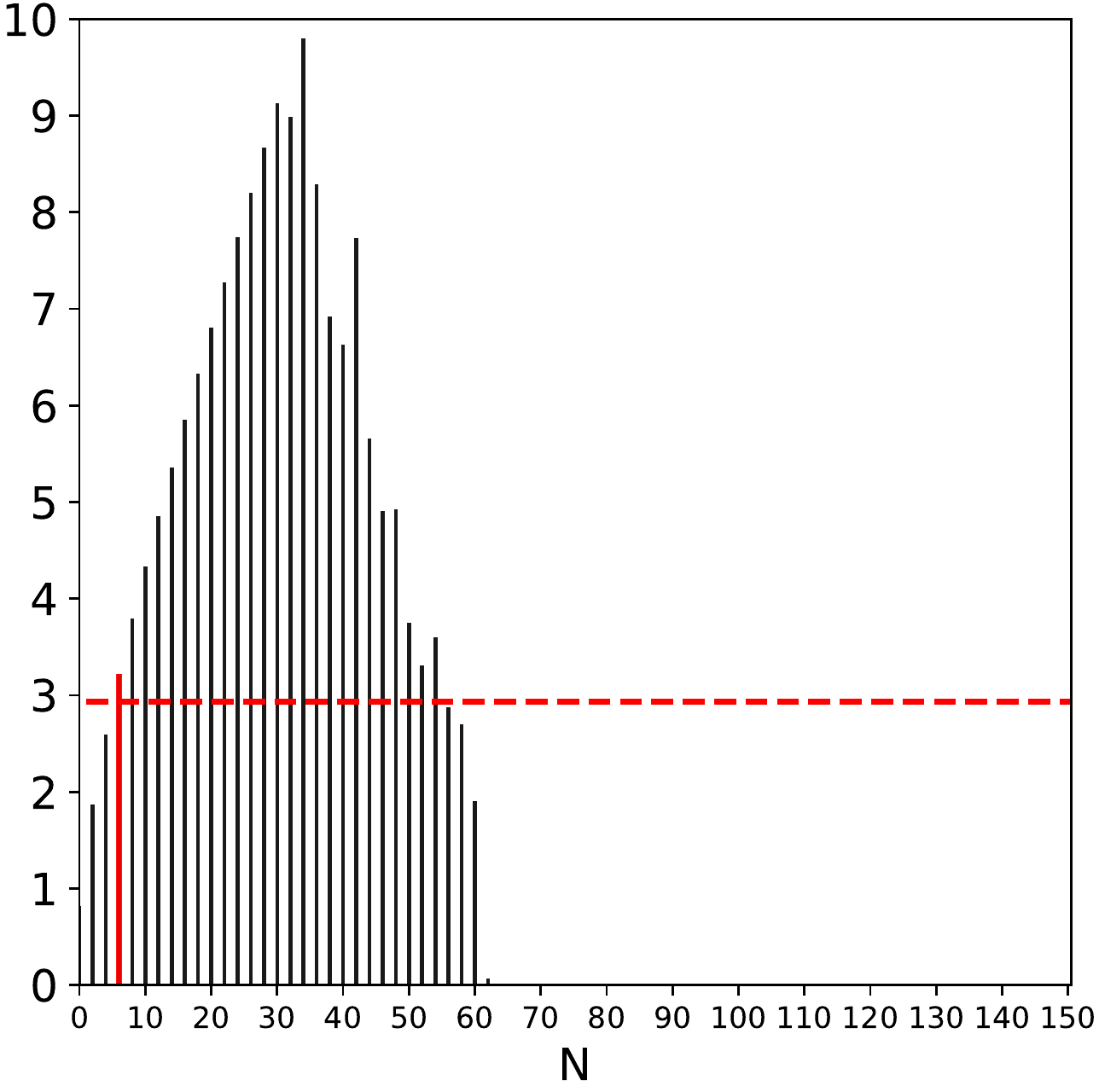}&
			\includegraphics[width=0.23\textwidth]{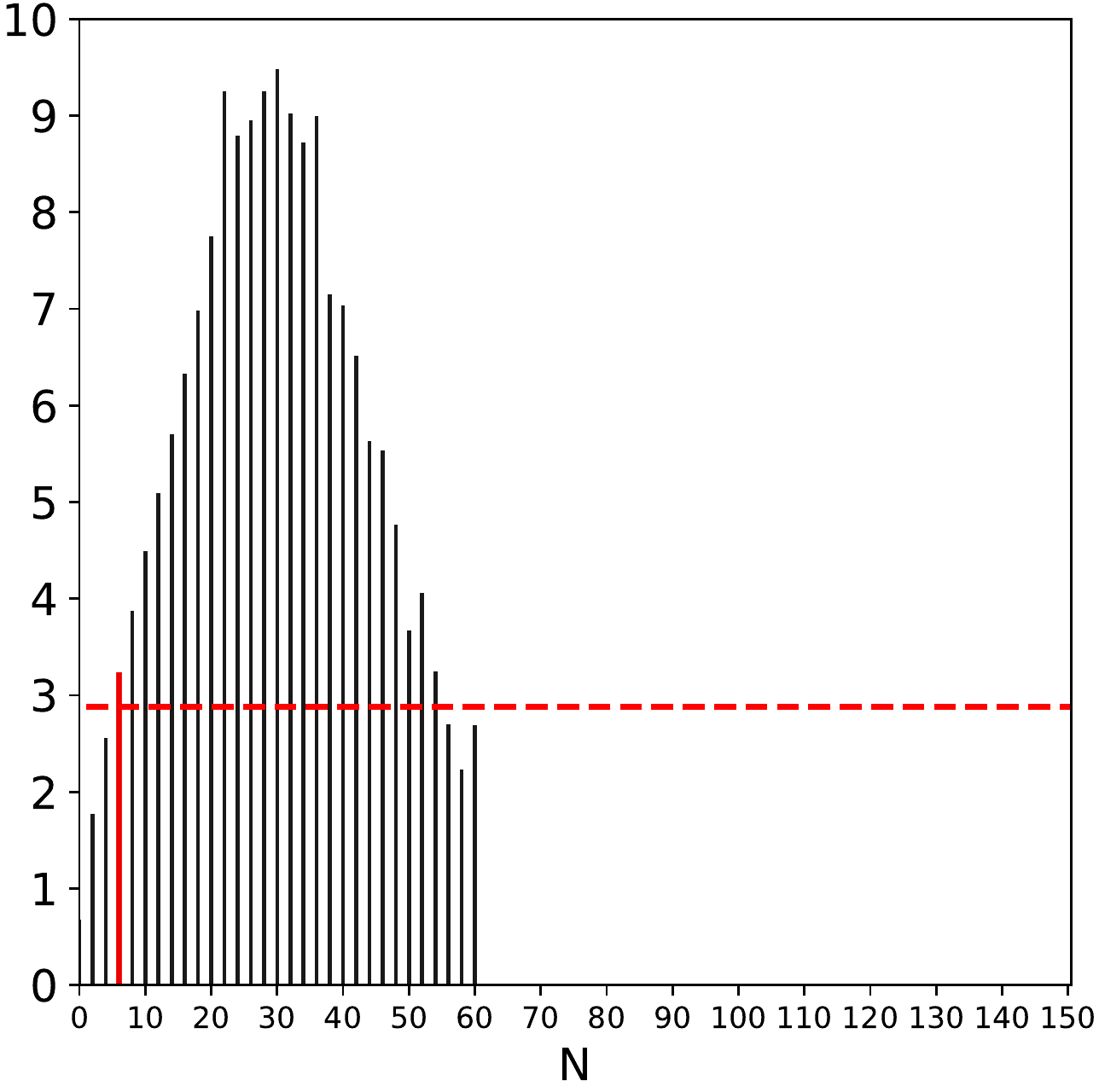} &	\includegraphics[width=0.23\textwidth]{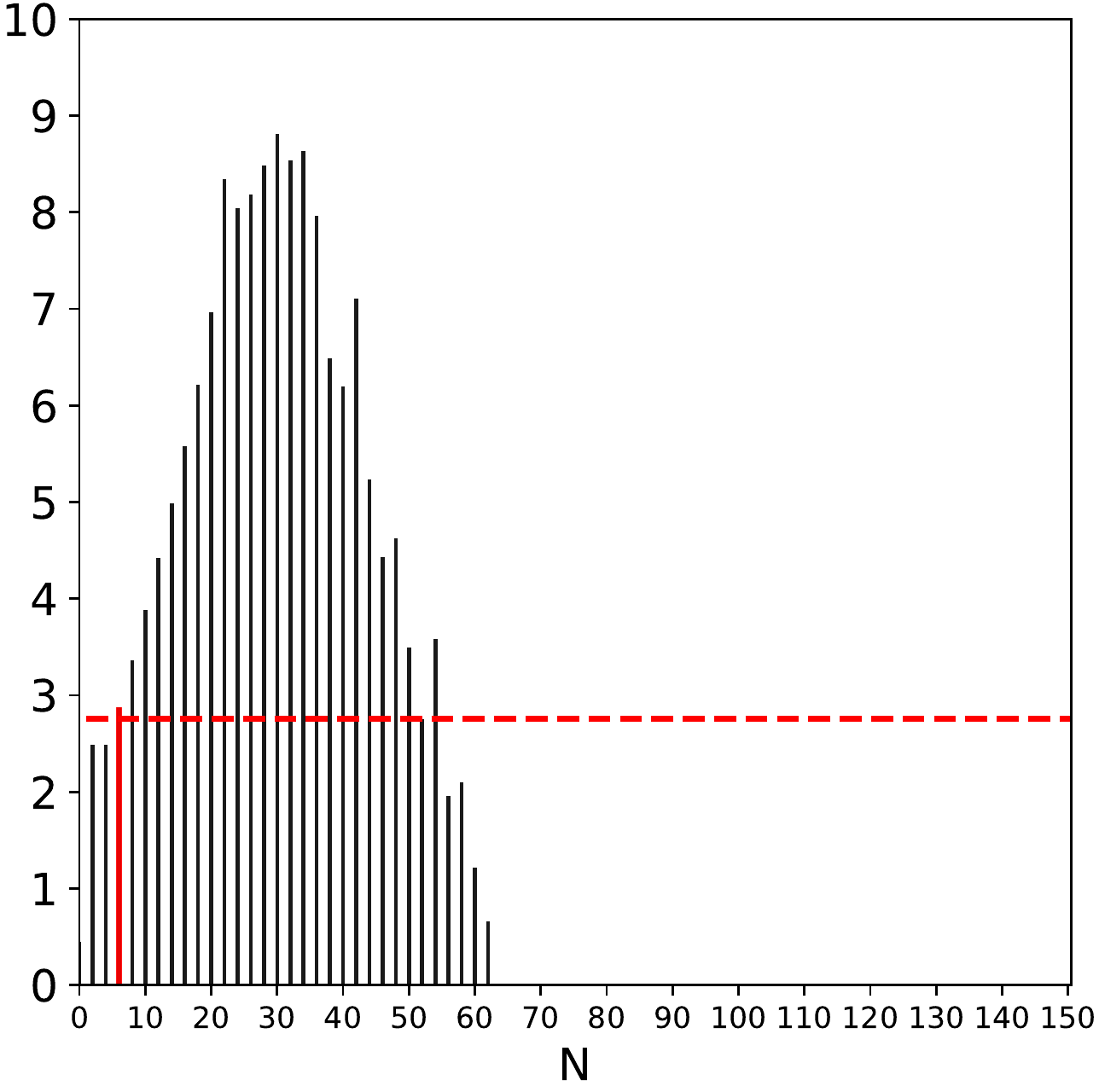} &
			\includegraphics[width=0.23\textwidth]{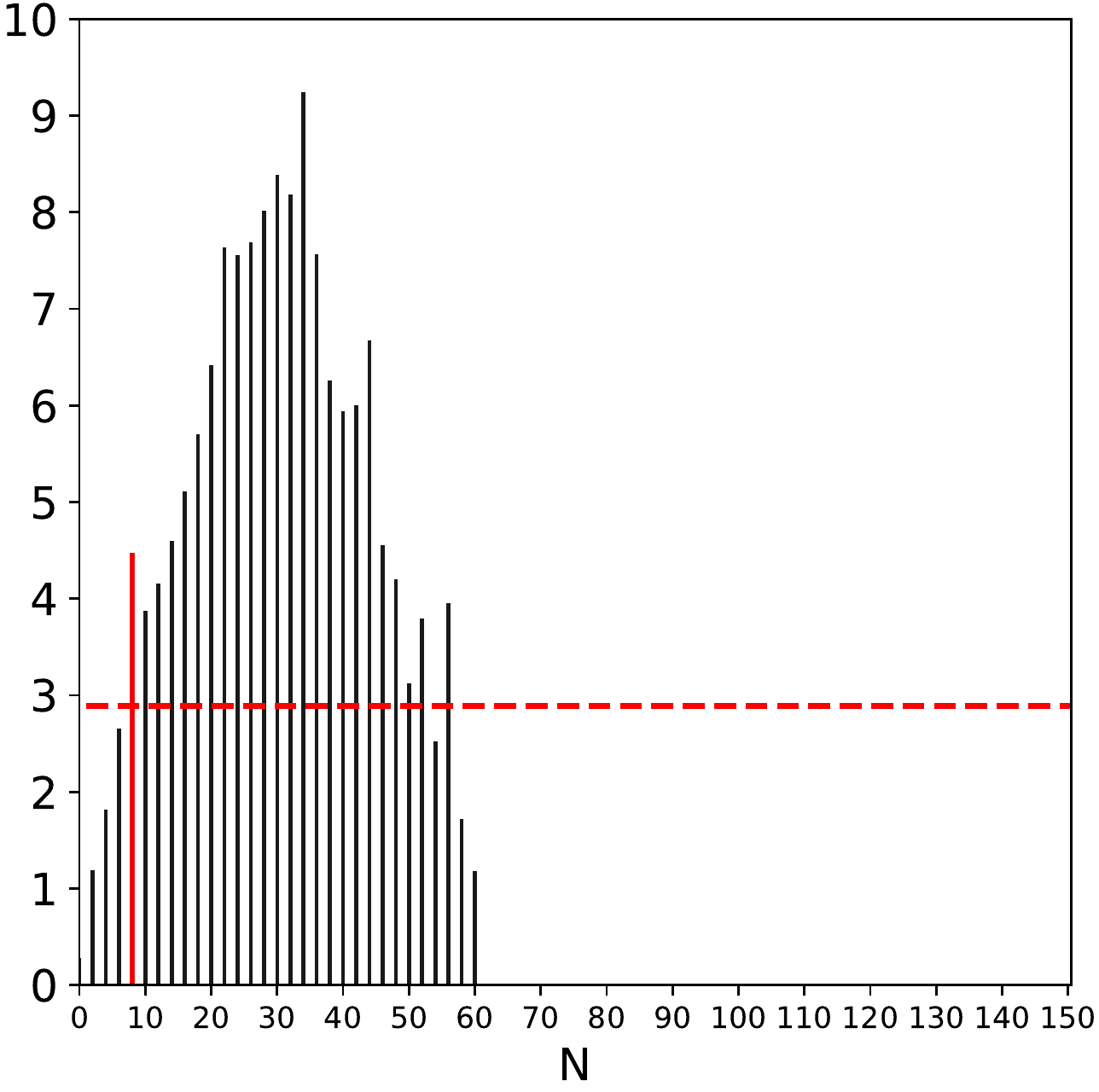}\\
			\begin{turn}{90}$\boldsymbol{b =1.0}$\end{turn}&
			\includegraphics[width=0.23\textwidth]{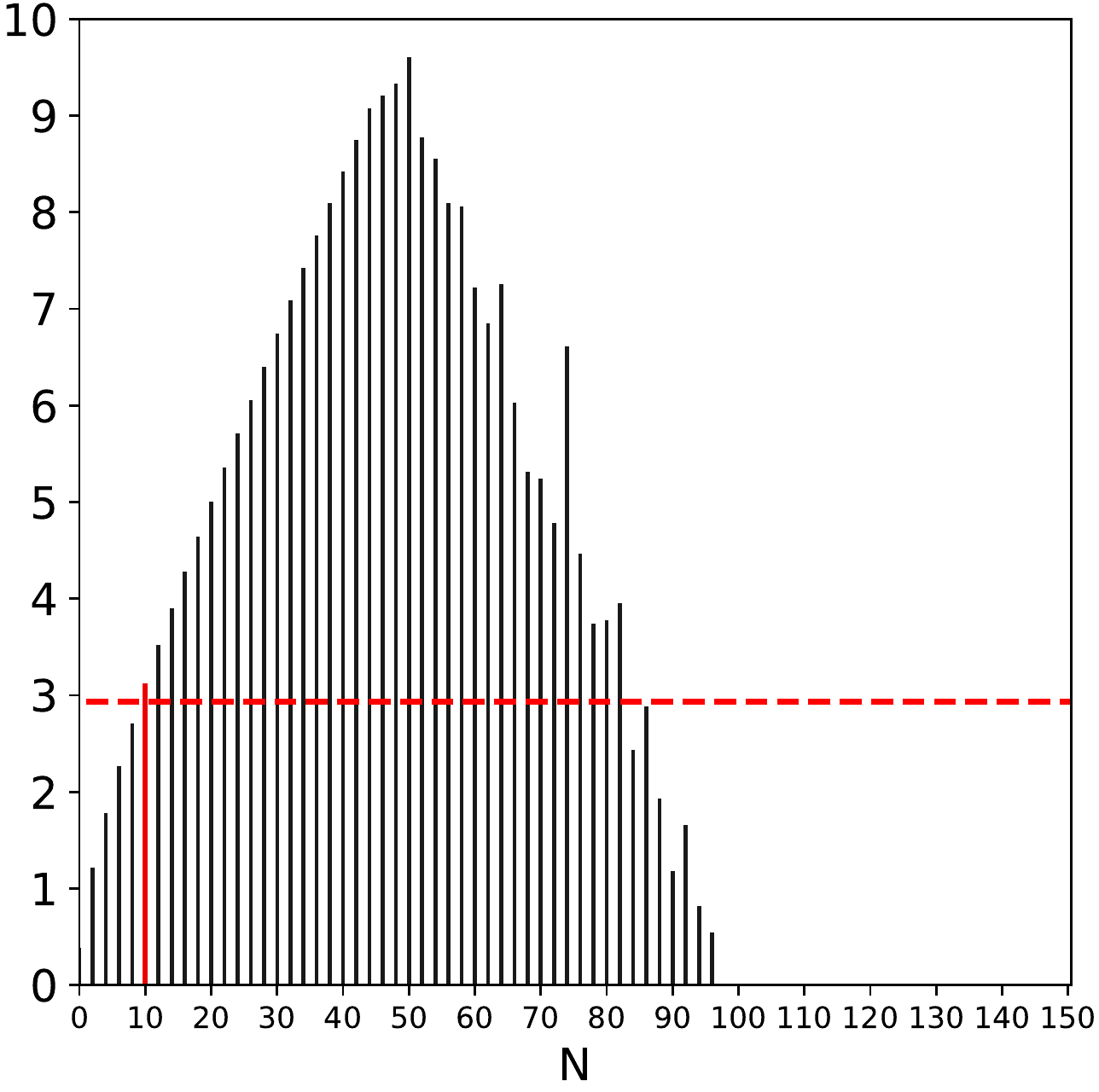}&
			\includegraphics[width=0.23\textwidth]{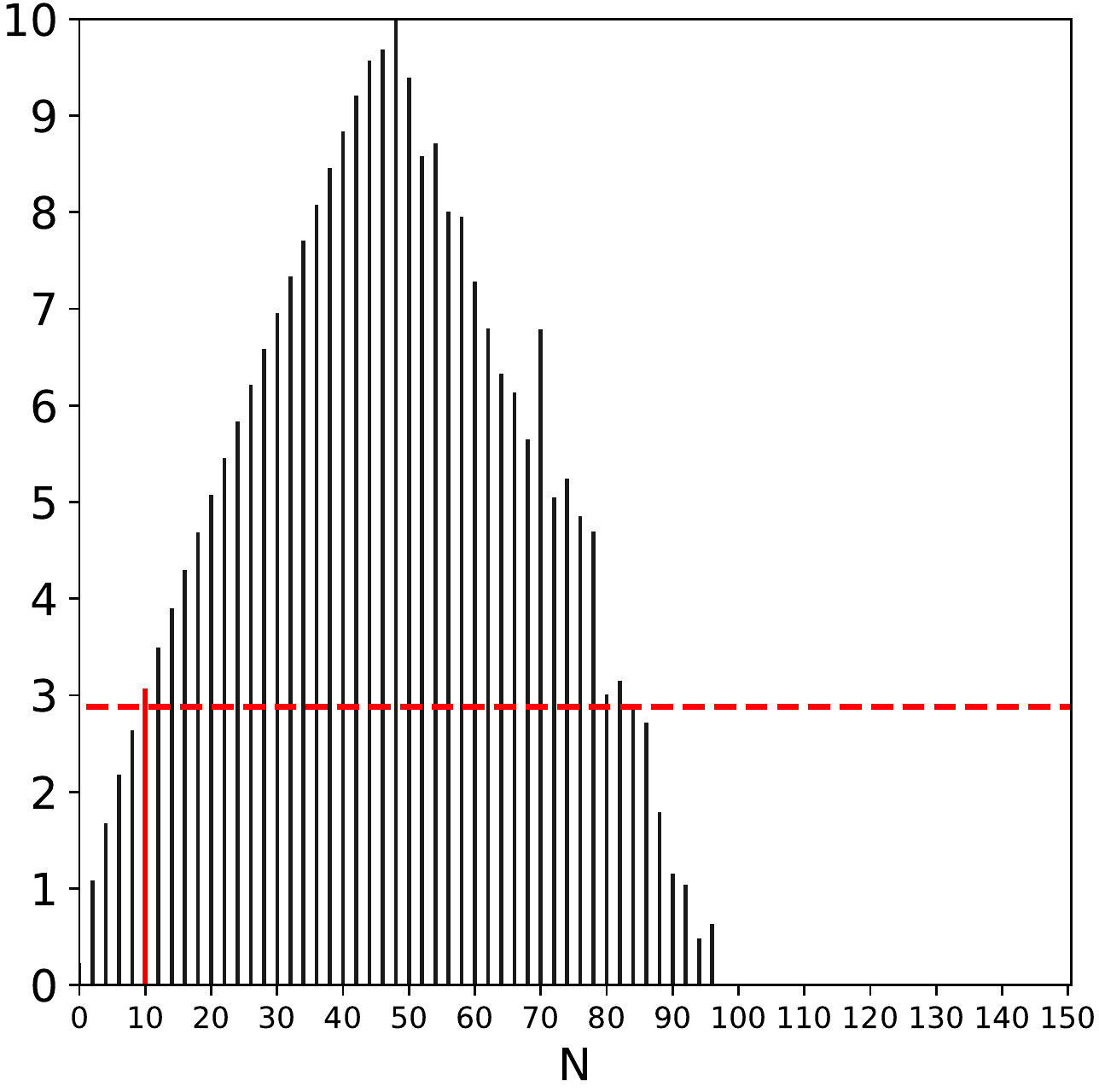} &	
			\includegraphics[width=0.23\textwidth]{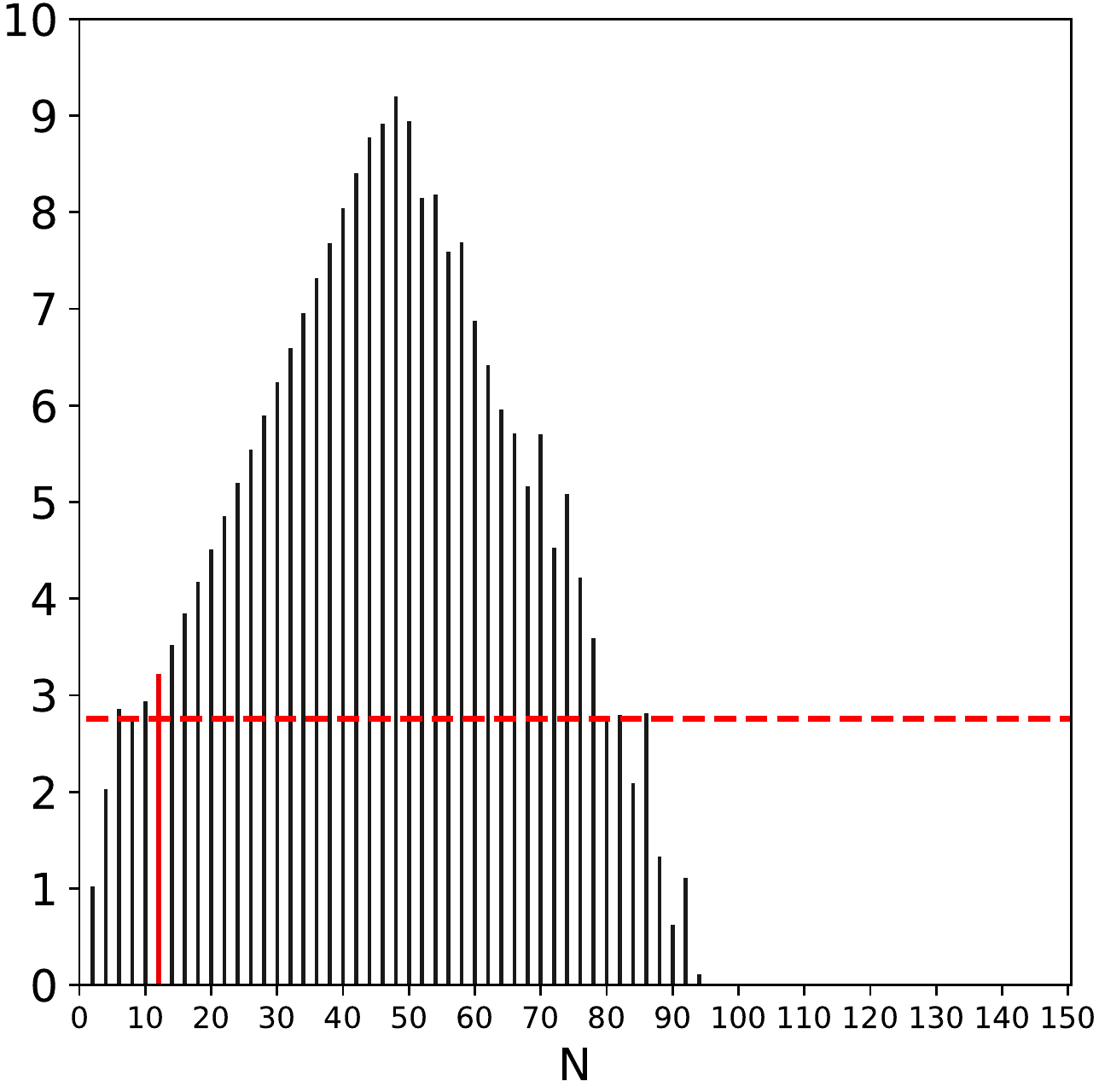} &
			\includegraphics[width=0.23\textwidth]{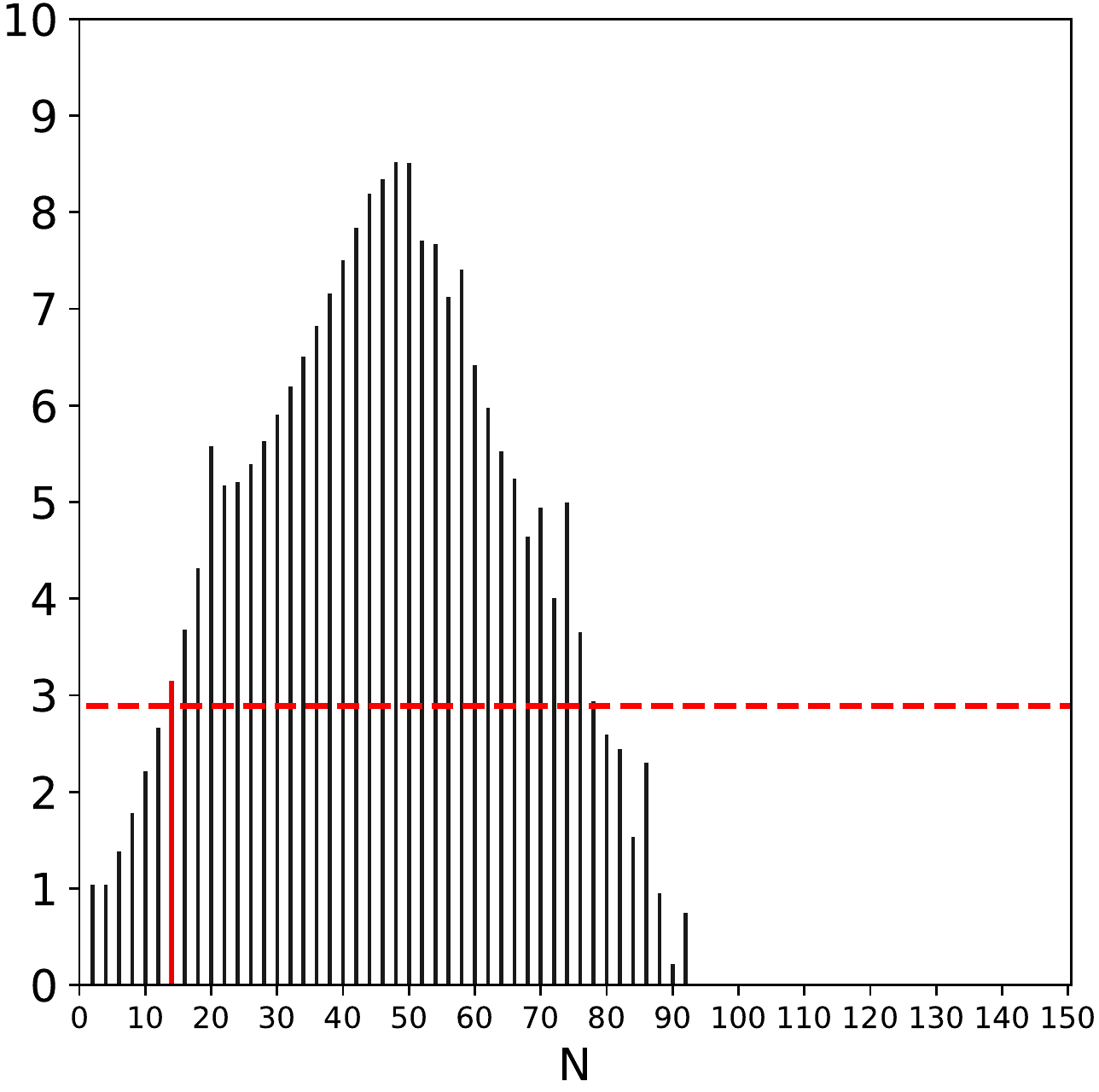}\\
			\begin{turn}{90}$\boldsymbol{b =2.0}$\end{turn}&
			\includegraphics[width=0.23\textwidth]{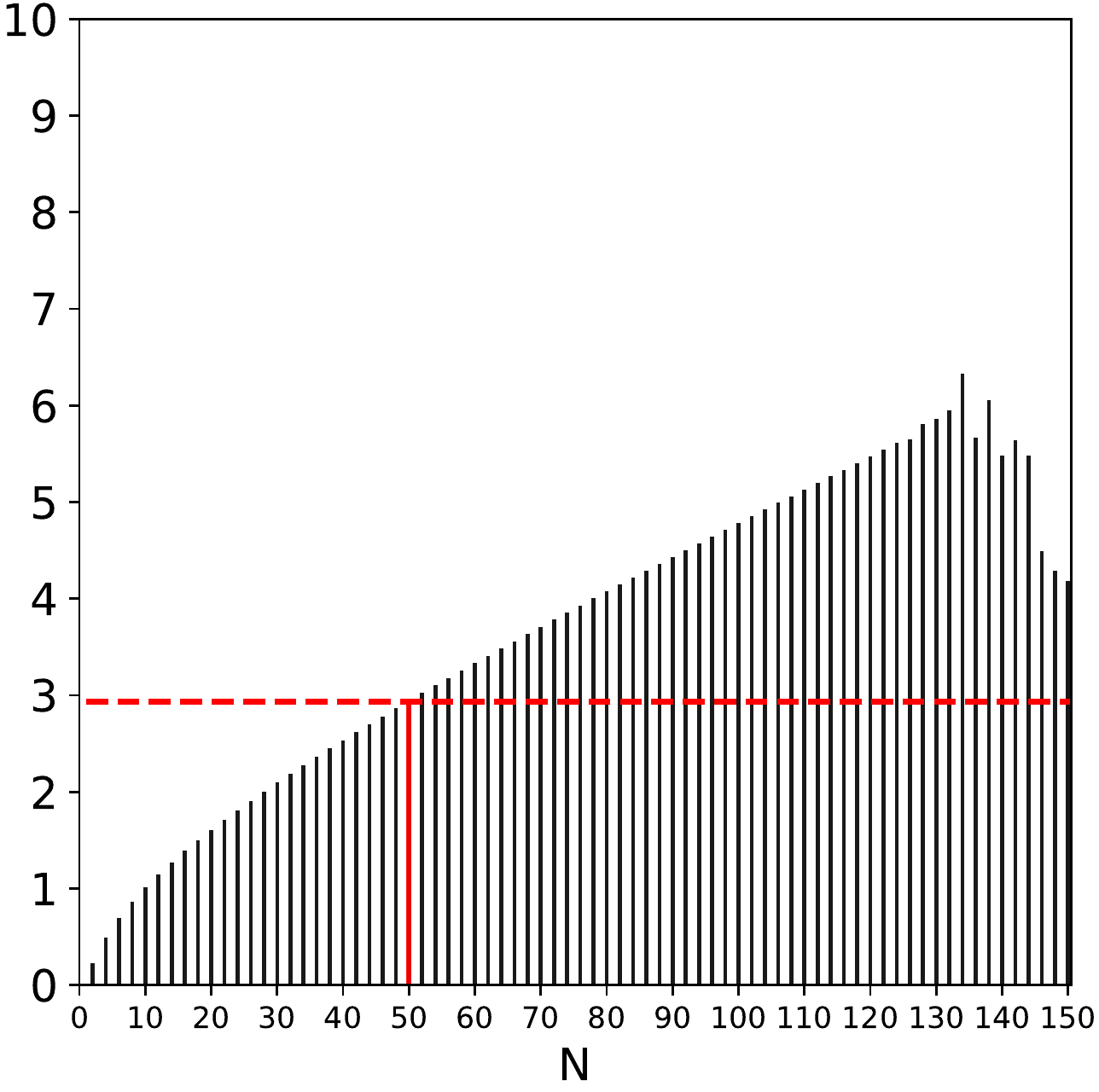}&
			\includegraphics[width=0.23\textwidth]{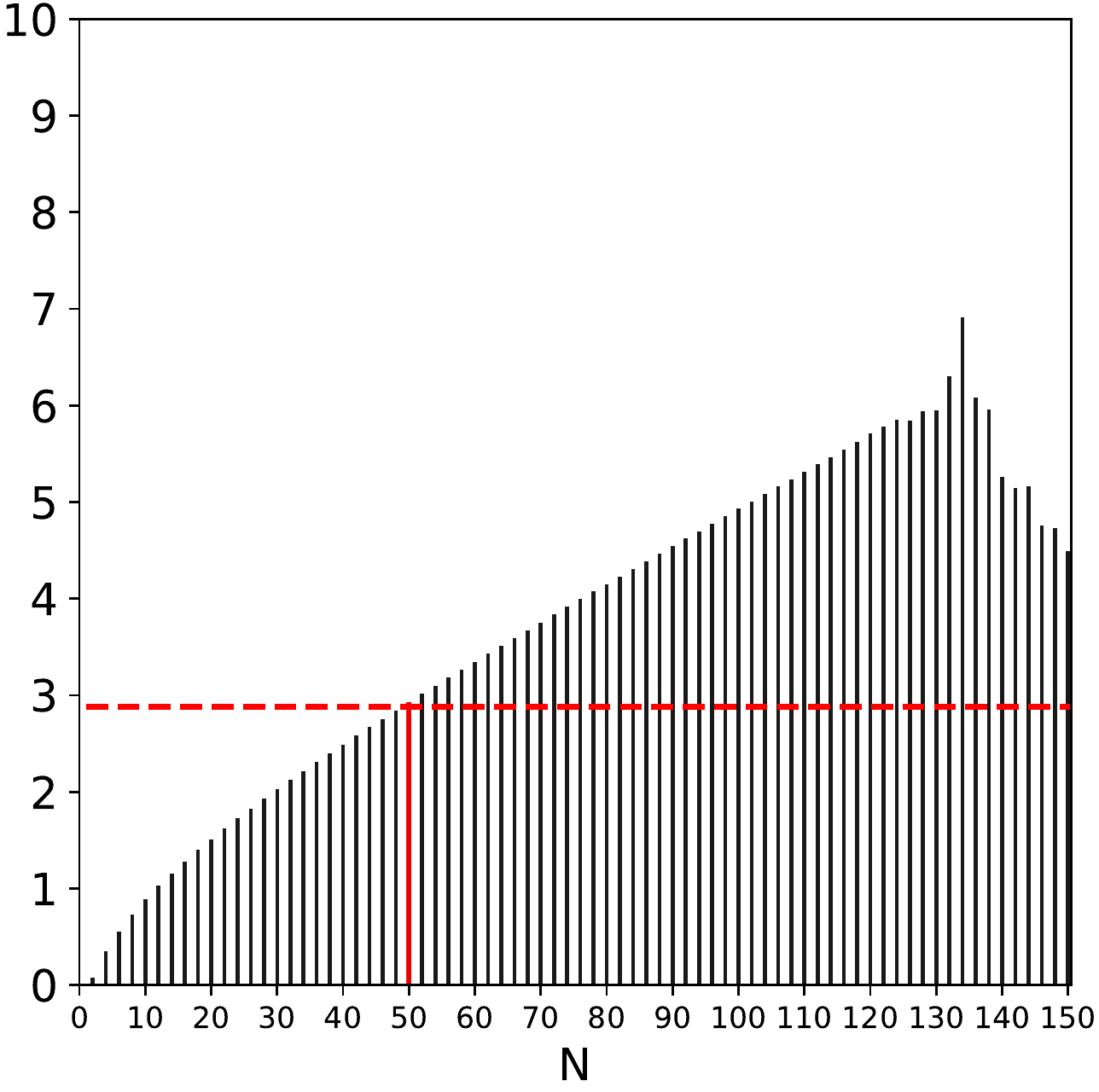} &
			\includegraphics[width=0.23\textwidth]{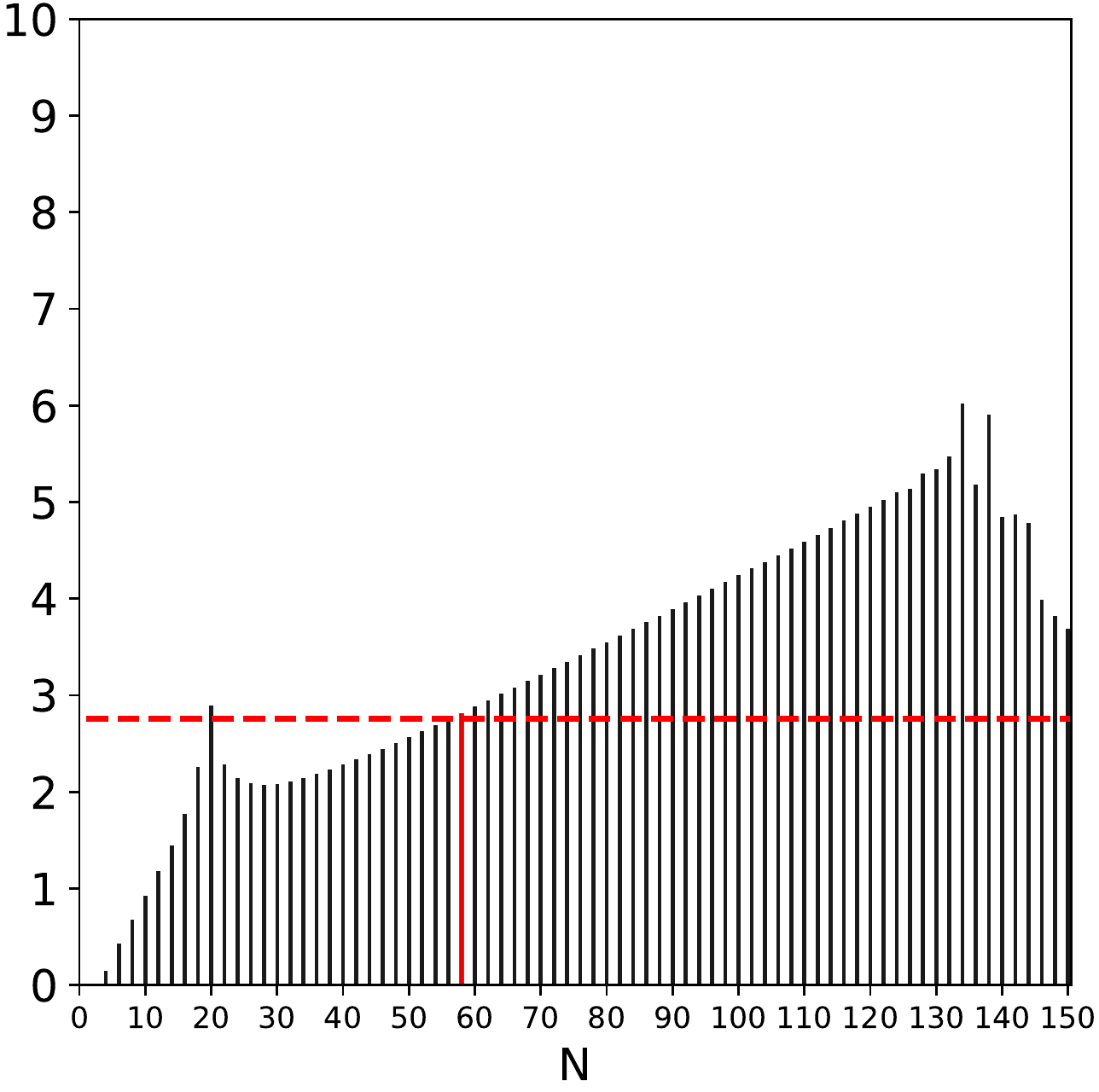} &
			\includegraphics[width=0.23\textwidth]{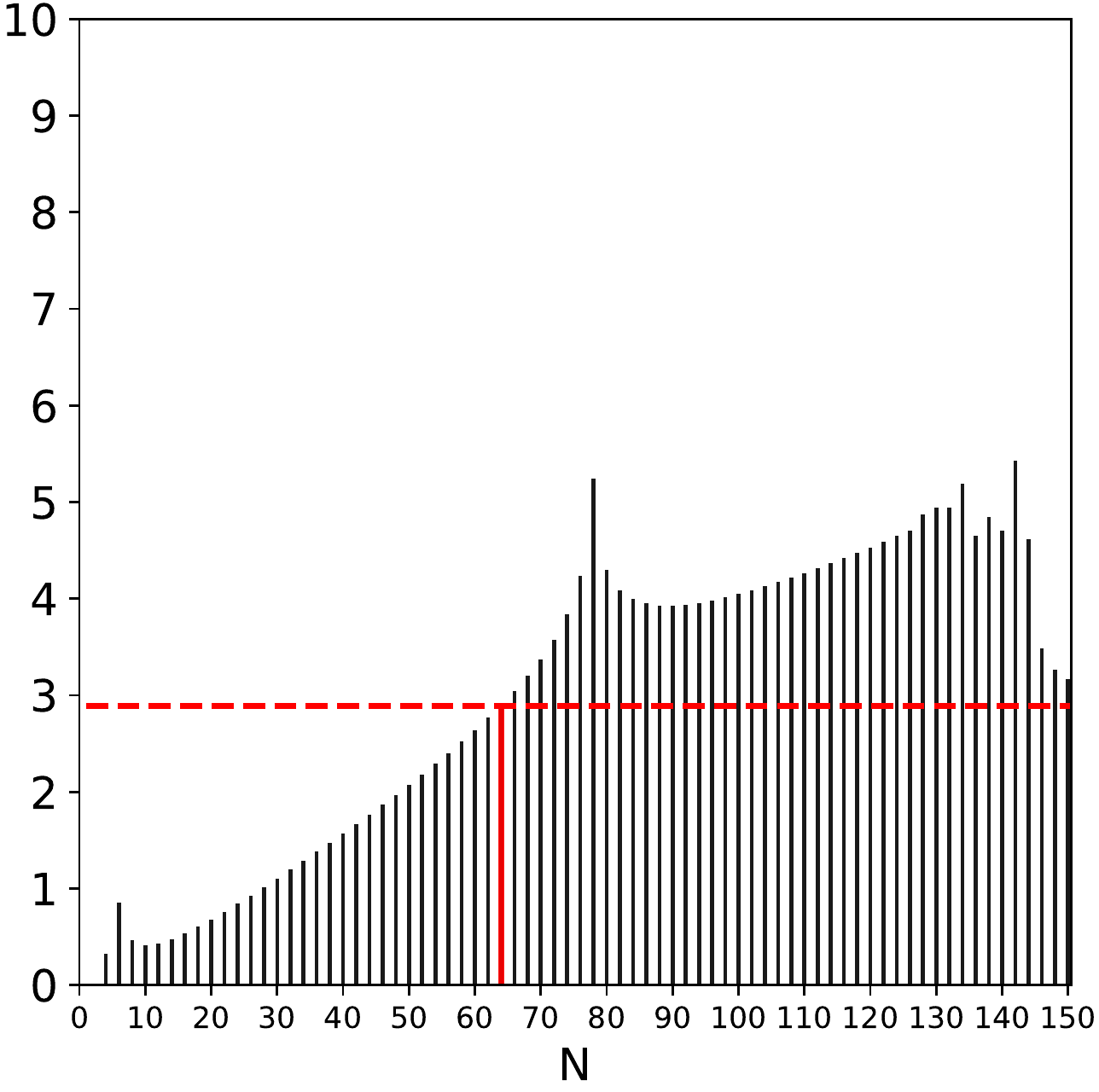}\\
			\begin{turn}{90}$\boldsymbol{b =3.0}$\end{turn}&
			\includegraphics[width=0.23\textwidth]{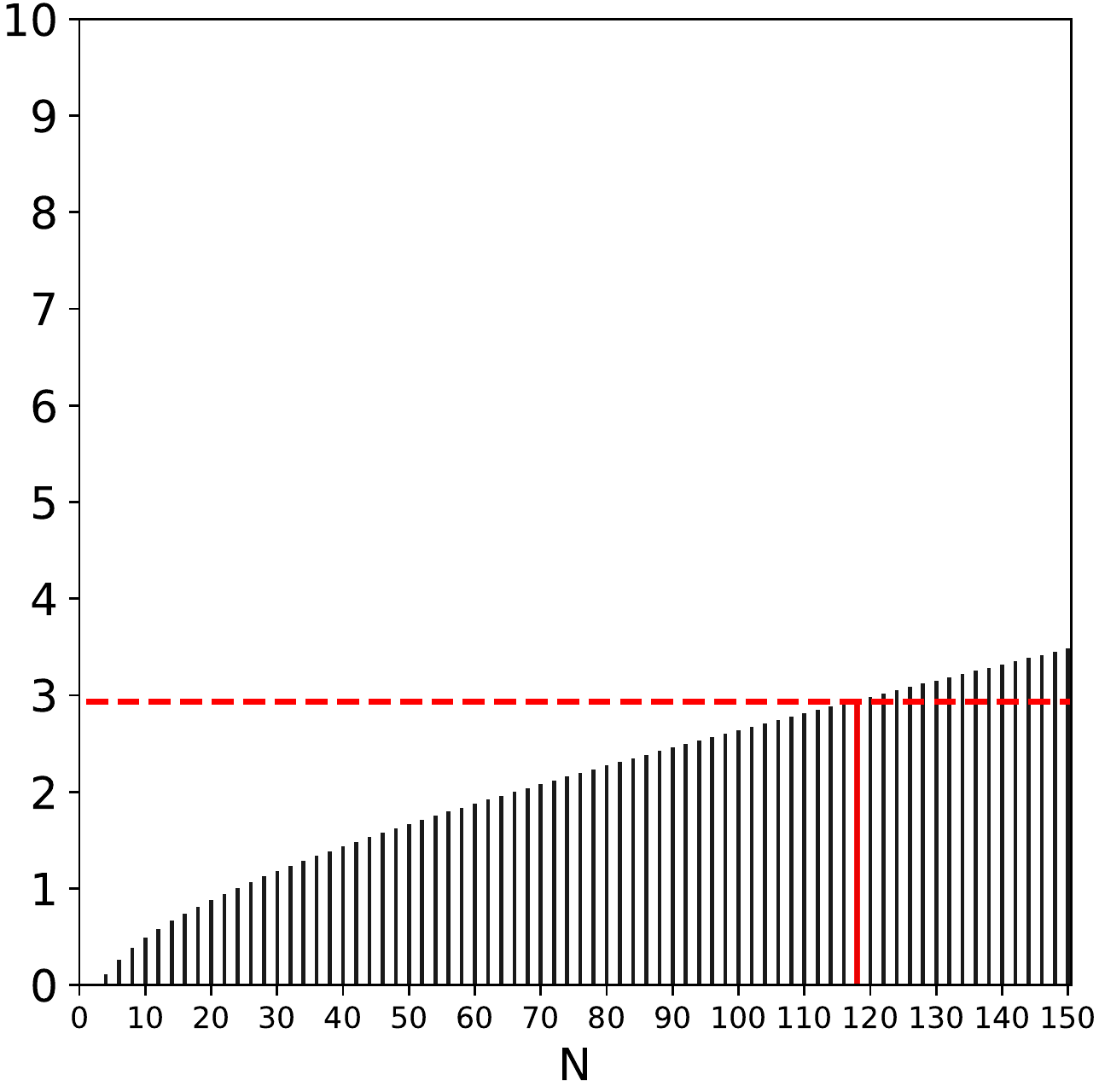}&
			\includegraphics[width=0.23\textwidth]{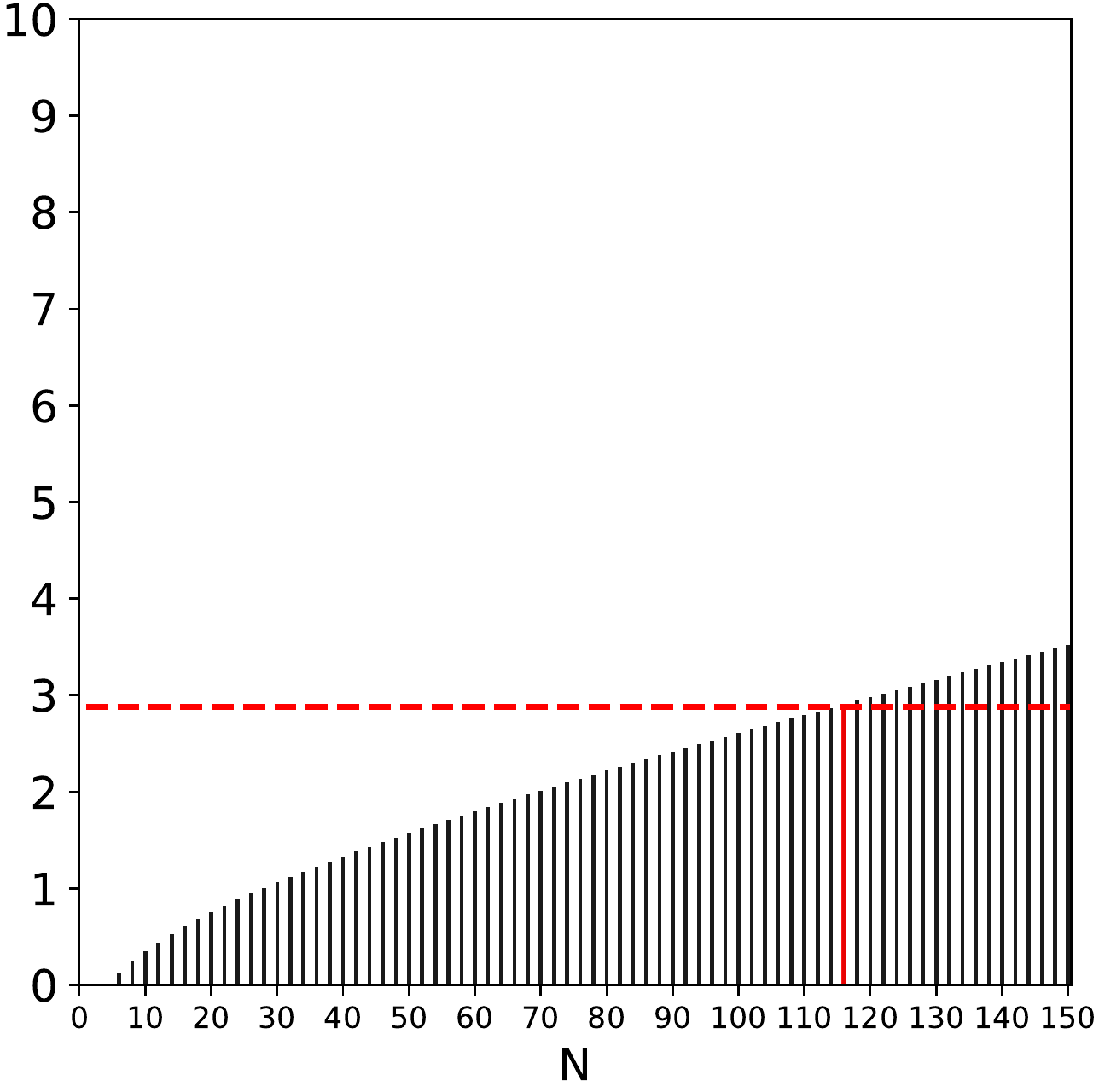} &	
			\includegraphics[width=0.23\textwidth]{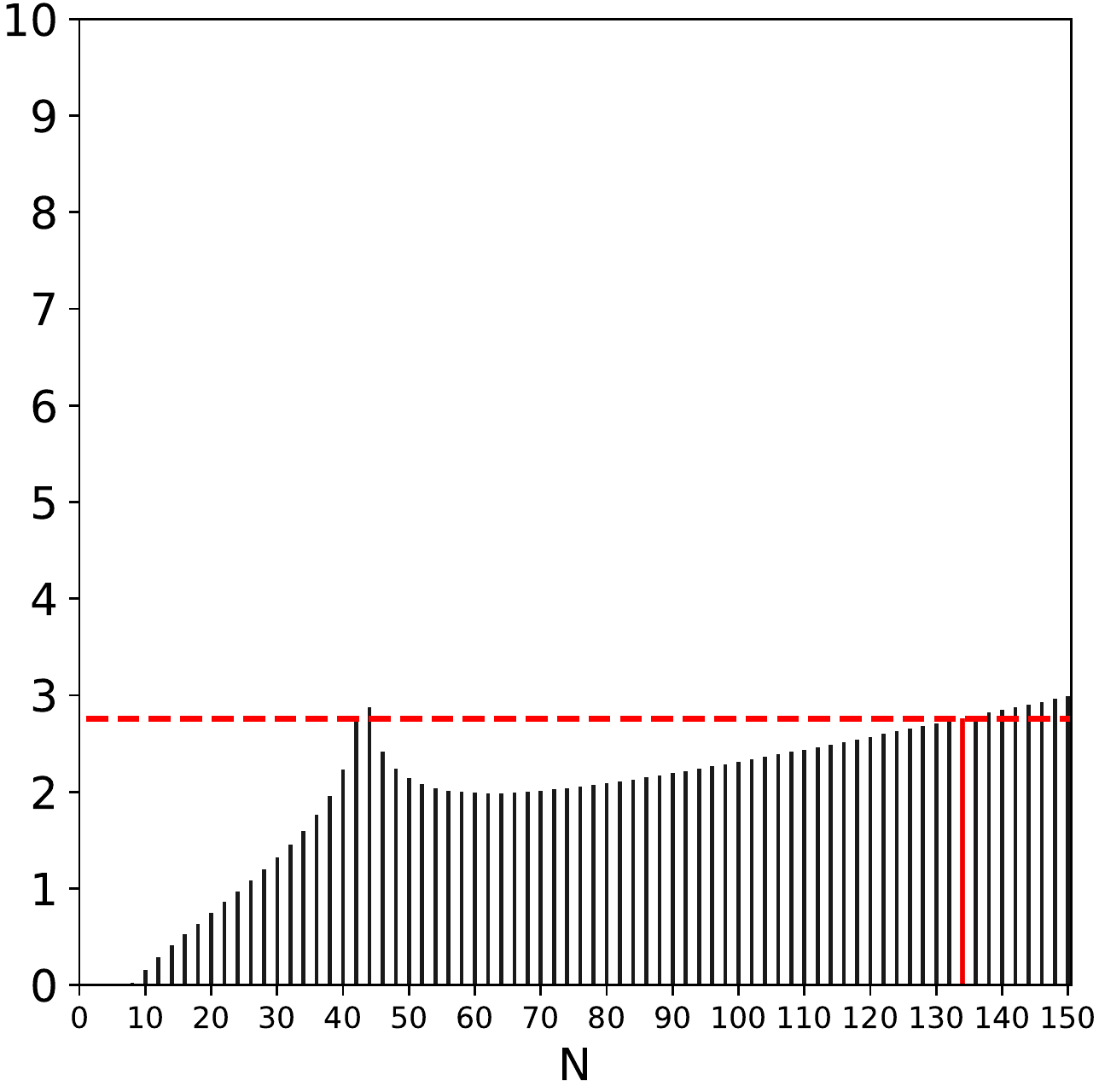} &
			\includegraphics[width=0.23\textwidth]{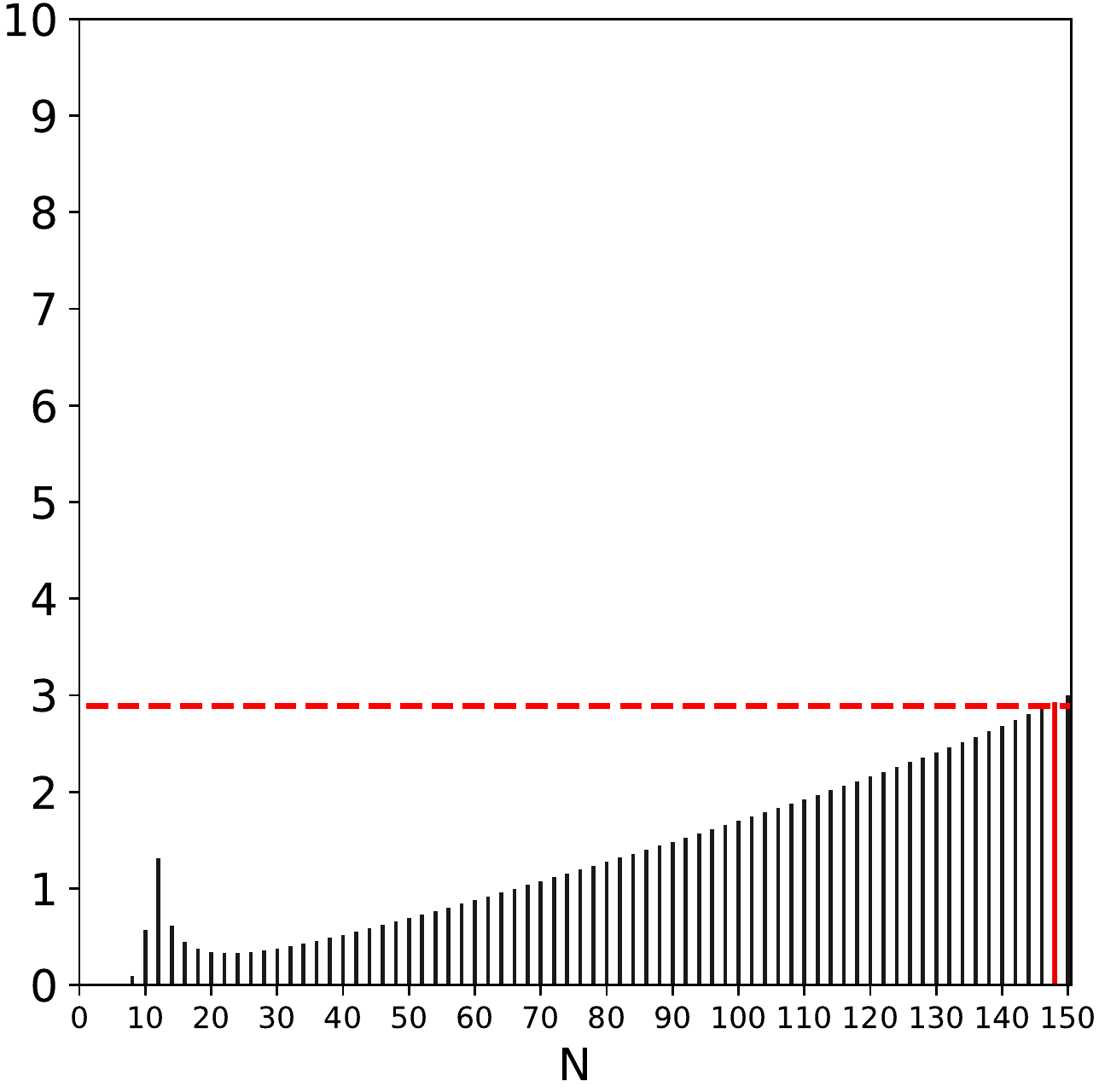}\\
		\end{tabular}
	}
	\caption{$\gamma_{a,b}^N$ as a function of $N$, when the underlying process is a BM with $\sigma_X(T;t)\approx 0.707$. The drift is $a=\mu_X(T;t)=0$. The dashed red horizontal lines indicate the accuracy of the MC method. The red vertical bars indicate when the Hermite series reaches the MC accuracy. \label{BMplot}}
\end{figure}

\begin{figure}[!tp]
	\setlength{\tabcolsep}{2pt}
	\resizebox{1\textwidth}{!}{
		\begin{tabular}{@{}>{\centering\arraybackslash}m{0.04\textwidth}@{}>{\centering\arraybackslash}m{0.24\textwidth}@{}>{\centering\arraybackslash}m{0.24\textwidth}@{}>{\centering\arraybackslash}m{0.24\textwidth}@{}>{\centering\arraybackslash}m{0.24\textwidth}@{}}
			& $\boldsymbol{K = 0.0}$&$\boldsymbol{K = 0.2}$ & $\boldsymbol{K = 0.6}$& $\boldsymbol{K = 1.0}$ \\
			\begin{turn}{90}$\boldsymbol{b =1.0}$\end{turn}&
			\includegraphics[width=0.23\textwidth]{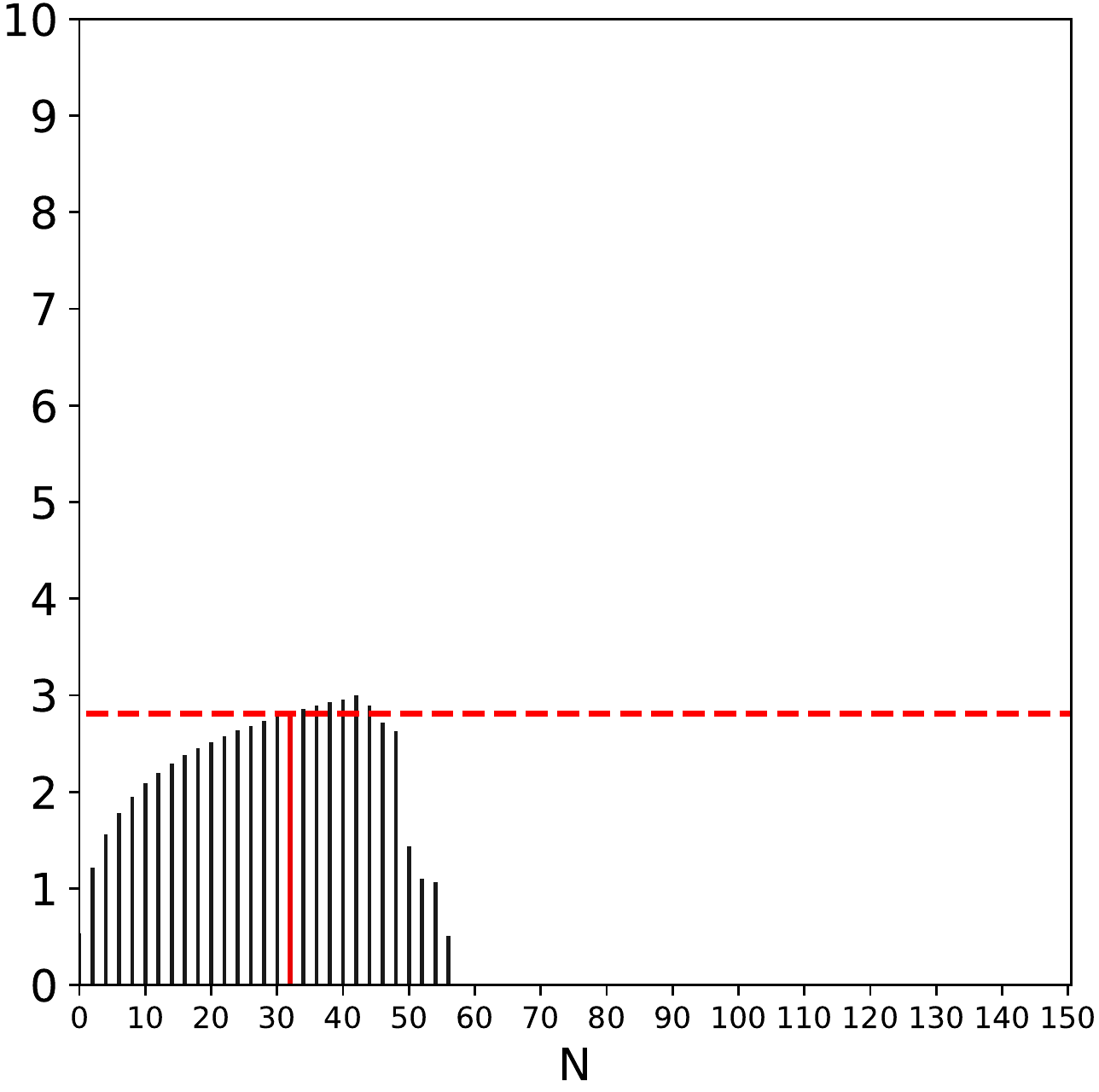}&
			\includegraphics[width=0.23\textwidth]{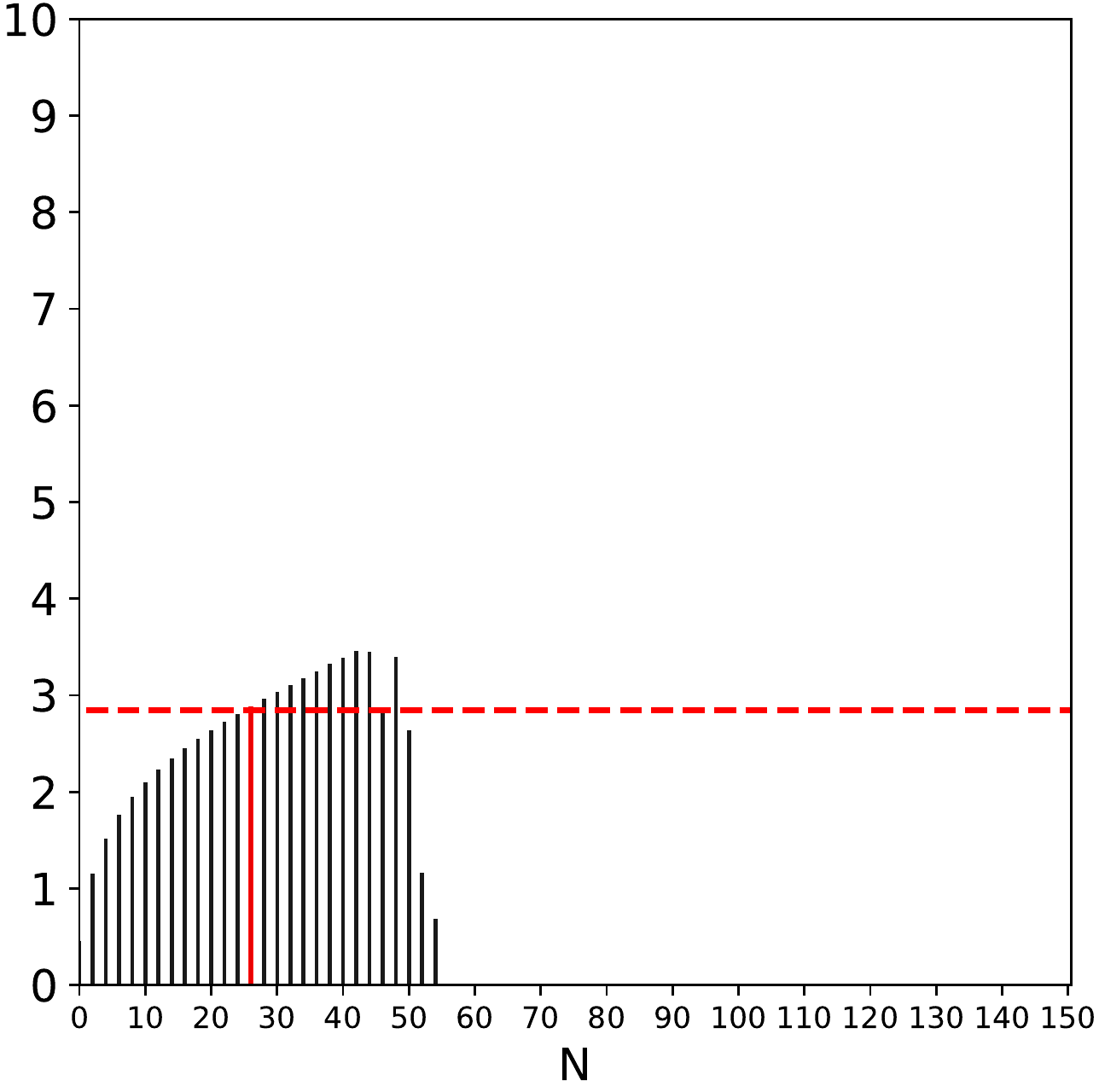} &
			\includegraphics[width=0.23\textwidth]{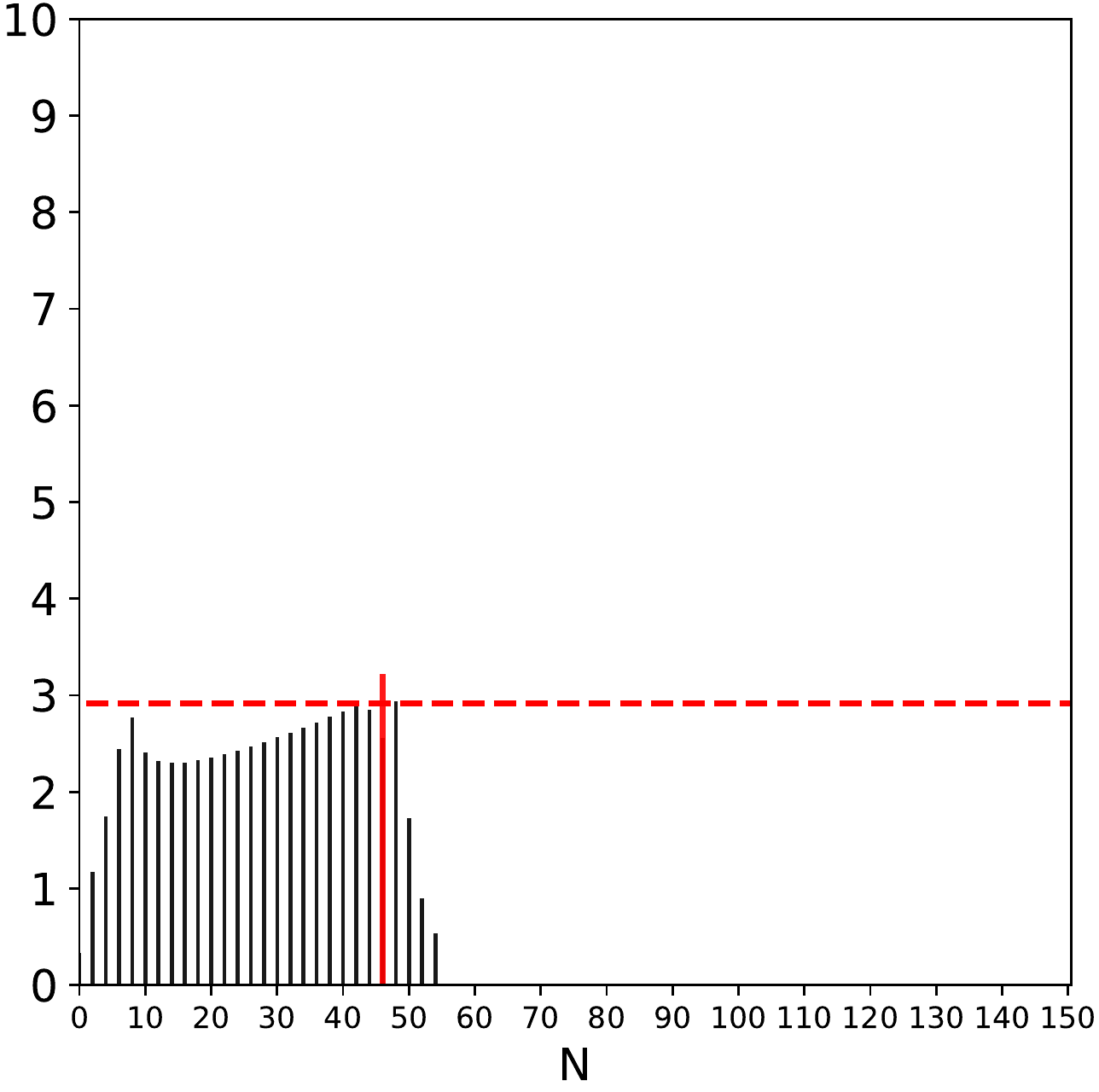} &
			\includegraphics[width=0.23\textwidth]{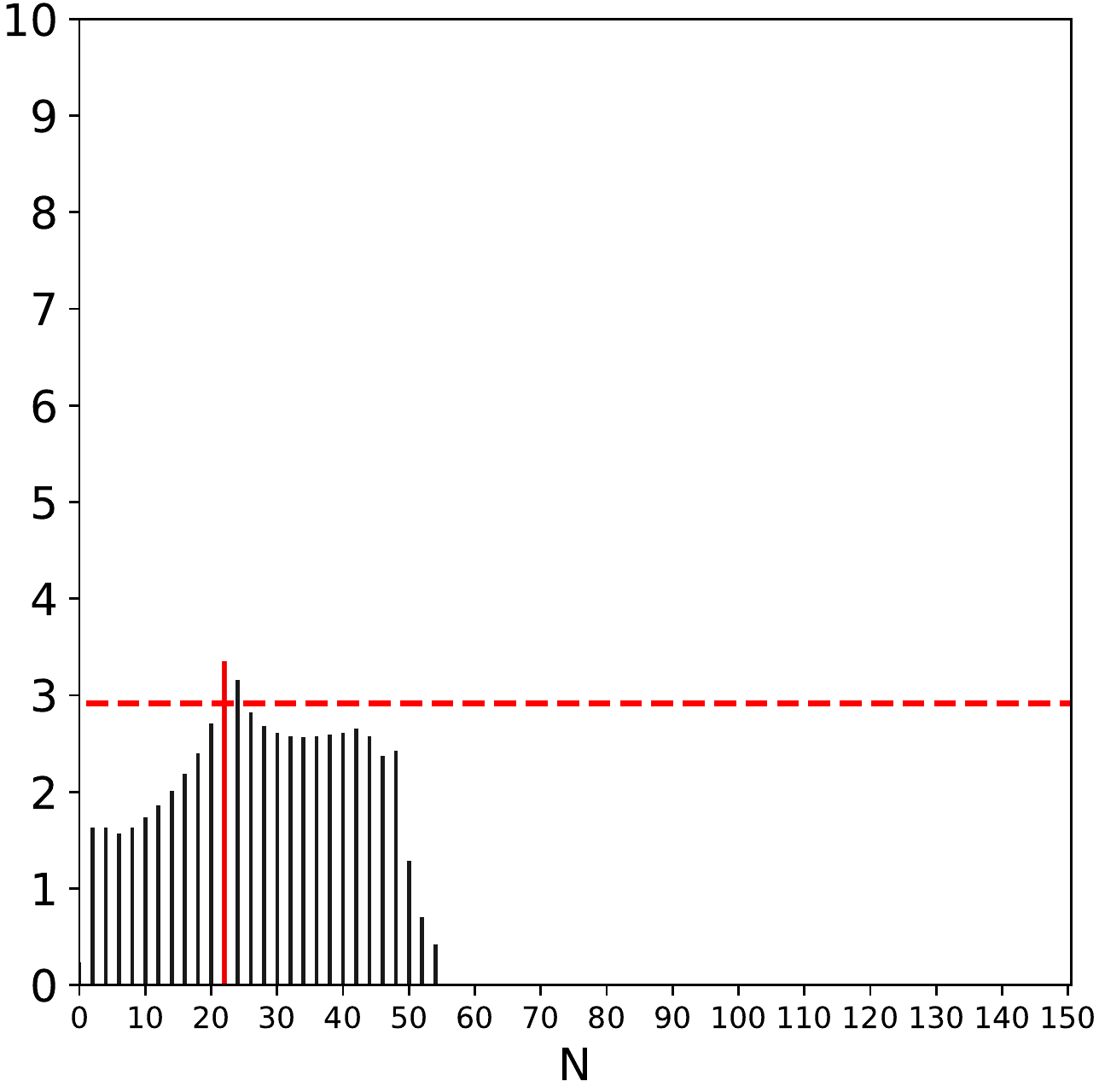}\\
			\begin{turn}{90}$\boldsymbol{b =1.2}$\end{turn}&
			\includegraphics[width=0.23\textwidth]{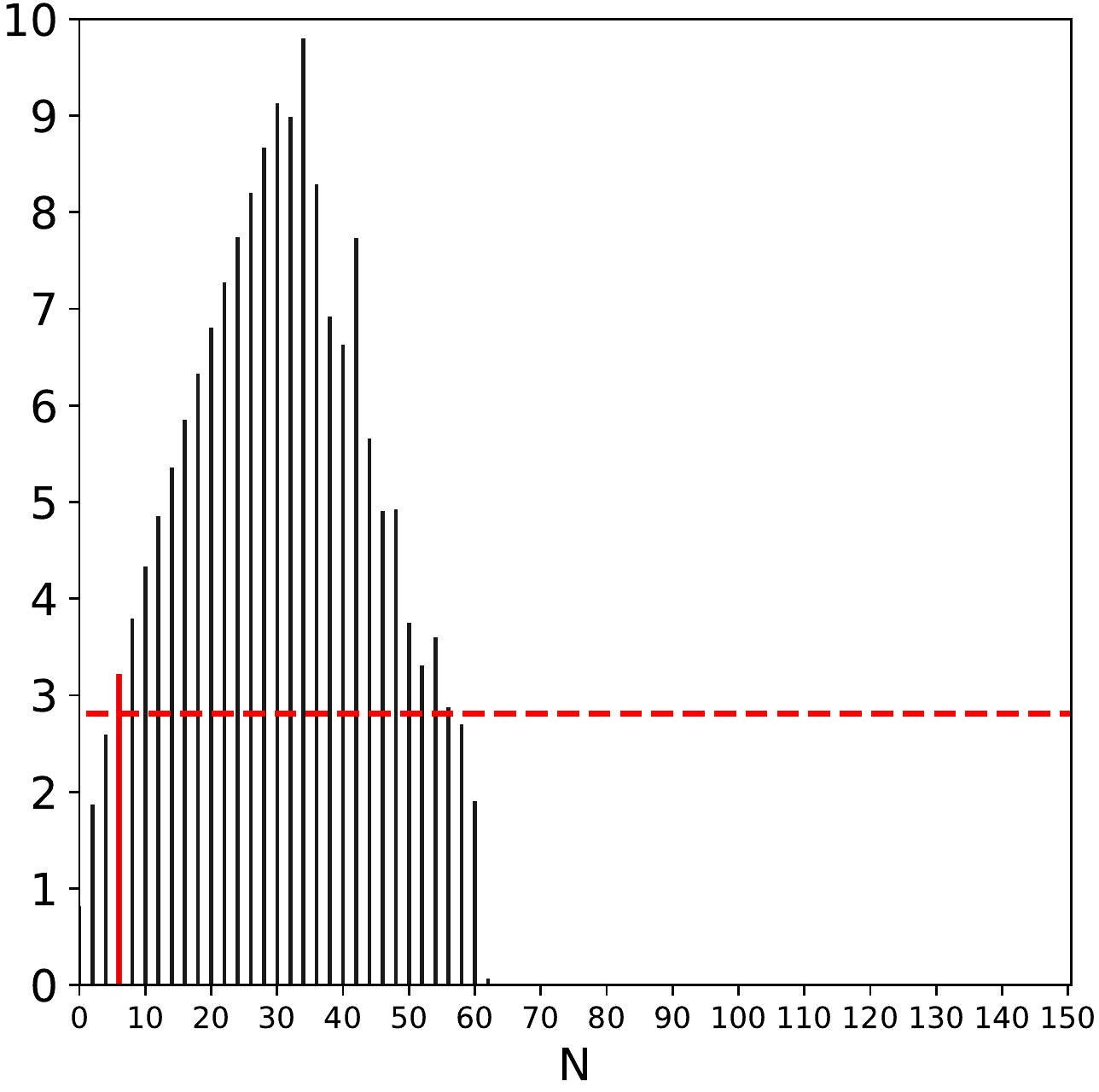}&
			\includegraphics[width=0.23\textwidth]{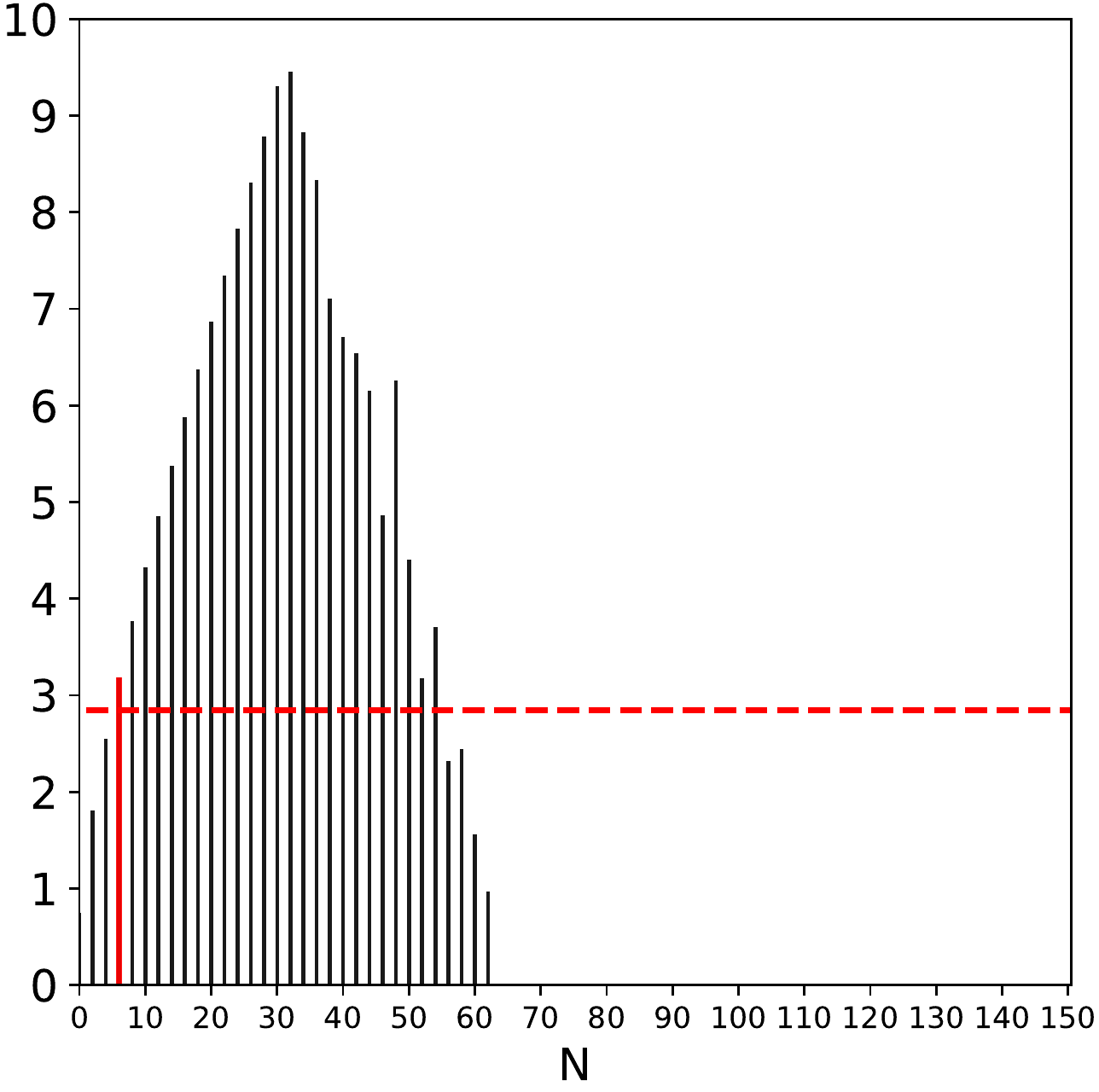} &
			\includegraphics[width=0.23\textwidth]{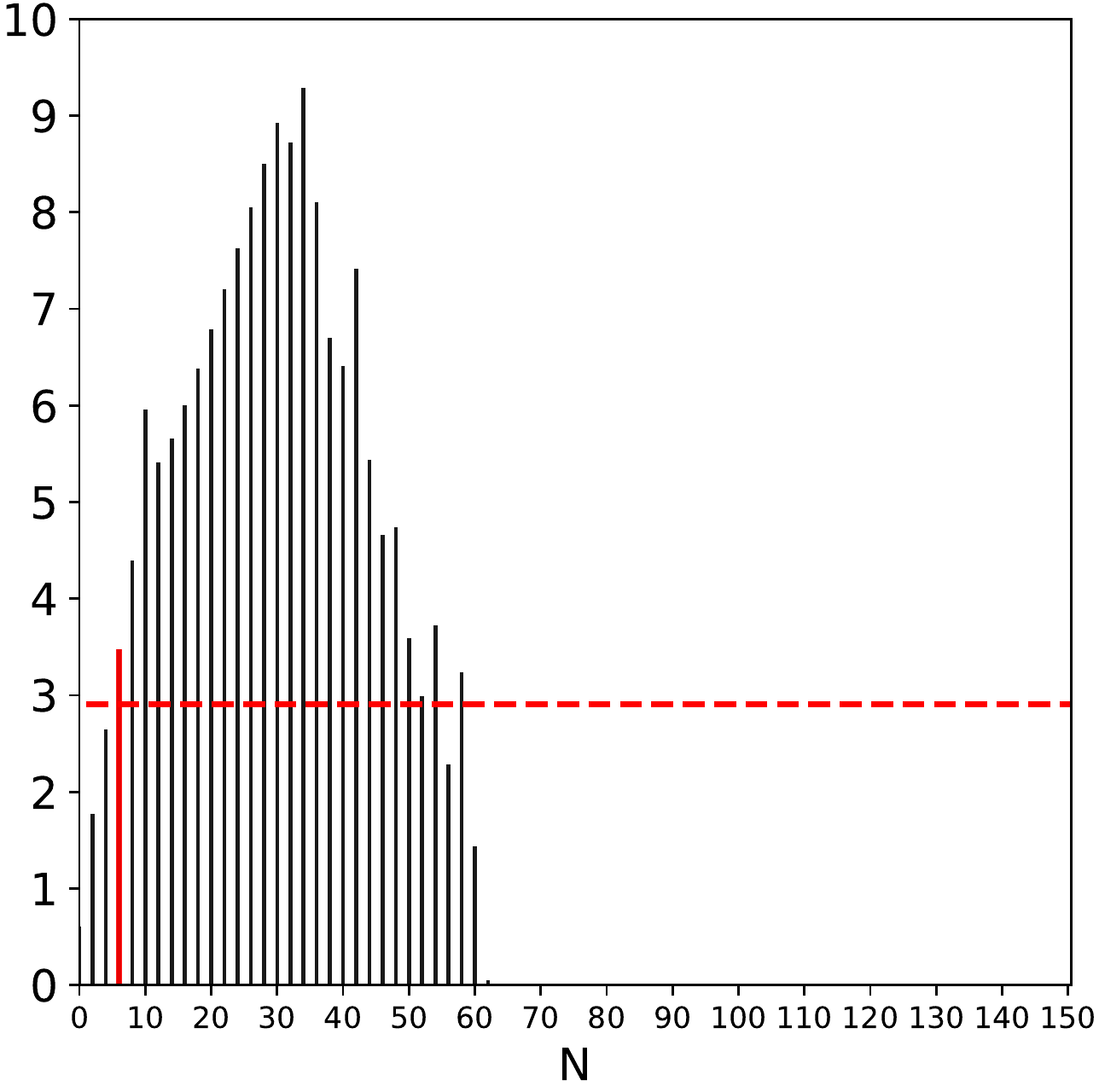} &
			\includegraphics[width=0.23\textwidth]{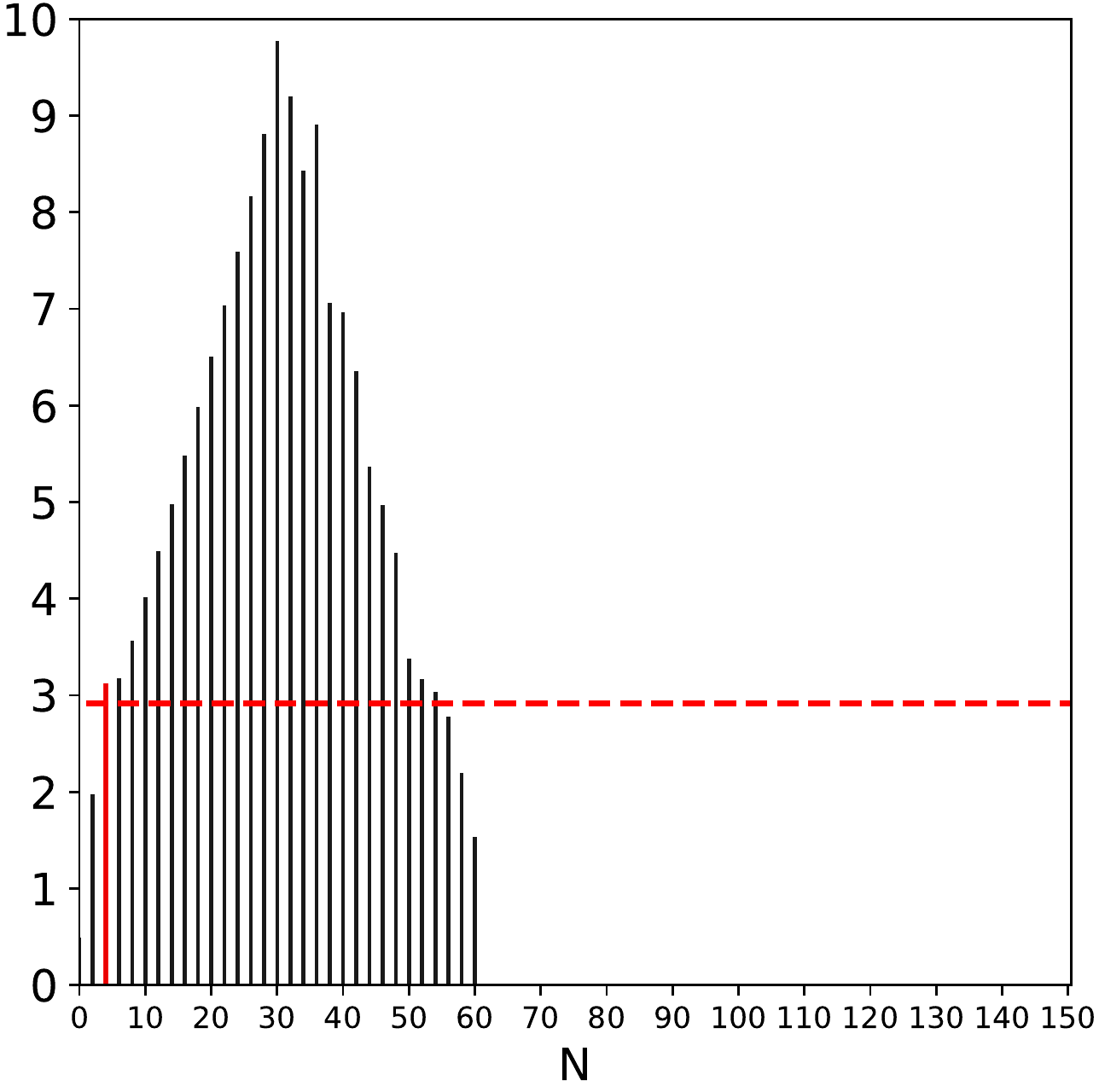}\\
			\begin{turn}{90}$\boldsymbol{b =2.0}$\end{turn}&
			\includegraphics[width=0.23\textwidth]{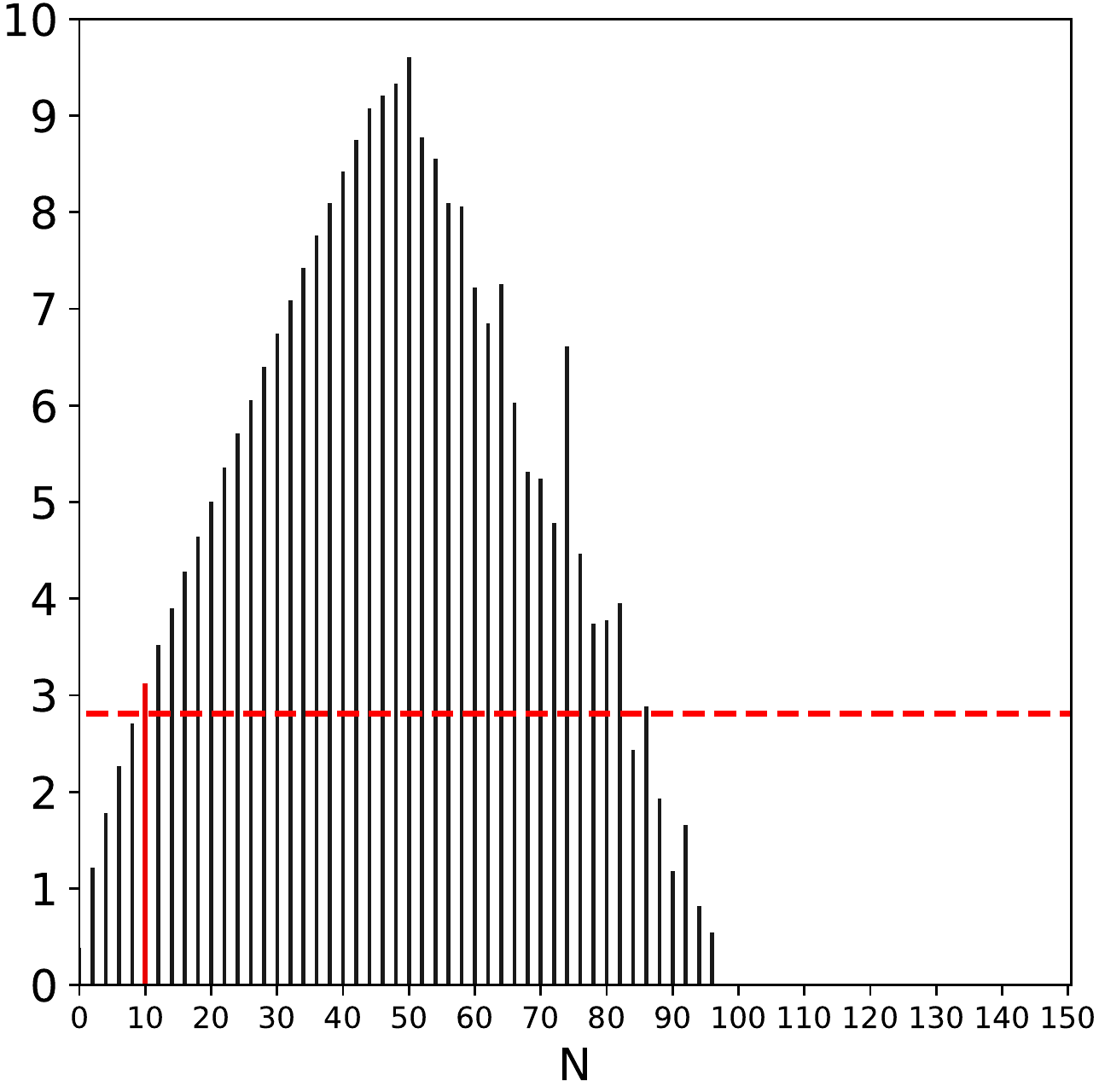}&
			\includegraphics[width=0.23\textwidth]{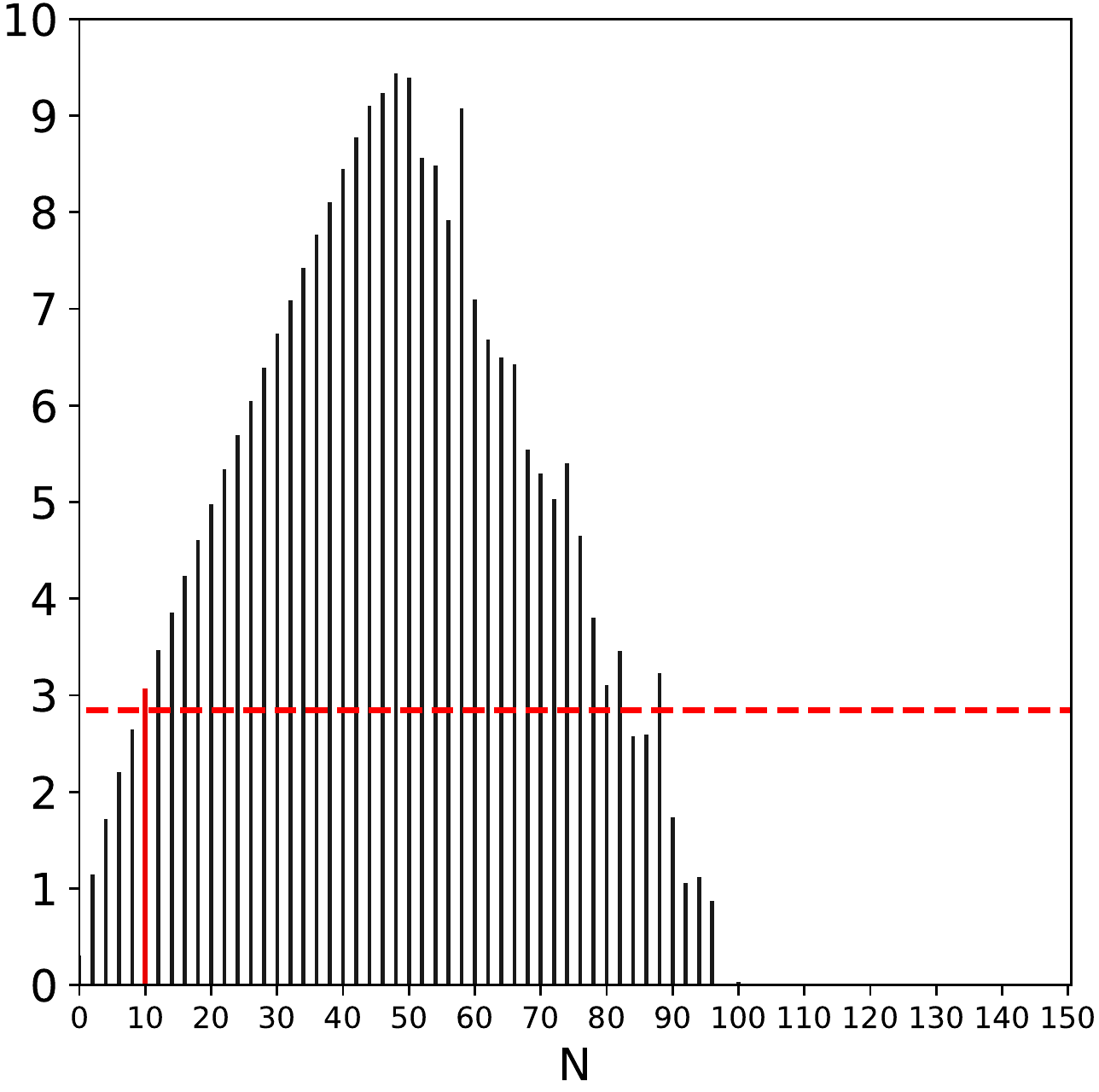} &
			\includegraphics[width=0.23\textwidth]{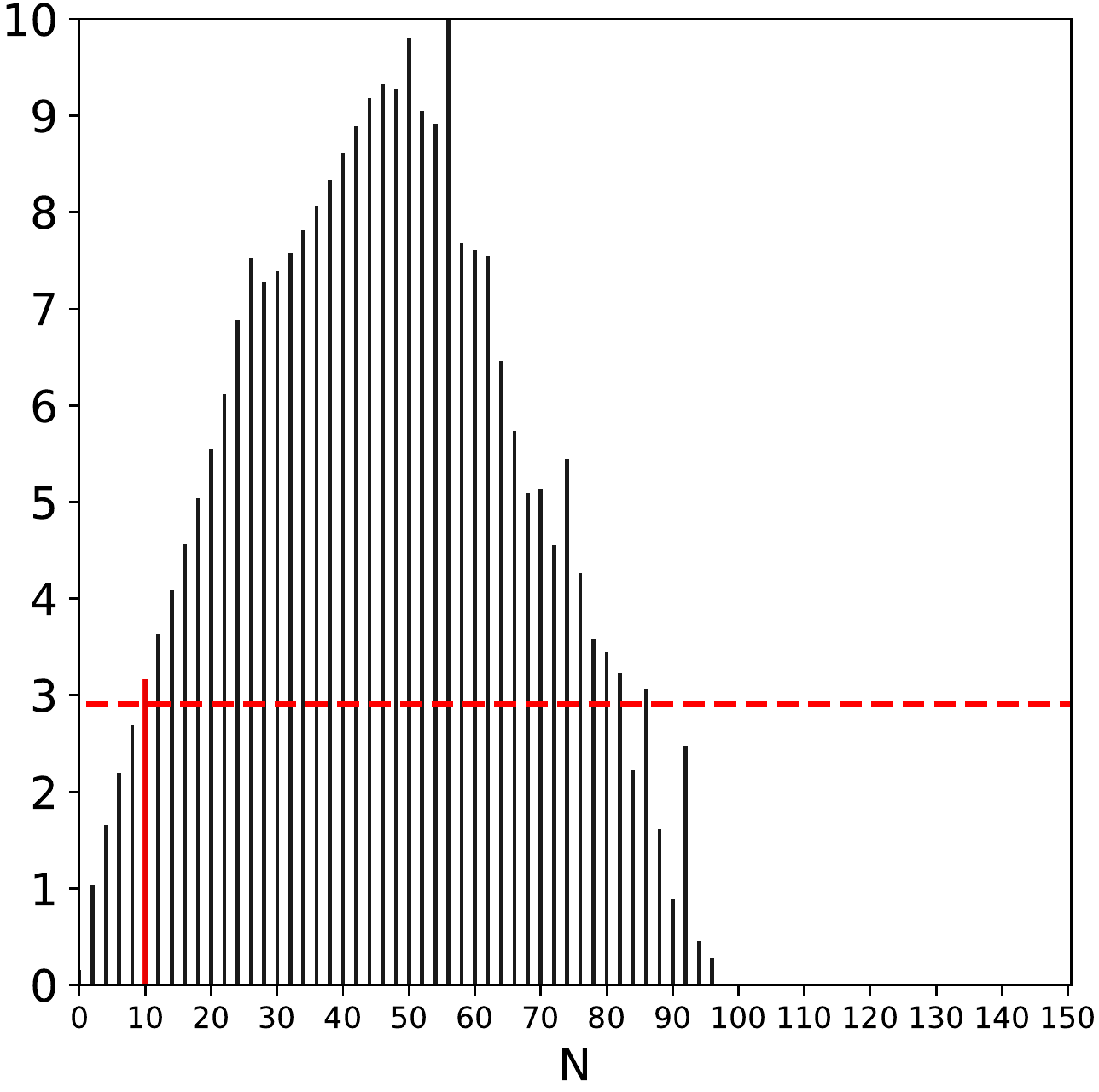} &
			\includegraphics[width=0.23\textwidth]{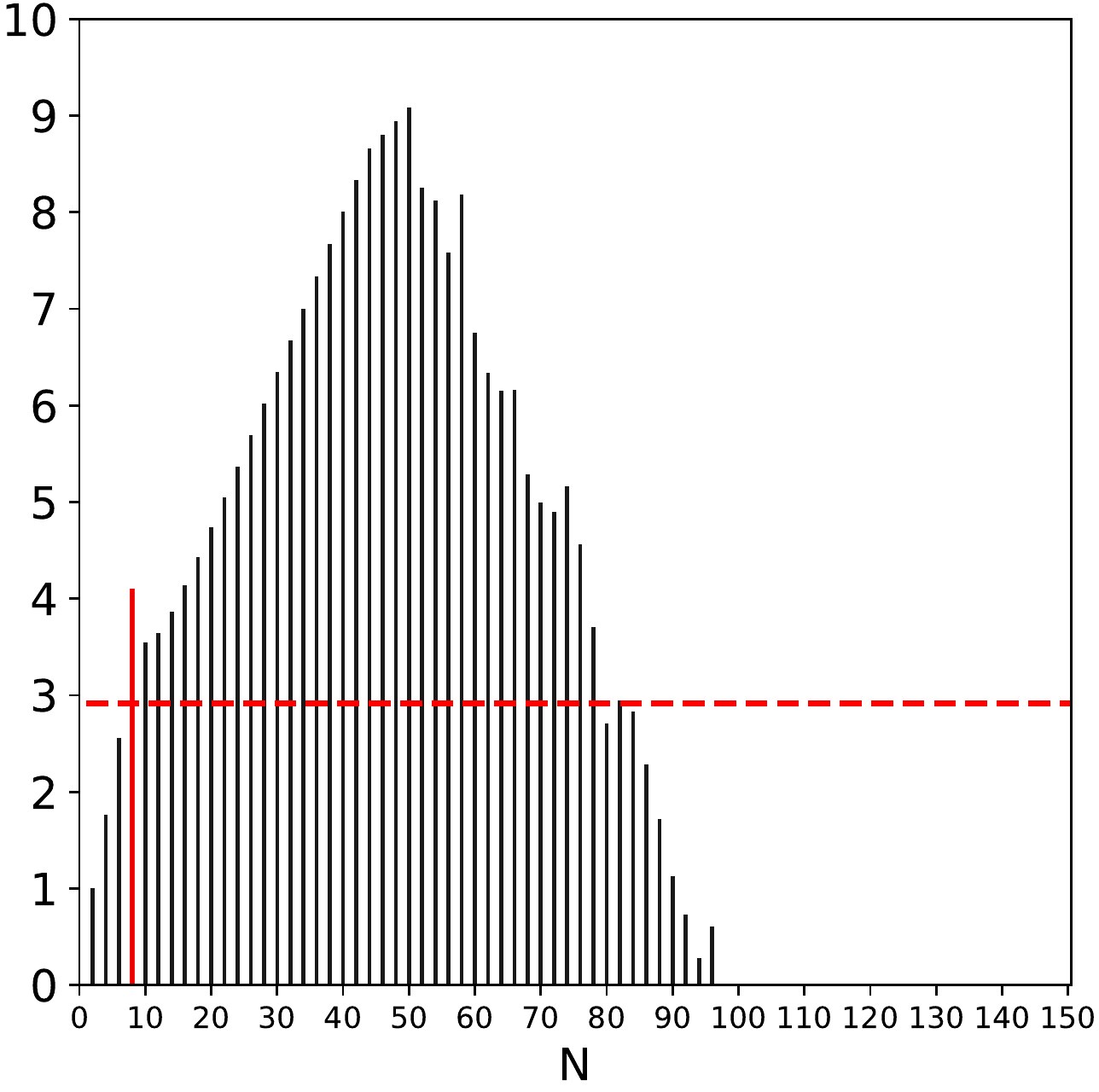}\\
			\begin{turn}{90}$\boldsymbol{b =4.0}$\end{turn}&
			\includegraphics[width=0.23\textwidth]{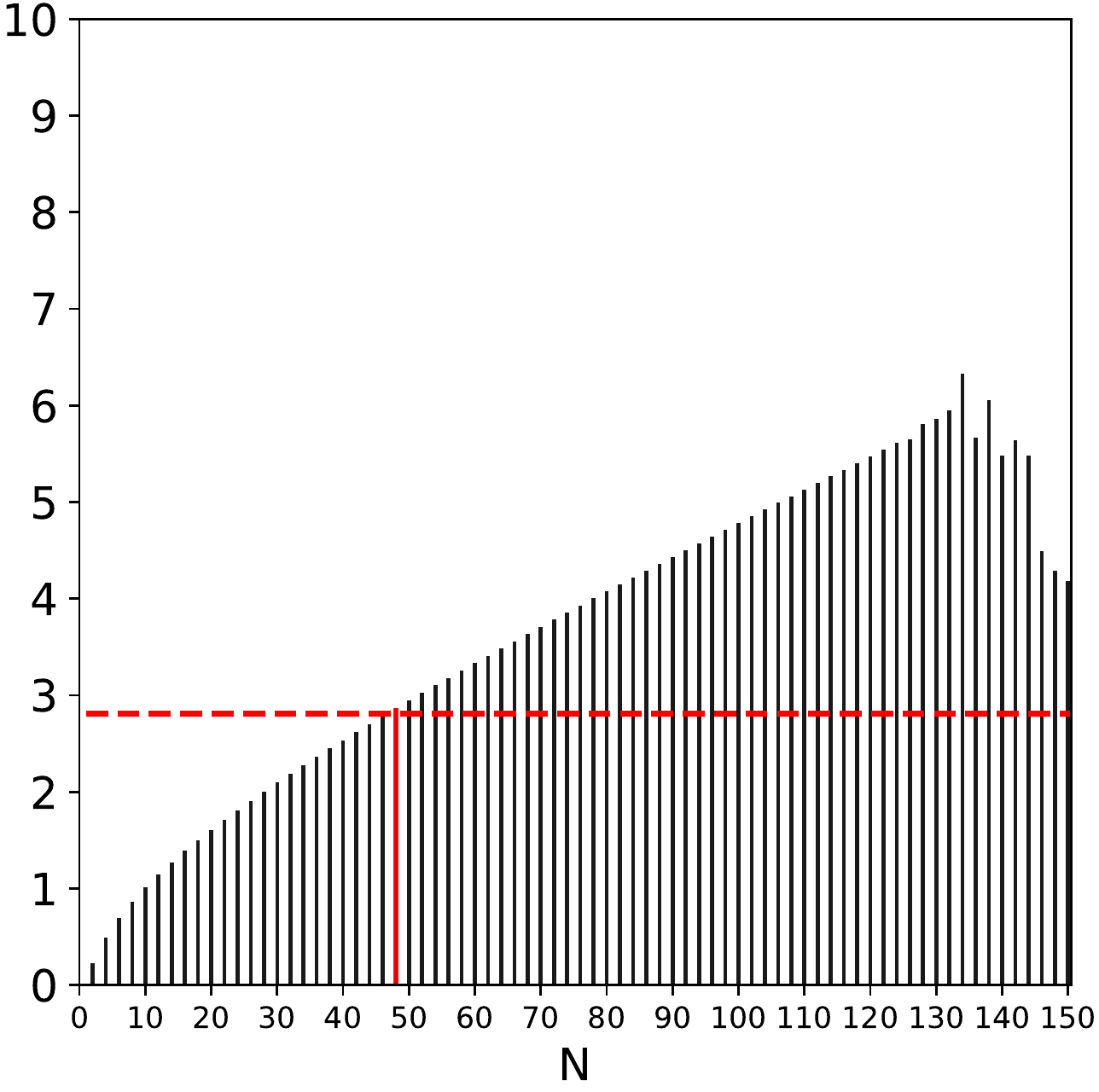}&
			\includegraphics[width=0.23\textwidth]{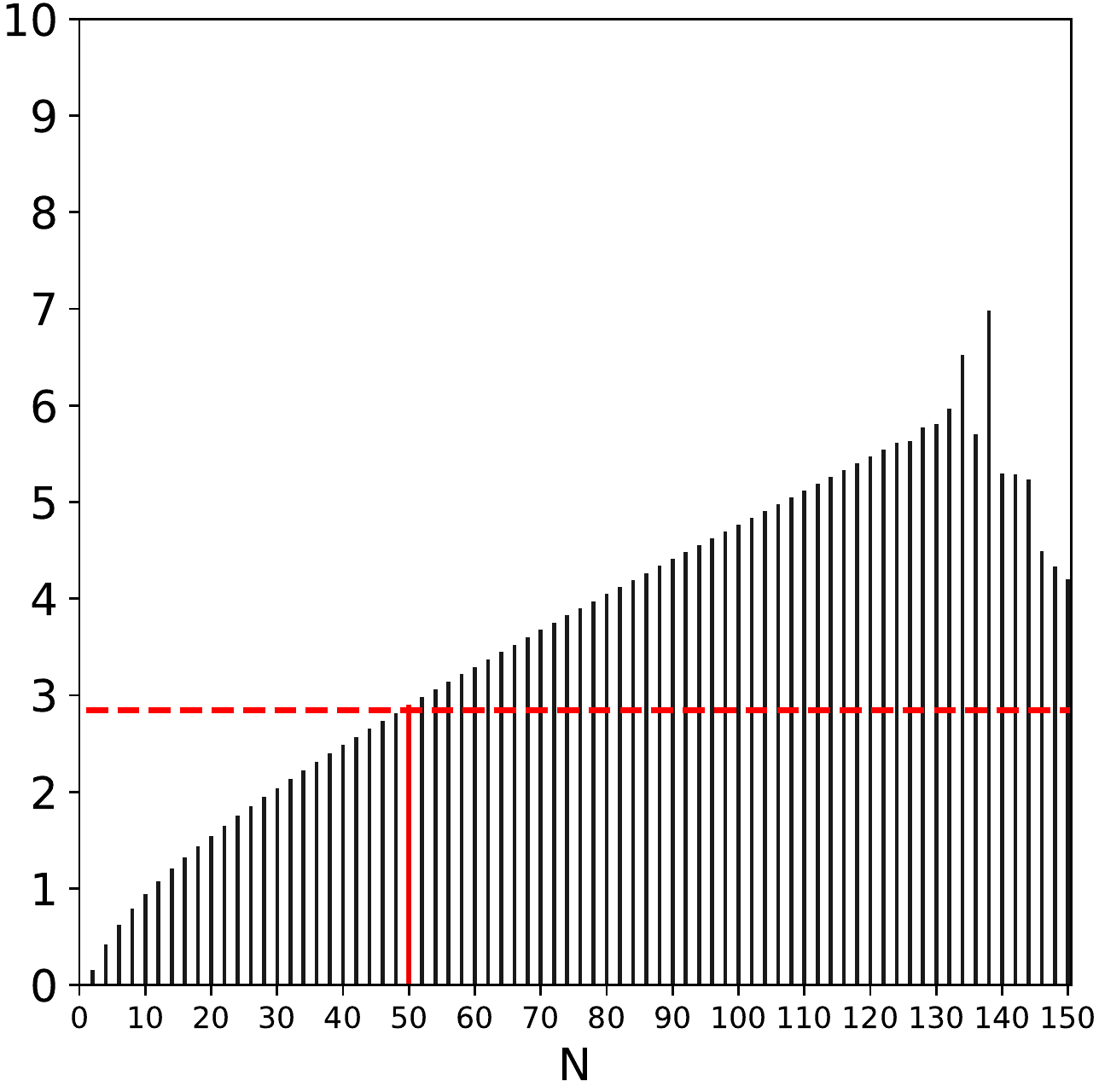} &
			\includegraphics[width=0.23\textwidth]{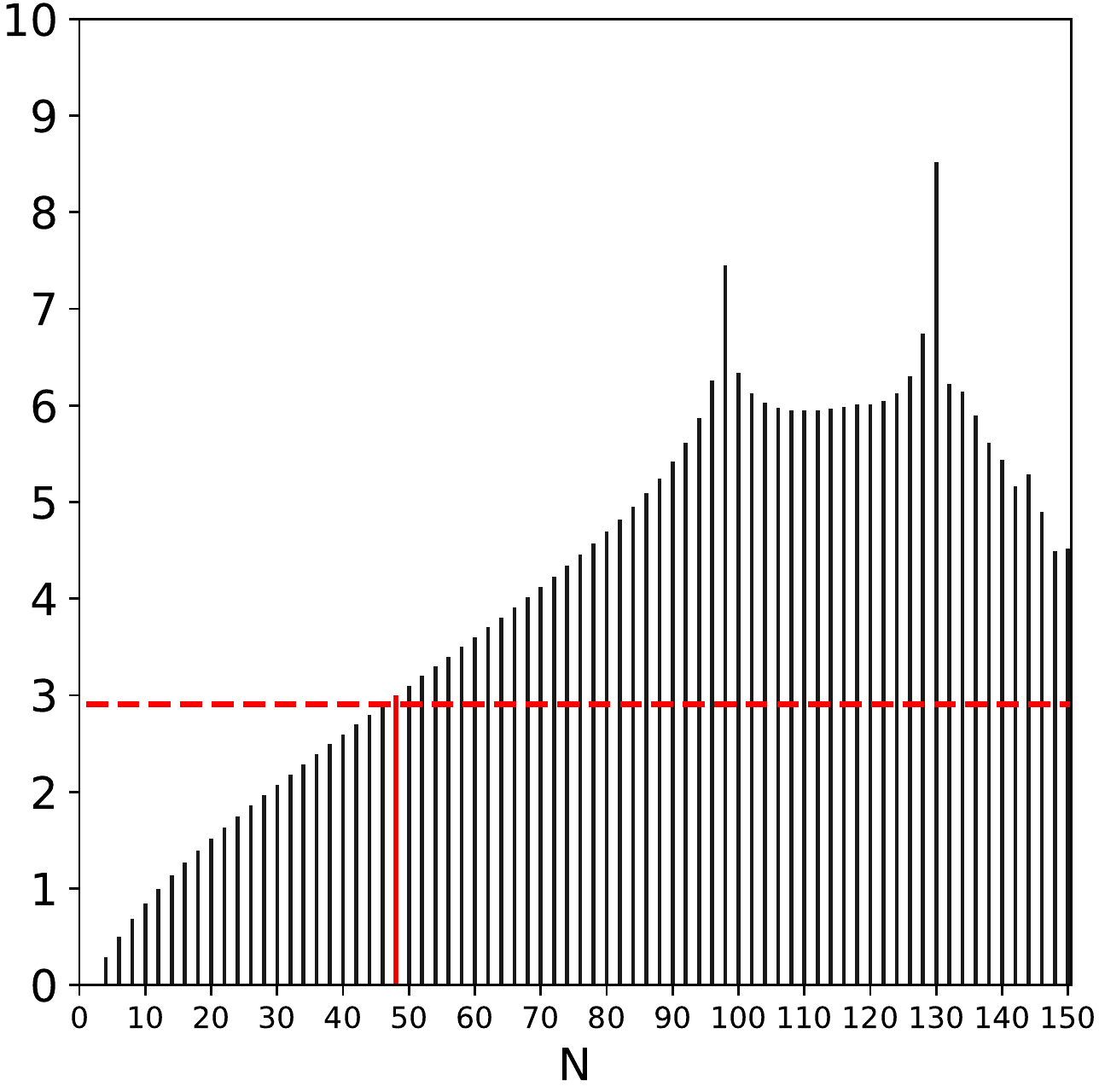} &
			\includegraphics[width=0.23\textwidth]{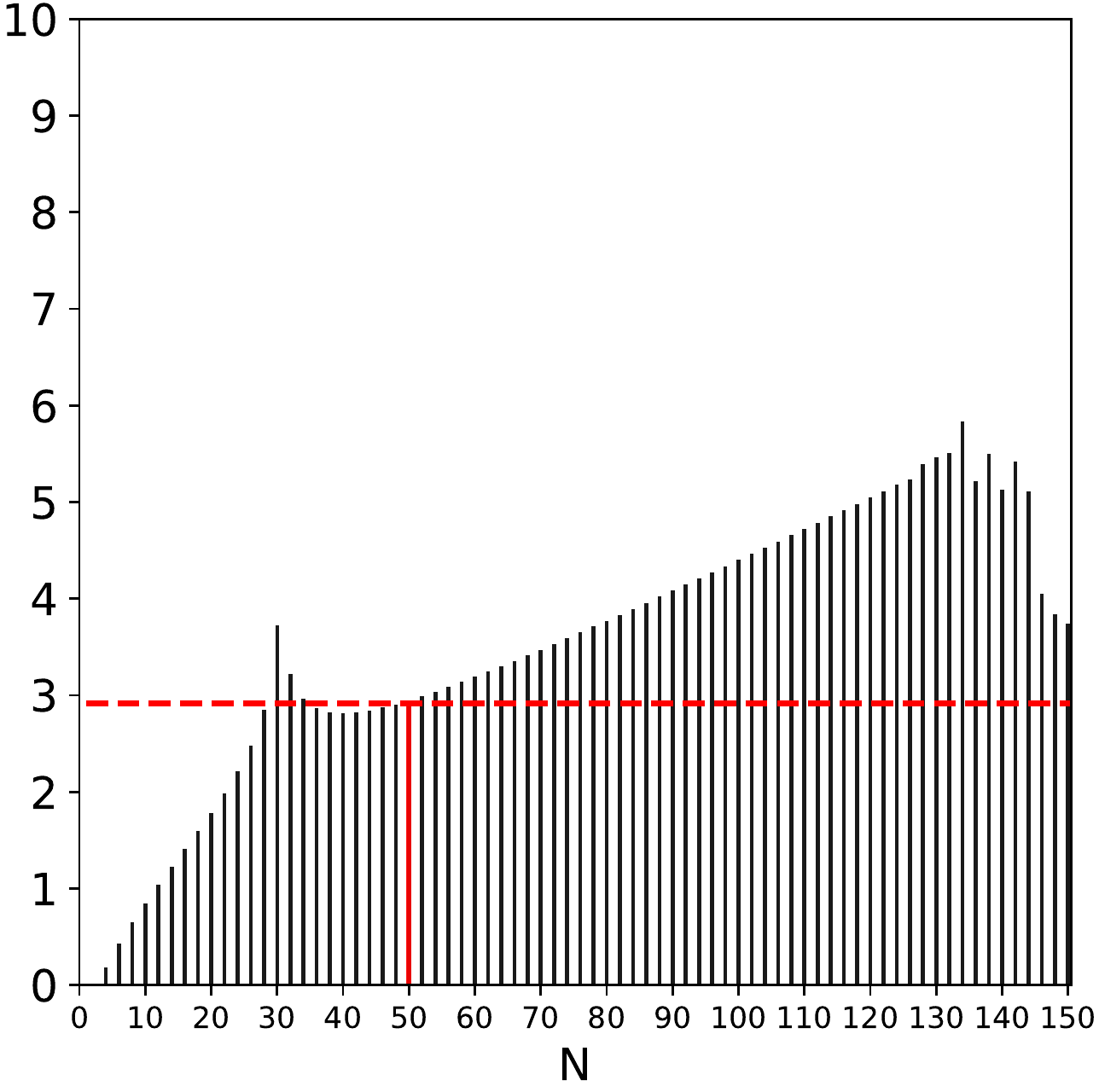}\\
			\begin{turn}{90}$\boldsymbol{b =6.0}$\end{turn}&
			\includegraphics[width=0.23\textwidth]{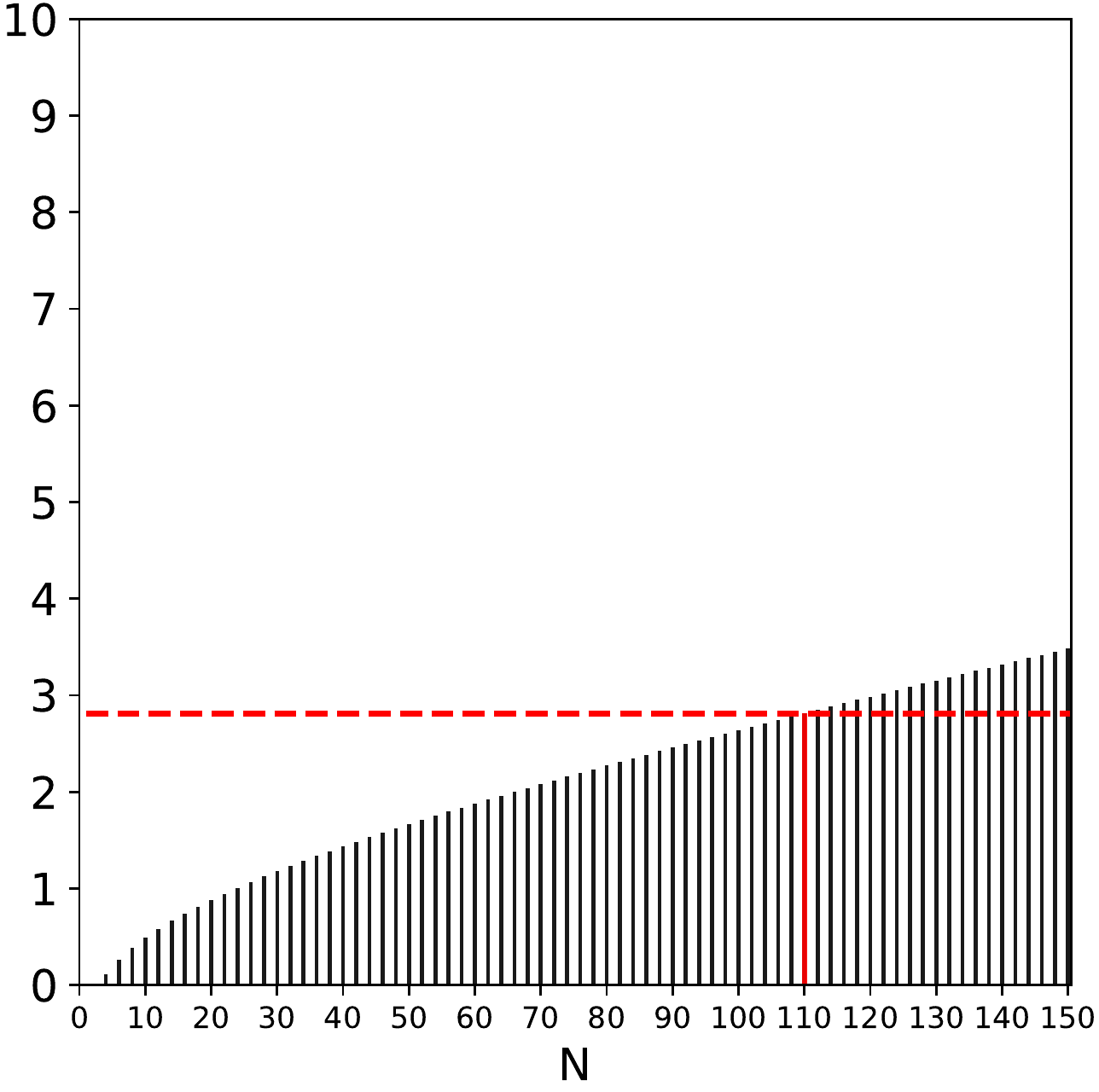}&
			\includegraphics[width=0.23\textwidth]{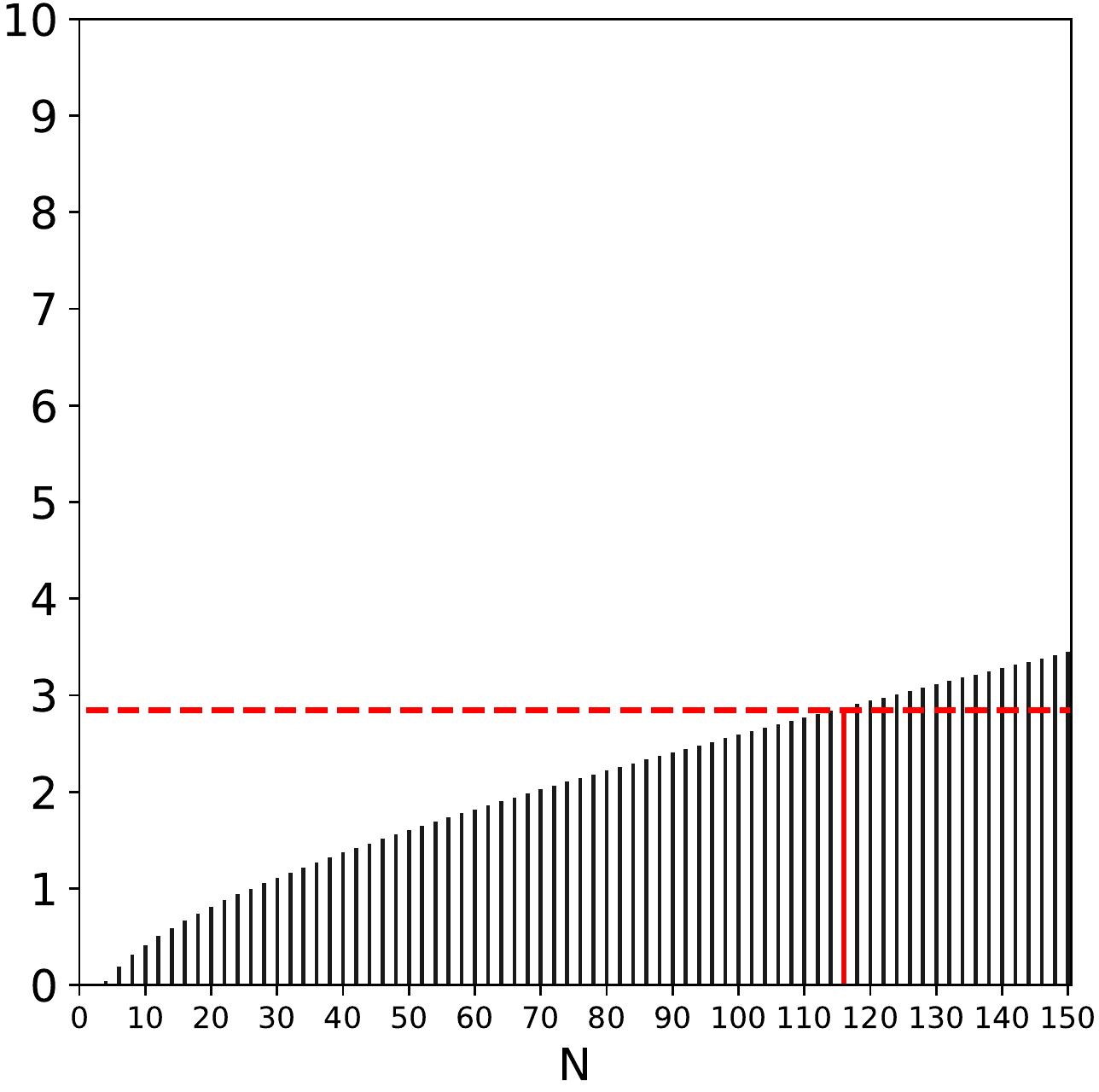} &
			\includegraphics[width=0.23\textwidth]{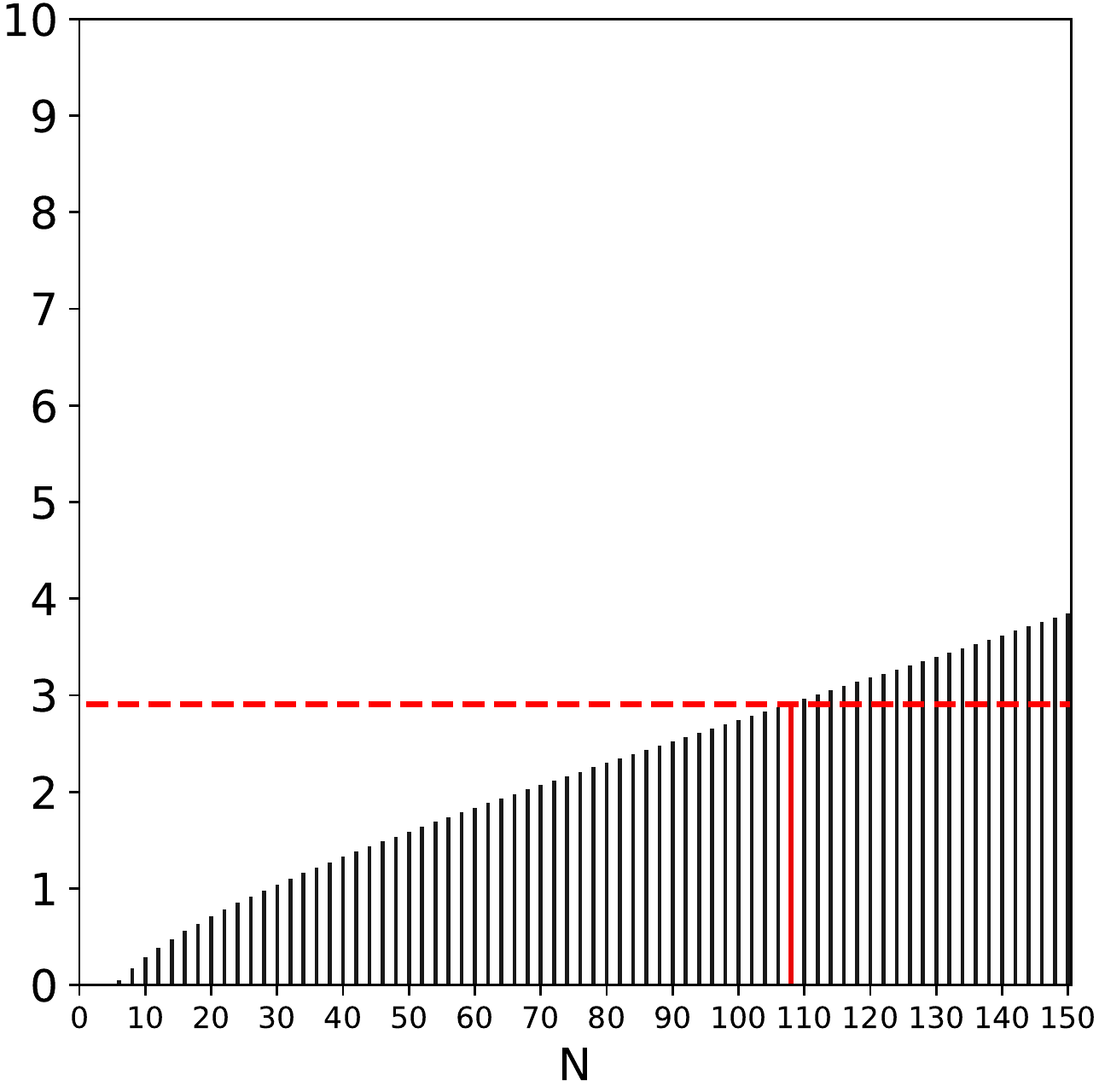} &
			\includegraphics[width=0.23\textwidth]{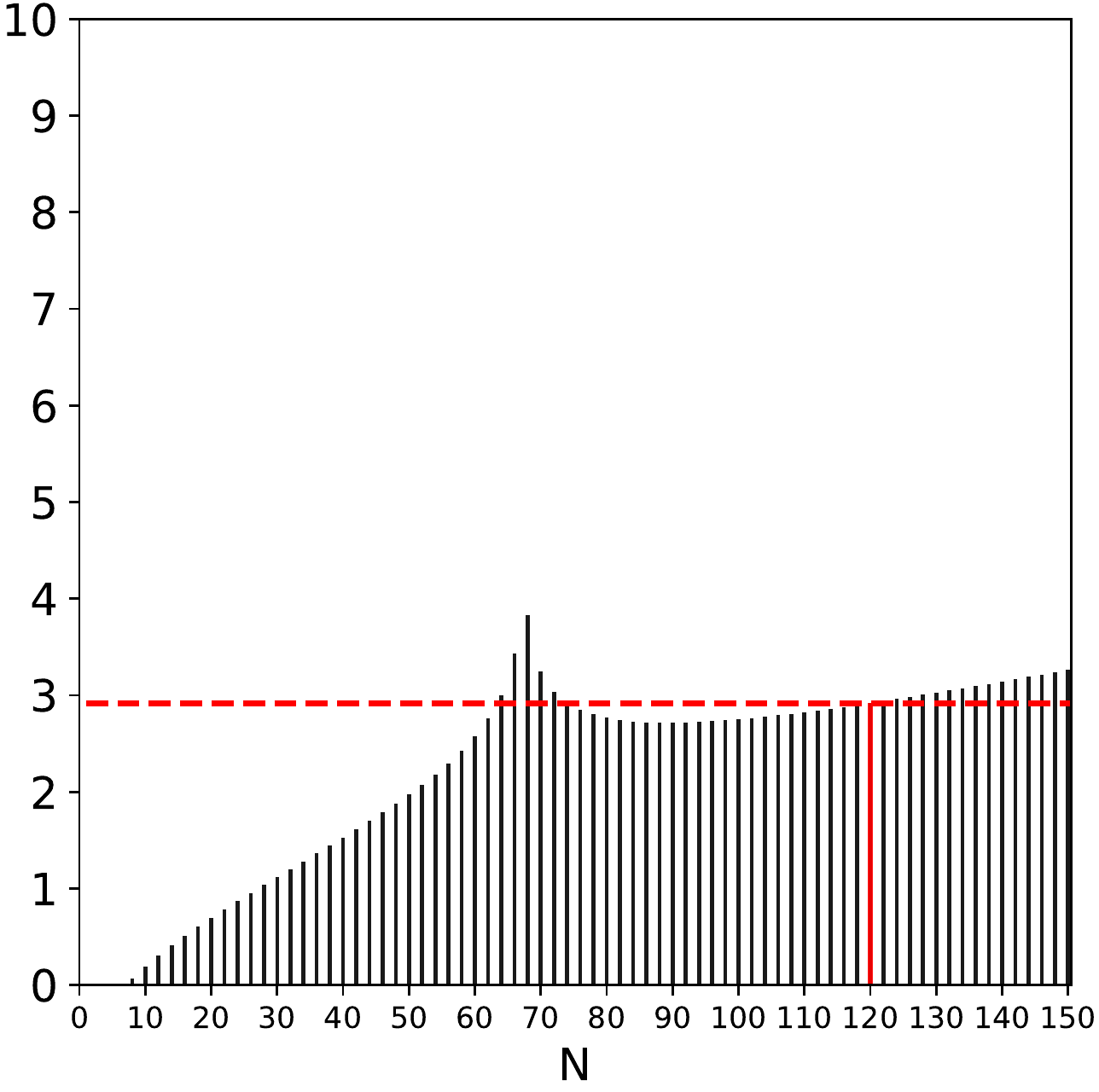}\\
		\end{tabular}
	}
	\caption{$\gamma_{a,b}^N$ as a function of $N$, when the underlying process is a BM with $\sigma_X(T;t)\approx 1.41$. The drift is $a=\mu_X(T;t)=0$. The dashed red horizontal lines indicate the accuracy of the MC method. The red vertical bars indicate when the Hermite series reaches the MC accuracy. \label{BMplot2}}
\end{figure}

In Figure \ref{BMplot} we have $T=1/2$ and $t=0$, hence $\sigma_X(T;t)\approx 0.707$. We observe that for $b=0.5$ the Hermite series is barely able to reach the accuracy of the Monte Carlo simulations, and it is not clear if we can actually consider it converging. This coincides indeed with the case $b = \underline{b}_{\sigma}$ in Proposition \ref{condb} which is expected to show instabilities. Things get better for $b=0.6$ and $b=1.0$: here we observe convergence of the Hermite approximation, reaching a level of accuracy of order $10^{-10}$. This convergence is slower for $b=1.0$ than for $b=0.6$. For $b=2.0$ and $b=3.0$ the convergence is even slower and to reach the same accuracy of the Monte Carlo simulations we need more than $50$ terms in the first case and more than $100$ terms in the second case, meaning that a bigger scale slows down the convergence. 

Very similar comments hold for Figure \ref{BMplot2}, where $T=2$, $t=0$ and $\sigma_X(T;t)\approx 1.41$. Here the plots show a similar behaviour to the ones in Figure  \ref{BMplot}, even if the values of the scale $b$ considered are different. More precisely, for Figure \ref{BMplot2} we consider values for $b$ which are two times (i.e. the double) the values used in Figure \ref{BMplot}. Since from formula  \eqref{Cab}, $C_{a,b}$ is proportional to $\frac{b}{\sigma}$, it seems reasonable to think that this phenomenon is related to the fact that the standard deviation of this second Brownian motion is exactly two times the standard deviation of the Brownian motion in Figure \ref{BMplot}. In other words, because of the ratio $\frac{b}{\sigma}$ that somehow controls the approximation error, in order to get the same accuracy we need to keep this ratio constant. Hence, if $\sigma$ doubles, also $b$ must double. Finally, we notice that the singularity in the sense of Proposition \ref{condb} is here expected for $b=1.0$, as indeed Figure \ref{BMplot2} shows.

Another phenomenon observable in both Figure \ref{BMplot} and \ref{BMplot2} is that, after reaching the best accuracy, the bars in the plots decrease. Moreover, some parts of the plots are empty, as for example, in the plot corresponding to $K=0.0$ and $b=1.0$ of Figure \ref{BMplot}, after $N=60$. This is because, after that, the values of $\gamma_{a,b}^N$ become negative, hence they don't appear in the plot. A negative $\gamma_{a,b}^N$ means in particular that the value of $\Pi_{K,N}^{a,b}(t)$ is completely far away from the true price value. We believe that this is due to numerical instabilities. In computing the approximation $\Pi_{K,N}^{a,b}(t)$ in Theorem \ref{price}, we need indeed the conditional moments of $X(T)$. It is clear that for high values of the truncation number $N$, we need then to calculate high-order moments of $X$, which create numerical instabilities due to rounding errors.

\subsection{Gaussian Ornstein--Uhlenbeck process}
We consider $X=Y$, where $Y$ is the Gaussian Ornstein--Uhlenbeck (OU) process defined by
\begin{equation}
	\label{OU}
	dY(t) = (b_0+b_1Y(t))\,dt + \sqrt{\sigma_0}\,dB(t),
\end{equation}
for $b_0, b_1, \sigma_0\in \R$ and $\sigma_0>0$. Then $\left. X(T)\right|\F_t\sim \mathcal{N}(\mu_X(T;t), \sigma_X^2(T;t) )$ with $$\mu_X(T;t) = X(t) \,e^{b_1(T-t)} + \frac{b_0}{b_1}\left(e^{b_1(T-t)}-1\right) \quad \mbox{ and } \quad \sigma_X^2(T;t) :=  \frac{\sigma_0}{2b_1}\left(e^{2b_1(T-t)}-1\right).$$
Moreover, since $Y$ is a polynomial process (thus $X$ is a polynomial process), the moments of $X(T)$ are given by Corollary \ref{prop2} and the price functional $\Pi_K(t)$ is again given in closed form by equation \eqref{closed_price}.

In Figure \ref{OUplot} and \ref{OUplot2} we report the numerical results for $(b_0, b_1, \sigma_0) = (-0.02,  0.01,  0.98)$, $T=2$ and $t=0$, which have been chosen to get $\mu_X(T;t)=X(0) =2.0$ for Figure \ref{OUplot} and $\mu_X(T;t)=X(0) =20.0$ for Figure \ref{OUplot2}. Moreover $\sigma_X(T;t)\approx 1.41$ for both figures as for the Brownian motion in Figure \ref{BMplot2}. Indeed, Figure \ref{OUplot} looks very similar to Figure \ref{BMplot2} and the behaviour of the approximation with respect to the scale $b$ is similar, as we would expect due to the fact the the volatility is the same. However, the maximum accuracy reached is lower than the one for the Brownian motion ($10^{-8}$ at best). Moreover, the numerical instabilities in the sense discussed above appear earlier, namely, with a smaller $N$. We believe that both these phenomena are related to the fact that $X(0) = 2.0>0$: the high-order moments of the process $X$ that we need to calculate for approximating the price reach here larger values than for the Brownian motion, which has mean zero. Hence the instabilities occur at an earlier stage. This phenomenon is even more emphasised in Figure \ref{OUplot2} where $X(0)=20.0>>0$. Here indeed the numerical instabilities start around $N=20$, which in some cases is not enough to reach a reasonable accuracy.

\begin{figure}[!tp]
	\setlength{\tabcolsep}{2pt}
	\resizebox{1\textwidth}{!}{
		\begin{tabular}{@{}>{\centering\arraybackslash}m{0.04\textwidth}@{}>{\centering\arraybackslash}m{0.24\textwidth}@{}>{\centering\arraybackslash}m{0.24\textwidth}@{}>{\centering\arraybackslash}m{0.24\textwidth}@{}>{\centering\arraybackslash}m{0.24\textwidth}@{}}
			& $\boldsymbol{K = 1.0}$&$\boldsymbol{K = 2.0}$ & $\boldsymbol{K = 3.0}$& $\boldsymbol{K = 4.0}$ \\
			\begin{turn}{90}$\boldsymbol{b =1.0}$\end{turn}&
			\includegraphics[width=0.23\textwidth]{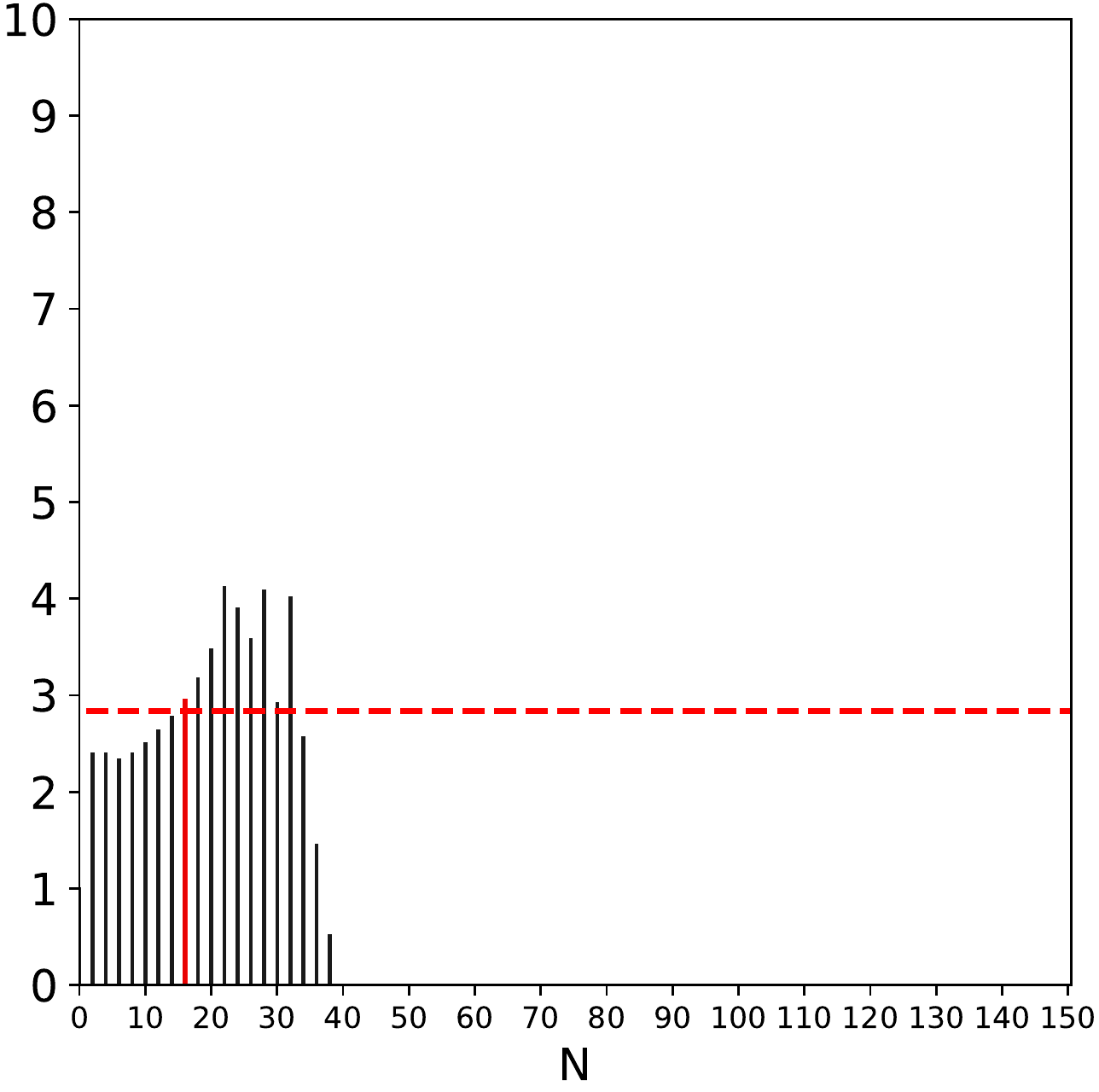}&
			\includegraphics[width=0.23\textwidth]{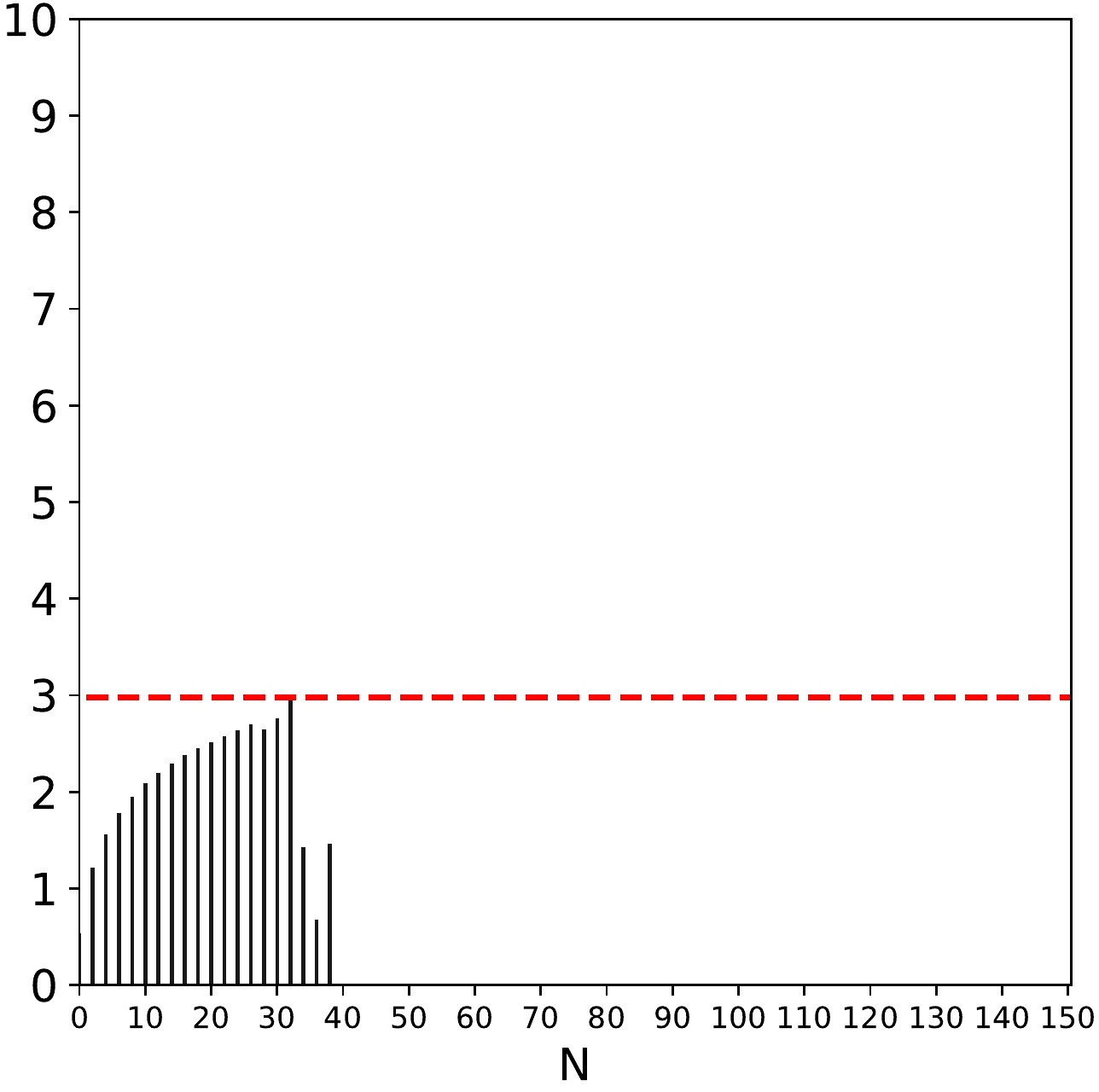} &	\includegraphics[width=0.23\textwidth]{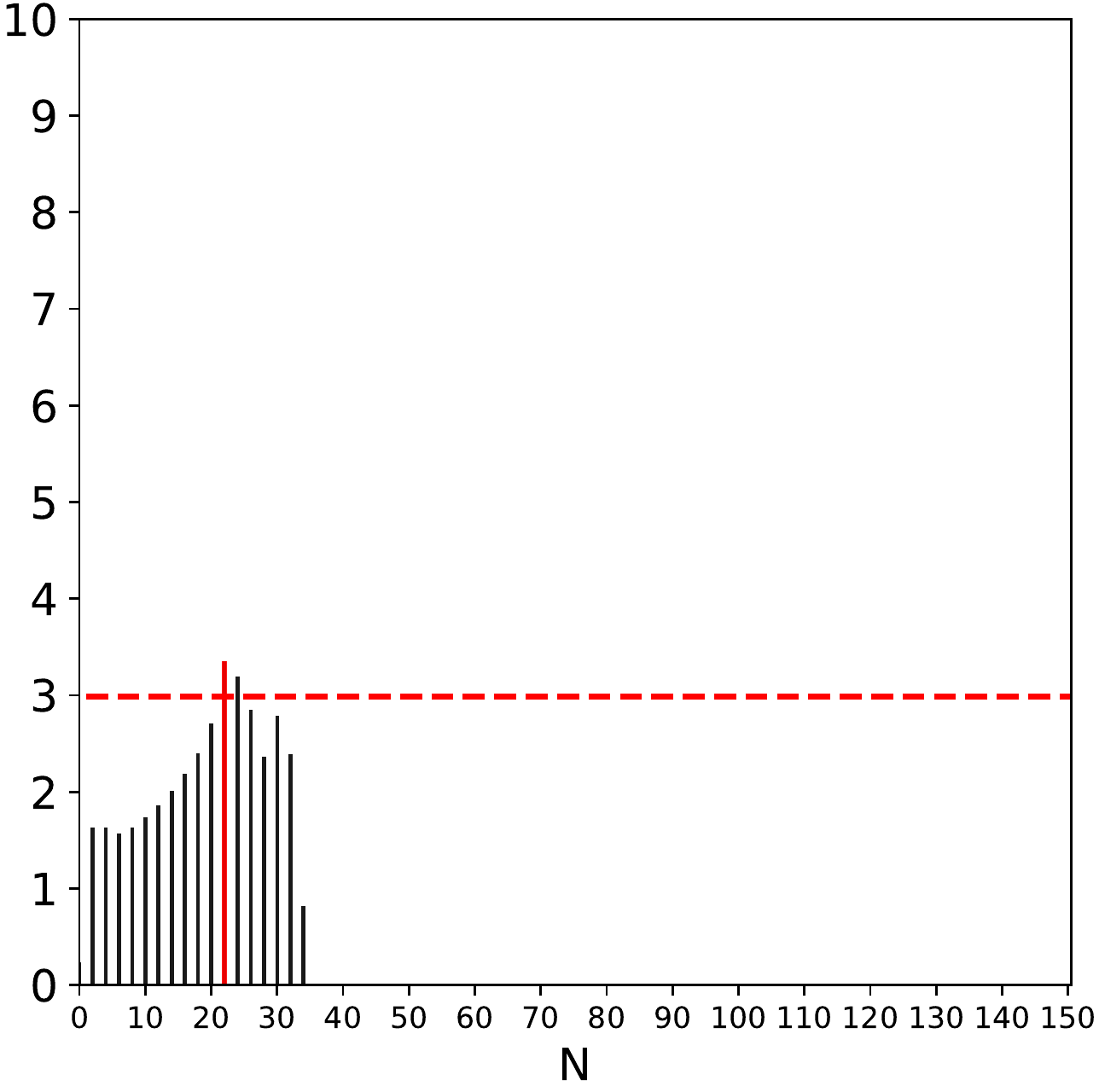} &
			\includegraphics[width=0.23\textwidth]{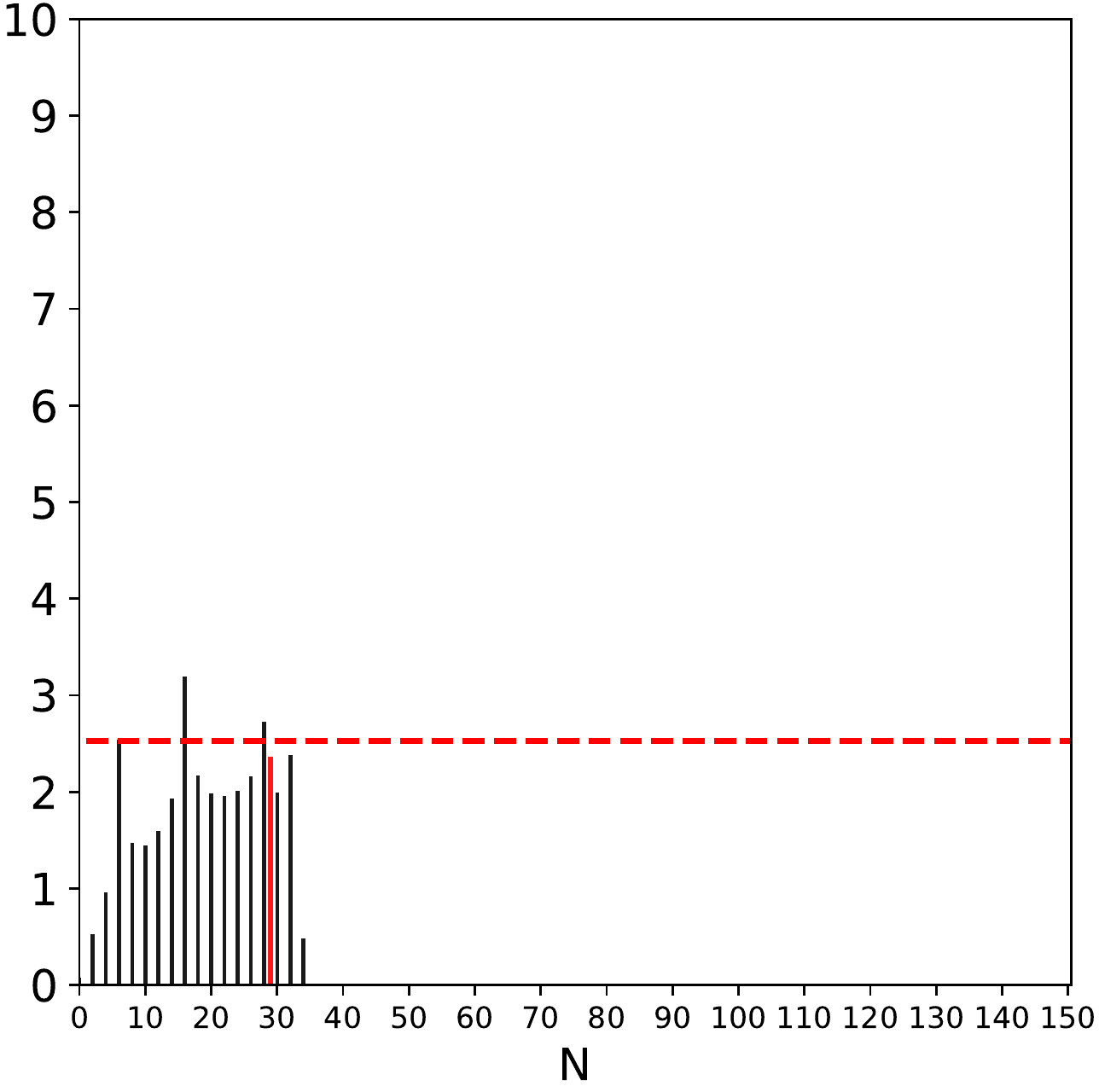}\\
			\begin{turn}{90}$\boldsymbol{b =1.2}$\end{turn}&
			\includegraphics[width=0.23\textwidth]{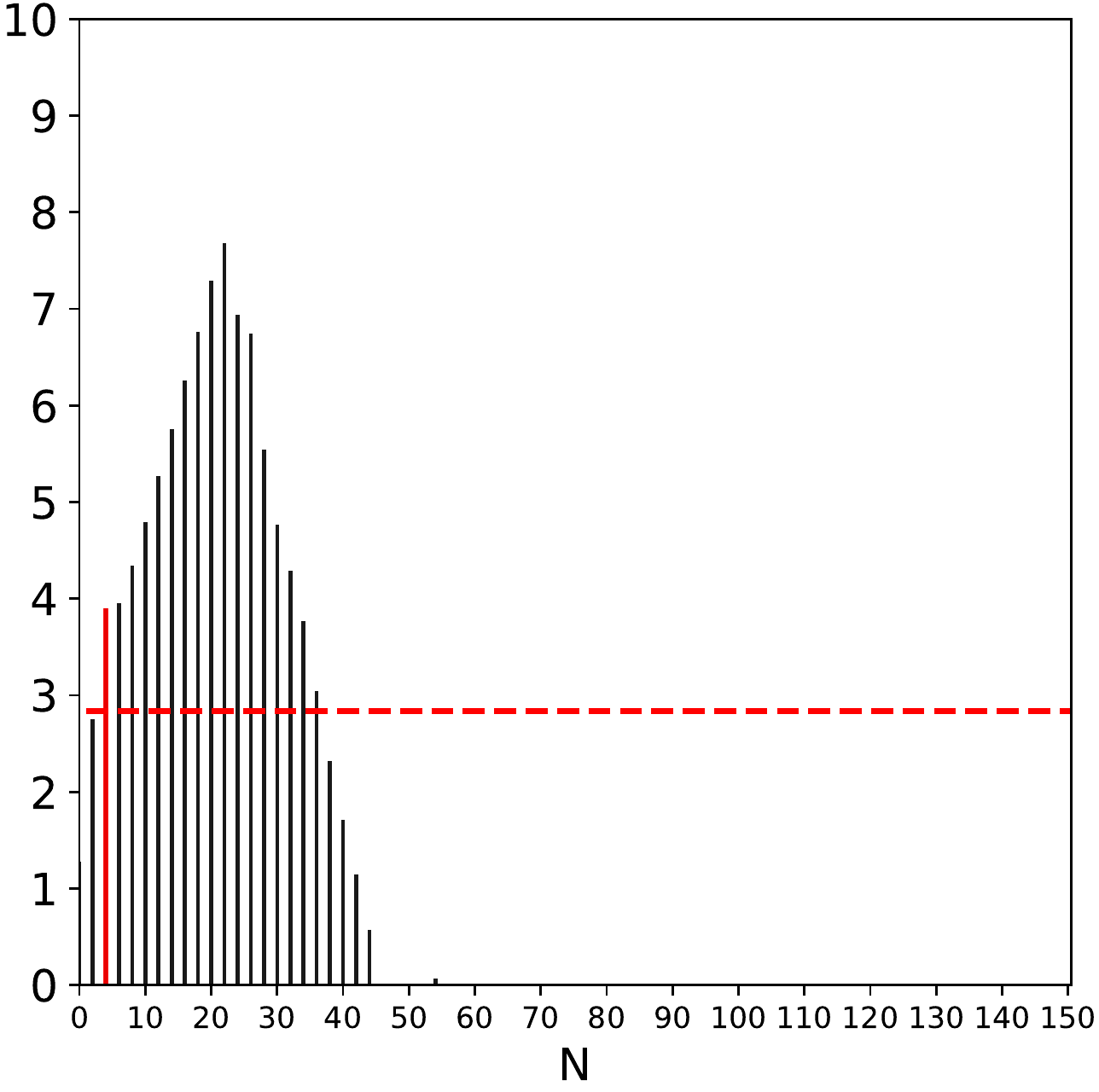}&
			\includegraphics[width=0.23\textwidth]{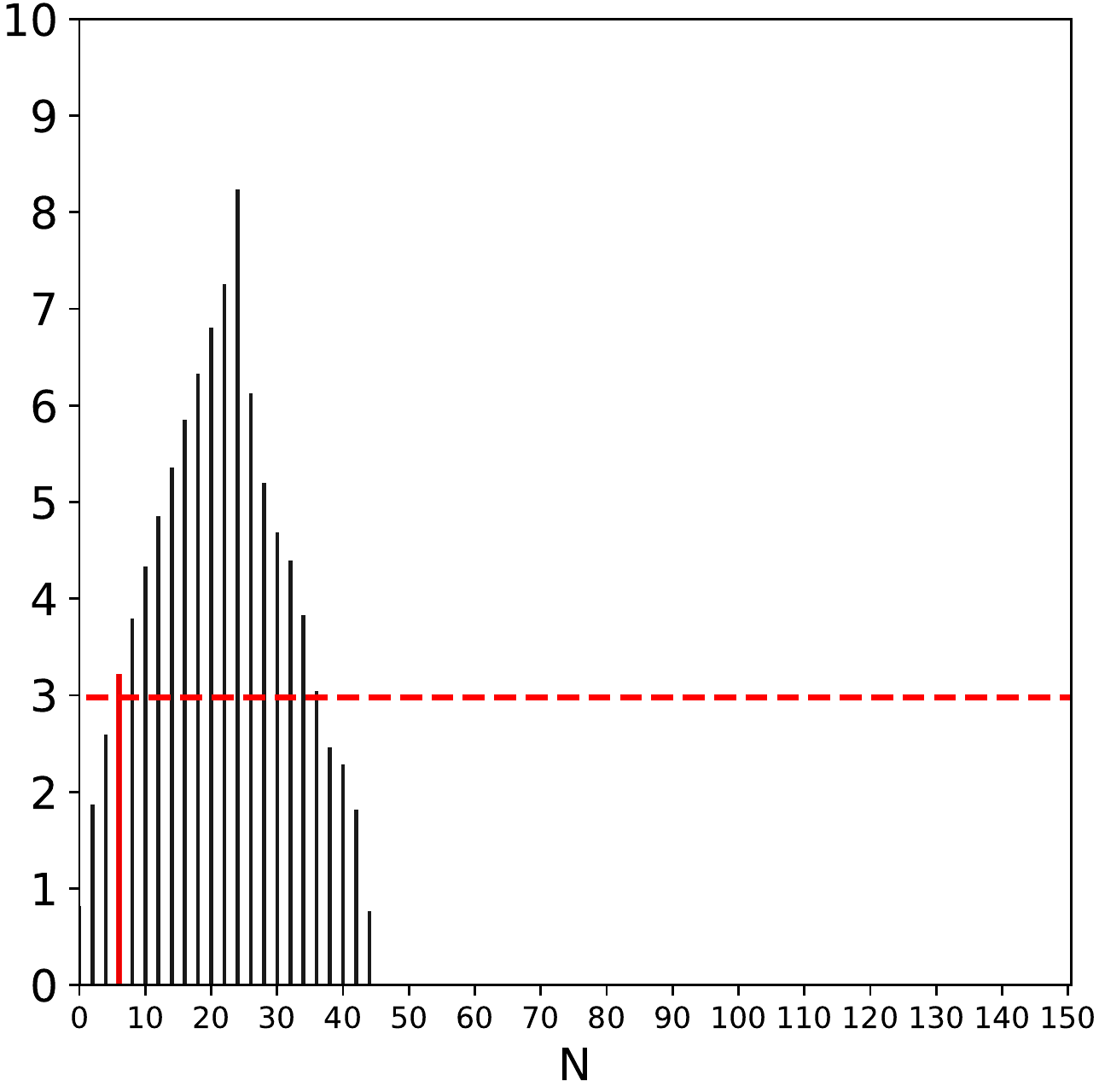} &	\includegraphics[width=0.23\textwidth]{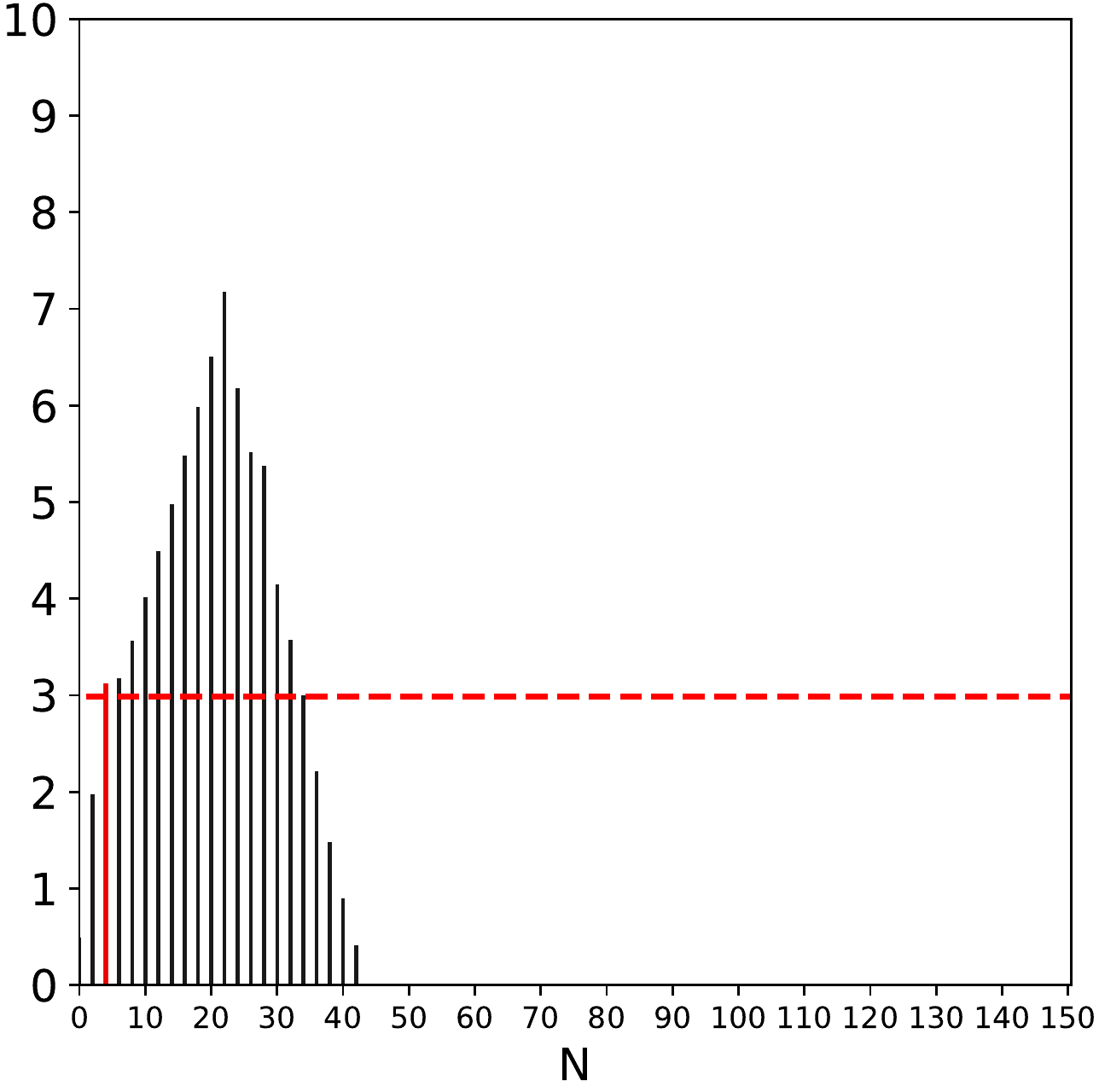} &
			\includegraphics[width=0.23\textwidth]{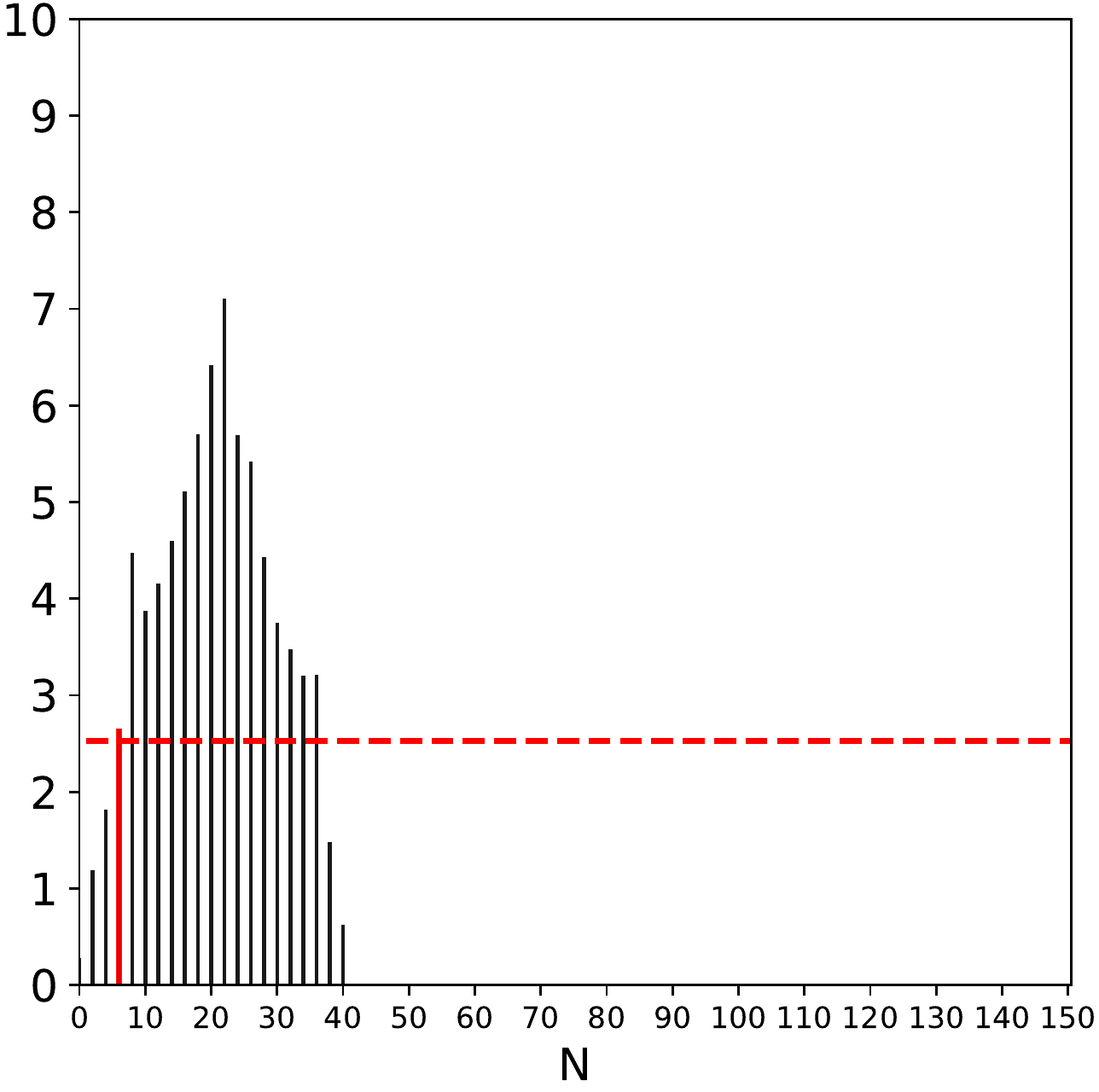}\\
			\begin{turn}{90}$\boldsymbol{b =2.0}$\end{turn}&
			\includegraphics[width=0.23\textwidth]{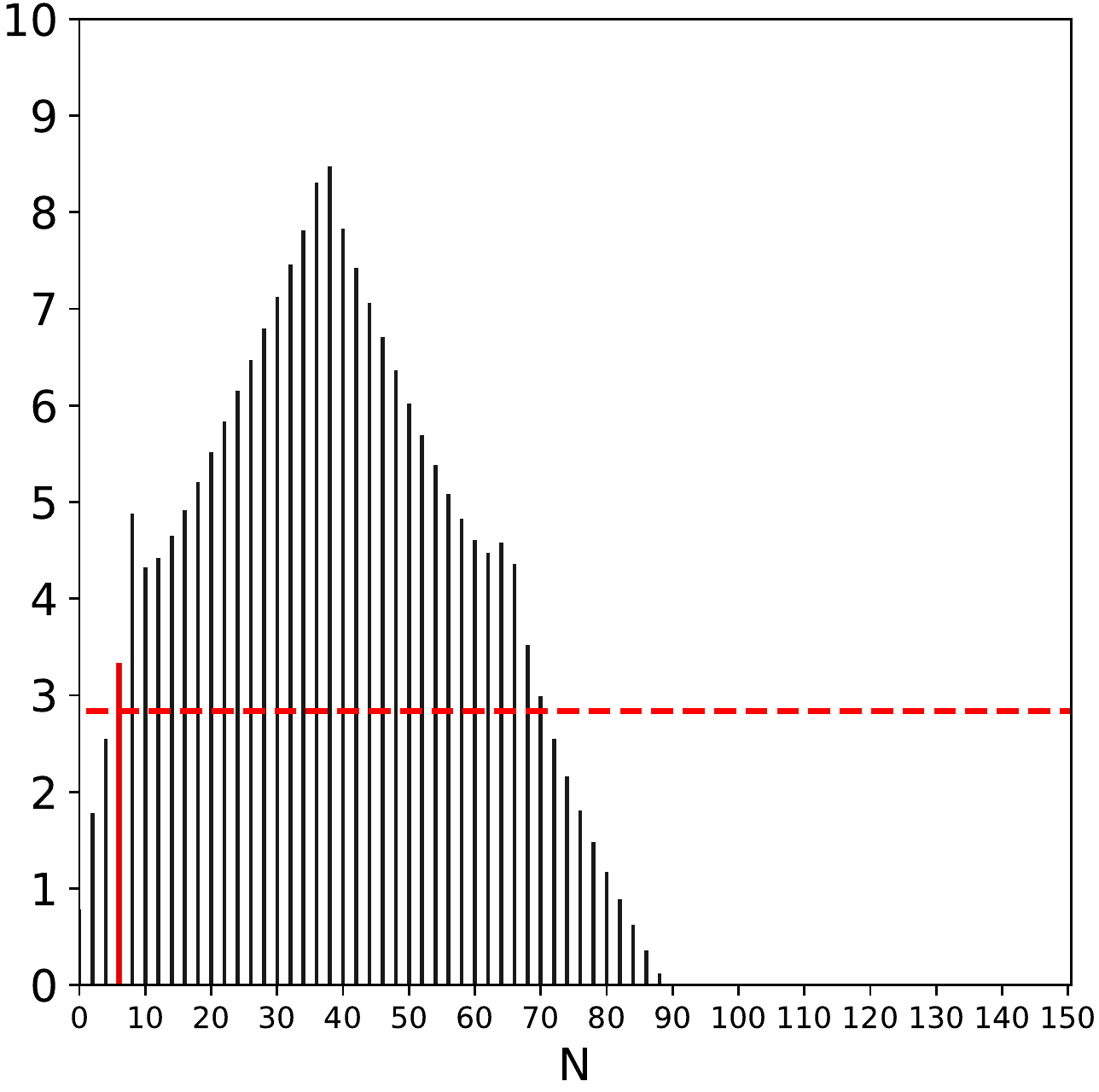}&
			\includegraphics[width=0.23\textwidth]{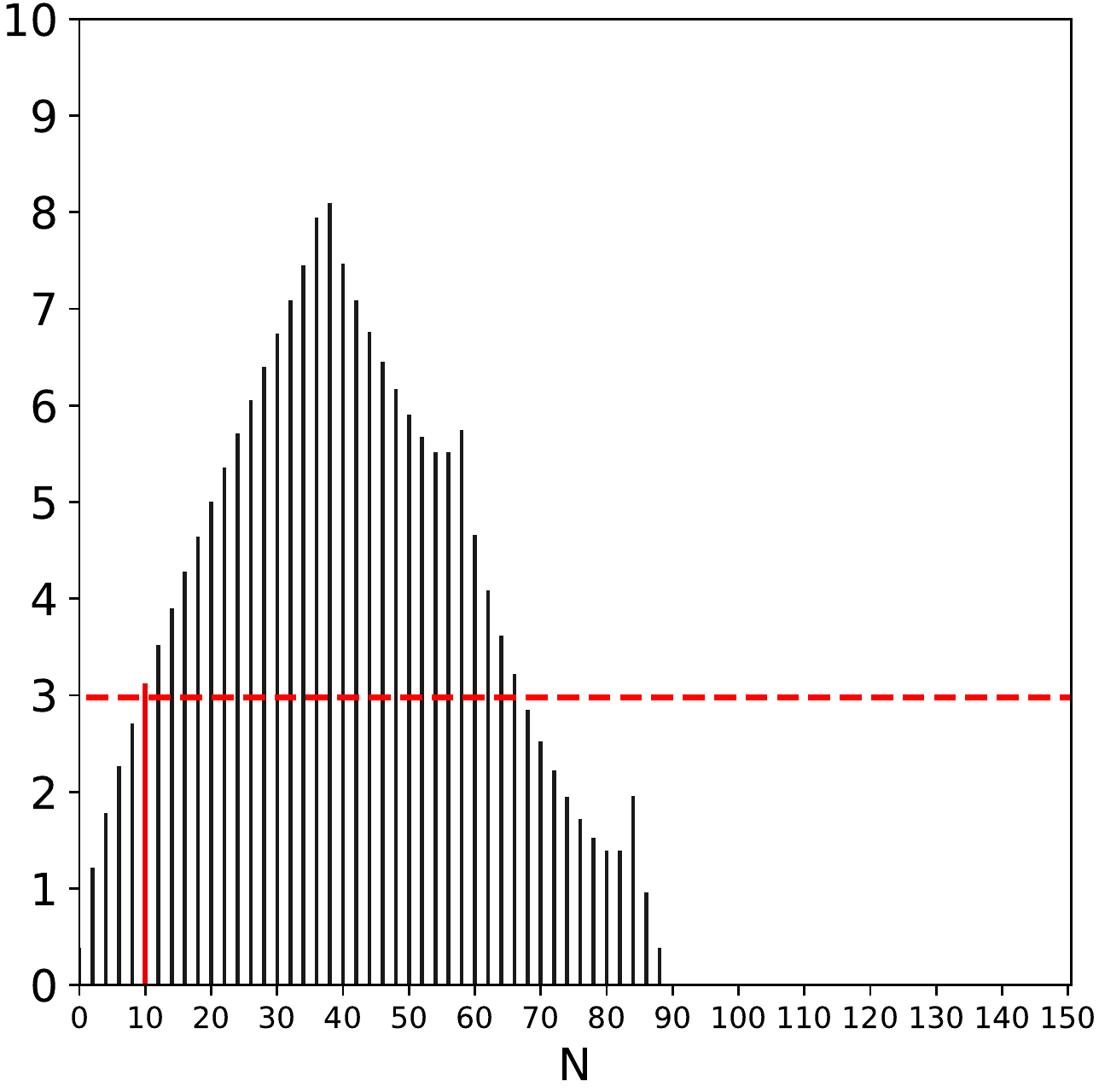} &	\includegraphics[width=0.23\textwidth]{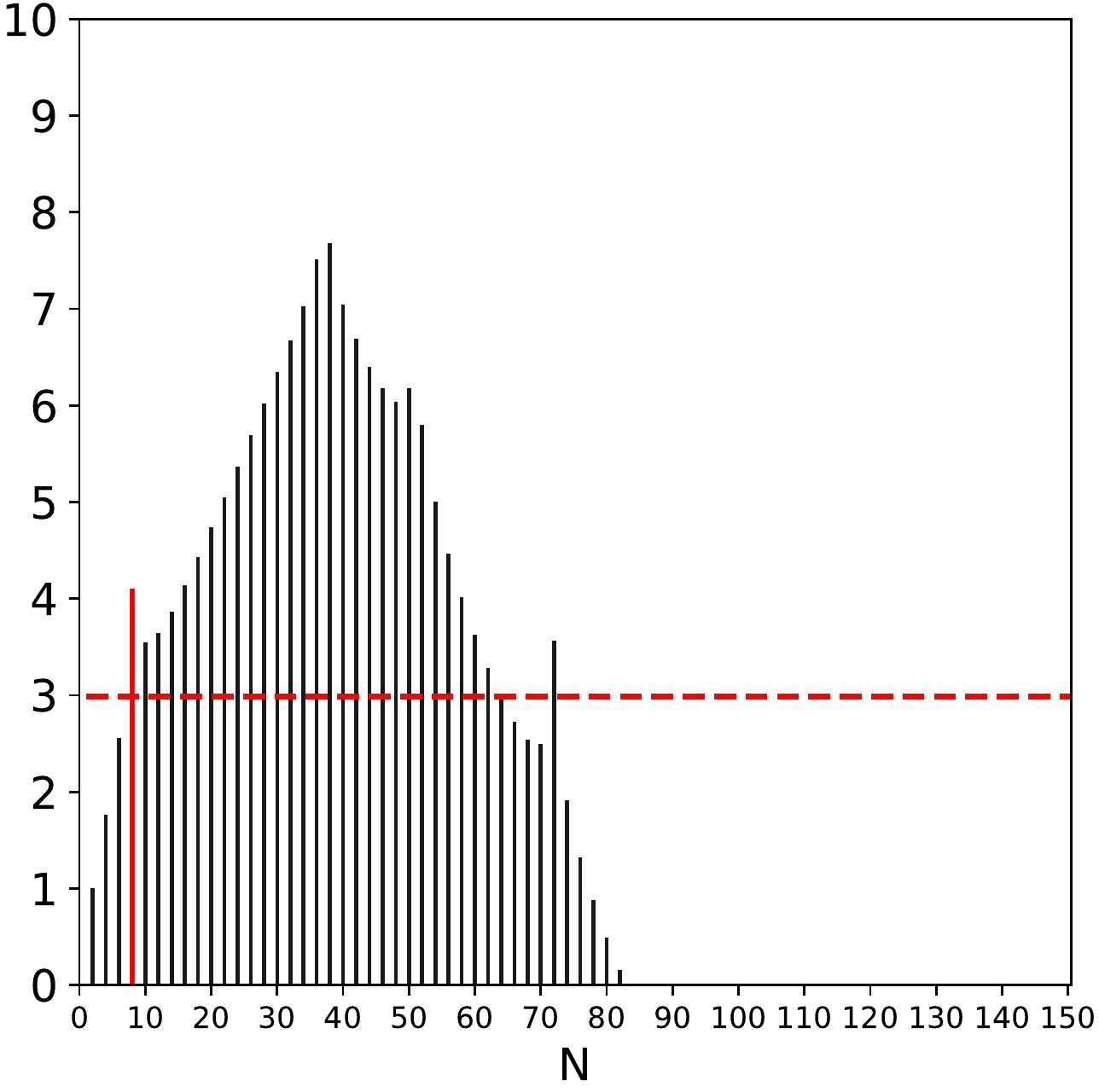} &
			\includegraphics[width=0.23\textwidth]{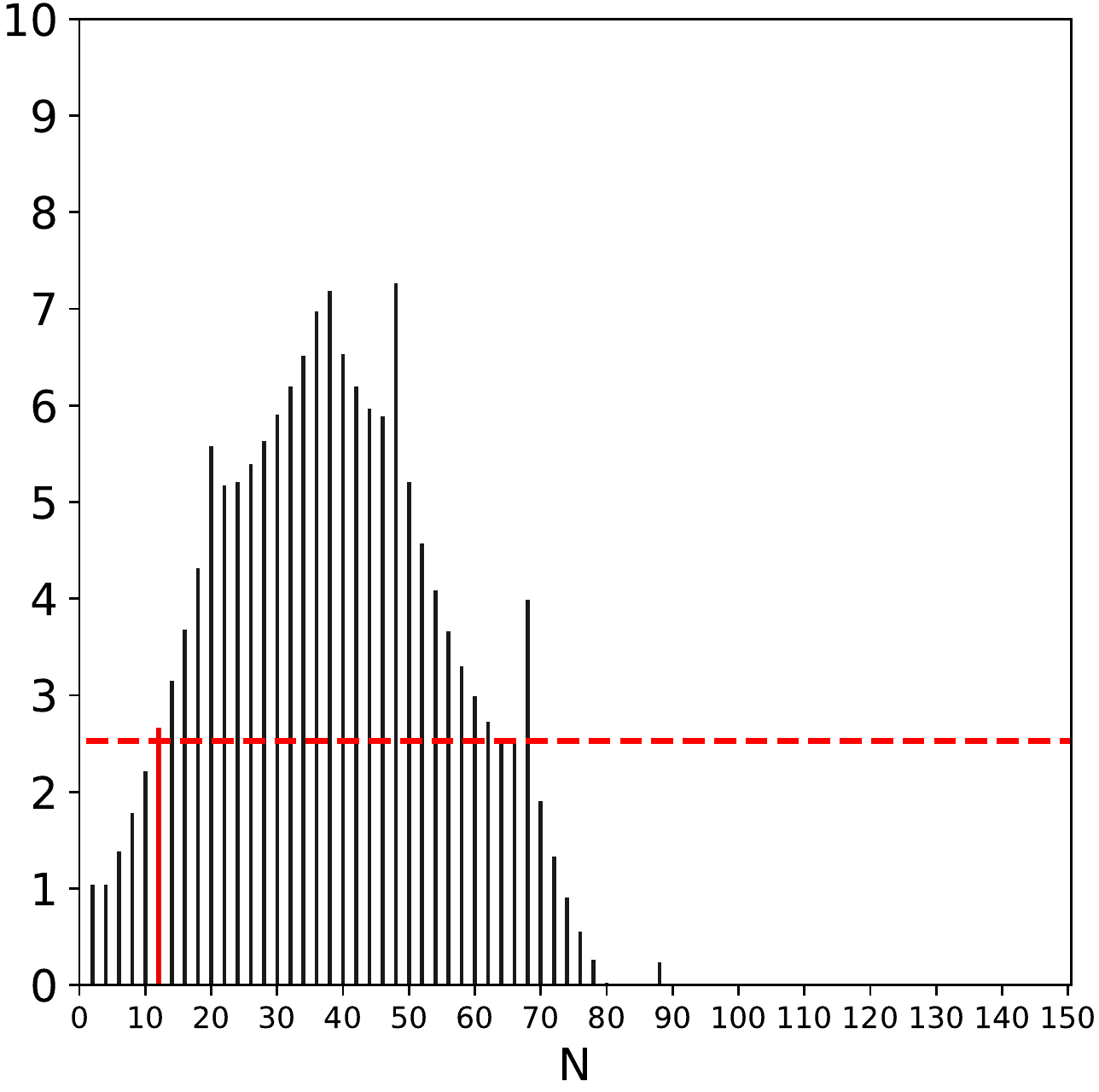}\\
			\begin{turn}{90}$\boldsymbol{b =4.0}$\end{turn}&
			\includegraphics[width=0.23\textwidth]{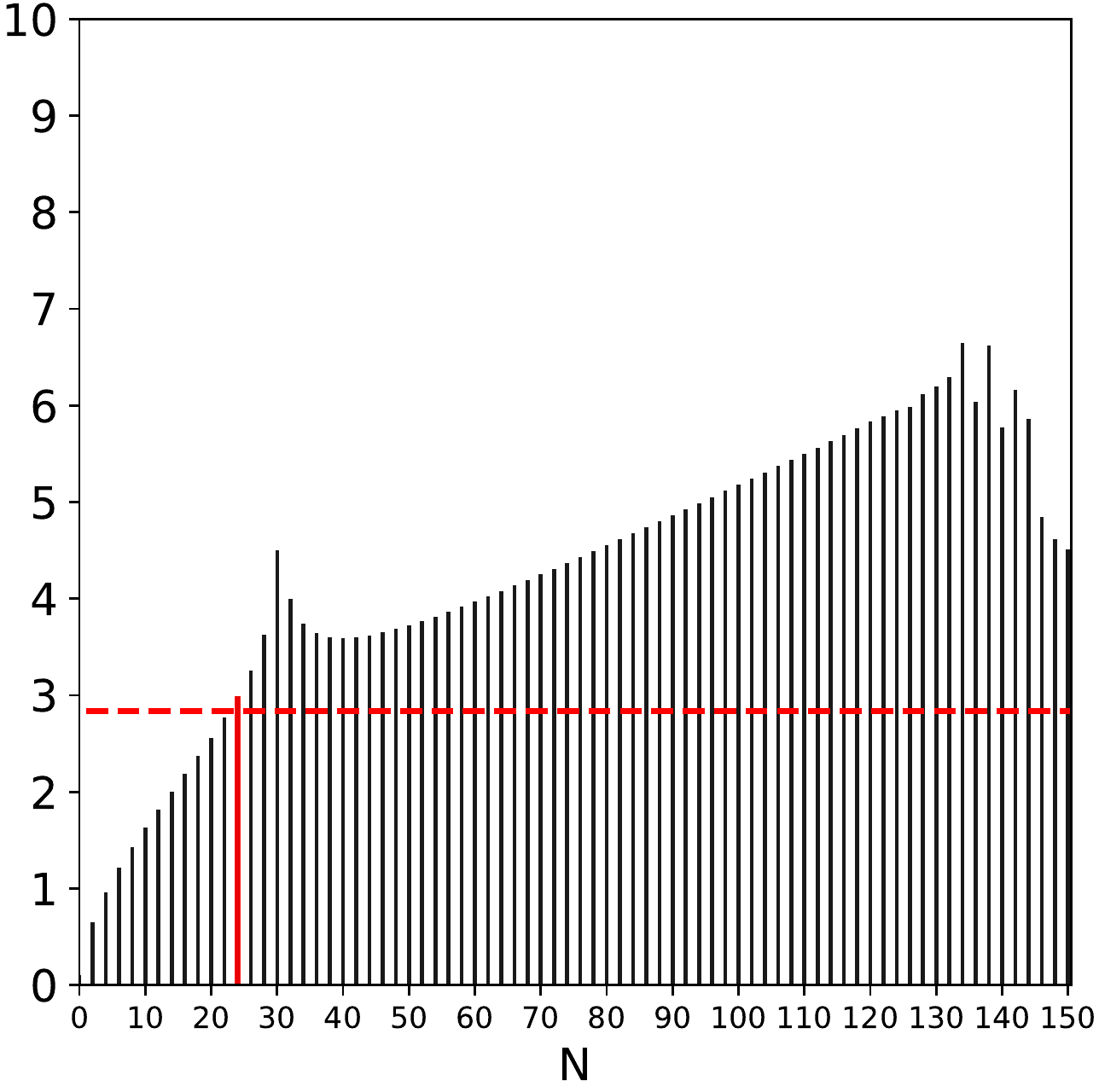}&
			\includegraphics[width=0.23\textwidth]{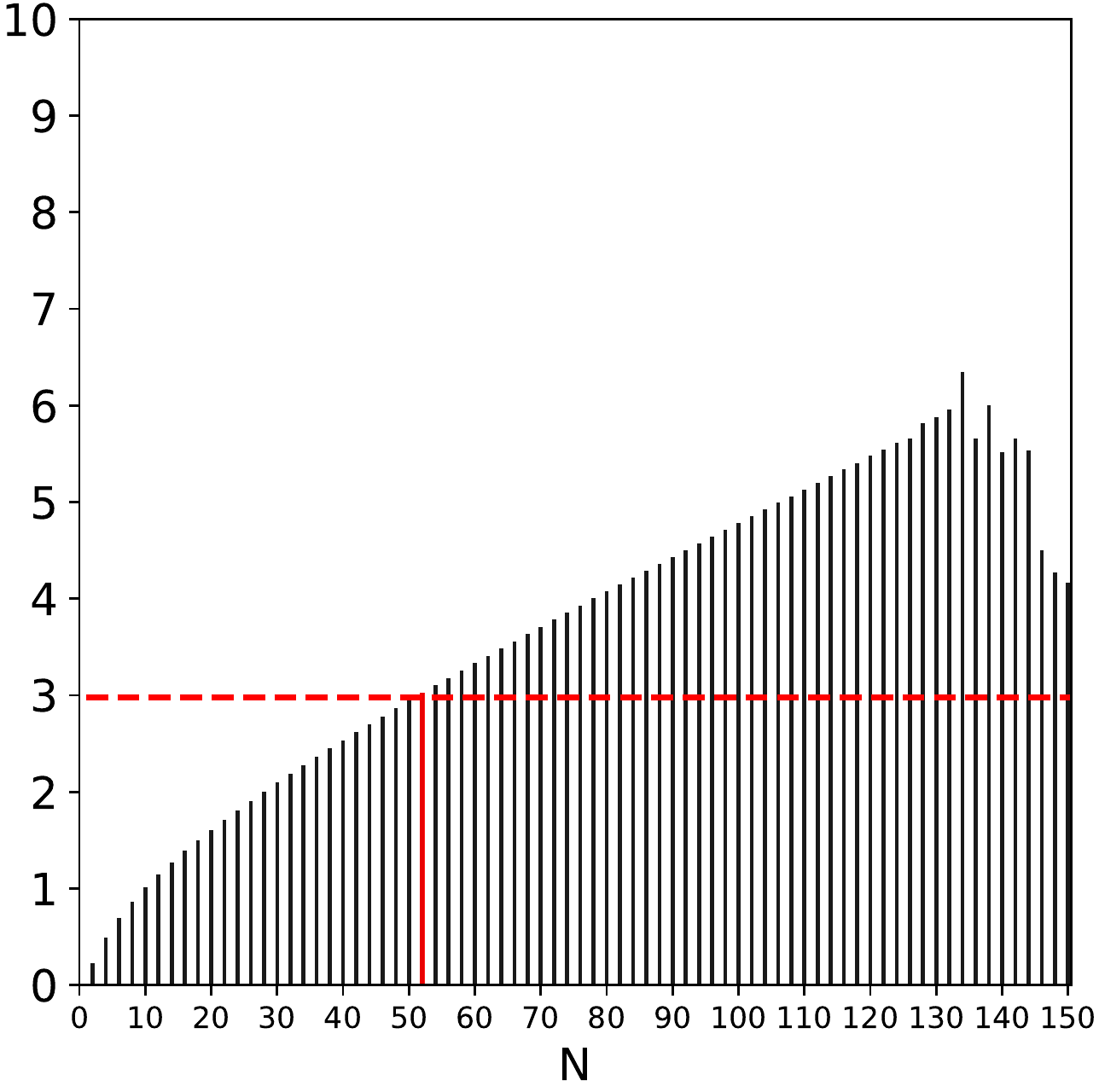} &	\includegraphics[width=0.23\textwidth]{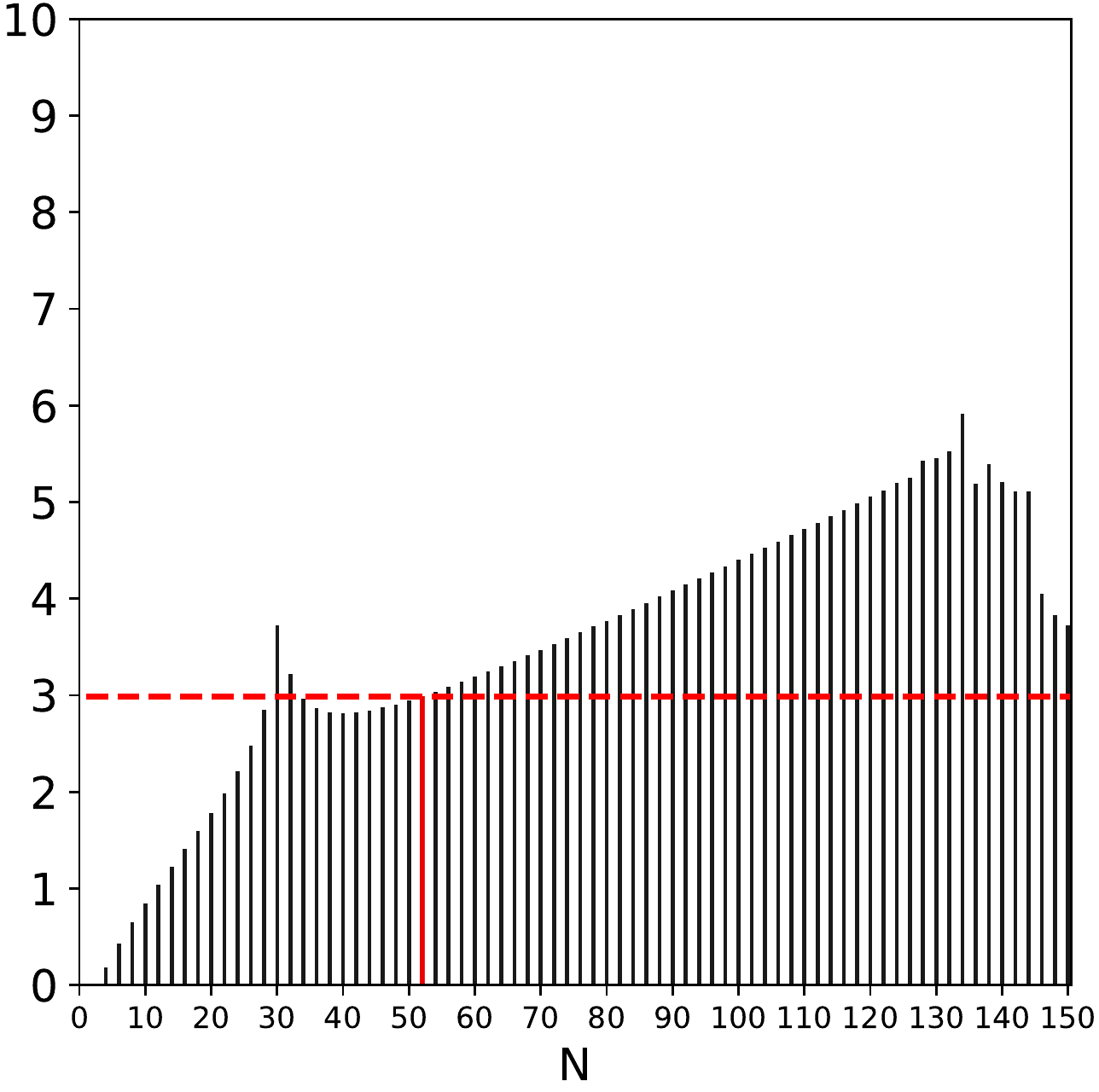} &
			\includegraphics[width=0.23\textwidth]{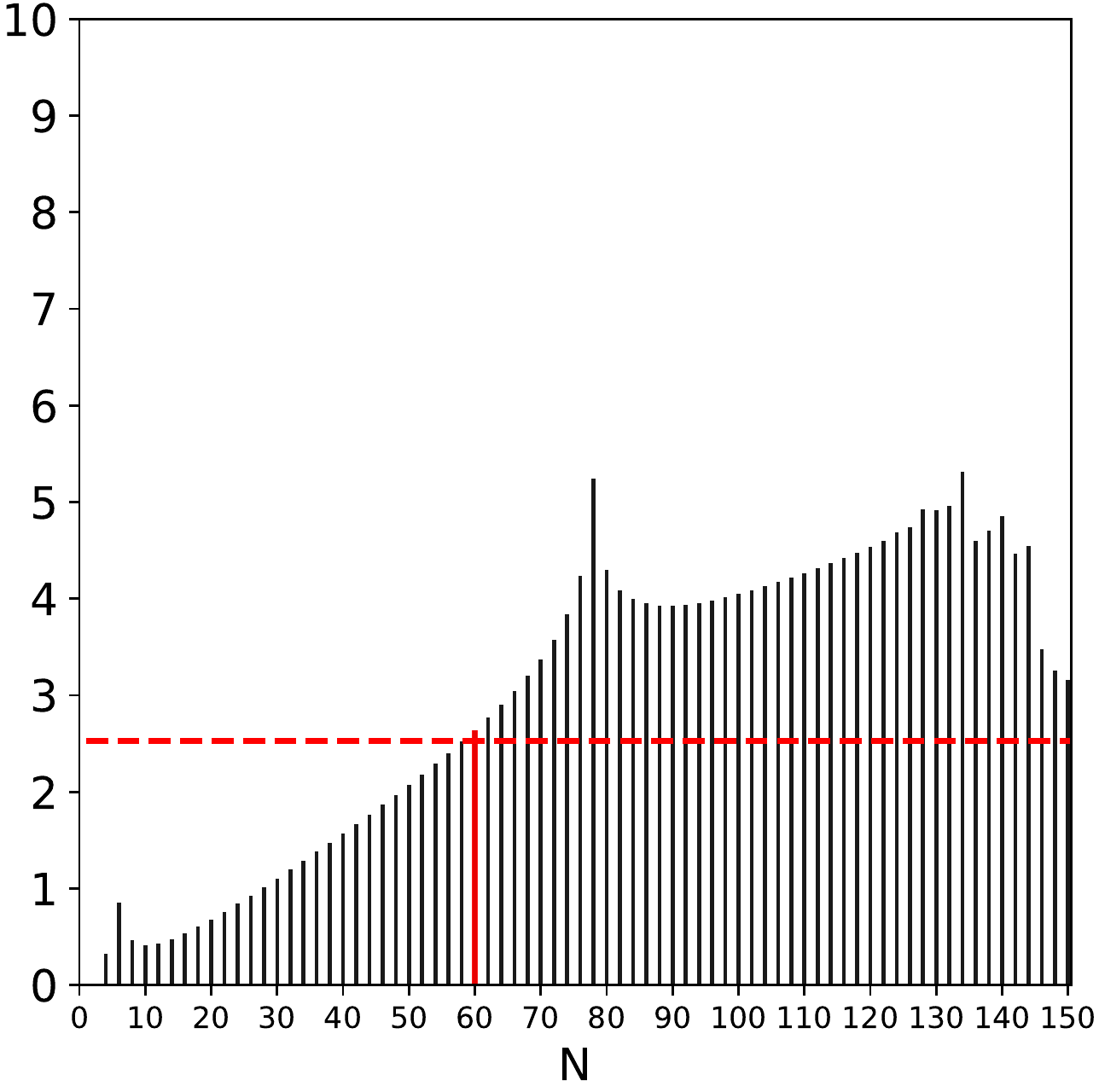}\\
			\begin{turn}{90}$\boldsymbol{b =6.0}$\end{turn}&
			\includegraphics[width=0.23\textwidth]{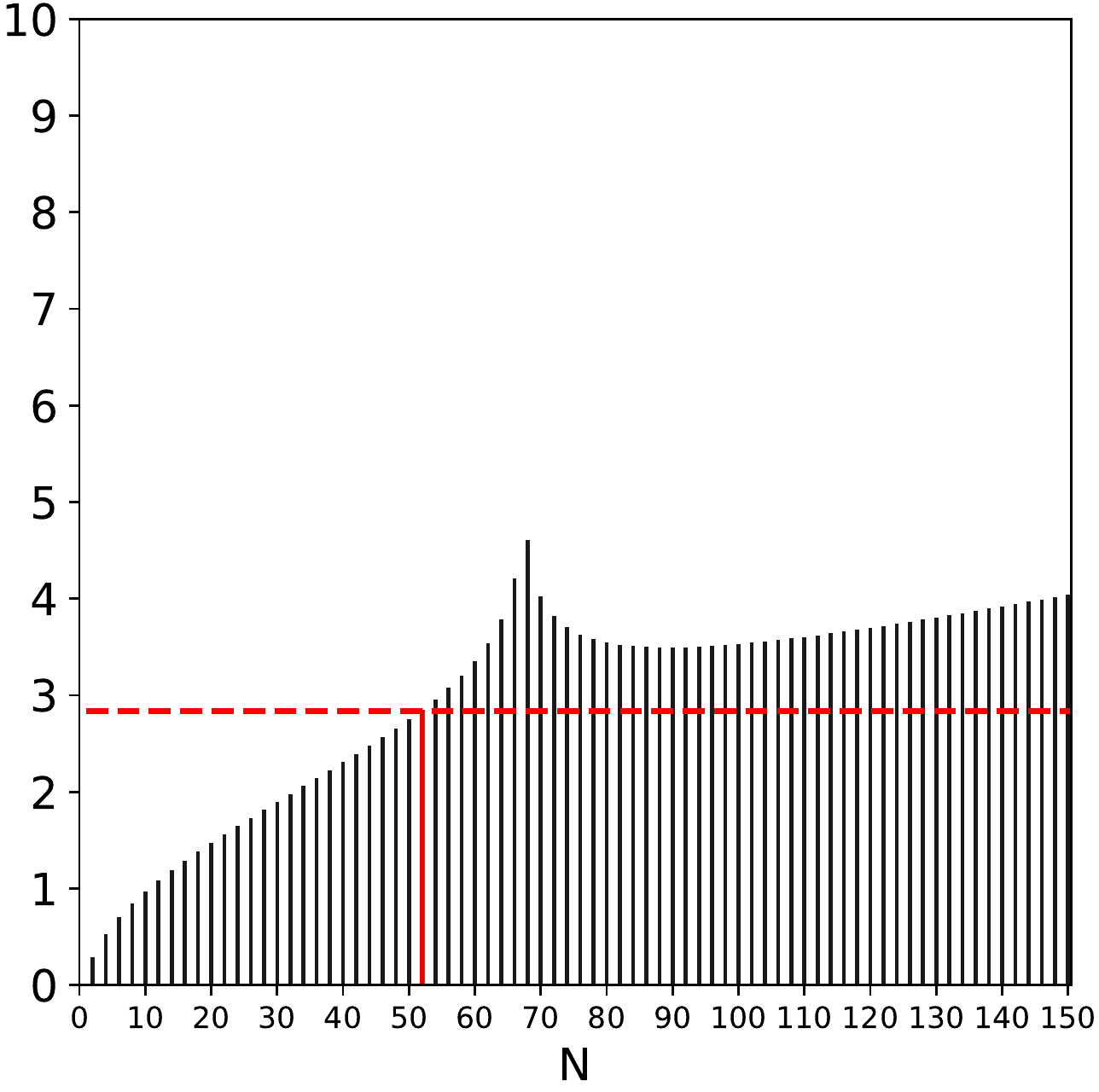}&
			\includegraphics[width=0.23\textwidth]{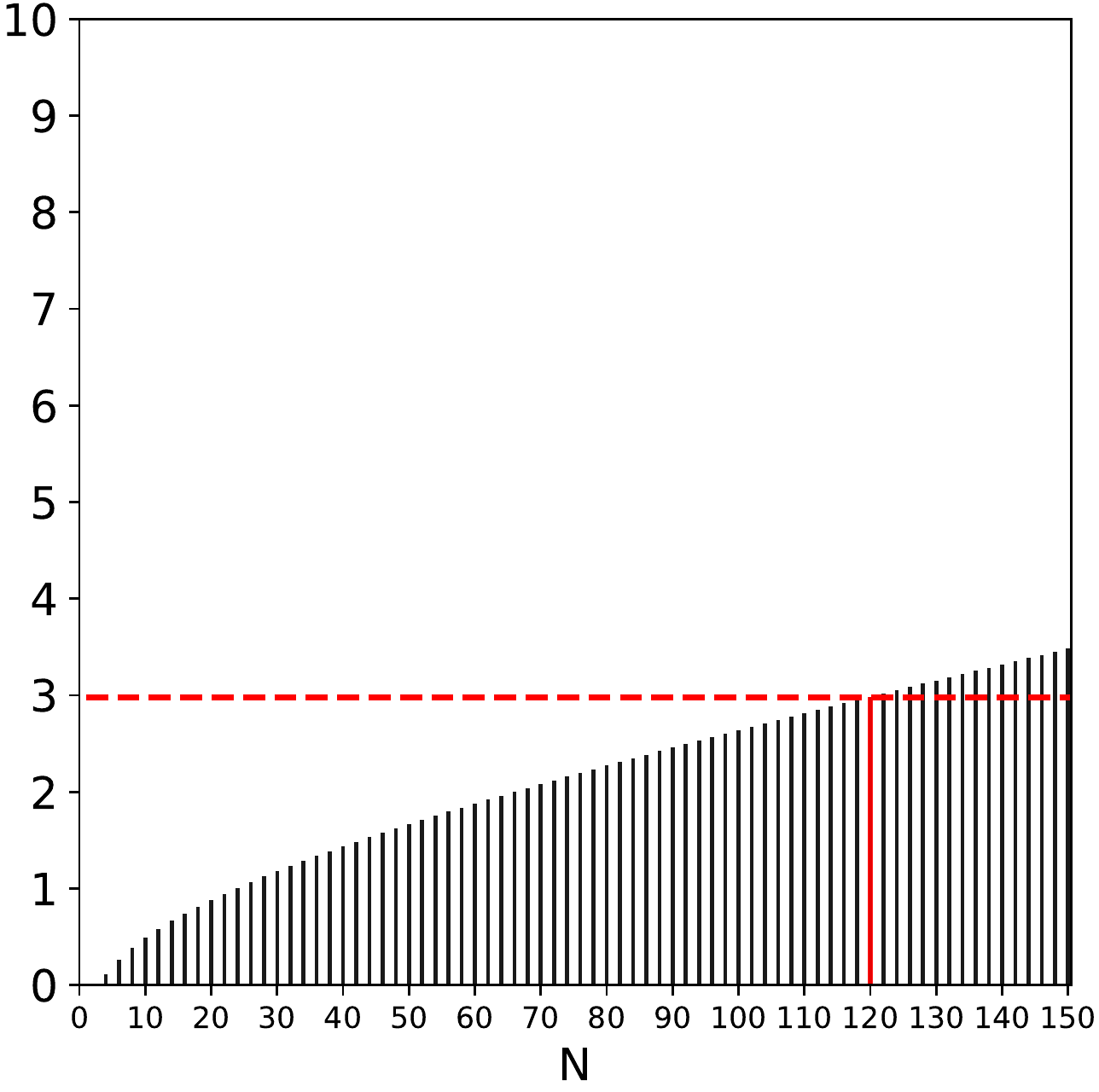} &	\includegraphics[width=0.23\textwidth]{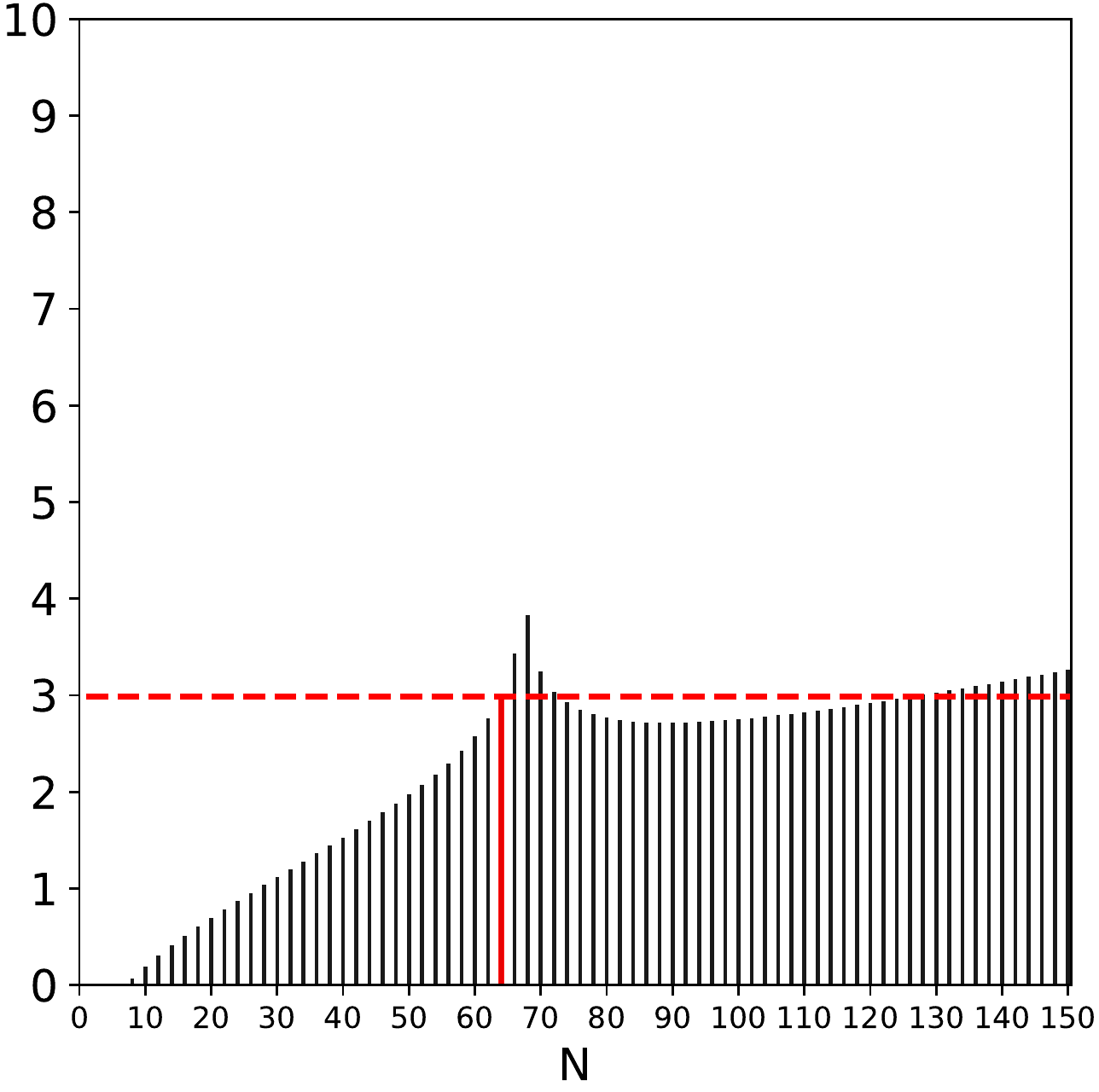} &
			\includegraphics[width=0.23\textwidth]{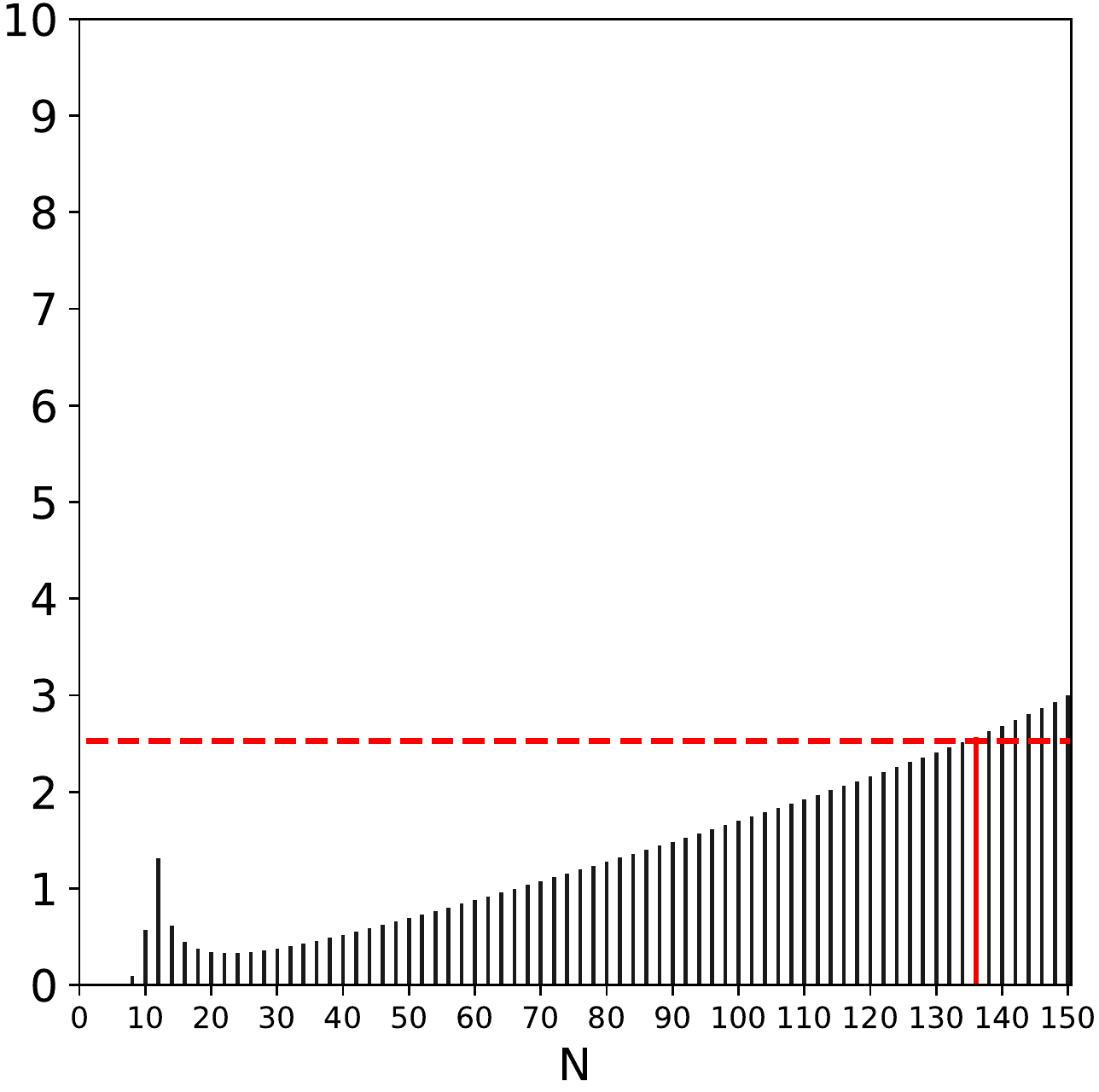}\\
		\end{tabular}
	}
	\caption{$\gamma_{a,b}^N$ as a function of $N$, when the underlying process is an OU with $\sigma_X(T;t)\approx 1.41$. The drift is $a= \mu_X(T;t)=2.0$. The dashed red horizontal lines indicate the accuracy of the MC method. The red vertical bars indicate when the Hermite series reaches the MC accuracy. \label{OUplot}}
\end{figure}

\begin{figure}[!tp]
	\setlength{\tabcolsep}{2pt}
	\resizebox{1\textwidth}{!}{
		\begin{tabular}{@{}>{\centering\arraybackslash}m{0.04\textwidth}@{}>{\centering\arraybackslash}m{0.24\textwidth}@{}>{\centering\arraybackslash}m{0.24\textwidth}@{}>{\centering\arraybackslash}m{0.24\textwidth}@{}>{\centering\arraybackslash}m{0.24\textwidth}@{}}
			& $\boldsymbol{K = 19.0}$&$\boldsymbol{K = 20.0}$ & $\boldsymbol{K = 21.0}$& $\boldsymbol{K = 22.0}$ \\
			\begin{turn}{90}$\boldsymbol{b =1.0}$\end{turn}&
			\includegraphics[width=0.23\textwidth]{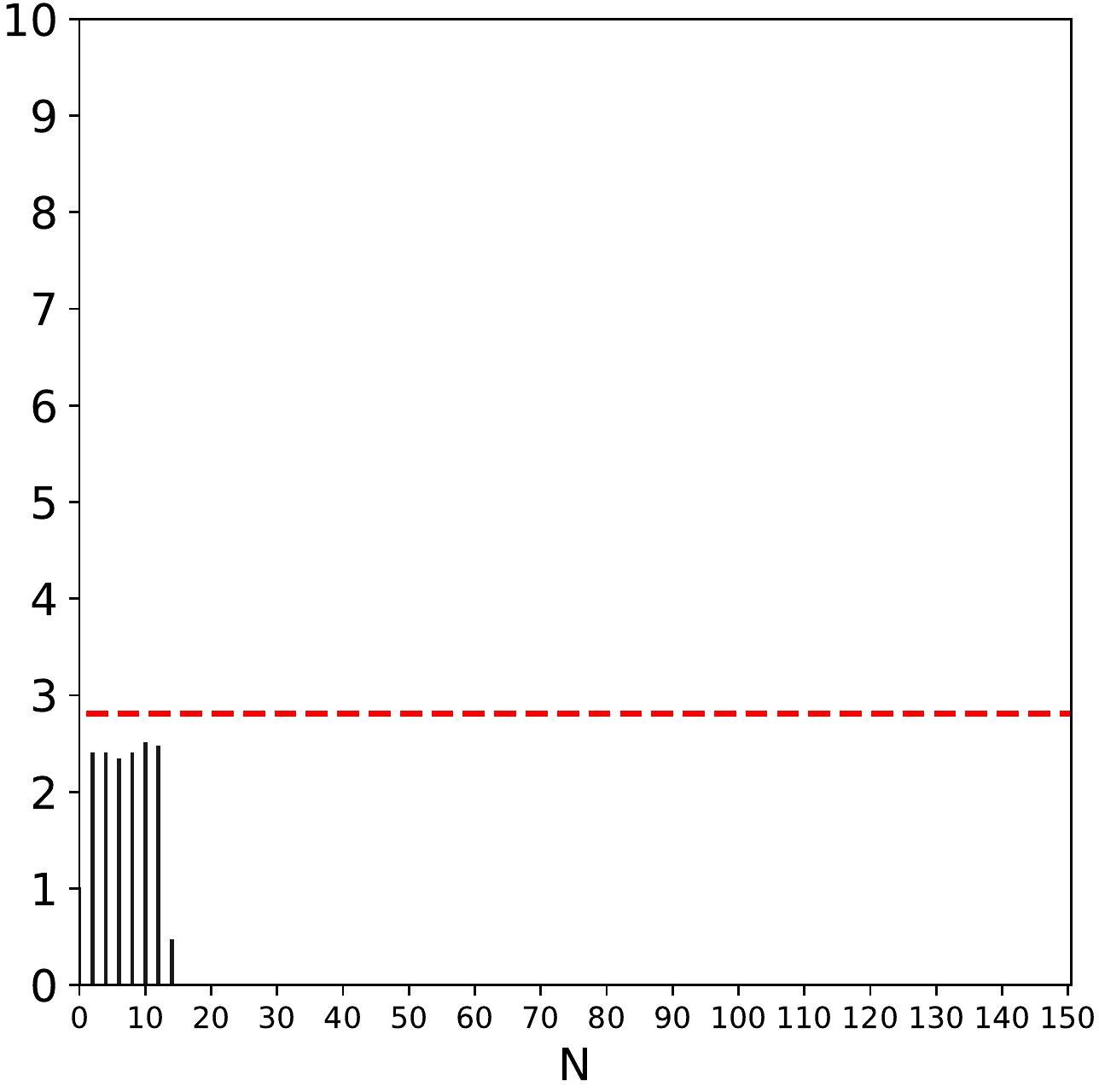}&
			\includegraphics[width=0.23\textwidth]{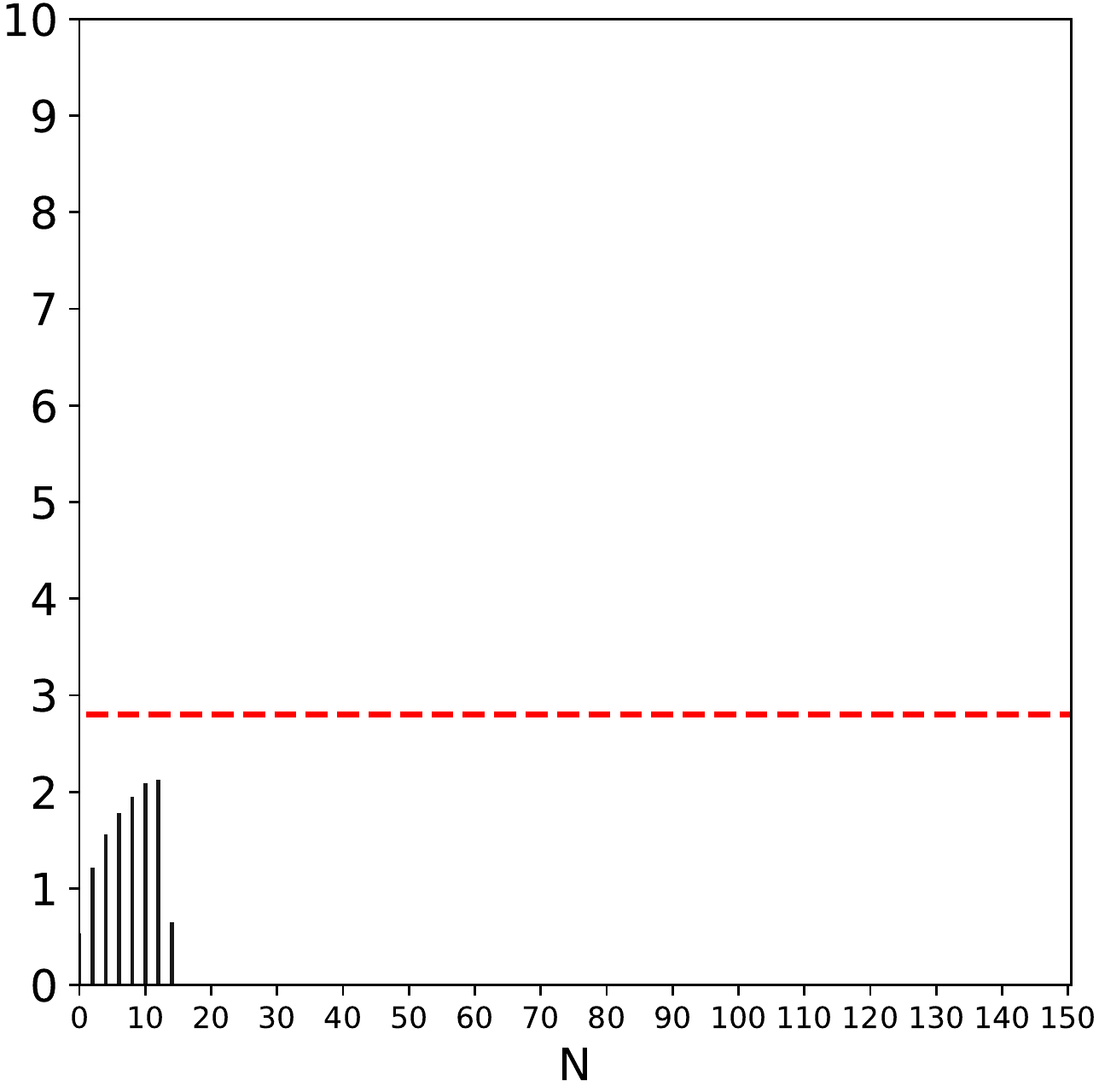} &	\includegraphics[width=0.23\textwidth]{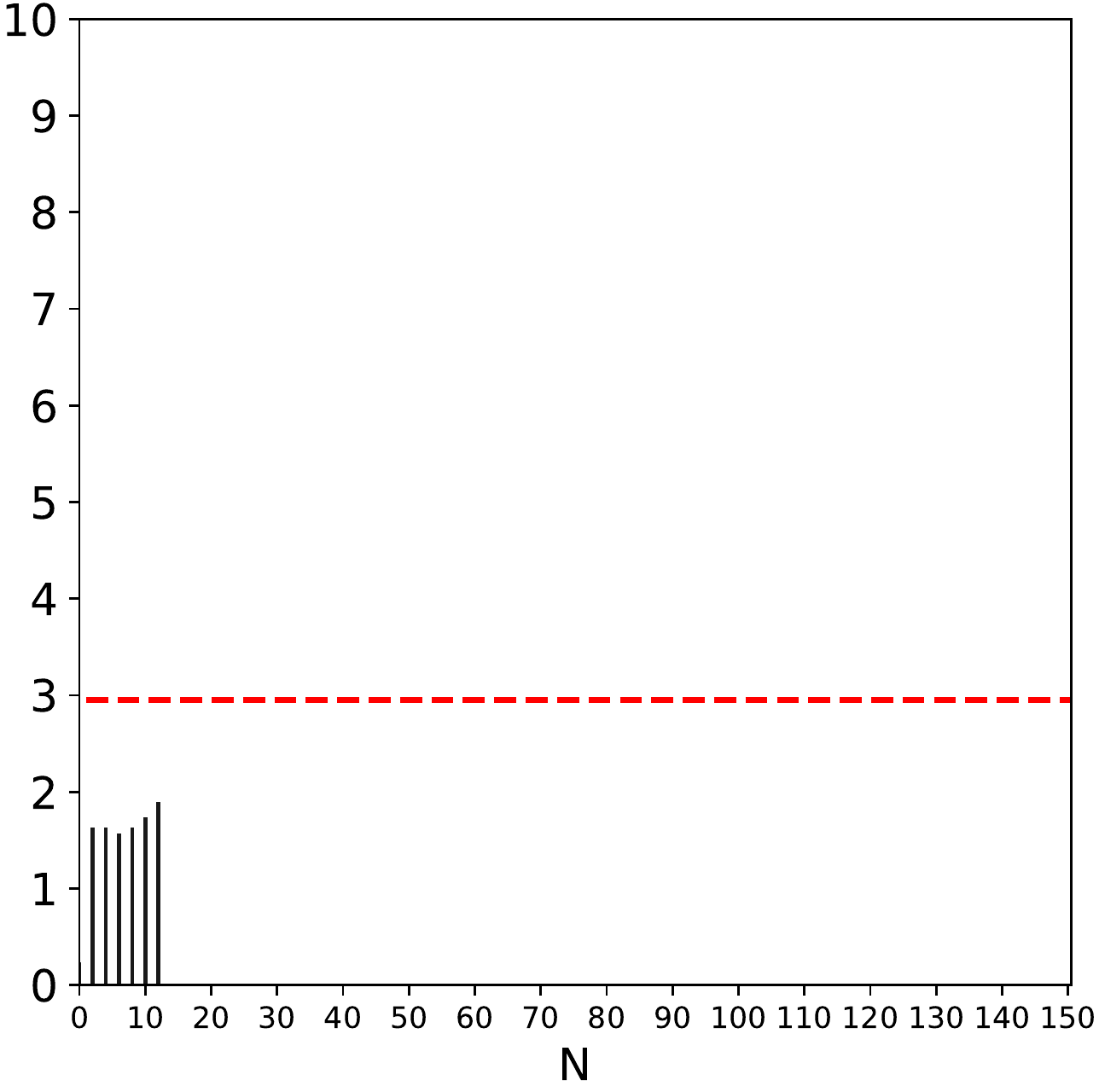} &
			\includegraphics[width=0.23\textwidth]{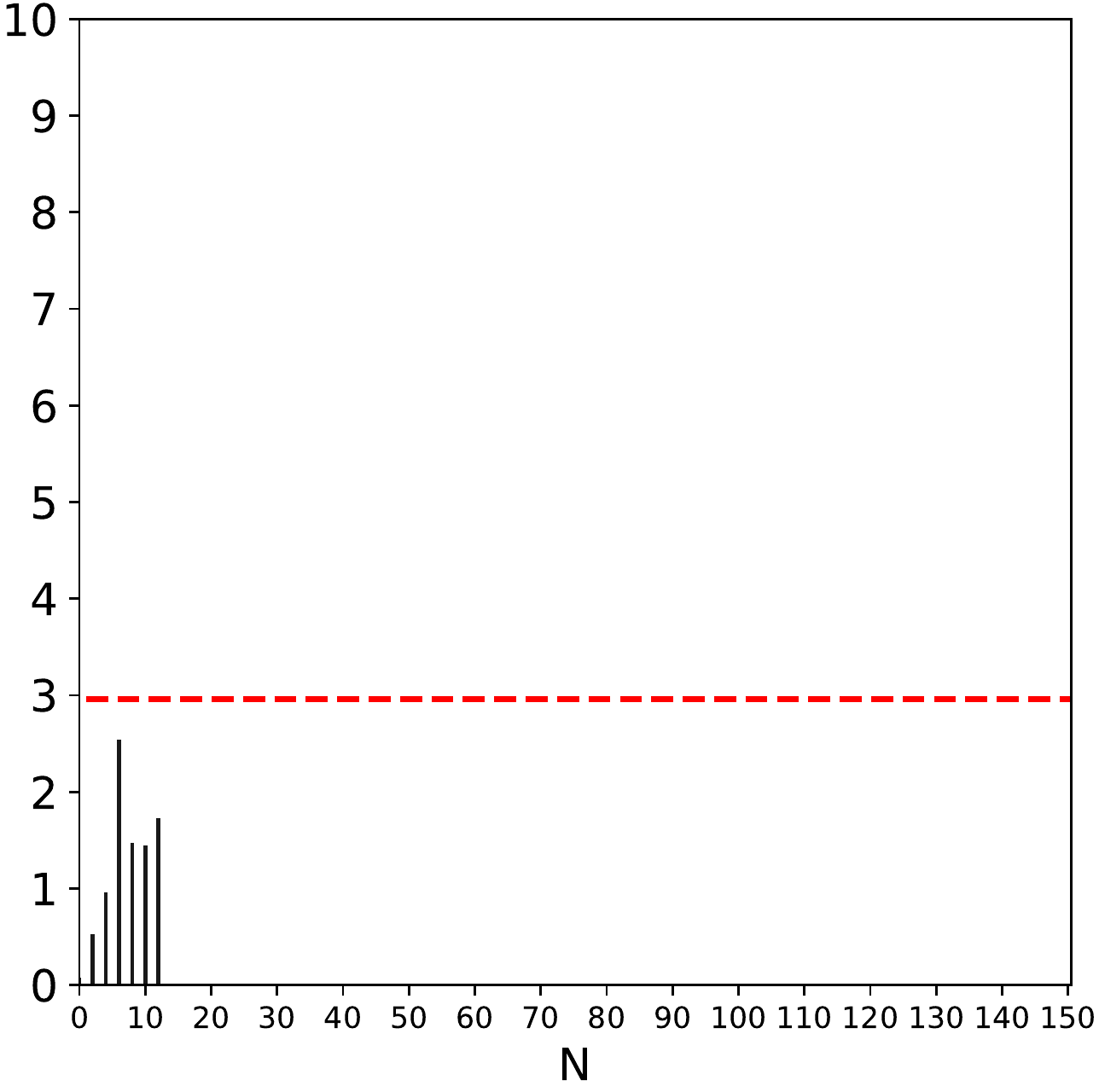}\\
			\begin{turn}{90}$\boldsymbol{b =1.2}$\end{turn}&
			\includegraphics[width=0.23\textwidth]{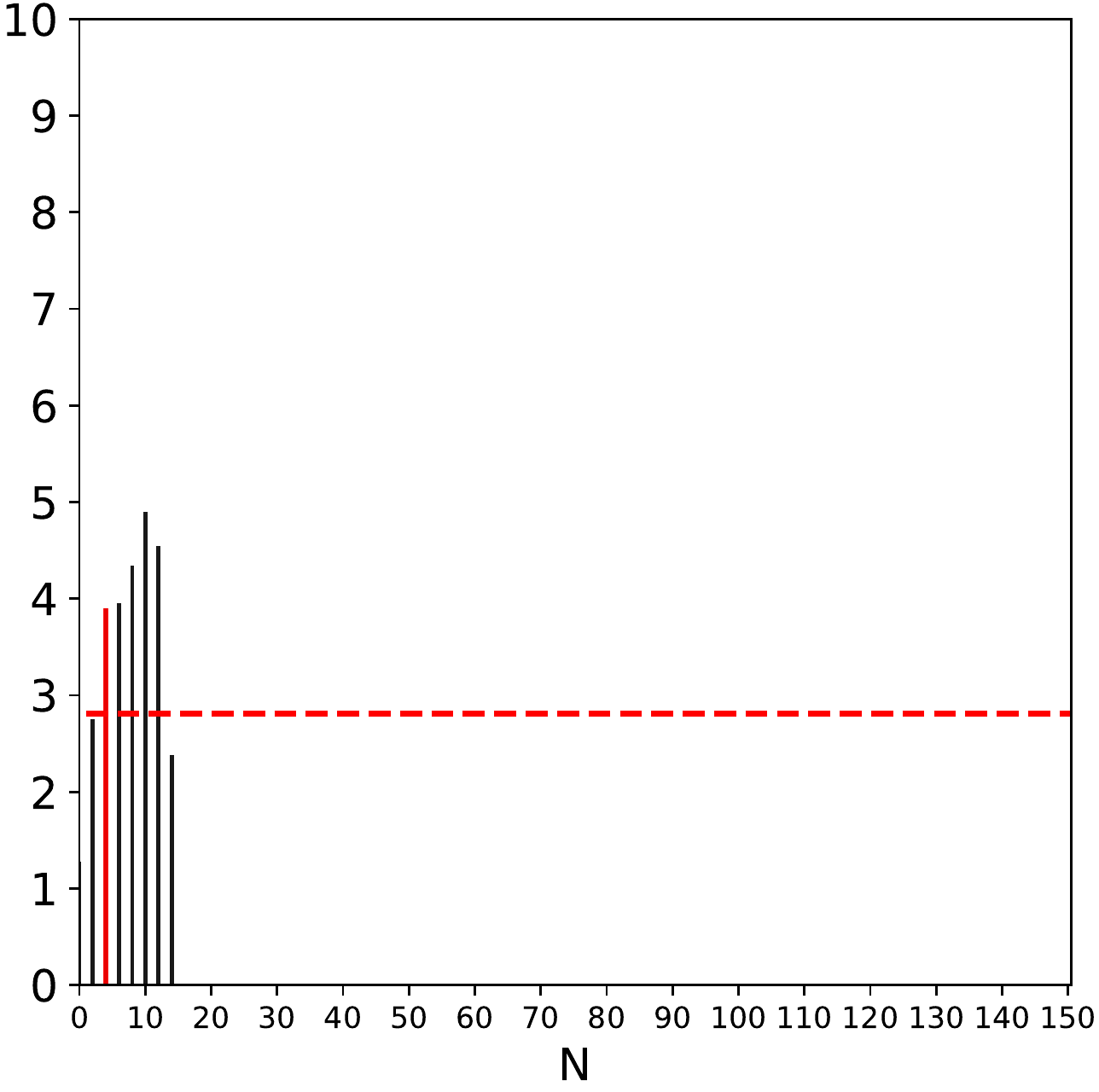}&
			\includegraphics[width=0.23\textwidth]{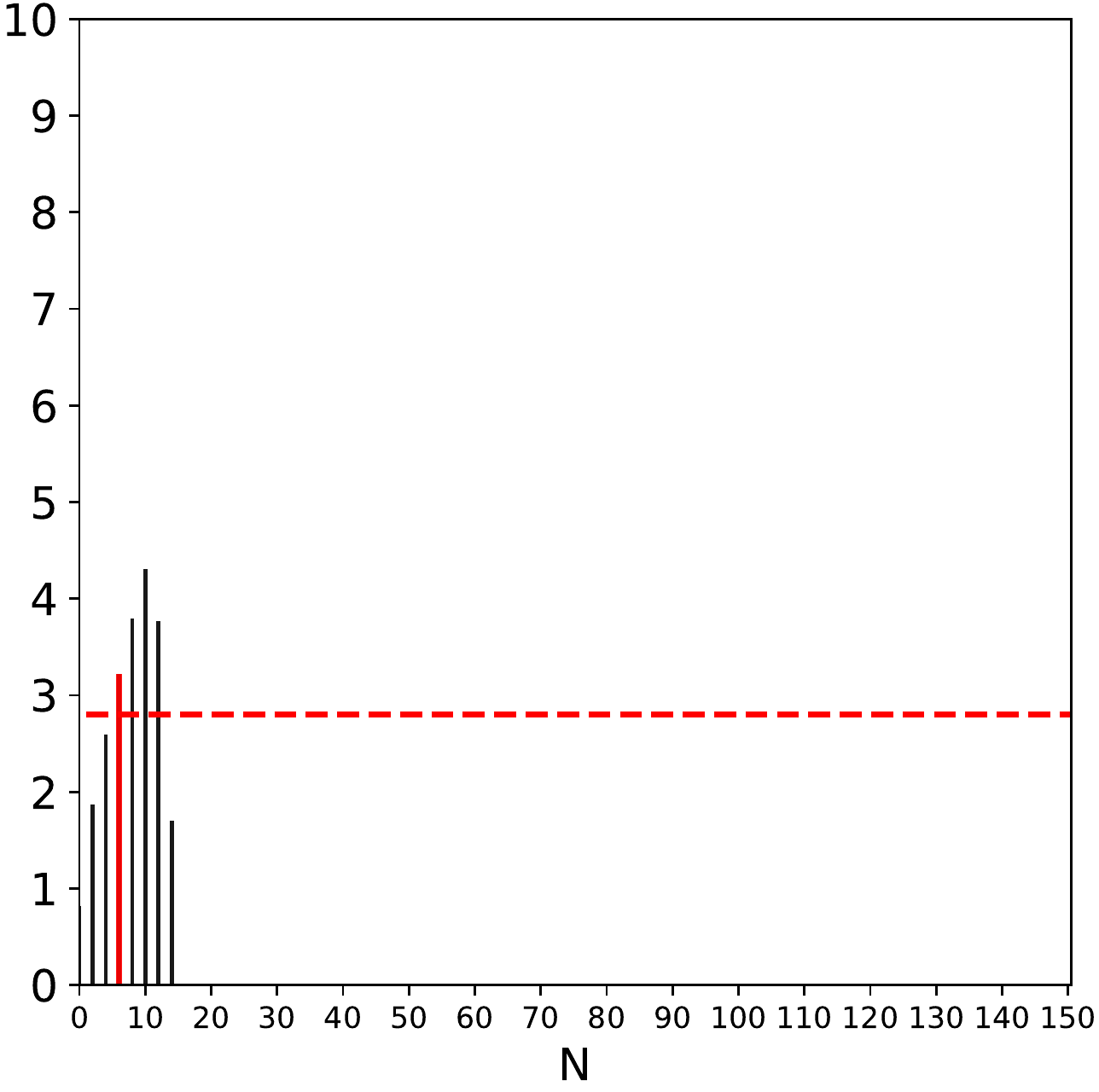} &	\includegraphics[width=0.23\textwidth]{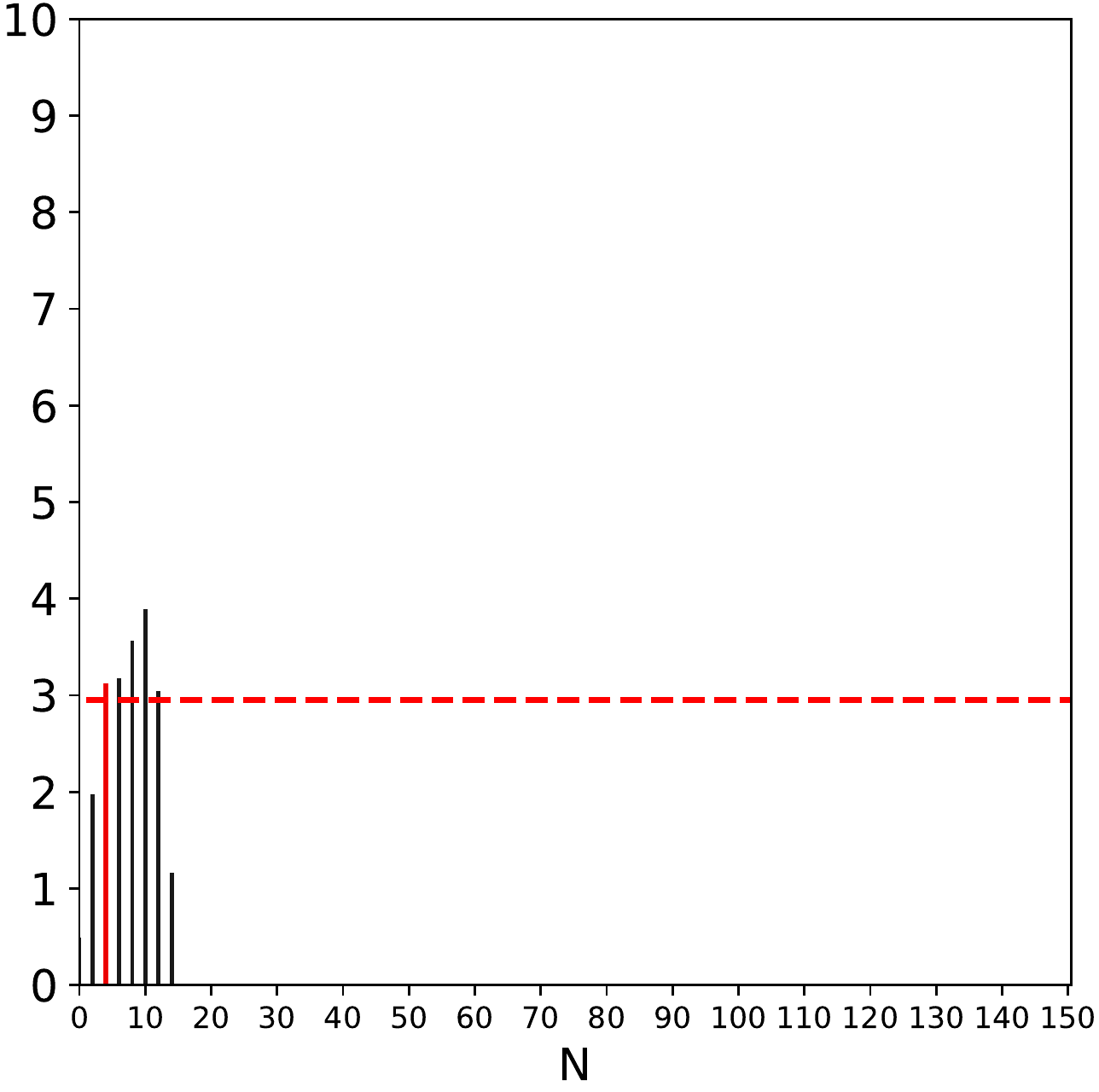} &
			\includegraphics[width=0.23\textwidth]{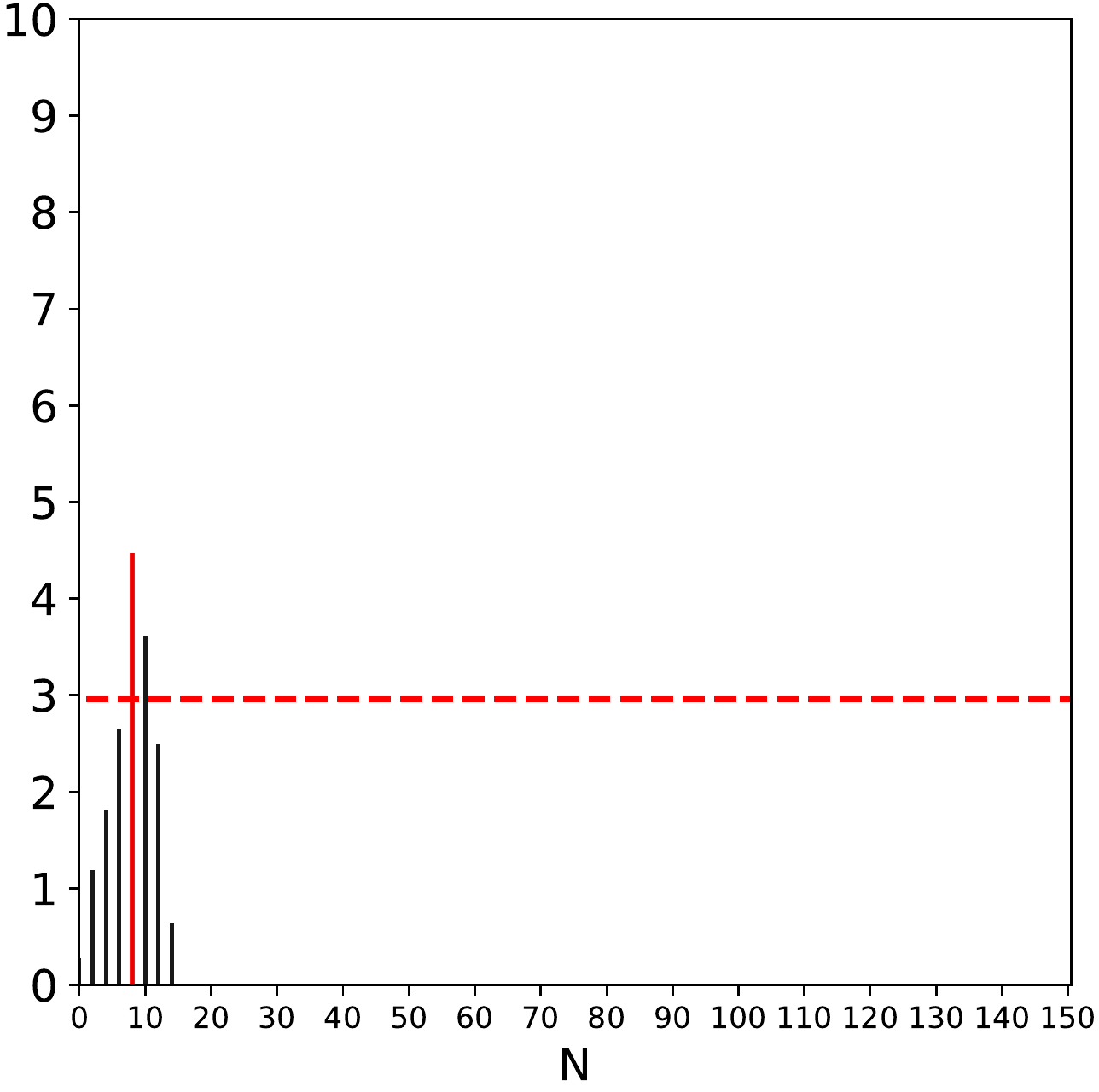}\\
			\begin{turn}{90}$\boldsymbol{b =2.0}$\end{turn}&
			\includegraphics[width=0.23\textwidth]{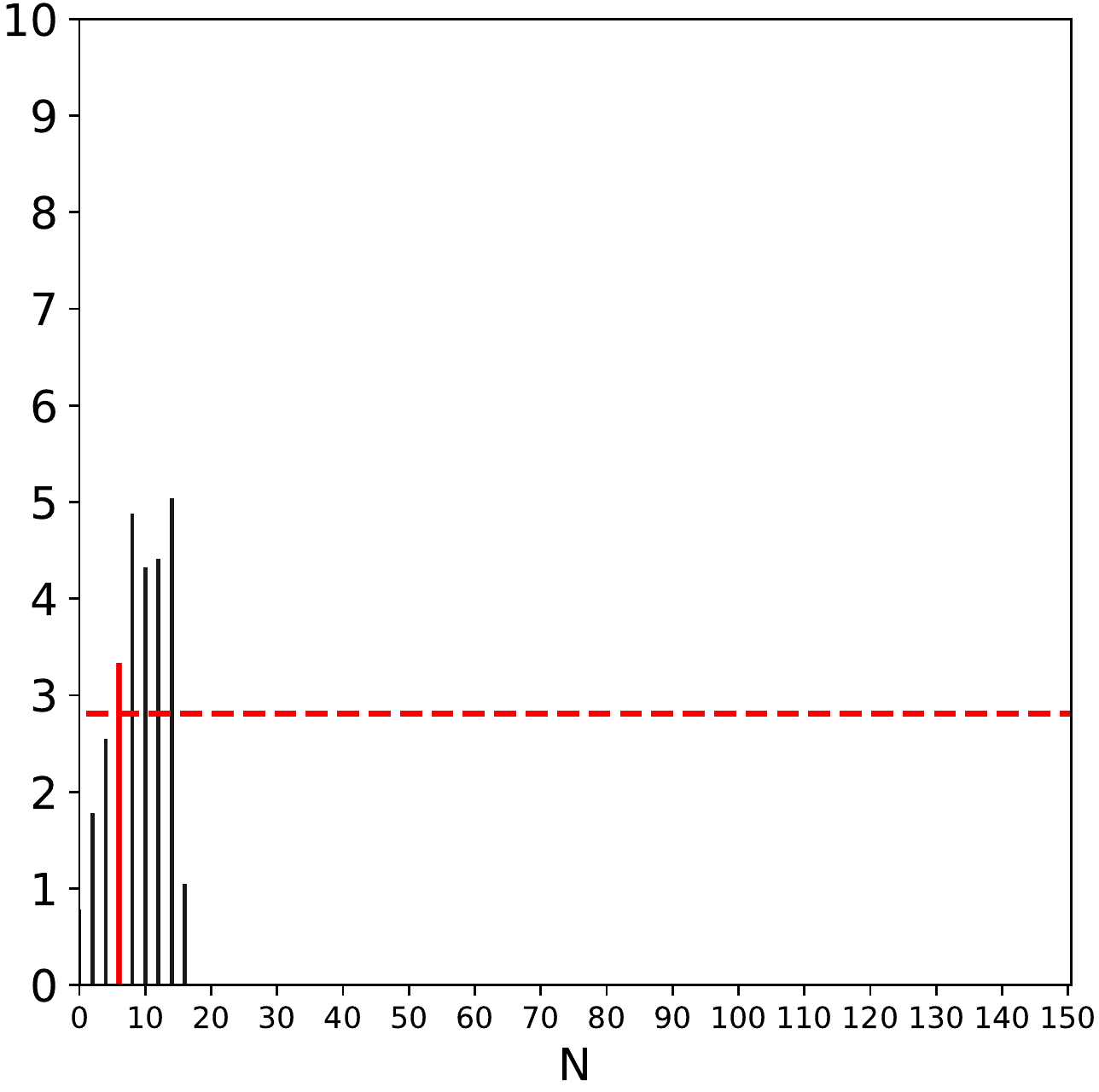}&
			\includegraphics[width=0.23\textwidth]{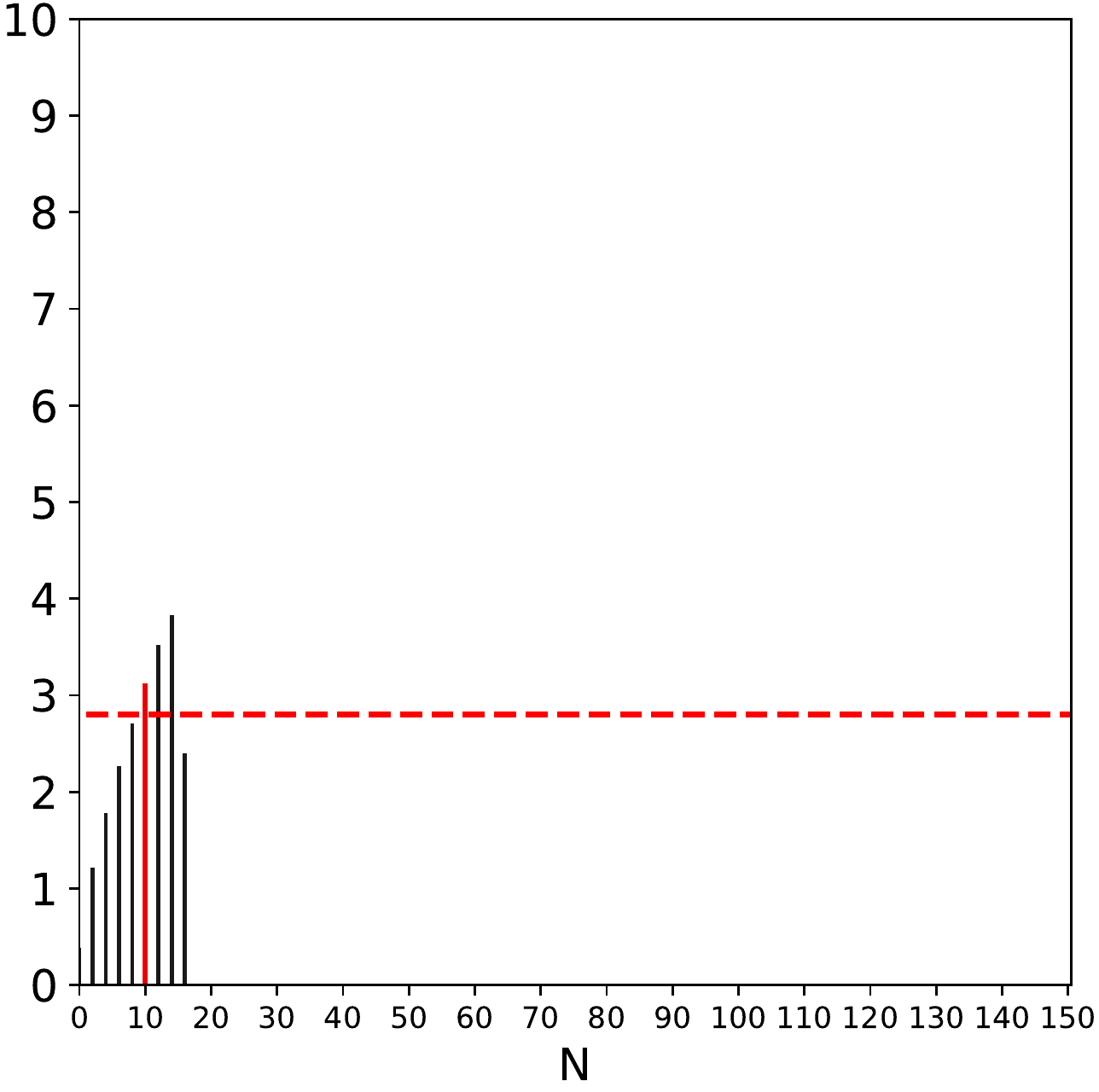} &	\includegraphics[width=0.23\textwidth]{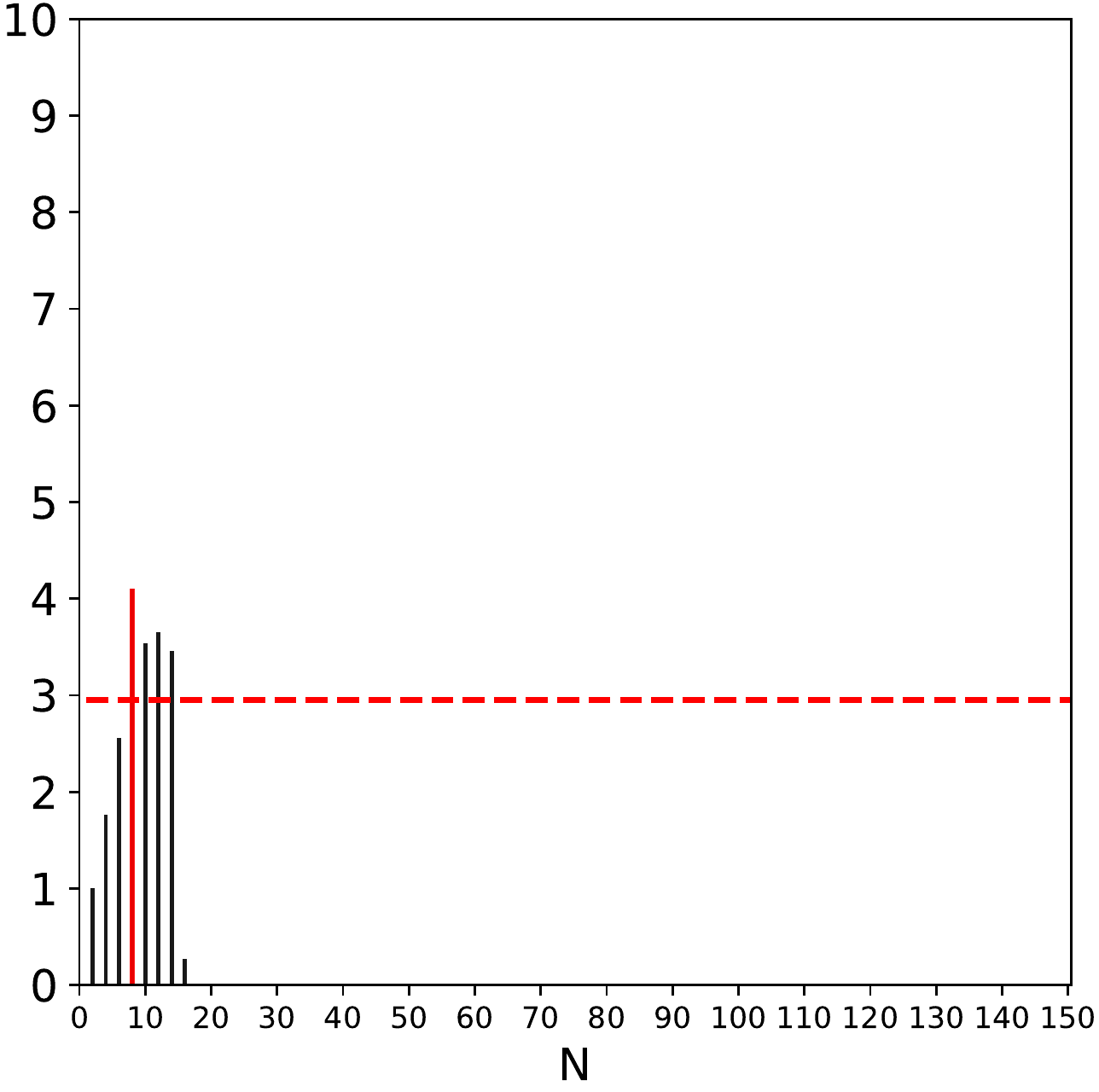} &
			\includegraphics[width=0.23\textwidth]{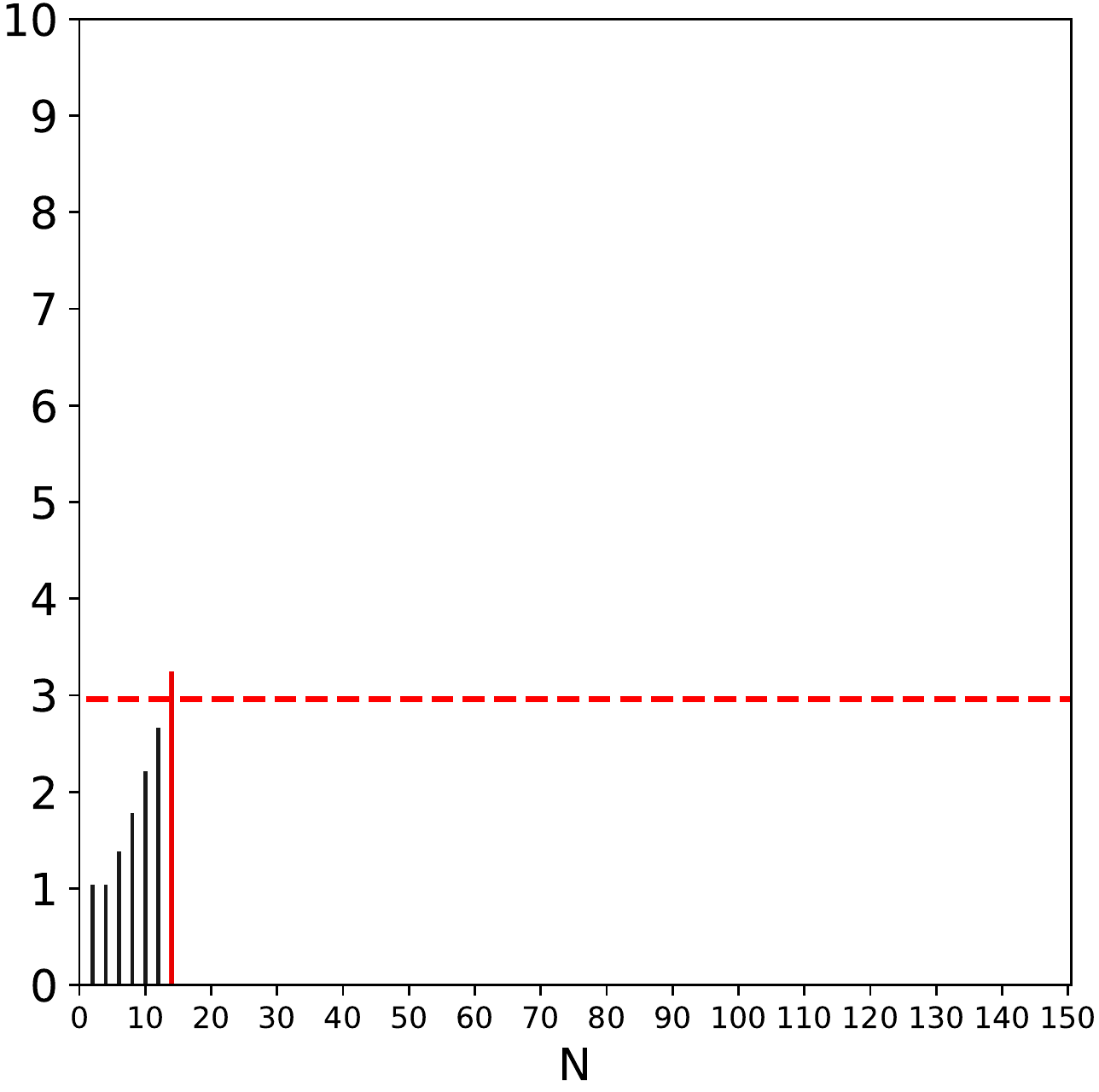}\\
			\begin{turn}{90}$\boldsymbol{b =4.0}$\end{turn}&
			\includegraphics[width=0.23\textwidth]{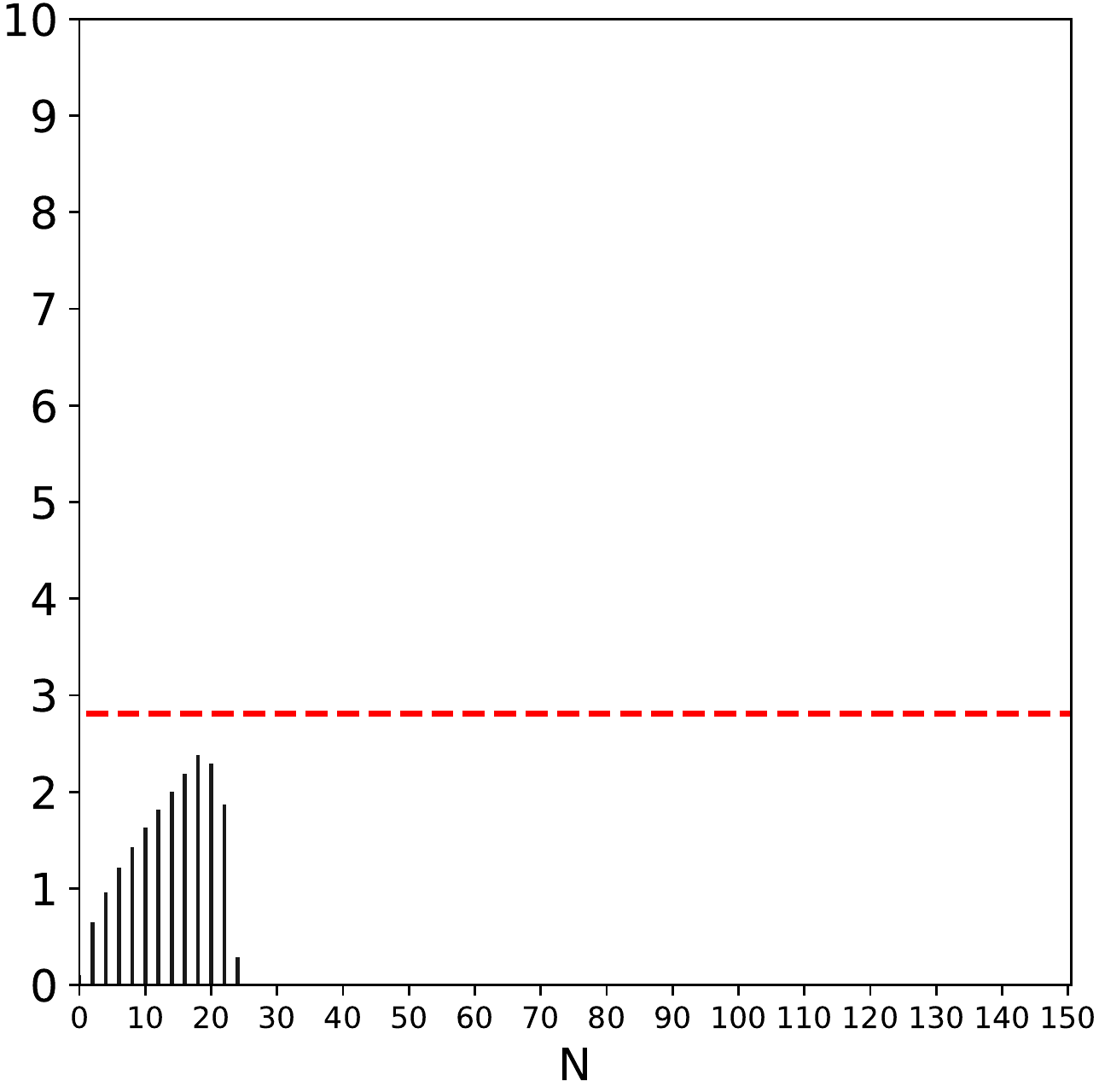}&
			\includegraphics[width=0.23\textwidth]{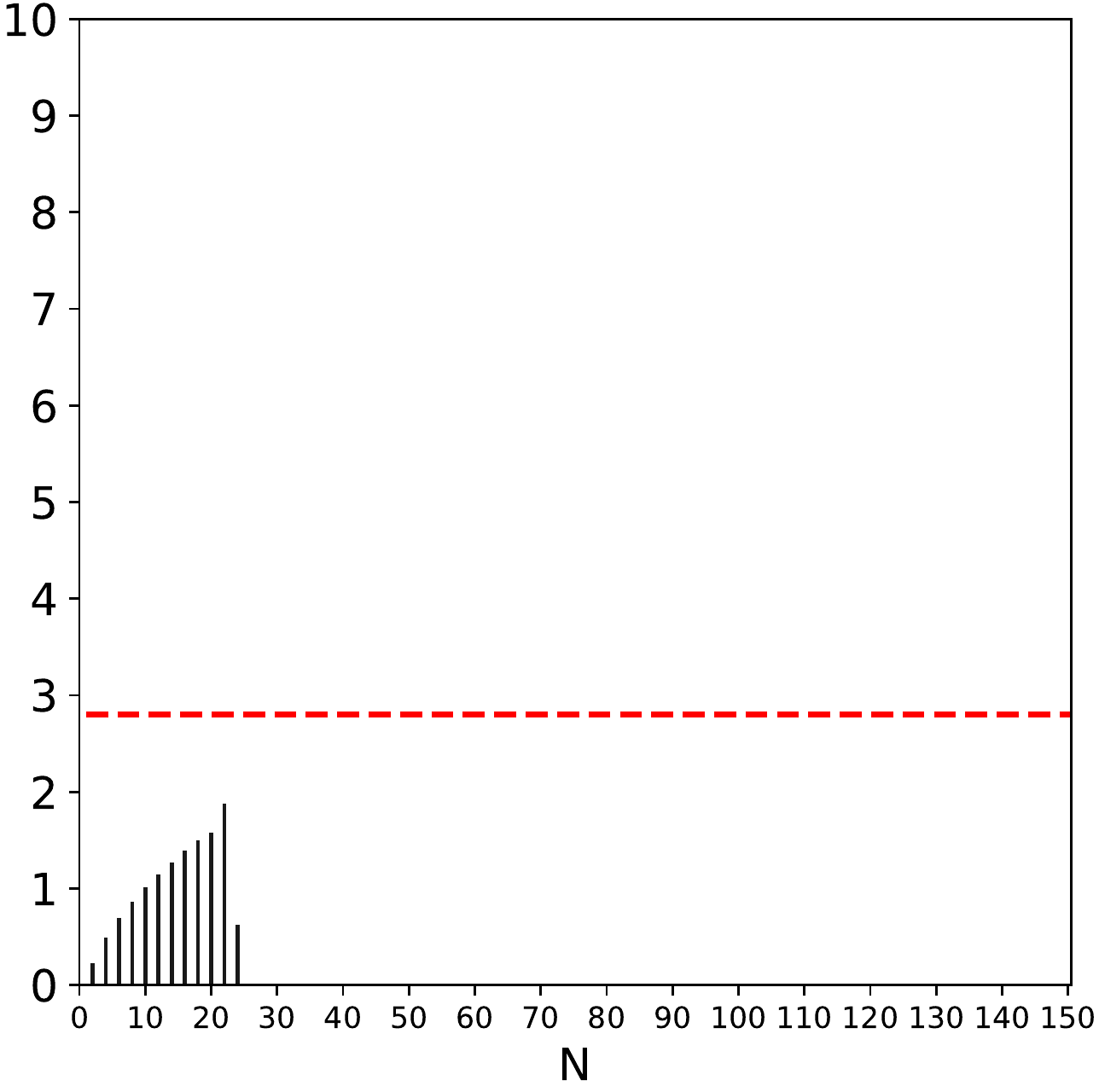} &	\includegraphics[width=0.23\textwidth]{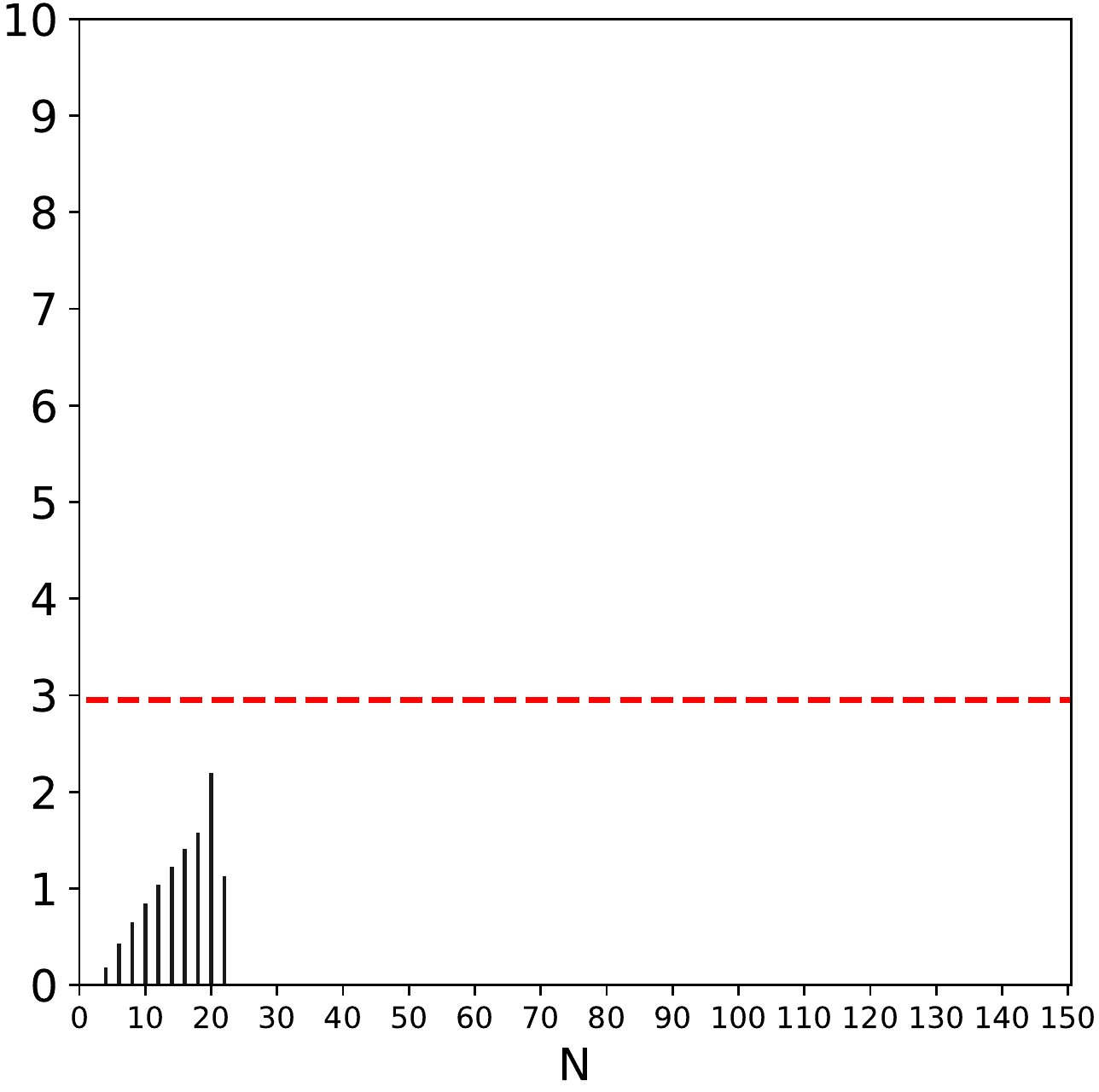} &
			\includegraphics[width=0.23\textwidth]{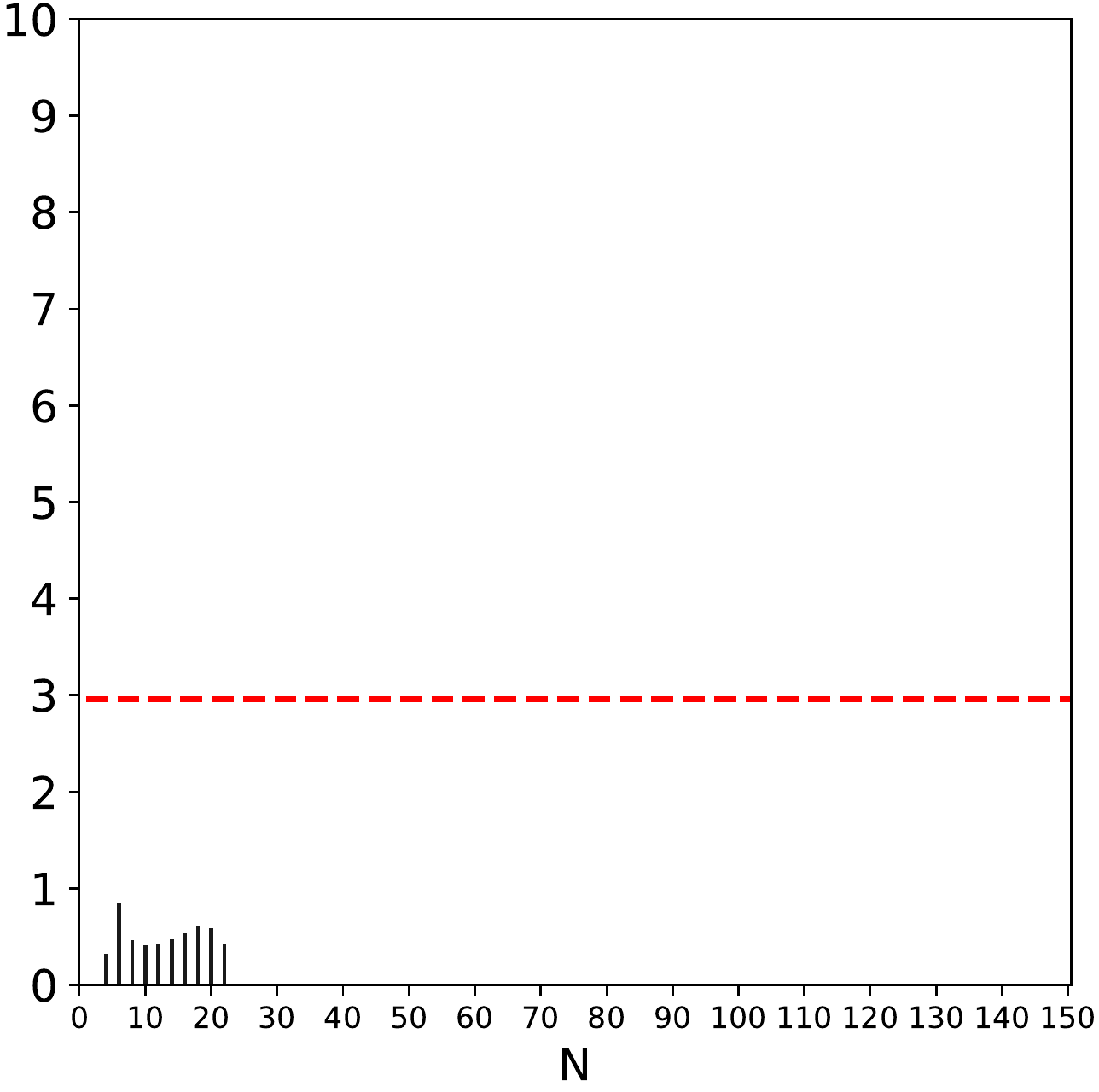}\\
			\begin{turn}{90}$\boldsymbol{b =6.0}$\end{turn}&
			\includegraphics[width=0.23\textwidth]{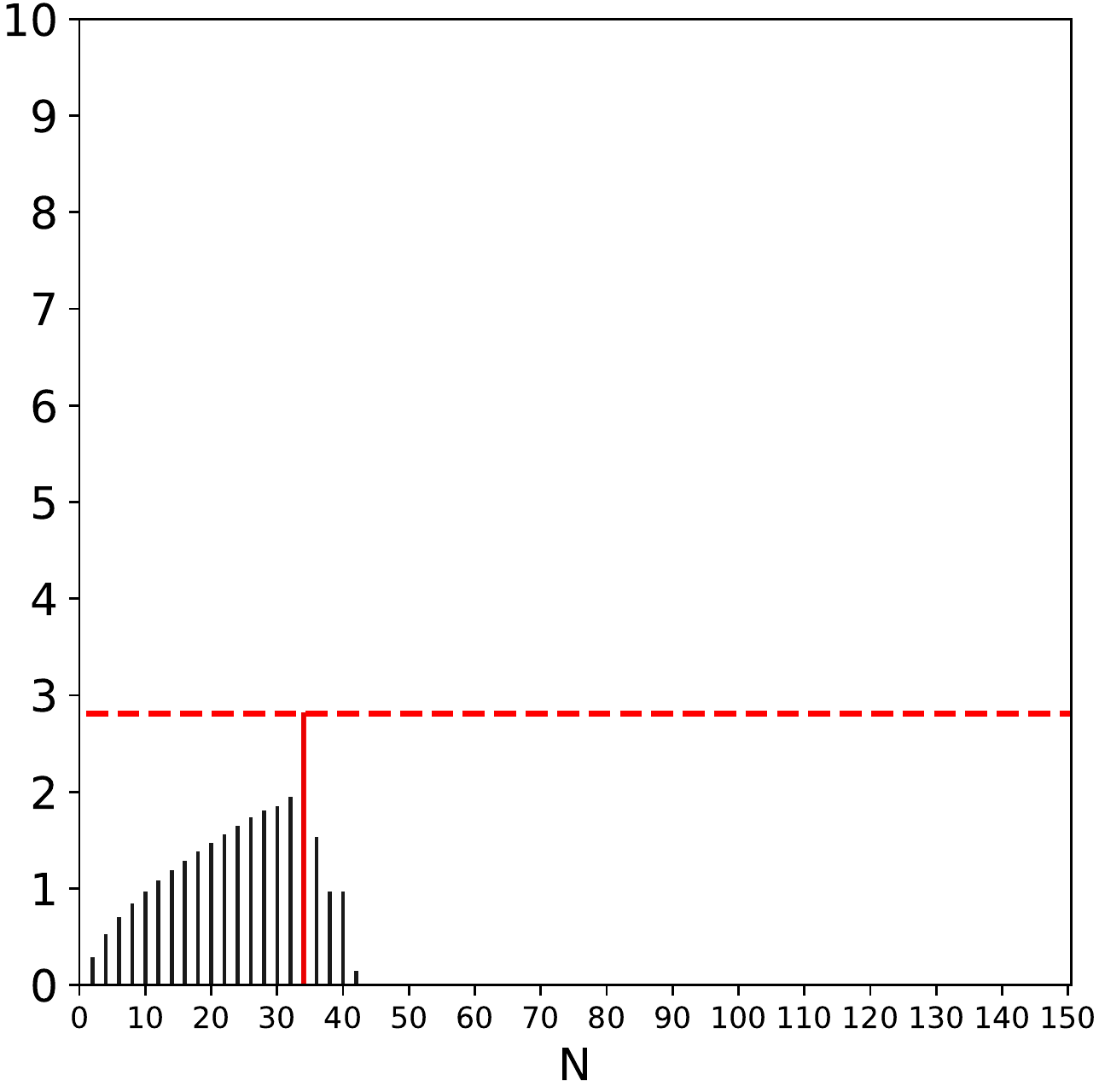}&
			\includegraphics[width=0.23\textwidth]{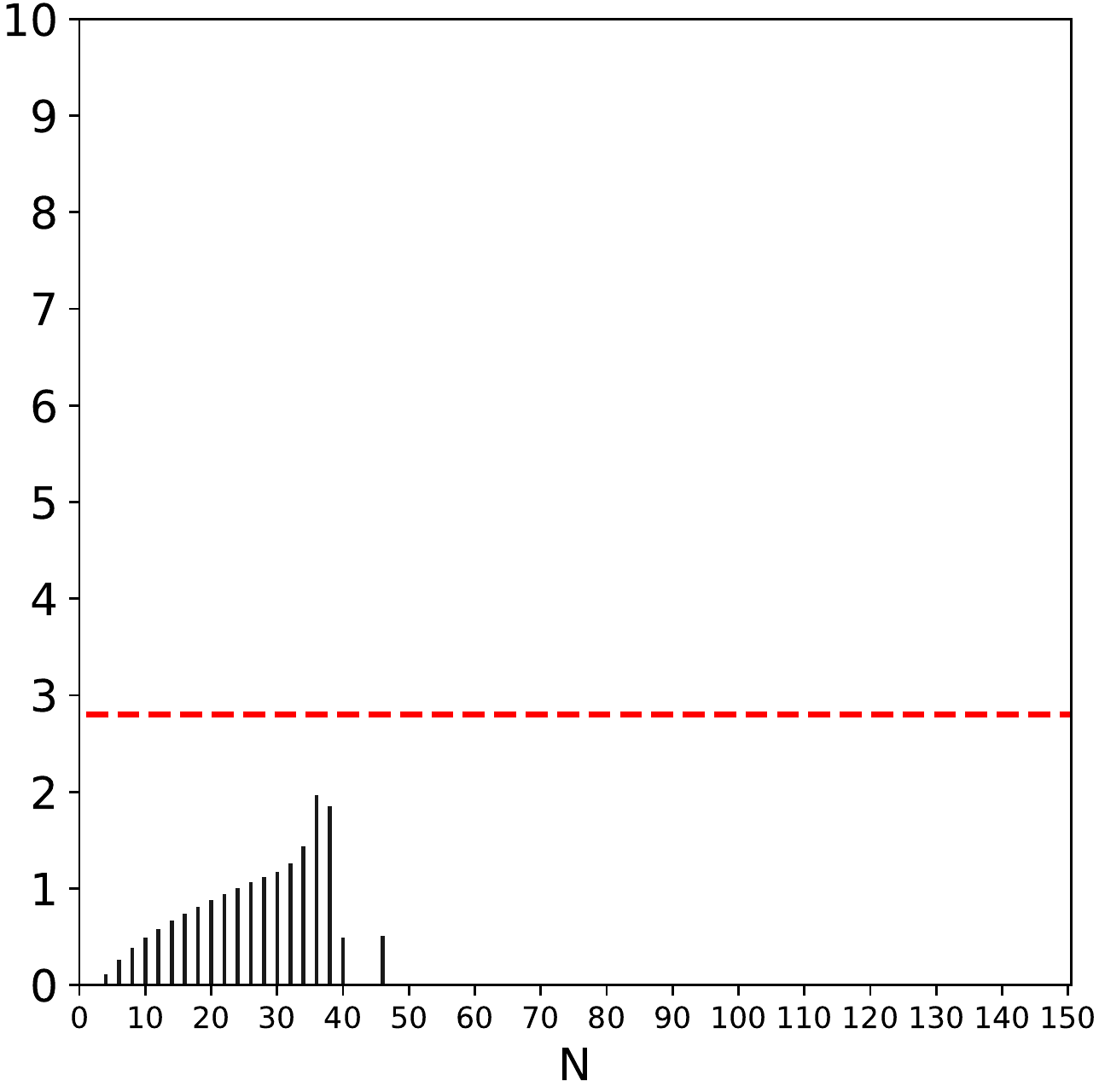} &	\includegraphics[width=0.23\textwidth]{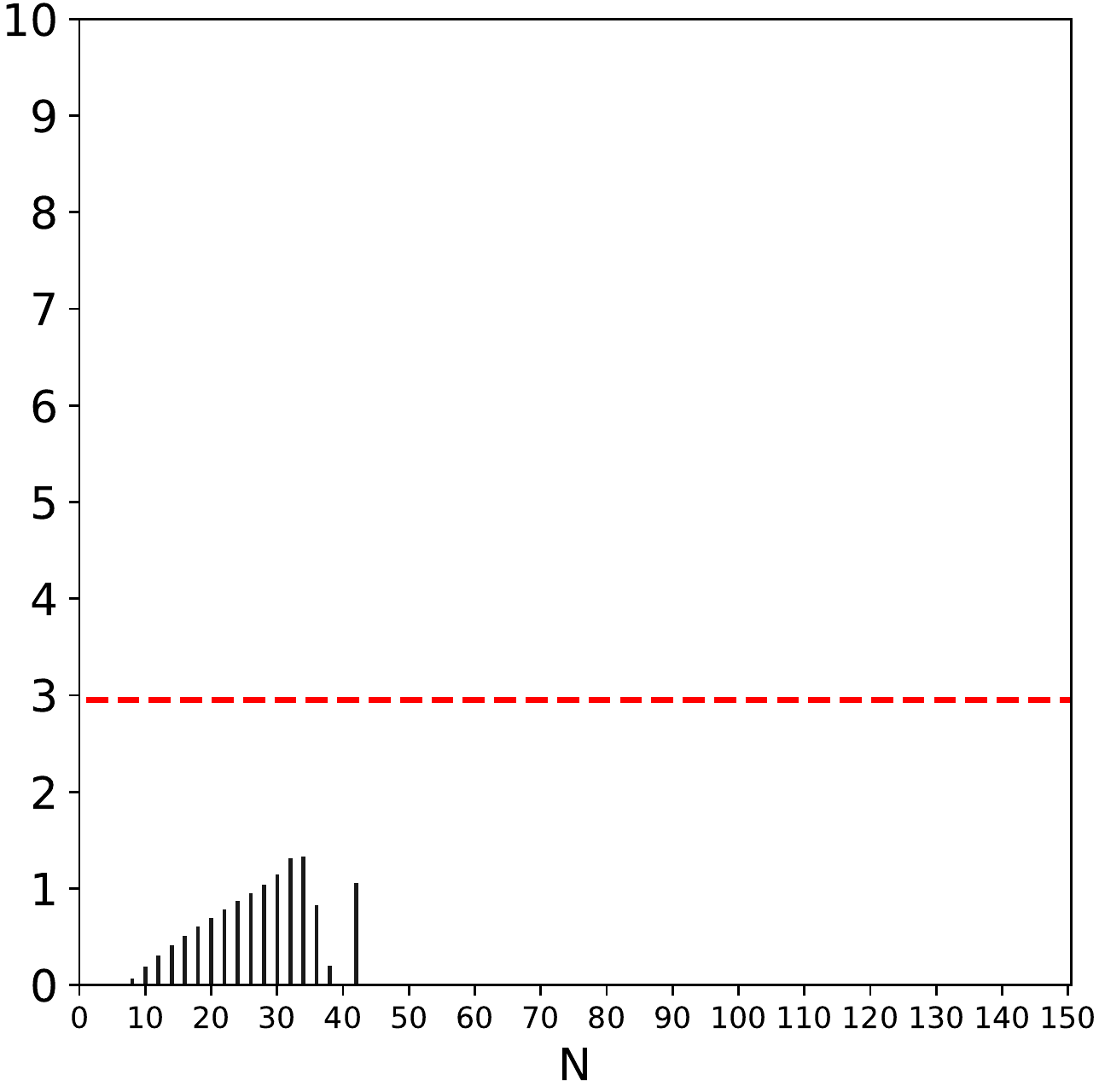} &
			\includegraphics[width=0.23\textwidth]{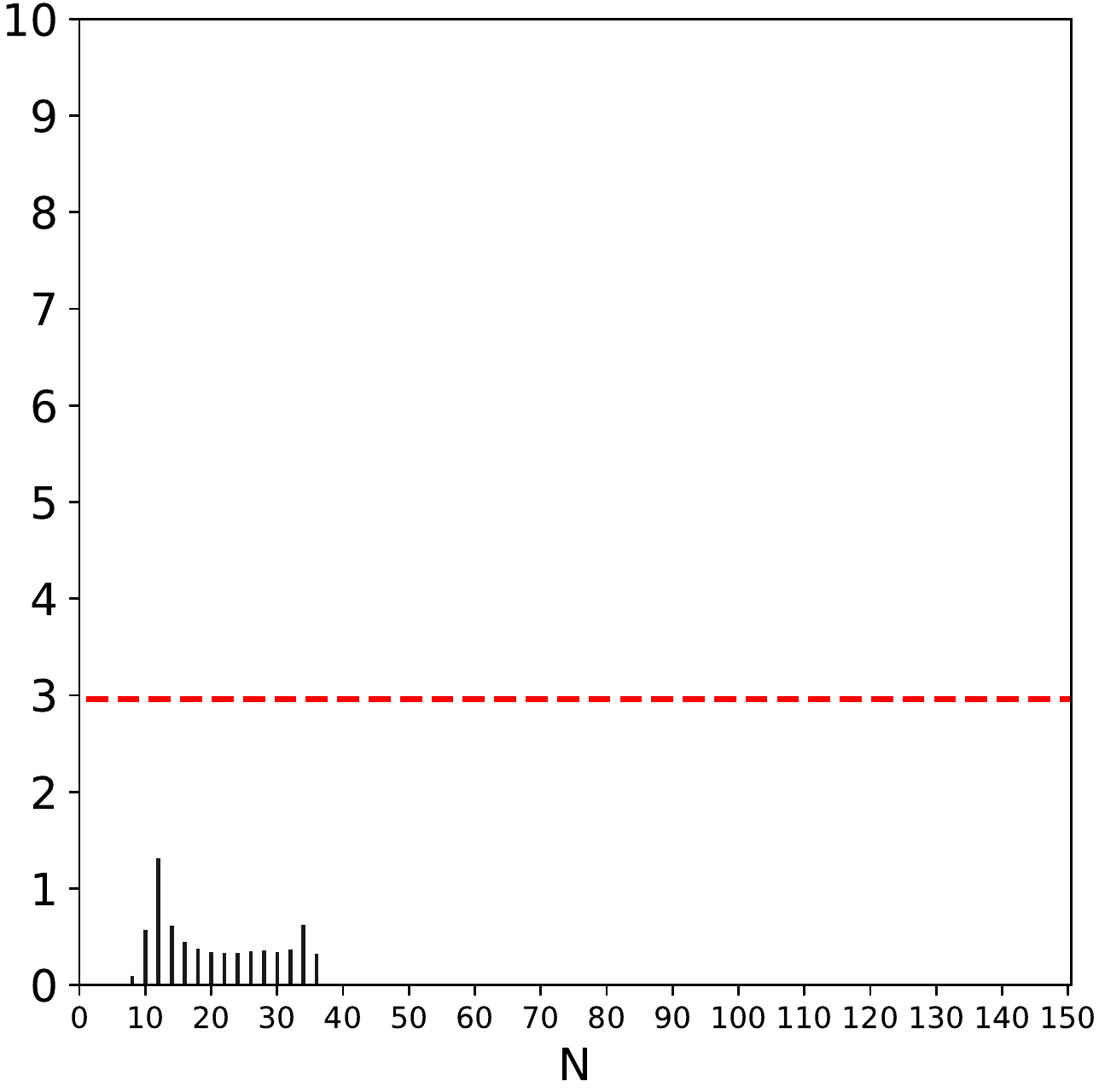}\\
		\end{tabular}
	}
	\caption{$\gamma_{a,b}^N$ as a function of $N$, when the underlying process is an OU with $\sigma_X(T;t)\approx 1.41$. The drift is $a= \mu_X(T;t)=20.0$. The dashed red horizontal lines indicate the accuracy of the MC method. The red vertical bars indicate when the Hermite series reaches the MC accuracy. \label{OUplot2}}
\end{figure}

\subsection{Polynomial jump-diffusion process}
We consider $X=Y$, where Y is the polynomial jump-diffusion process following the dynamics
\begin{equation*}
dY(t) = (b_0+b_1Y(t))\,dt + \sqrt{\sigma_0}\,dB(t)+\int_{\R}z\tilde{N}(dt,dz)
\end{equation*}
where $\tilde{N}(dt, dz)$ is a compensated Poisson random measure with compensator $\nu(dz)dt$. In this case, the jump measure $\ell(x, dz)$ is given by $\int_{\R}f(z)\ell(x,dz)= \int_{\R} f(\delta(x,z))\nu(dz)$, see \cite[Example 2.1]{benth2019} for details. Moreover, we consider $\nu$ to be the Lévy measure of a normal inverse Gaussian (NIG) process with
location parameter $\mu\in \R$, scale $\delta> 0$, asymmetry parameter $\beta$ and steepness parameter $\alpha$ (see \cite{barndorff1997}). 
Since $Y$ is a polynomial process, the moments of $X(T)$ are given in closed form by Corollary \ref{prop2}. However, there is no explicit price functional to be used as benchmark for the experiments. 

In Figure \ref{JDplot} we report the results for a jump-diffusion process with $(b_0, b_1, \sigma_0) =( -0.02, 0.01, 0.49)$, $(\alpha, \beta, \mu, \delta)= (1.0, 0, 0, 0.05)$, $T=2$, $t=0$ and initial condition $X(0) = 2.0$. The mean and standard deviation, calculated with the moment formula in Corollary \ref{prop2}, are $\mu_X(T;t)=2.00$ and $\sigma_X(T;t)=1.00$ respectively. Since no closed price formula is available in this case, we cannot access the accuracy of the approximation as before. The plots of Figure \ref{JDplot} show then the approximated price with generalized Hermite polynomials compared with the approximated price via Monte Carlo. For each experiment, we report two plots: one showing the convergence/non-convergence at large scale and one zoomed.

In particular, we use values for $b$ which have the same scale with respect to $\underline{b}_{\sigma}= \frac{\sqrt{\sigma}}{2}$ (for $\sigma = \sigma_X(T;t)$) of Figure \ref{BMplot}. Even if $\underline{b}_{\sigma}$ has been introduced in Proposition \ref{condb} for a Gaussian random variable, the results in Figure \ref{JDplot} are in line with the previous ones. Specifically, we see that for higher values of the scale the convergence is slower but more stable at the same time. For example, for $b = 2.0\underline{b}_{\sigma}$, the convergence is reached around $N=17$, but after $N=47$, due to numerical instabilities, the series starts to diverge. On the other hand, for $b = 6.0\underline{b}_{\sigma}$, we see that $N=60$ terms are not enough to reach convergence. Last, we notice that for $b = \underline{b}_{\sigma}$ there is no convergence, as expected in line with Proposition \ref{condb}. This is visible both for $K=1.0$ and $K=2.0$ in the zoomed plots: we see that the option price oscillates around the MC price, without reaching convergence. After a certain number of iterations (around $N=30$), because of numerical instabilities, the series starts to diverge.

\begin{figure}[!tp]
	\setlength{\tabcolsep}{2pt}
	\resizebox{1\textwidth}{!}{
		\begin{tabular}{@{}>{\centering\arraybackslash}m{0.04\textwidth}@{}>{\centering\arraybackslash}m{0.24\textwidth}@{}>{\centering\arraybackslash}m{0.24\textwidth}@{}>{\centering\arraybackslash}m{0.24\textwidth}@{}>{\centering\arraybackslash}m{0.24\textwidth}@{}}
			& $\boldsymbol{K = 1.0}$&$\boldsymbol{K = 1.0}$ (zoom) & $\boldsymbol{K = 2.0}$& $\boldsymbol{K = 2.0}$ (zoom) \\
			\begin{turn}{90}$\boldsymbol{b \!=1.0\, \underline{b}_{\sigma}\! \approx\!0.71}$\end{turn}&
			\includegraphics[width=0.23\textwidth]{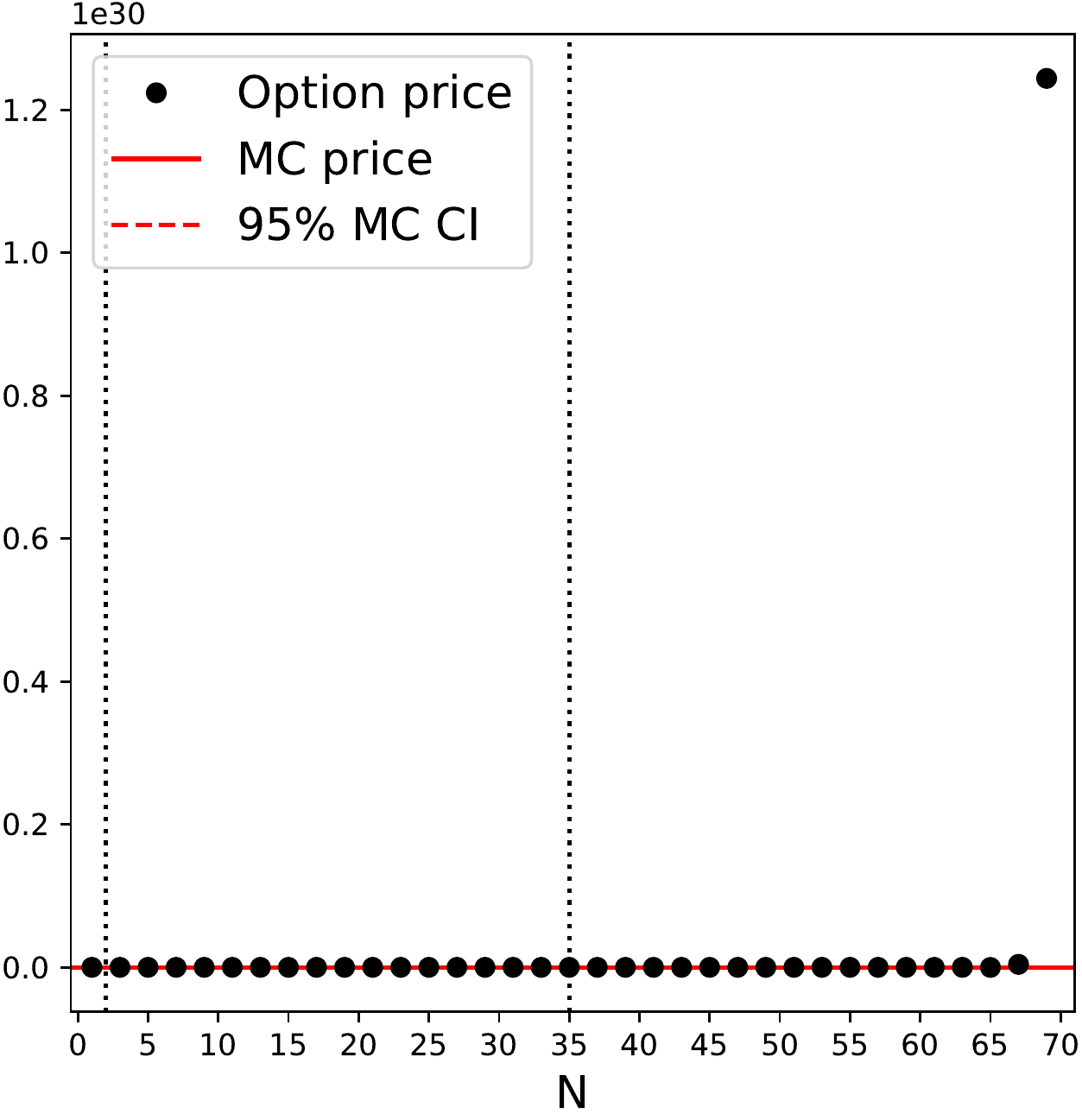}&
			\includegraphics[width=0.23\textwidth]{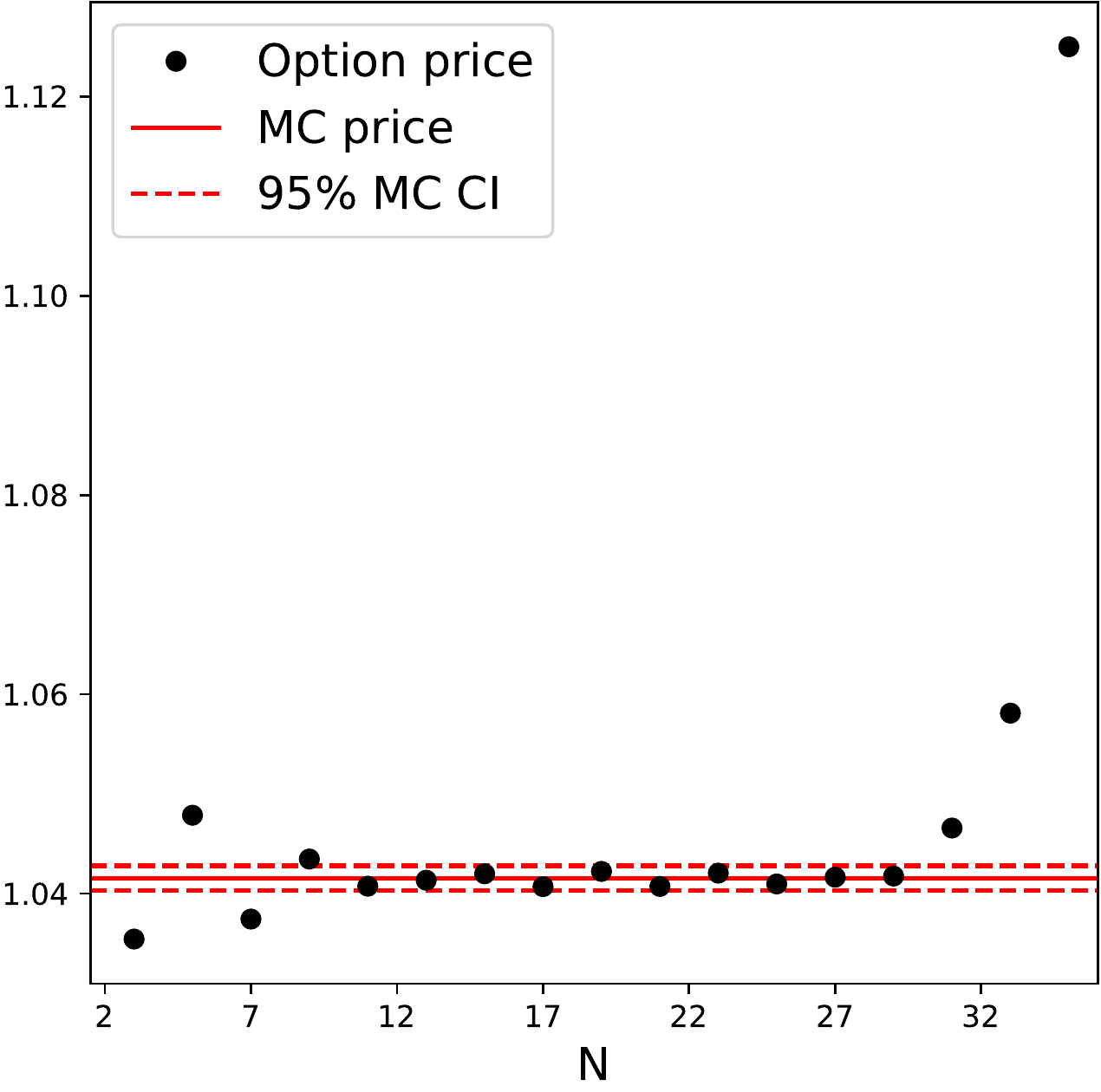} &	\includegraphics[width=0.23\textwidth]{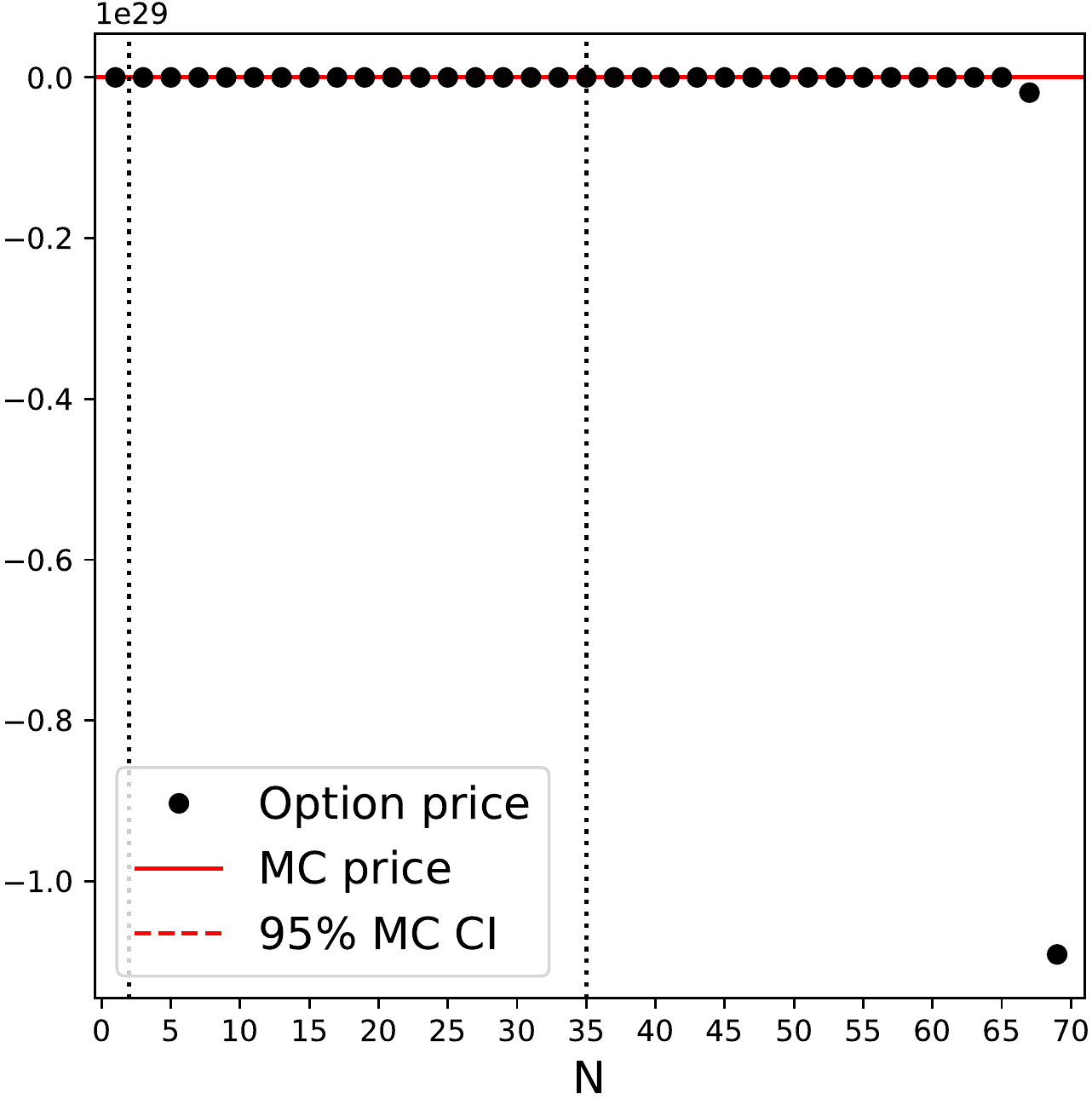} &
			\includegraphics[width=0.23\textwidth]{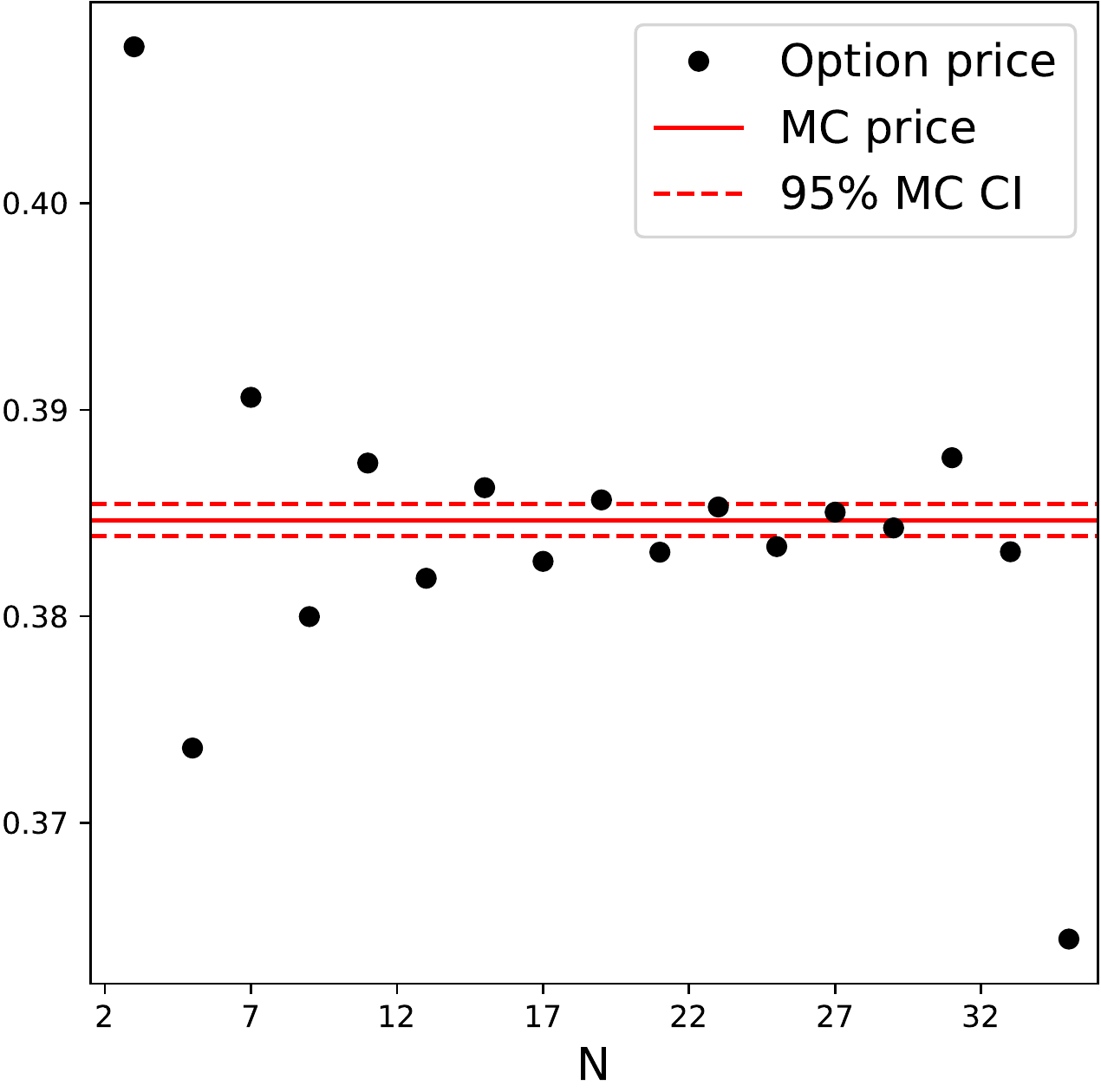}\\
			\begin{turn}{90}$\boldsymbol{b \!=1.2\, \underline{b}_{\sigma} \!\approx\!0.85}$\end{turn}&
			\includegraphics[width=0.23\textwidth]{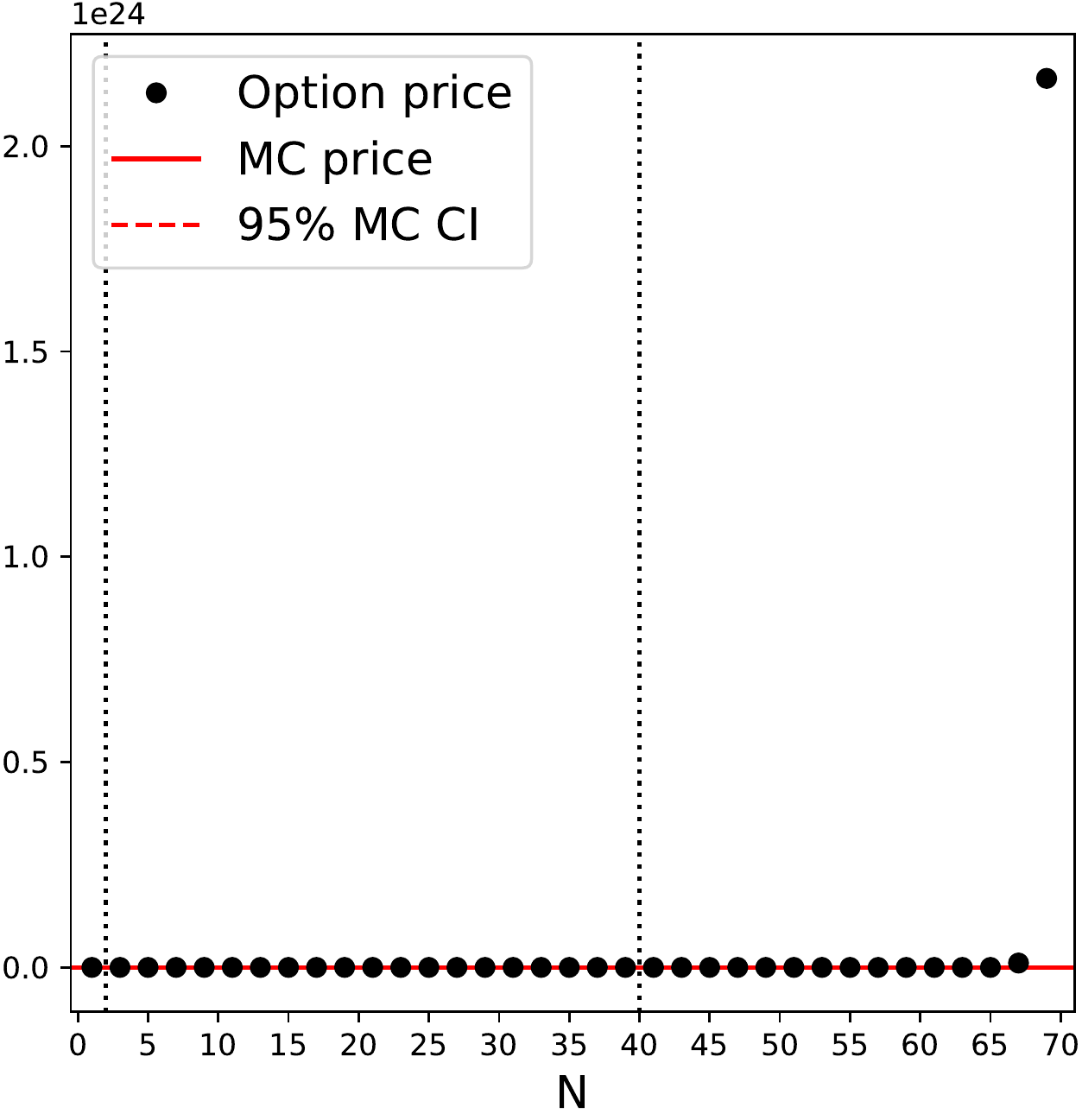}&
			\includegraphics[width=0.23\textwidth]{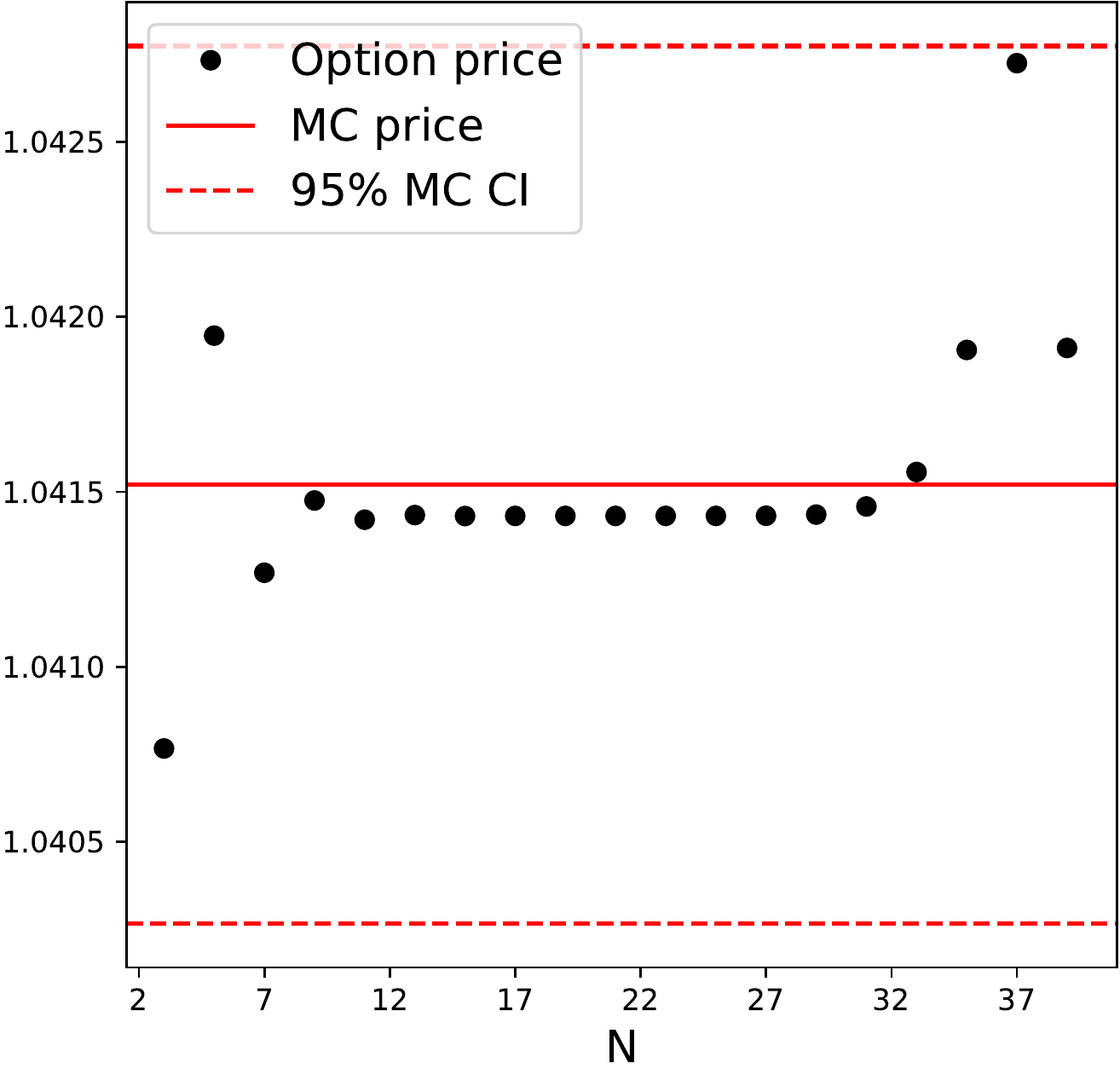} &	\includegraphics[width=0.23\textwidth]{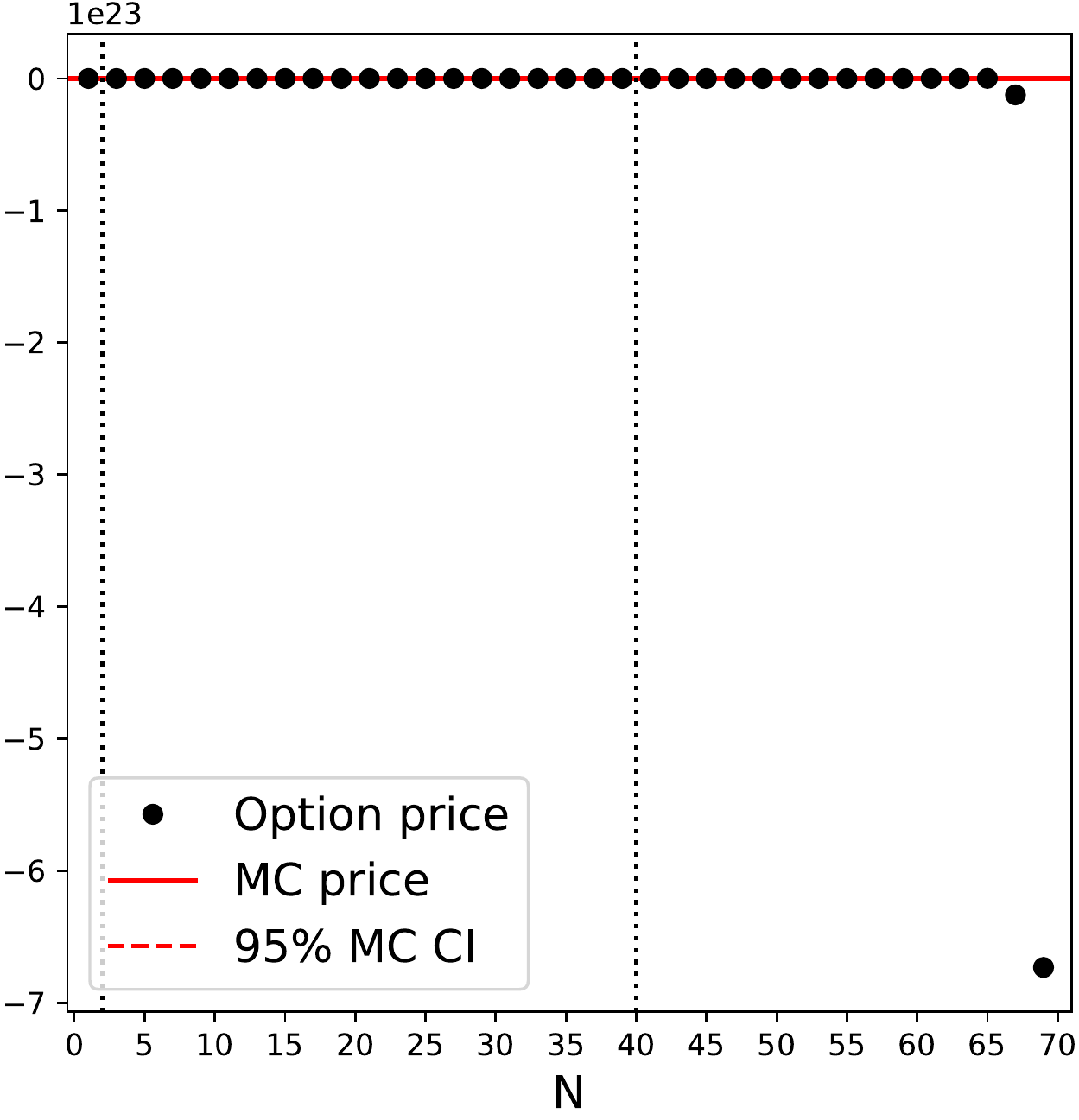} &
			\includegraphics[width=0.23\textwidth]{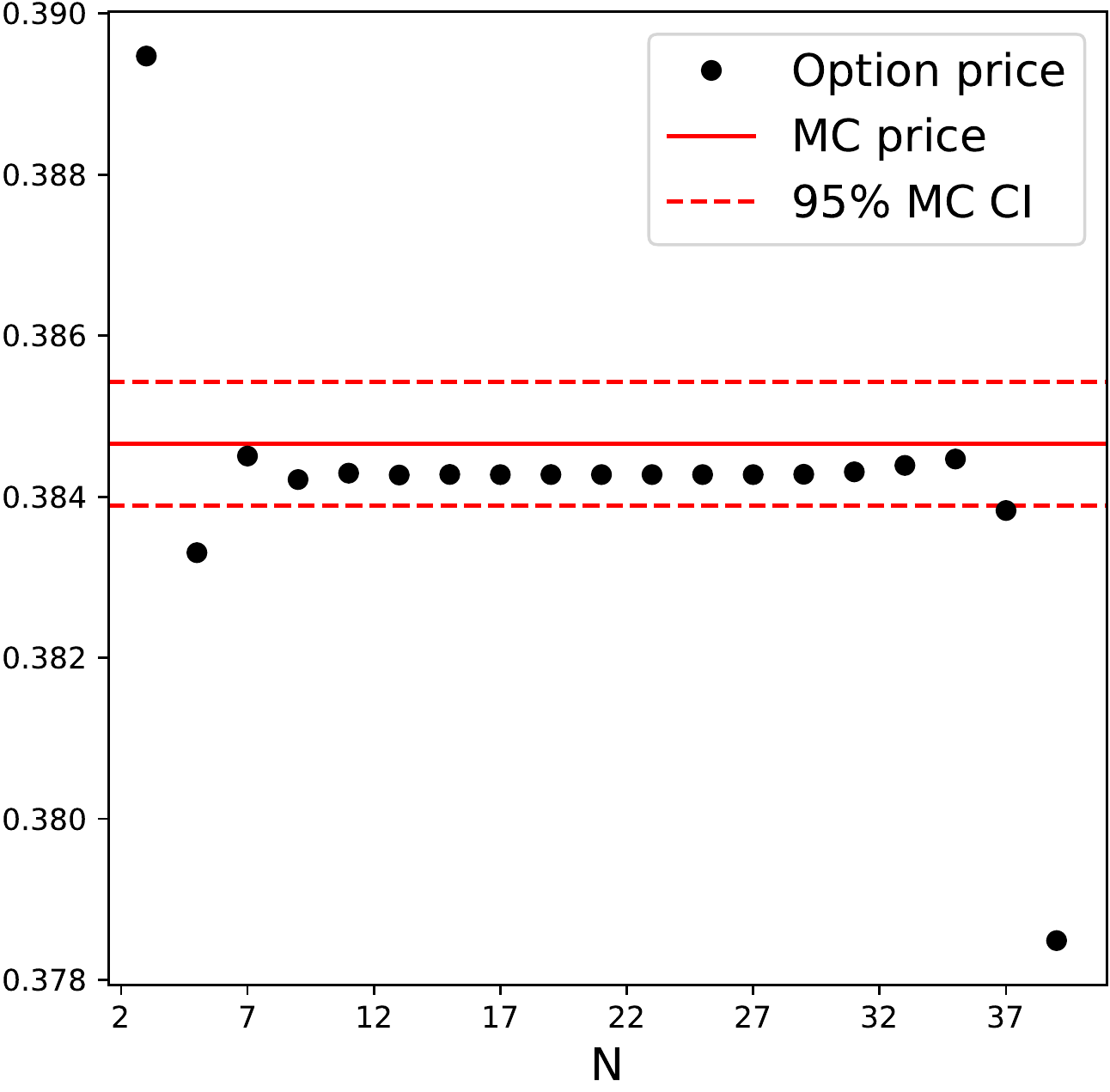}\\
			\begin{turn}{90}$\boldsymbol{b \!=2.0\, \underline{b}_{\sigma} \!\approx\!1.42}$\end{turn}&
			\includegraphics[width=0.23\textwidth]{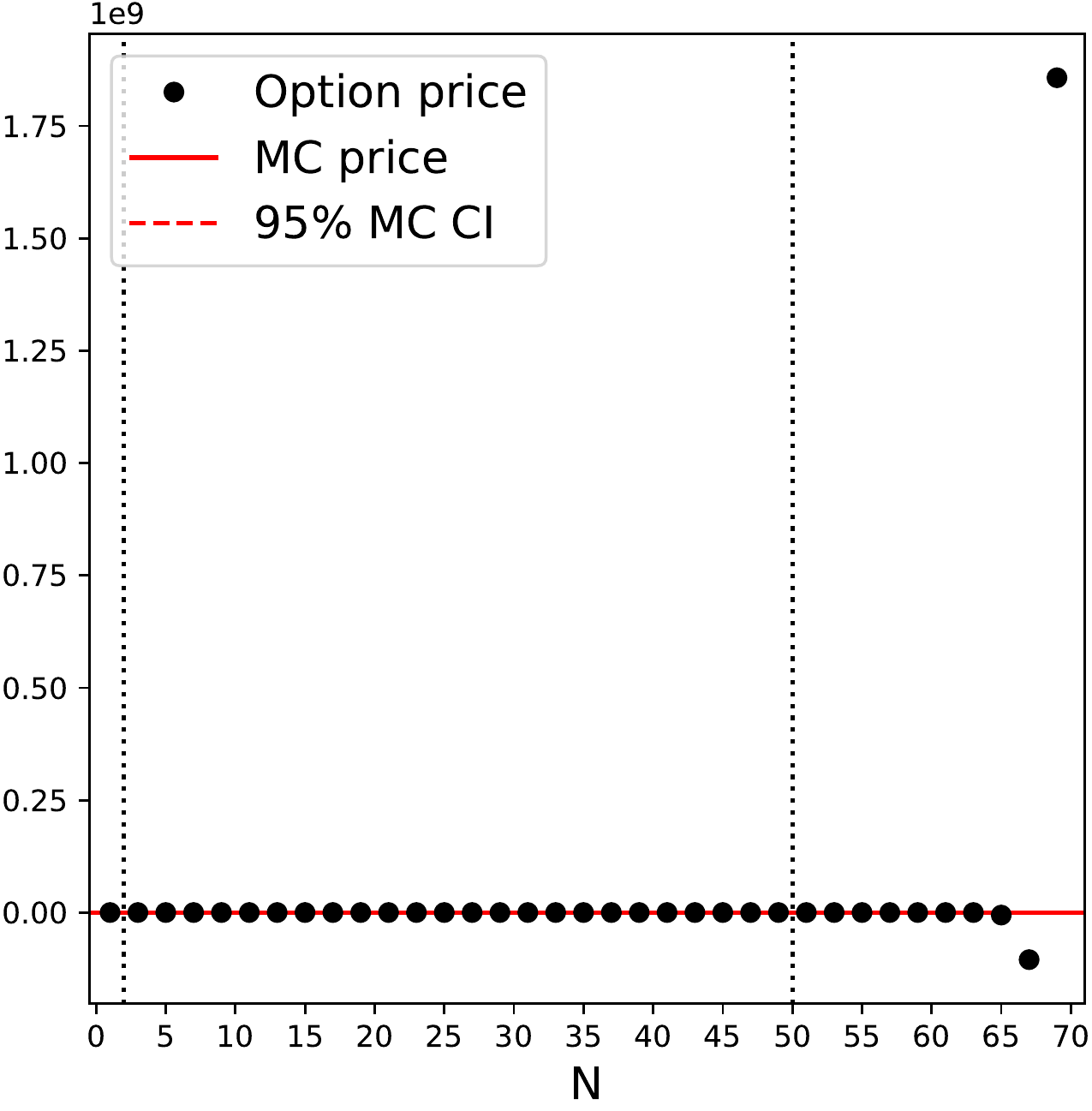}&
			\includegraphics[width=0.23\textwidth]{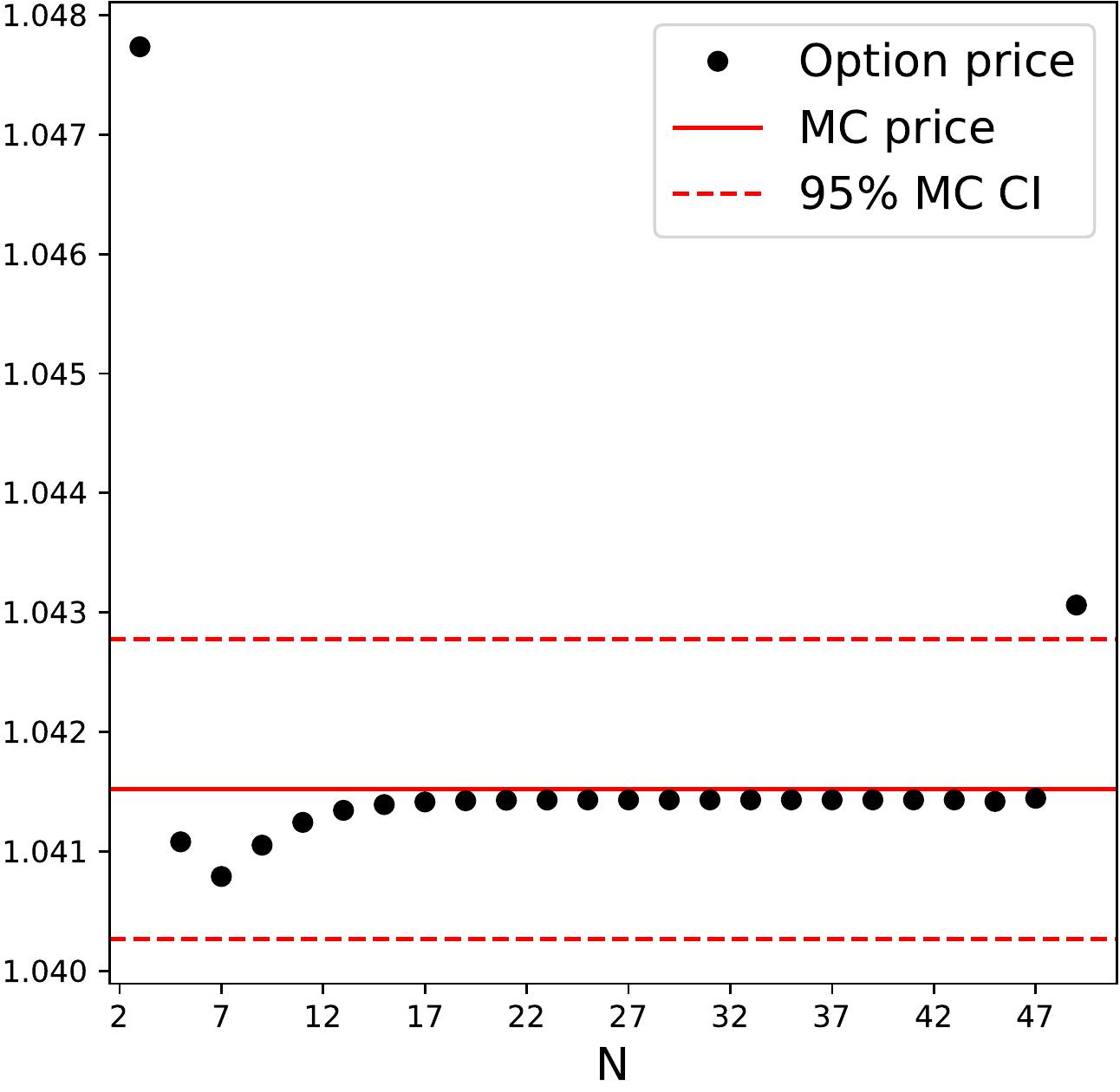} &	\includegraphics[width=0.23\textwidth]{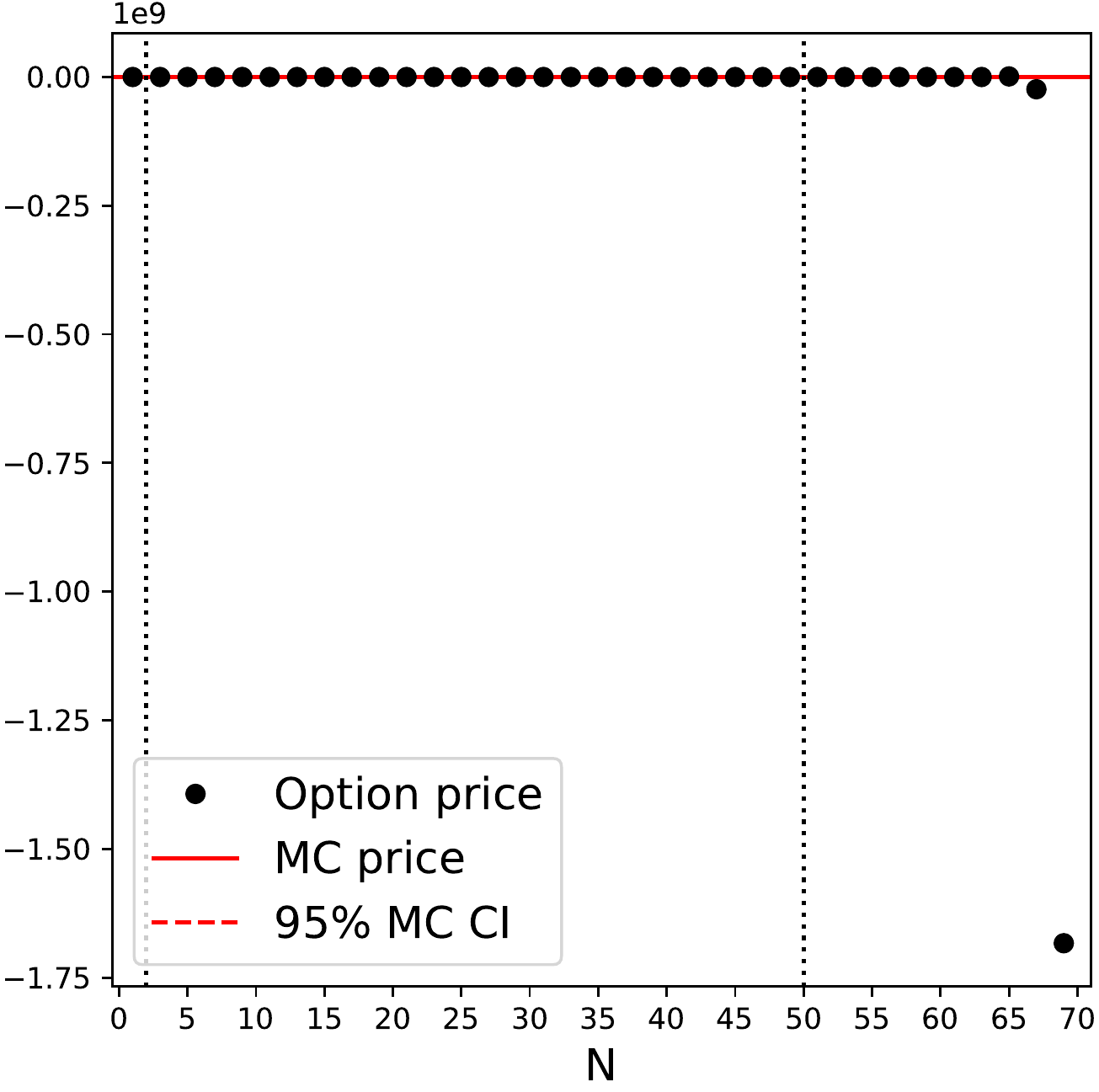} &
			\includegraphics[width=0.23\textwidth]{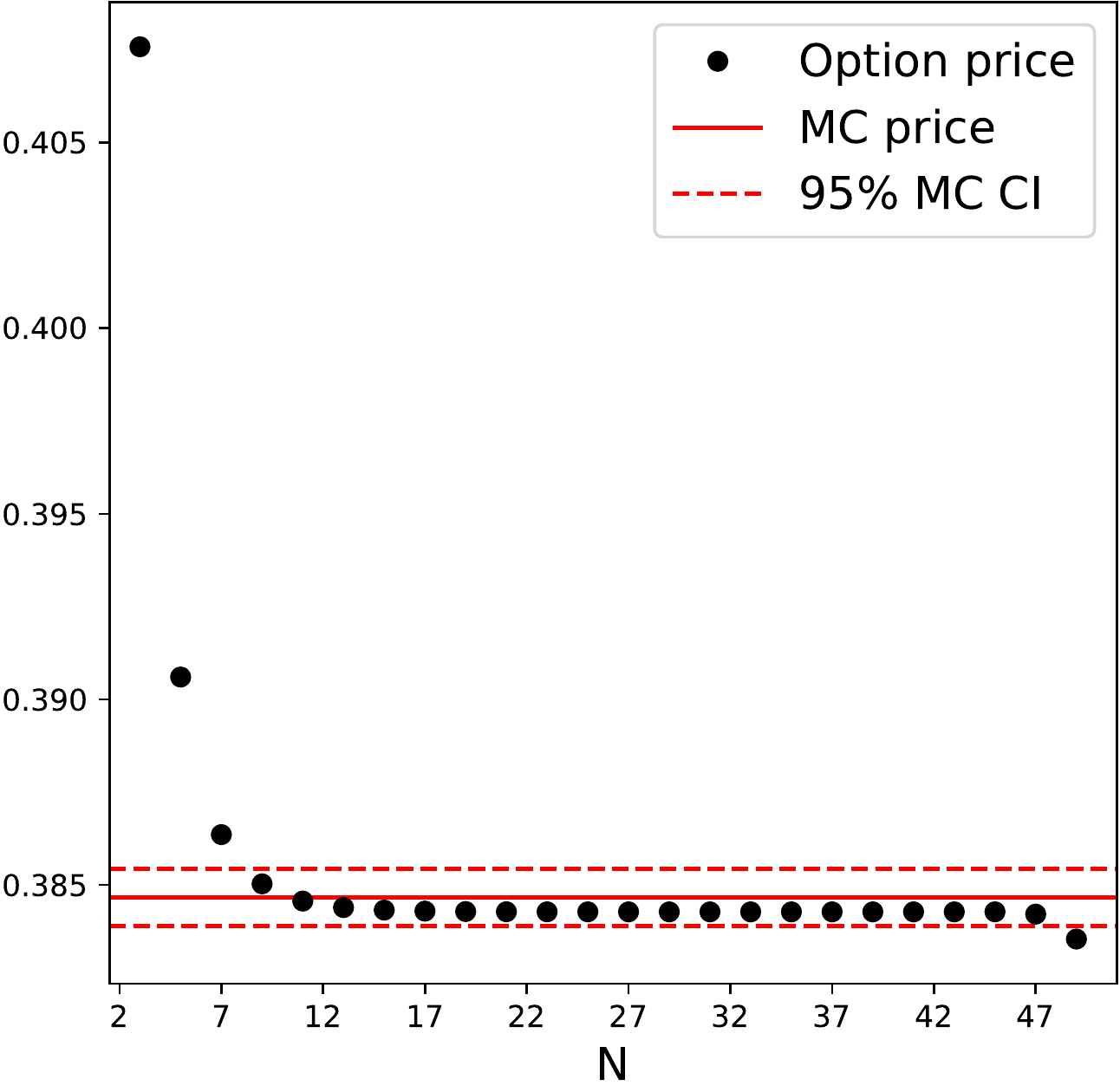}\\
			\begin{turn}{90}$\boldsymbol{b \!=4.0\, \underline{b}_{\sigma} \!\approx\!2.84}$\end{turn}&
			\includegraphics[width=0.23\textwidth]{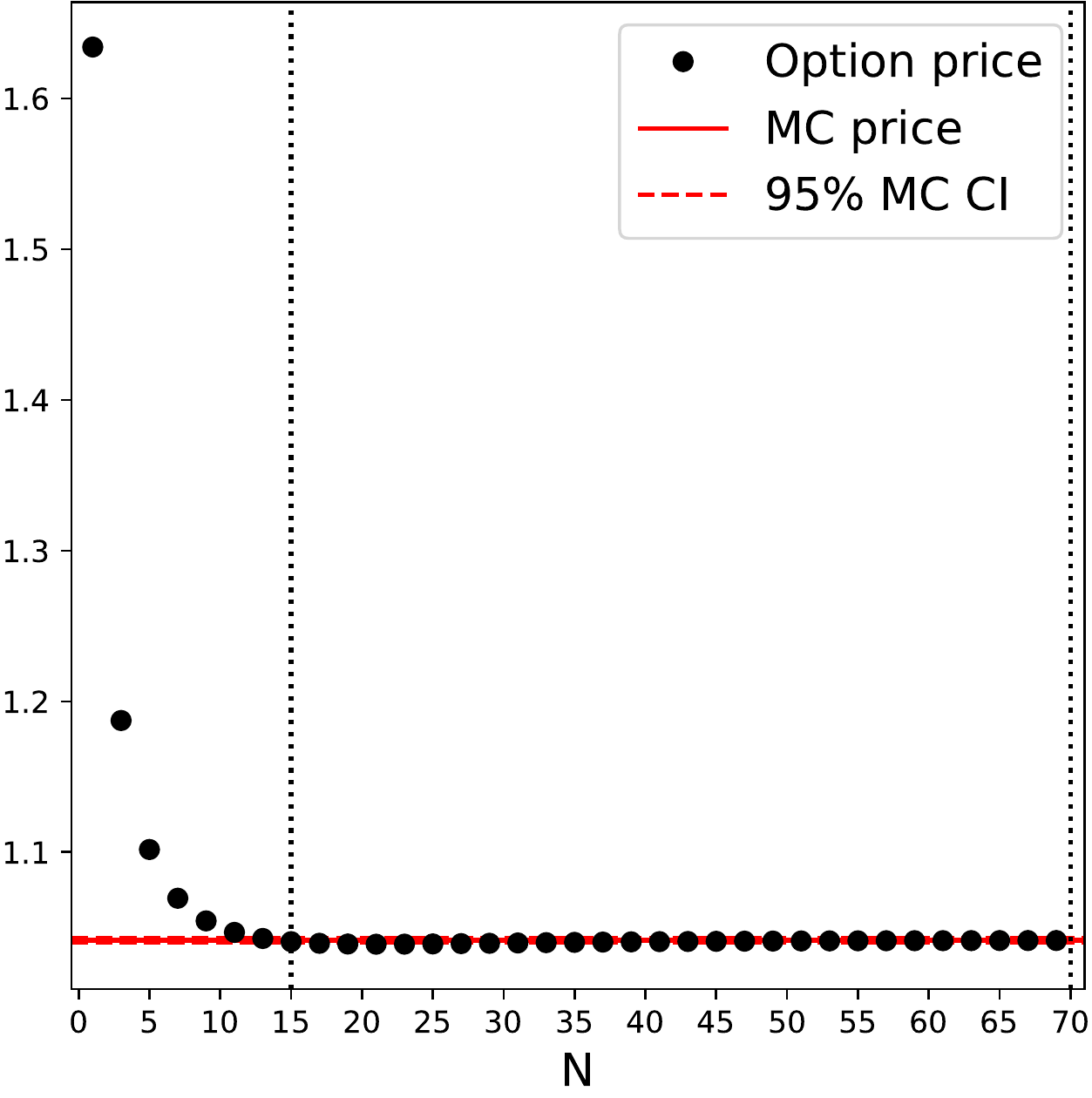}&
			\includegraphics[width=0.23\textwidth]{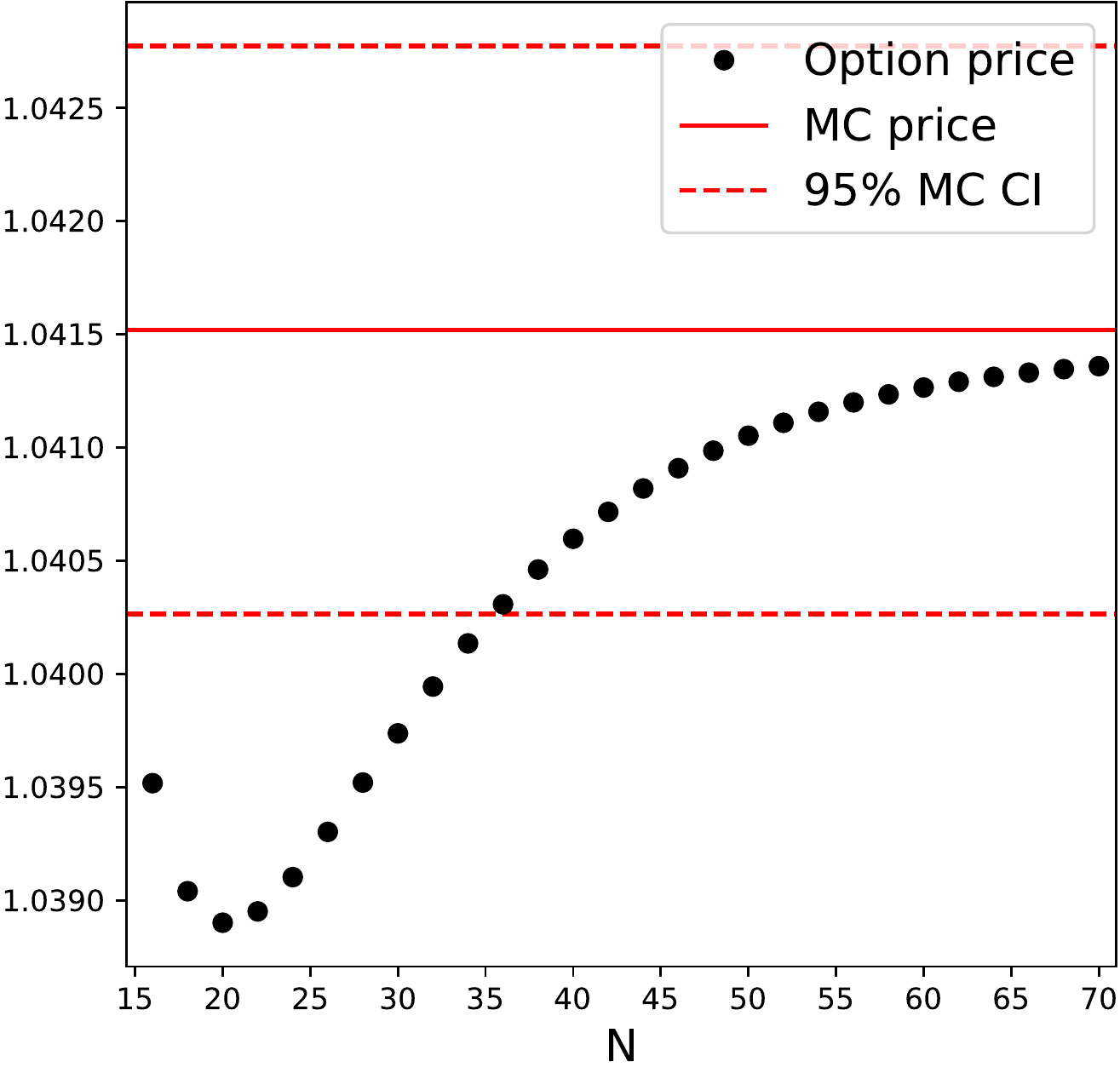} &	\includegraphics[width=0.23\textwidth]{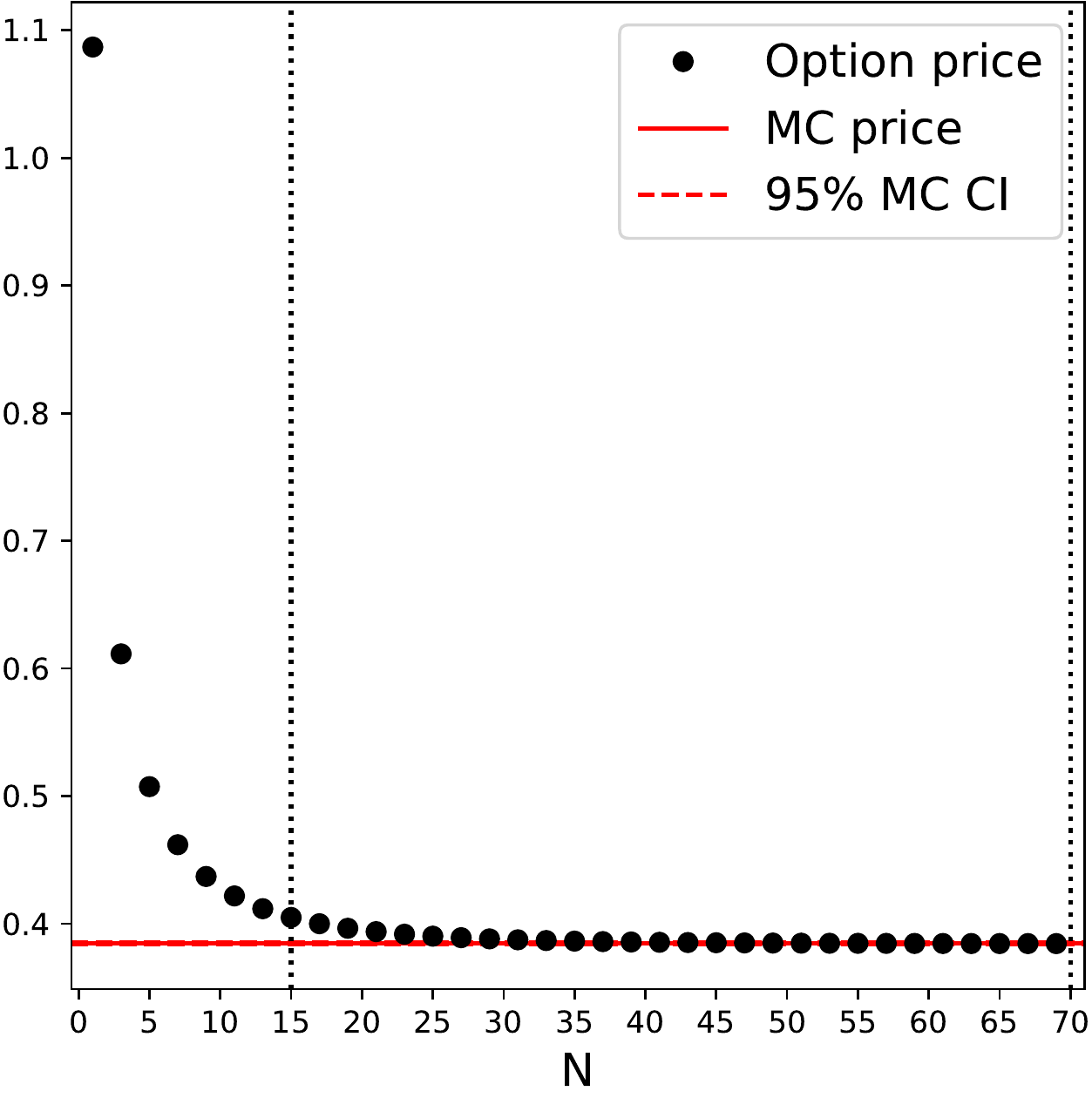} &
			\includegraphics[width=0.23\textwidth]{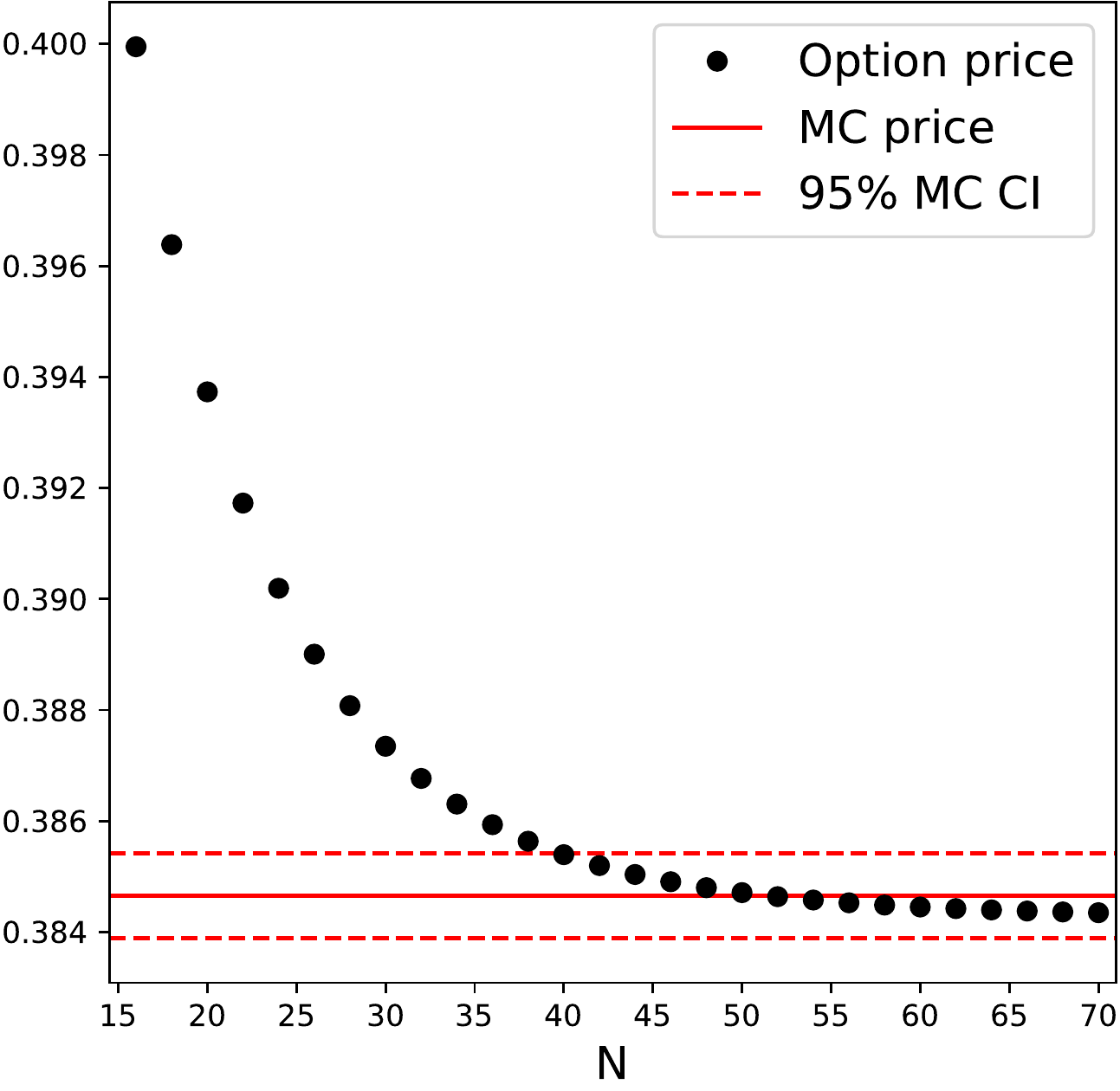}\\
			\begin{turn}{90}$\boldsymbol{b \!=6.0\, \underline{b}_{\sigma} \!\approx\!4.25}$\end{turn}&
			\includegraphics[width=0.23\textwidth]{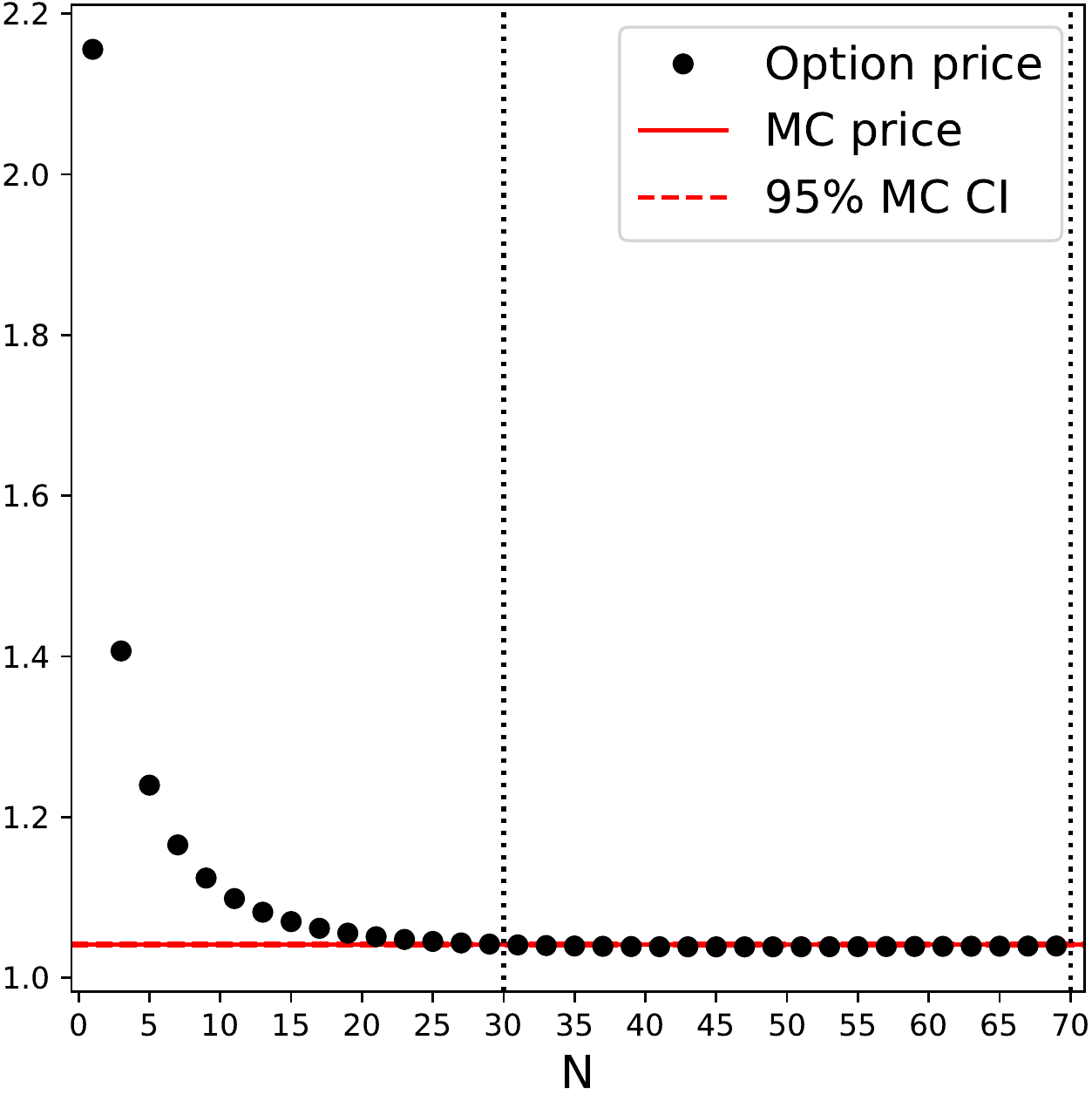}&
			\includegraphics[width=0.23\textwidth]{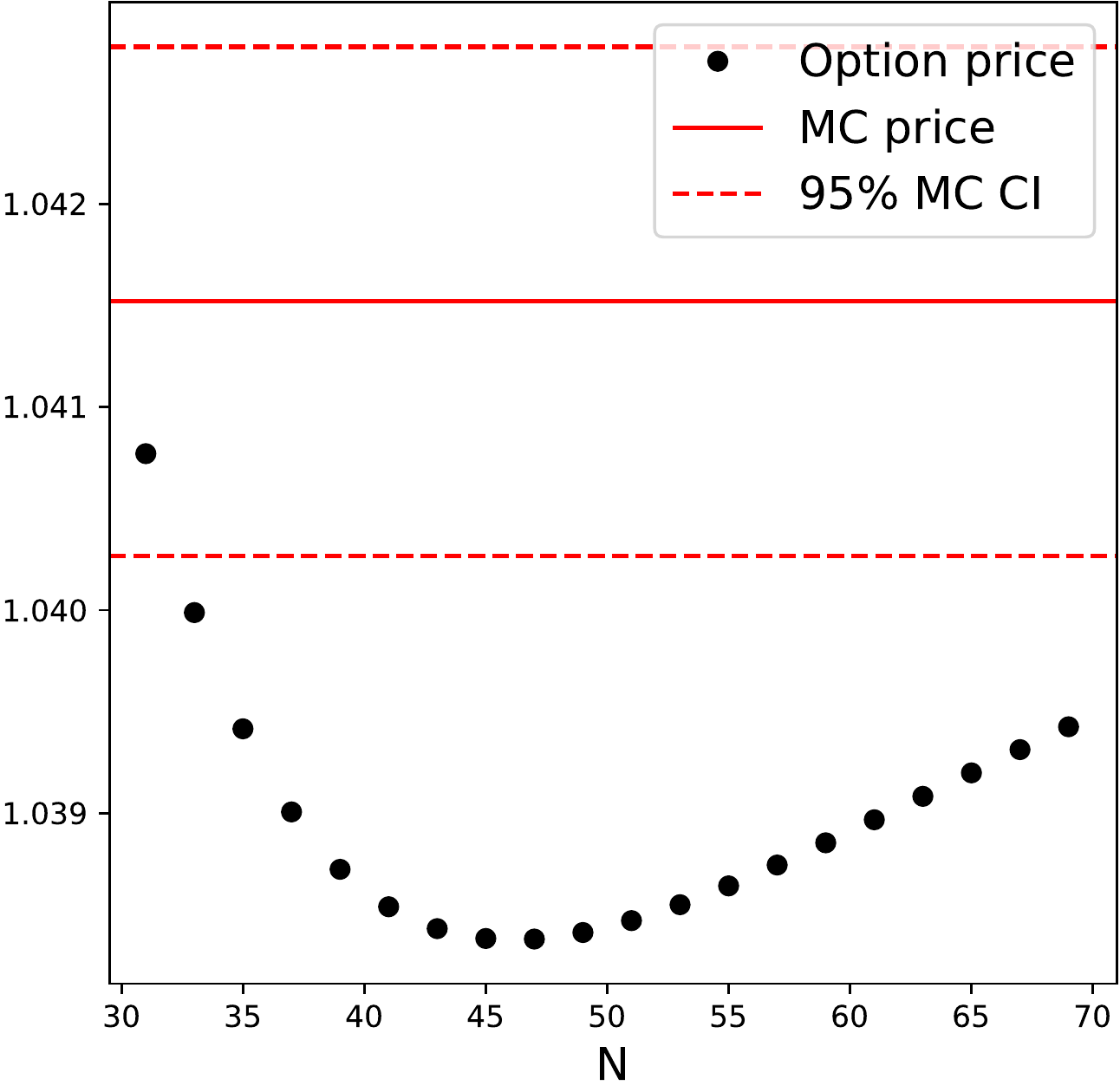} &	\includegraphics[width=0.23\textwidth]{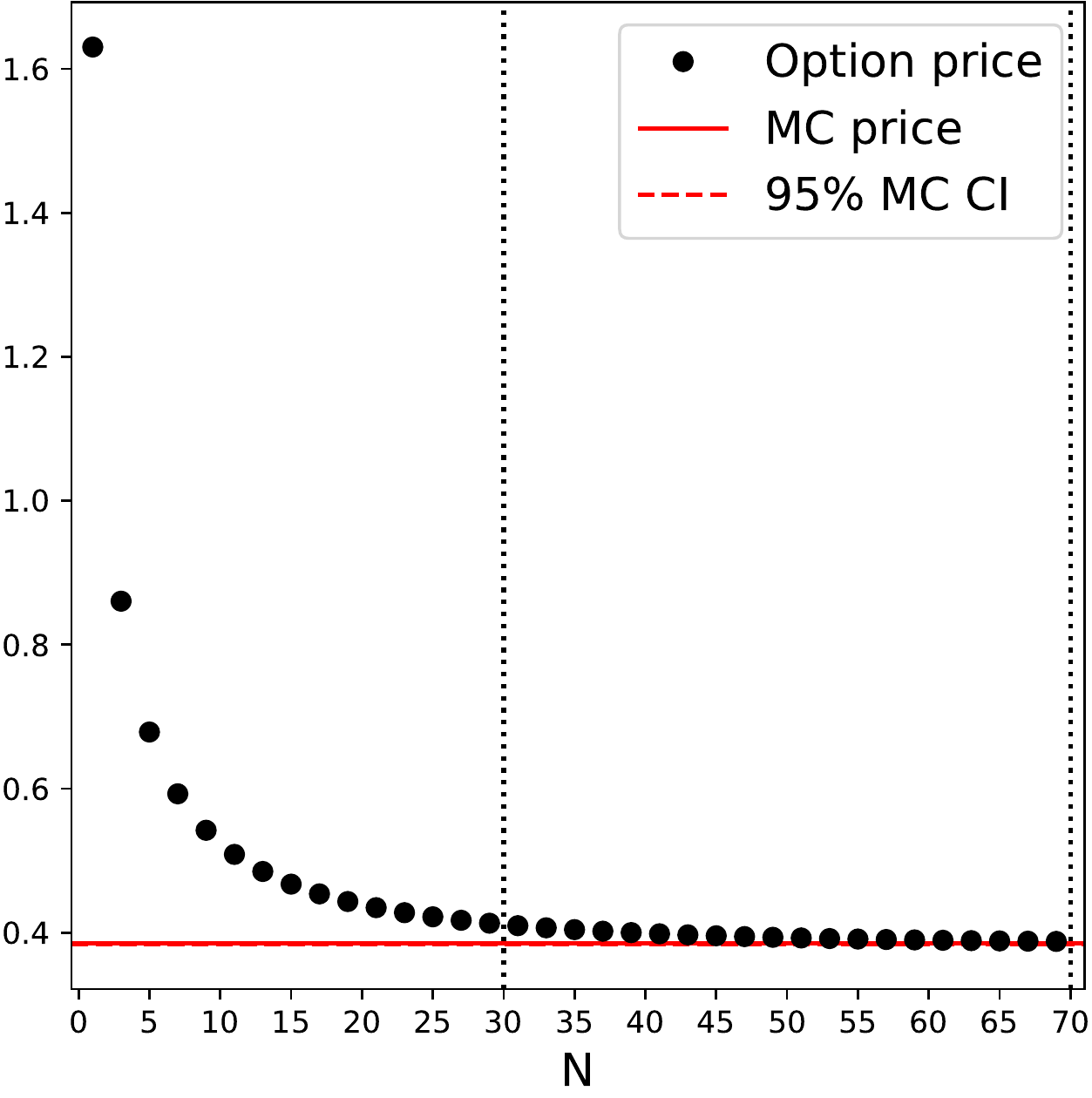} &
			\includegraphics[width=0.23\textwidth]{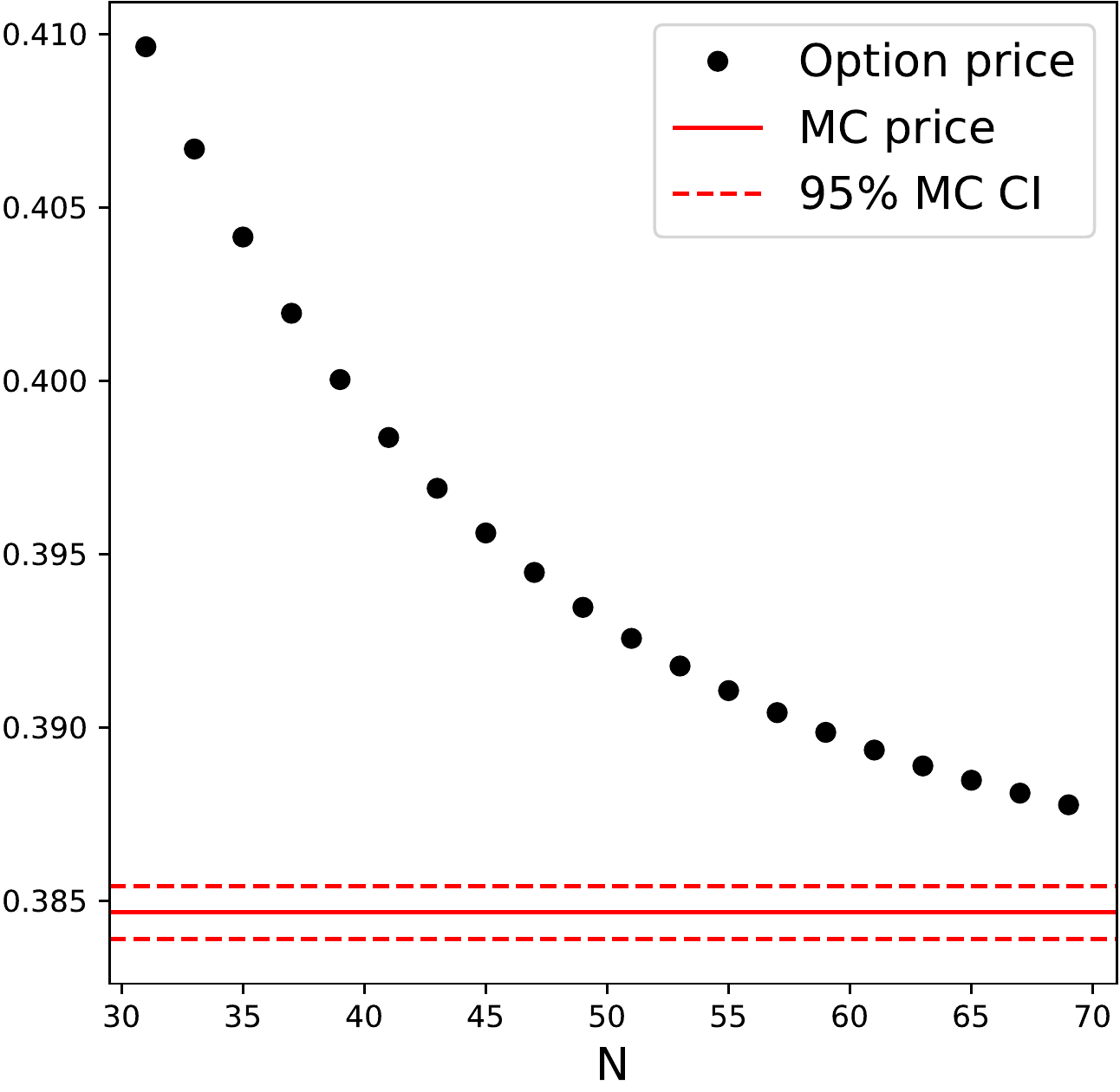}\\
		\end{tabular}
	}
	\caption{$\Pi_{K,N}^{a,b}(t)$ as a function of $N$, when the underlying process is a jump-diffusion with NIG measure and $\sigma_X(T;t)\approx 1.414$. The drift is $a=\mu_X(T;t)=2.0$. The black dots represent the option price, the solid red line the MC price, the two red dashed lines are the $95\%$ confidence interval for MC. The plots in the second and fourth columns are a zoomed subplot of the plots in the first and third columns. \label{JDplot}}
\end{figure}

\subsection{Asian options}
\label{examplesection}
We shall now test the option pricing formula with correlators of Theorem \ref{price2} for discretely sampled arithmetic Asian options. In particular, we shall make use of the insight learnt from the previous experiments and deal only with the Gaussian OU process and the polynomial-jump diffusion process as introduced above. For the Gaussian OU, closed price formula is available for benchmarking our approach, as we shall derive now. For the  polynomial-jump diffusion, no closed formula is available.

We first consider the OU process $Y$ introduced in equation \eqref{OU}. Then 
the average process $X$ is
\begin{equation*}
	X(T) 
	=  \frac{1}{m+1}\sum_{j=0}^m \left\{ Y(t)e^{b_1(s_j-t)}+\frac{b_0}{b_1}\left(e^{b_1(s_j-t)}-1\right)+\sqrt{\sigma_0}\int_{t}^{s_j}e^{b_1(s_j-v)}dB(v)  \right\}.
\end{equation*}
In particular, the random variables $\{Y(s_j)\}_{j=0}^m$ are not independent. We can however rewrite their sum as the sum of some other random variables $\{Z_j\}_{j=0}^m$ which are independent. 
\begin{proposition}
	\label{distribution}
	For $s_{-1}:=t$, the random variable $X(T)$ equals in distribution the weighted sum of $m+1$ independent random variables $\{Z_j\}_{j=0}^m$, namely $X(T)\mathop{=}\limits^d \frac{1}{m+1}\sum_{j=0}^m Z_j$,
	where $Z_j$ is defined by
	\begin{equation*}
		Z_j := Y(t)e^{b_1(s_j-t)}+\frac{b_0}{b_1}\left(e^{b_1(s_j-t)}-1\right)+\sqrt{\sigma_0}\int_{s_{j-1}}^{s_j}\left(\sum_{k=j}^m e^{b_1(s_k-v)} \right) dB(v) \qquad \mbox{for }j=0, \dots, m.
	\end{equation*}
\end{proposition}

As a direct consequence of Proposition \ref{distribution}, we find that $\left. X(T)\right|\F_t\sim \mathcal{N}(\mu_X(T;t), \sigma_X^2(T;t) )$ with
\begin{align*}
	&\mu_X(T;t) = \frac{1}{m+1} \sum_{j=0}^m \left(Y(t) \,e^{b_1(s_j-t)} + \frac{b_0}{b_1}\left(e^{b_1(s_j-t)}-1\right)\right) \quad \mbox{and }\\
	&\sigma_X^2(T;t) = \frac{\sigma_0}{(m+1)^2} \sum_{j=0}^m\sum_{k_1=j}^{m}\sum_{k_2=j}^{m} \frac{e^{b_1(s_{k_1}+s_{k_2}-2s_{j-1})}-e^{b_1(s_{k_1}+s_{k_2}-2s_{j})}}{2b_1},
\end{align*}
which, together with equation \eqref{closed_price}, gives us a benchmark for the experiments.

In Figure \ref{OUplot_corr} and \ref{JDplot_corr} we report the results for, respectively, the Gaussian OU and the polynomials jump-diffusion process with $a = \mu_X(T;t)$ and $b=2.0\underline{b}_{\sigma}$, since in the previous experiments this was the value of the scale performing the best. All the experiments are in line with the previous ones: the accuracy of the approximation increases with $N$ increasing, until a certain value after which it starts decreasing. We see that the Hermite approximation performs well also for a path-dependent option, whose evaluation requires the correlator formula instead of the moment formula for polynomial processes. There is no significant difference between $m=0$, $m=1$ and $m=2$. Unfortunately, due to computational constraints, it is not possible to test the Hermite approximation for $m>2$.

For practical purposes, one needs a way to understand when to truncate the series, i.e., how to choose the value for $N$. We then propose the following stopping criterion: for each $N$ one calculates
\begin{equation}
	\label{criterio}
	\tilde{\gamma}_{a,b}^N:=-\log\left(\frac{\left|\Pi_{K,N-1}^{a,b}(t) -\Pi_{K,N}^{a,b}(t) \right|}{\Pi_{K,N-1}^{a,b}(t)} \right).
\end{equation}
If $\tilde{\gamma}_{a,b}^N>4$, then the contribution of the $N$-th term to the price approximation is smaller than $10^{-4}$ and we truncate the series.  In Figure \ref{JDplot_corr}  the prices obtained by this criterion are marked with a red star.

\begin{figure}[!tp]
	\setlength{\tabcolsep}{2pt}
	\resizebox{1\textwidth}{!}{
		\begin{tabular}{@{}>{\centering\arraybackslash}m{0.04\textwidth}@{}>{\centering\arraybackslash}m{0.24\textwidth}@{}>{\centering\arraybackslash}m{0.24\textwidth}@{}>{\centering\arraybackslash}m{0.24\textwidth}@{}>{\centering\arraybackslash}m{0.24\textwidth}@{}}
			& $\boldsymbol{K = 1.0}$&$\boldsymbol{K = 2.0}$ & $\boldsymbol{K = 3.0}$& $\boldsymbol{K = 4.0}$ \\
			\begin{turn}{90}$\boldsymbol{m = 1}$\end{turn}&
			\includegraphics[width=0.23\textwidth]{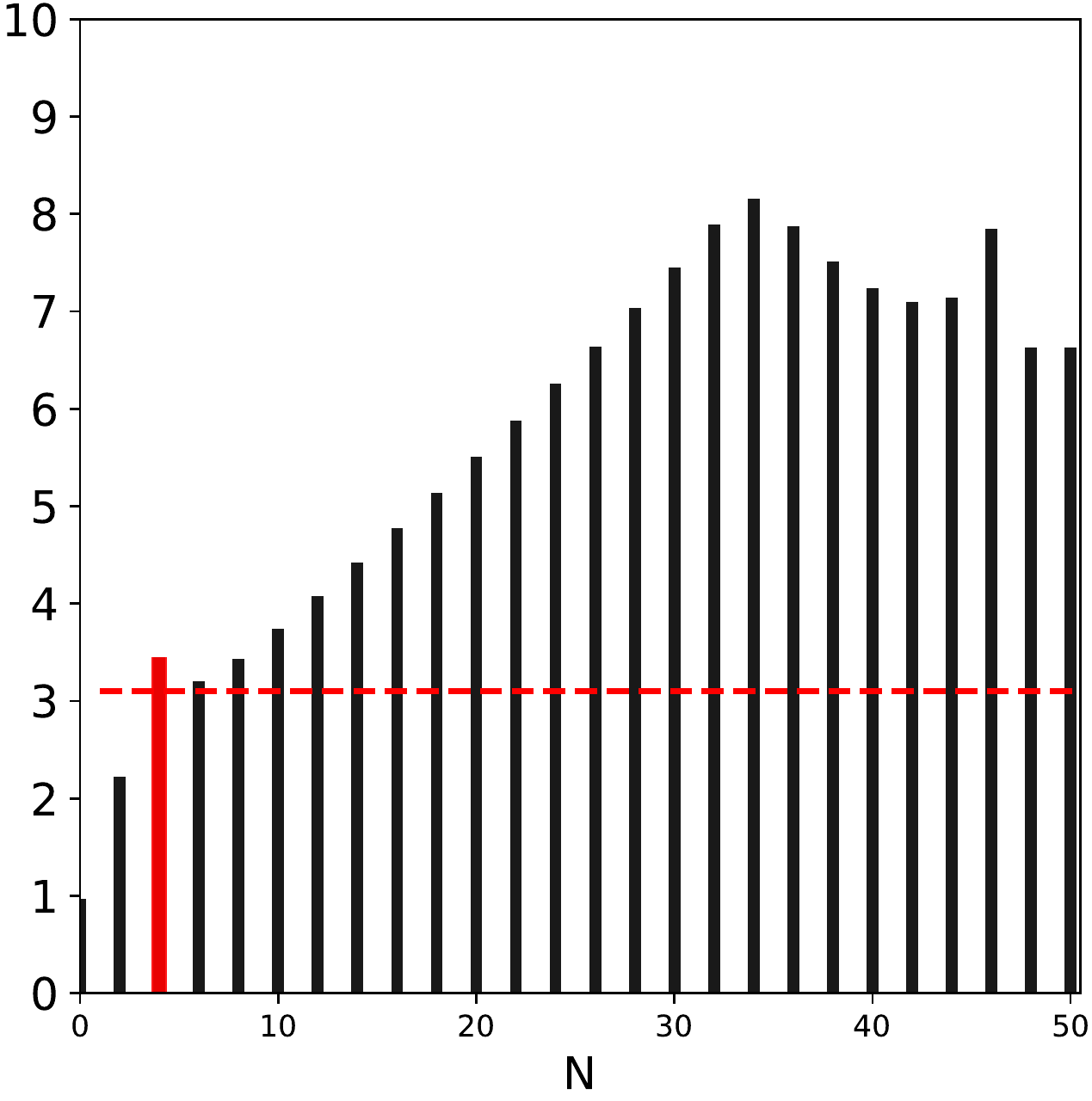}&
			\includegraphics[width=0.23\textwidth]{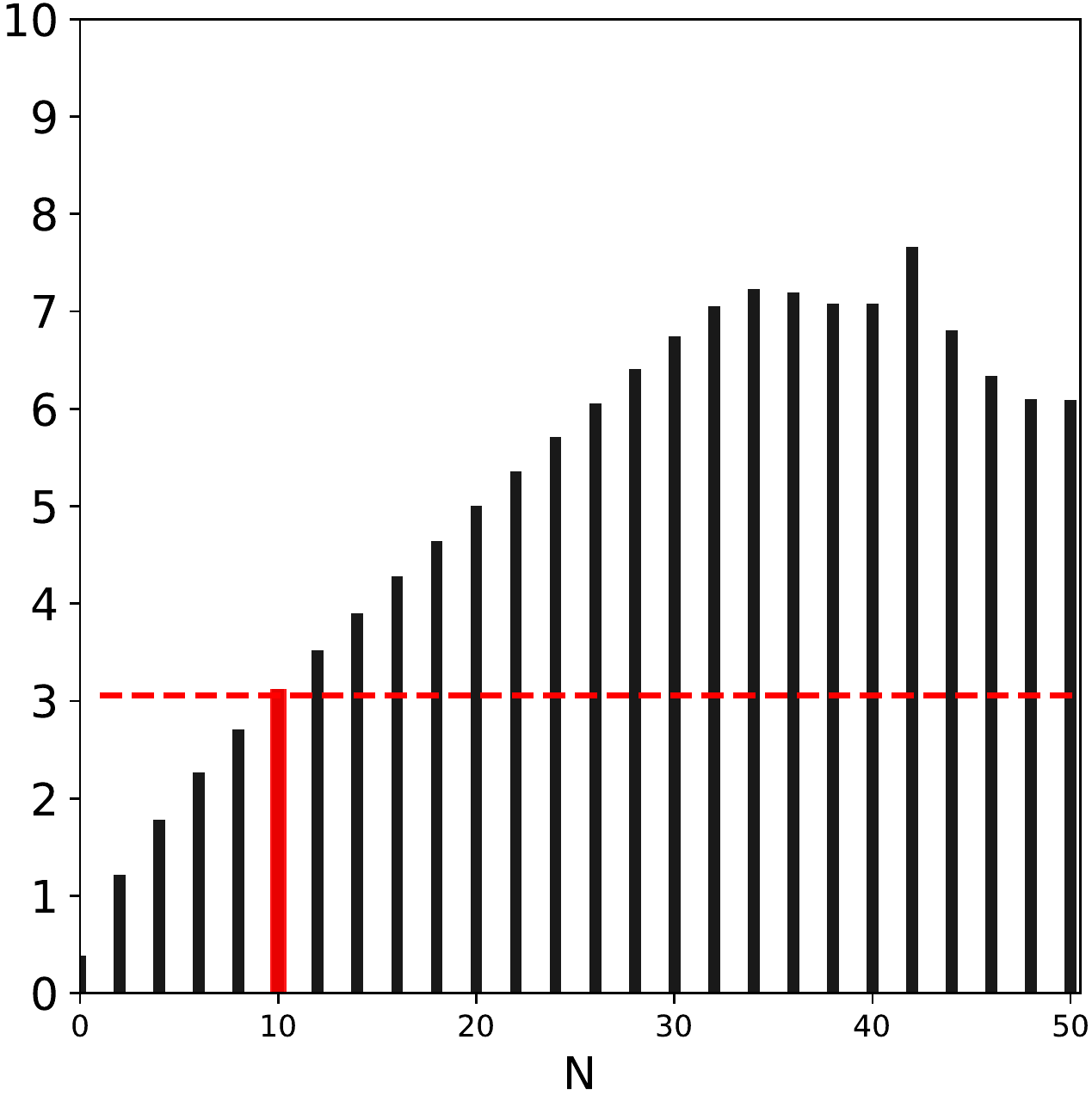} &	\includegraphics[width=0.23\textwidth]{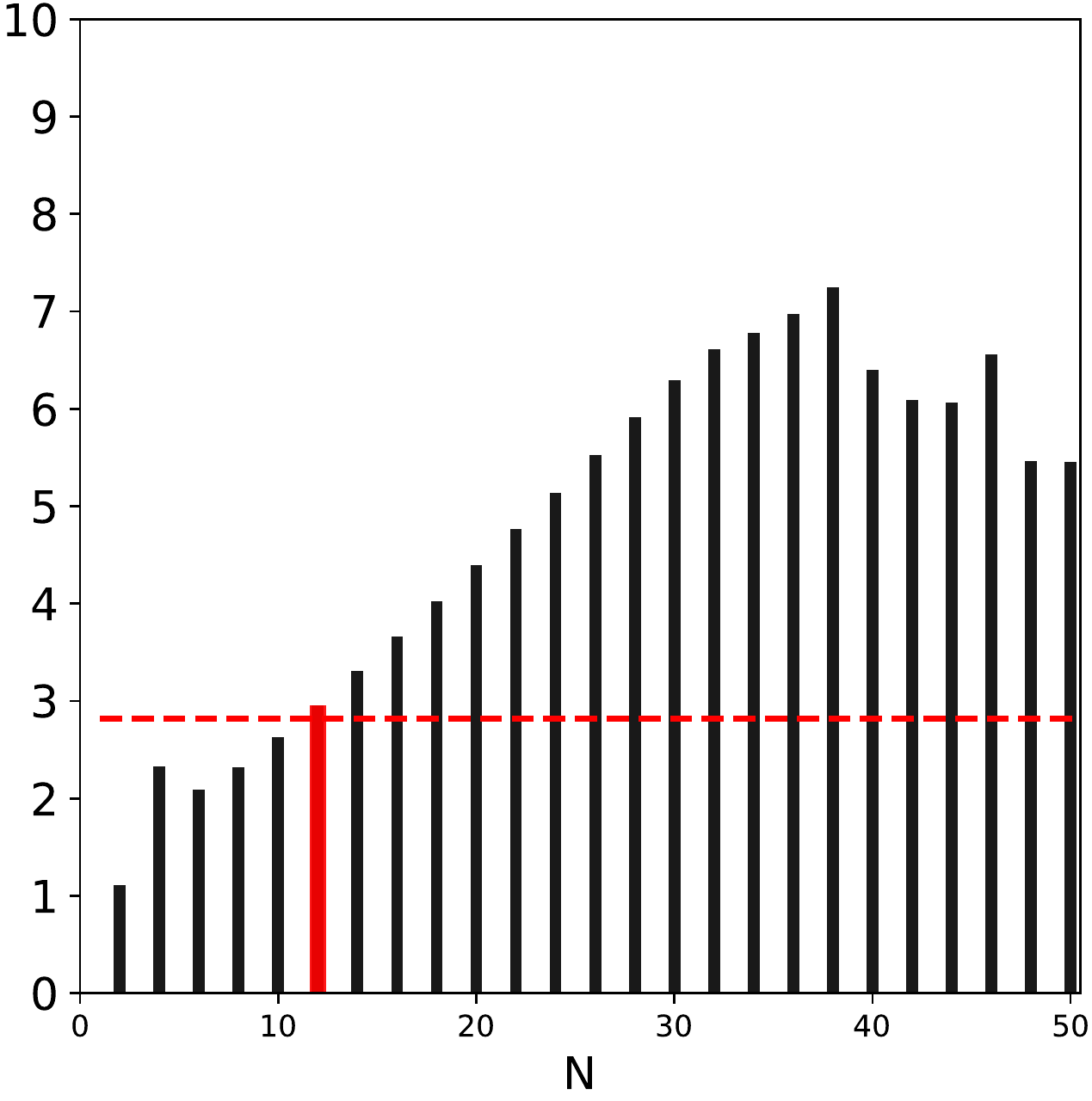} &
			\includegraphics[width=0.23\textwidth]{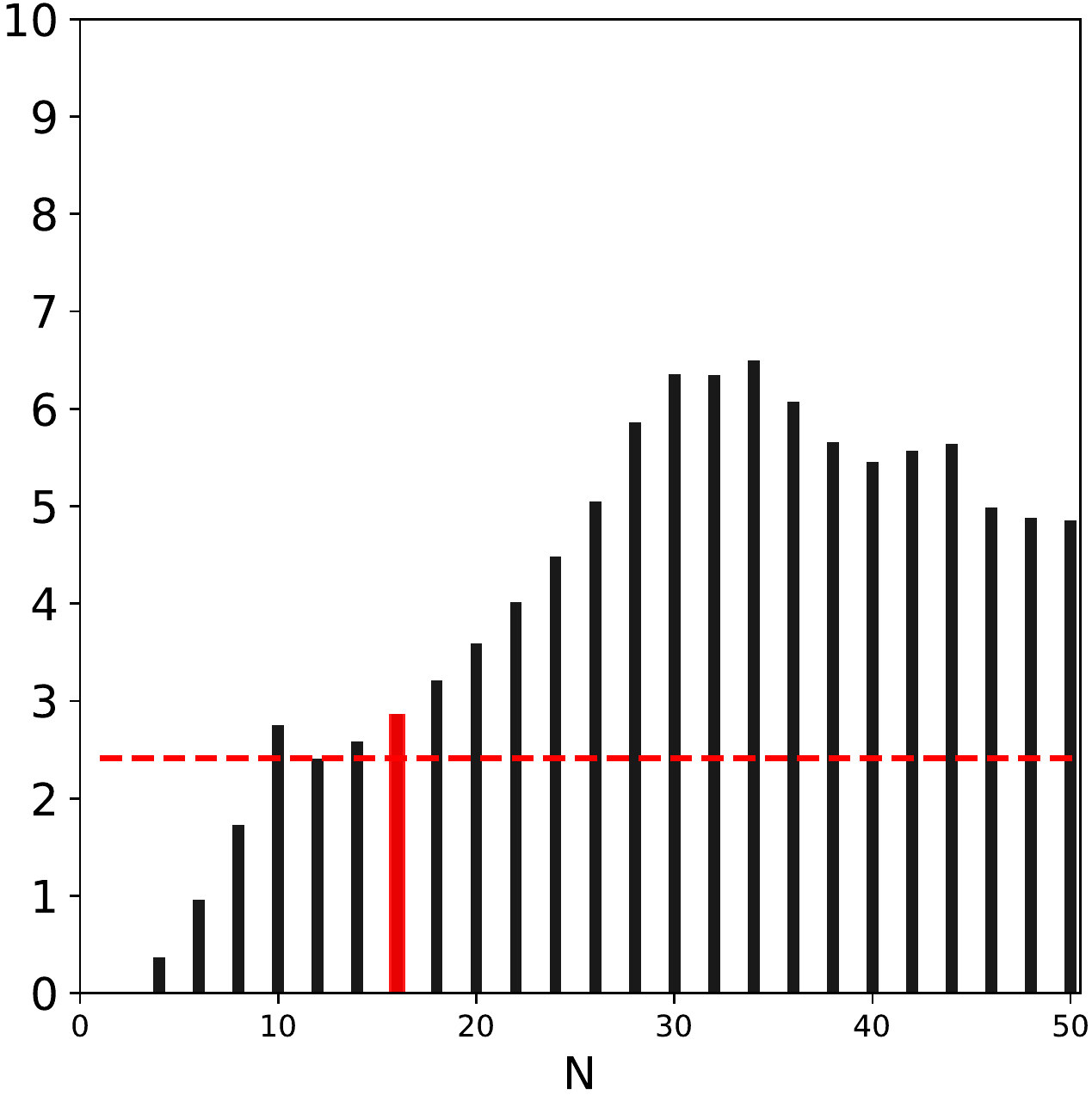}\\
			\begin{turn}{90}$\boldsymbol{m=2}$\end{turn}&
			\includegraphics[width=0.23\textwidth]{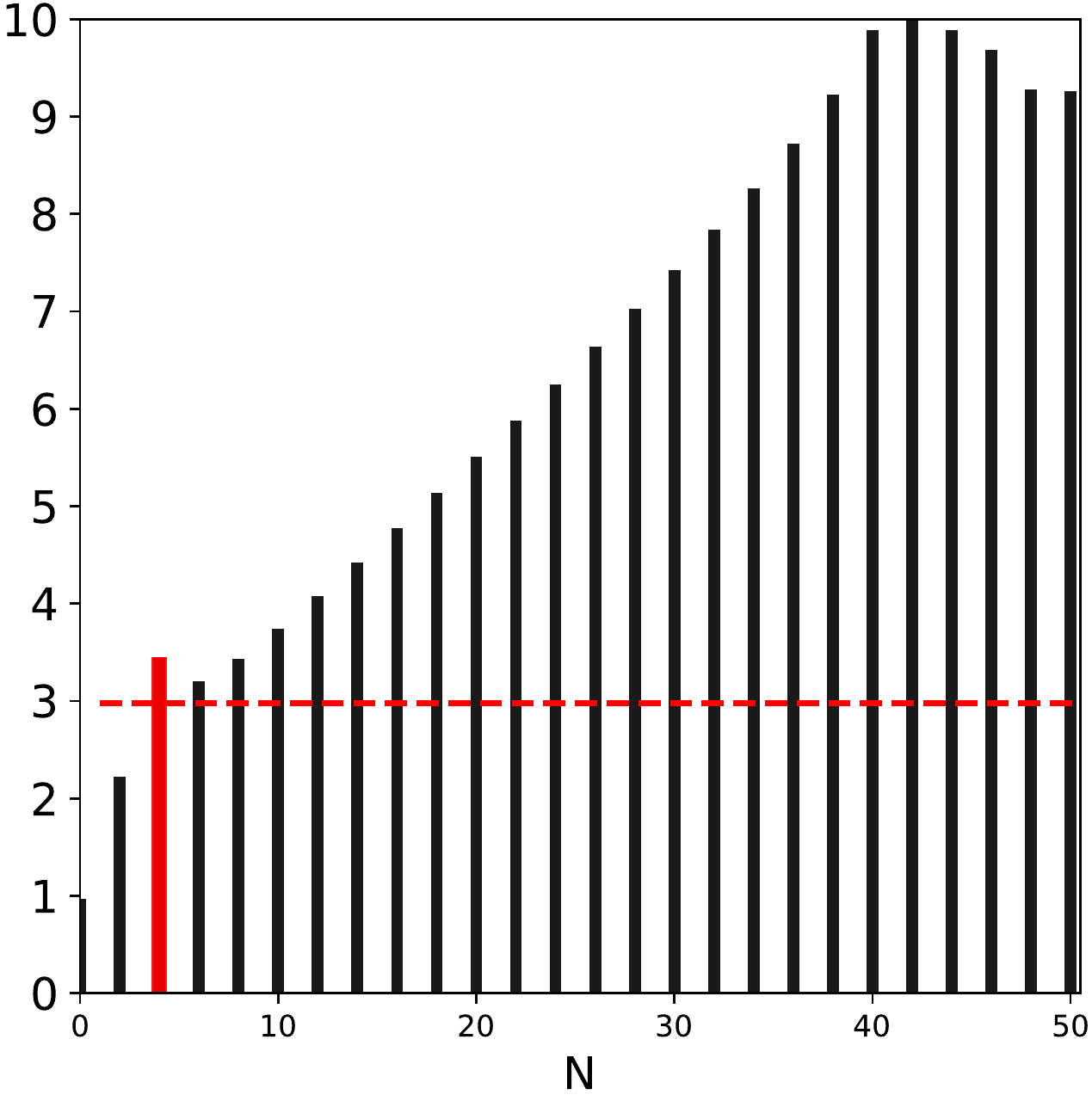}&
			\includegraphics[width=0.23\textwidth]{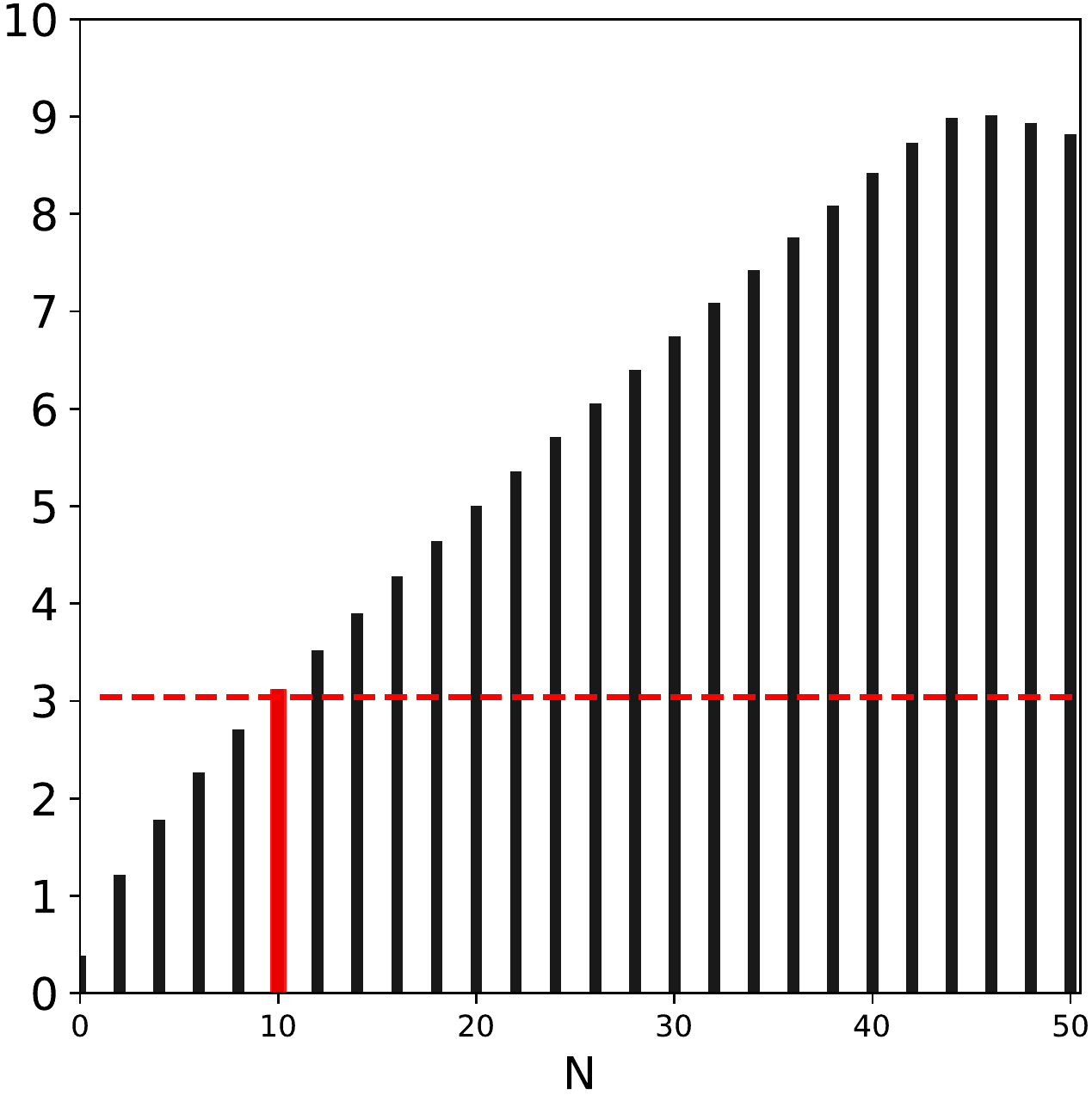} &	\includegraphics[width=0.23\textwidth]{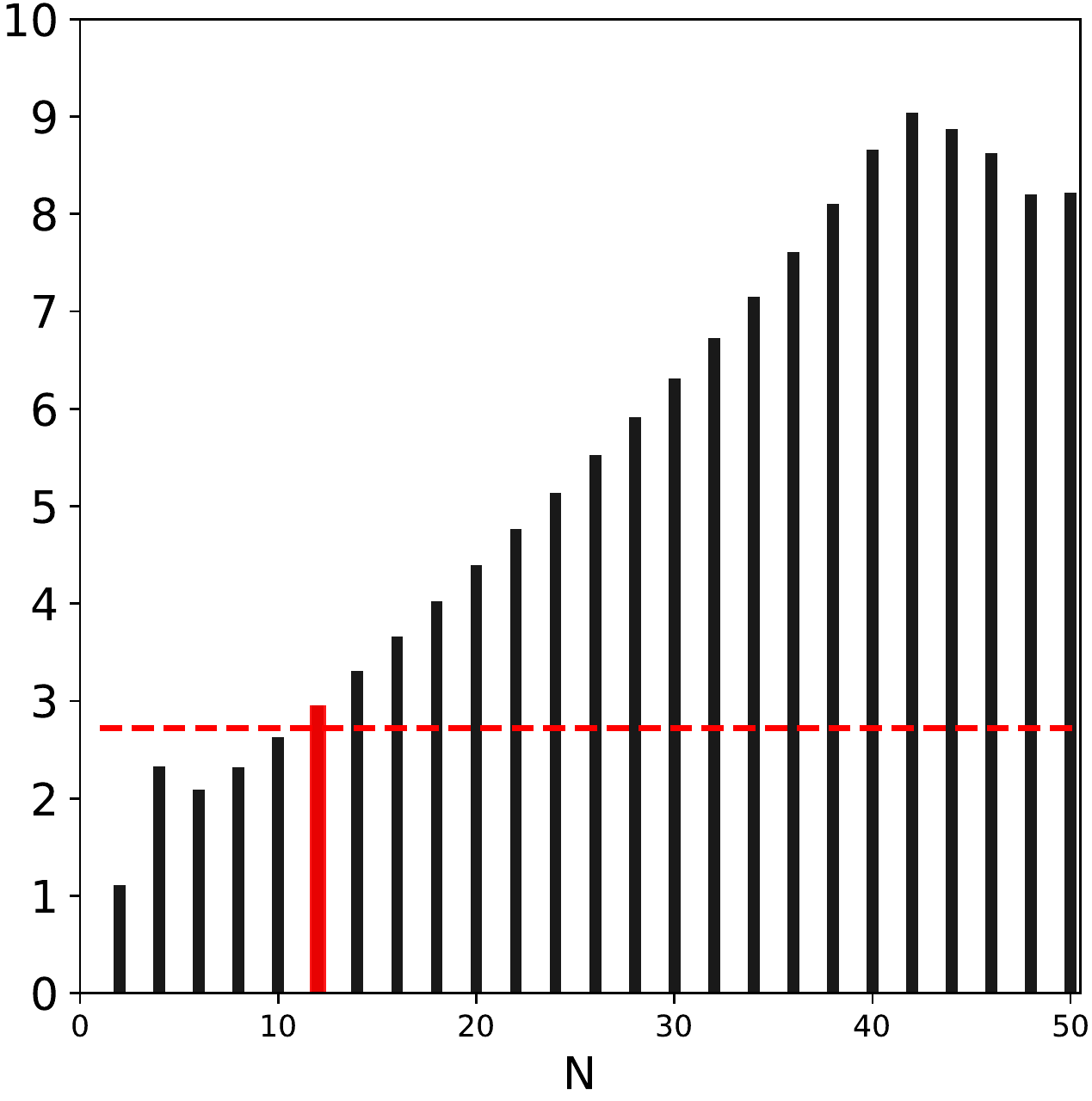} &
			\includegraphics[width=0.23\textwidth]{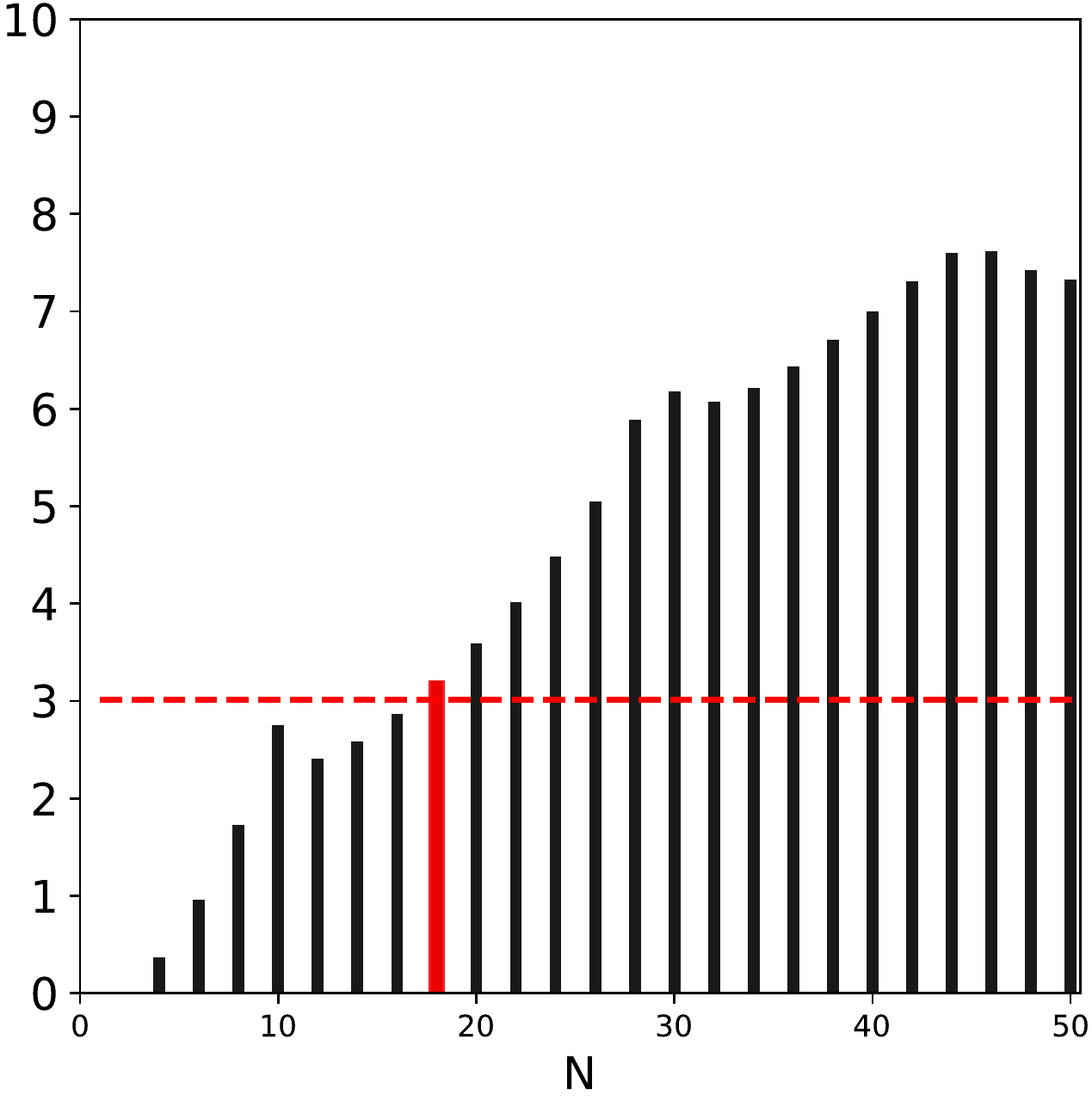}\\
		\end{tabular}
	}
	\caption{$\gamma_{a,b}^N$ as a function of $N$, when the underlying process is an $(m+1)$-point weighted OU with $\sigma_X(T;t)=1.0$. The drift is $a= \mu_X(T;t)=2.0$ and the scale is $b = 2.0 \underline{b}_{\sigma}$. The dashed red horizontal lines indicate the accuracy of the MC method. The red vertical bars indicate when the Hermite series reaches the MC accuracy. \label{OUplot_corr}}
\end{figure}

\begin{figure}[!tp]
	\setlength{\tabcolsep}{2pt}
	\resizebox{1\textwidth}{!}{
		\begin{tabular}{@{}>{\centering\arraybackslash}m{0.04\textwidth}@{}>{\centering\arraybackslash}m{0.24\textwidth}@{}>{\centering\arraybackslash}m{0.24\textwidth}@{}>{\centering\arraybackslash}m{0.24\textwidth}@{}>{\centering\arraybackslash}m{0.24\textwidth}@{}}
			& $\boldsymbol{K = 1.0}$&$\boldsymbol{K = 1.0}$ (zoom) & $\boldsymbol{K = 2.0}$& $\boldsymbol{K = 2.0}$ (zoom)\\
			\begin{turn}{90}$\boldsymbol{m = 1}$\end{turn}&
			\includegraphics[width=0.23\textwidth]{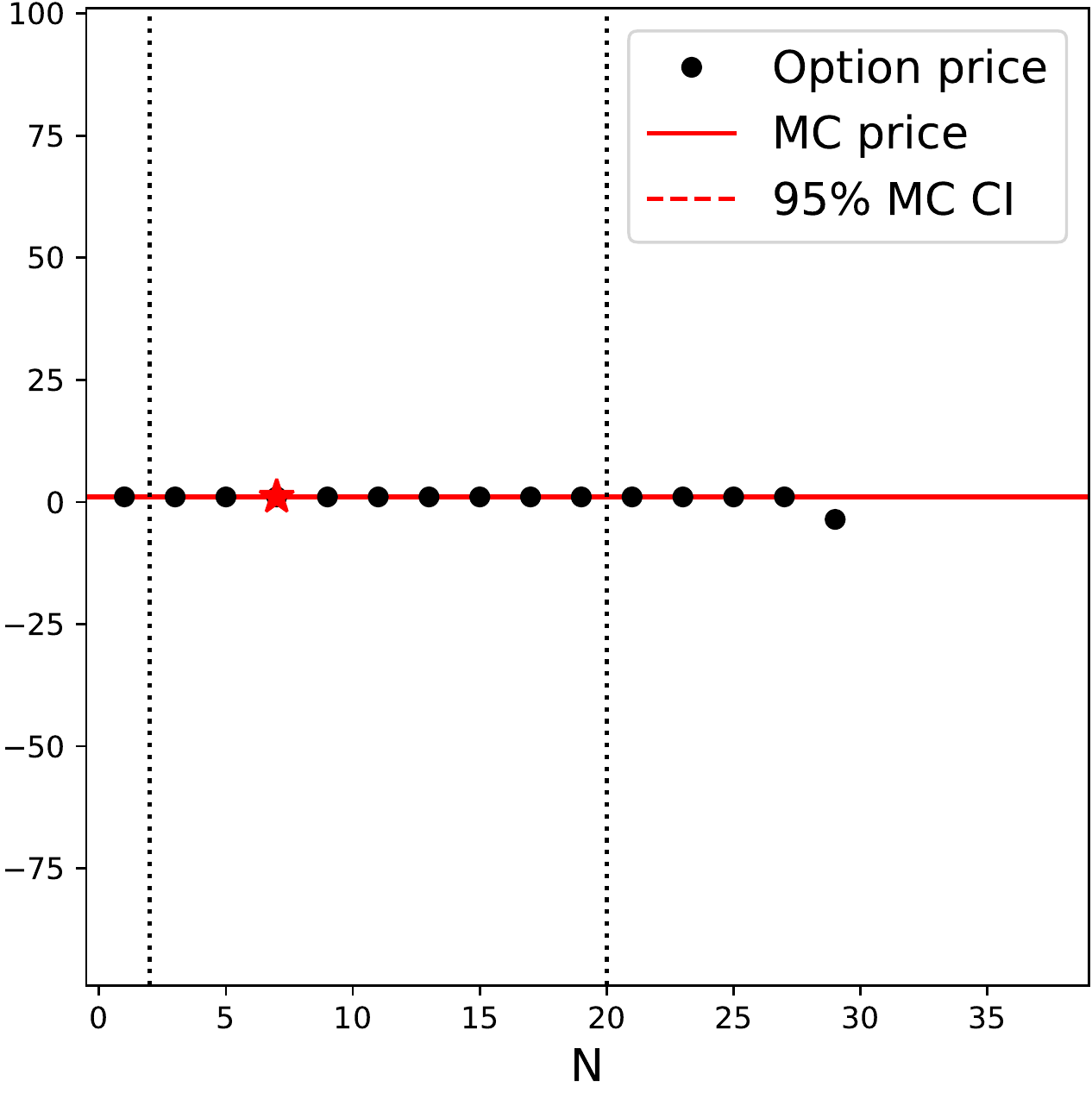}&
			\includegraphics[width=0.23\textwidth]{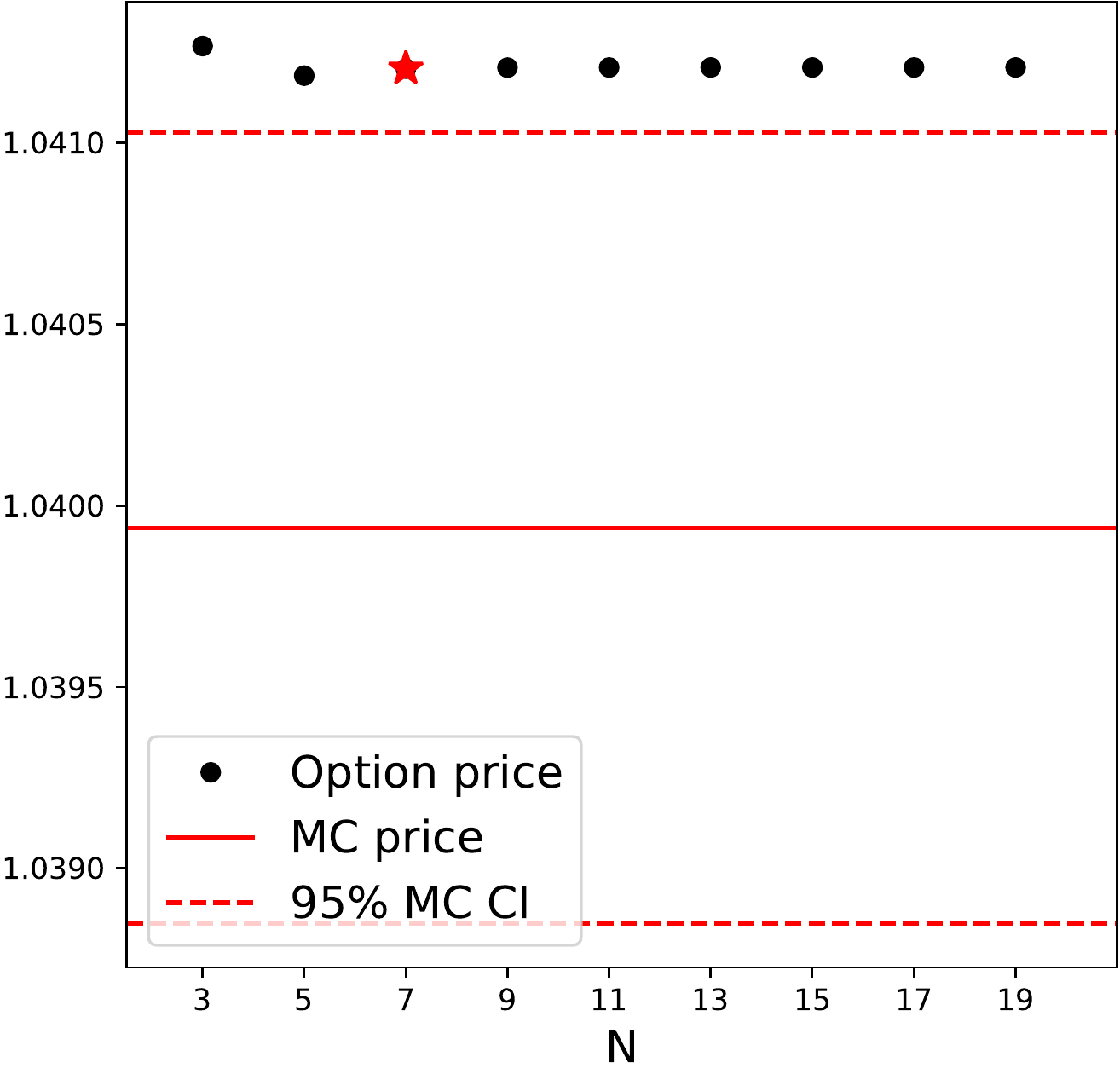} &	\includegraphics[width=0.23\textwidth]{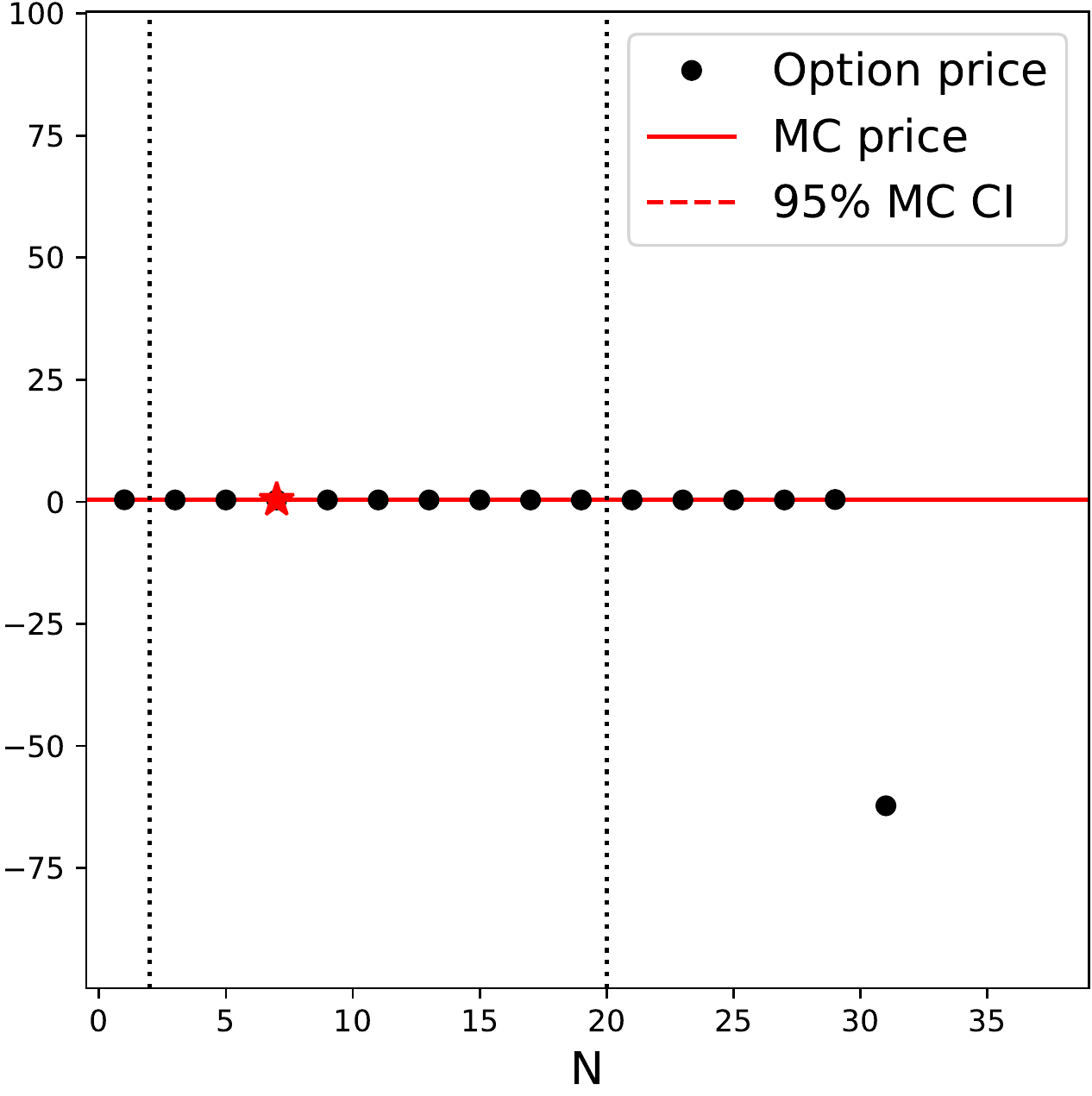} &
			\includegraphics[width=0.23\textwidth]{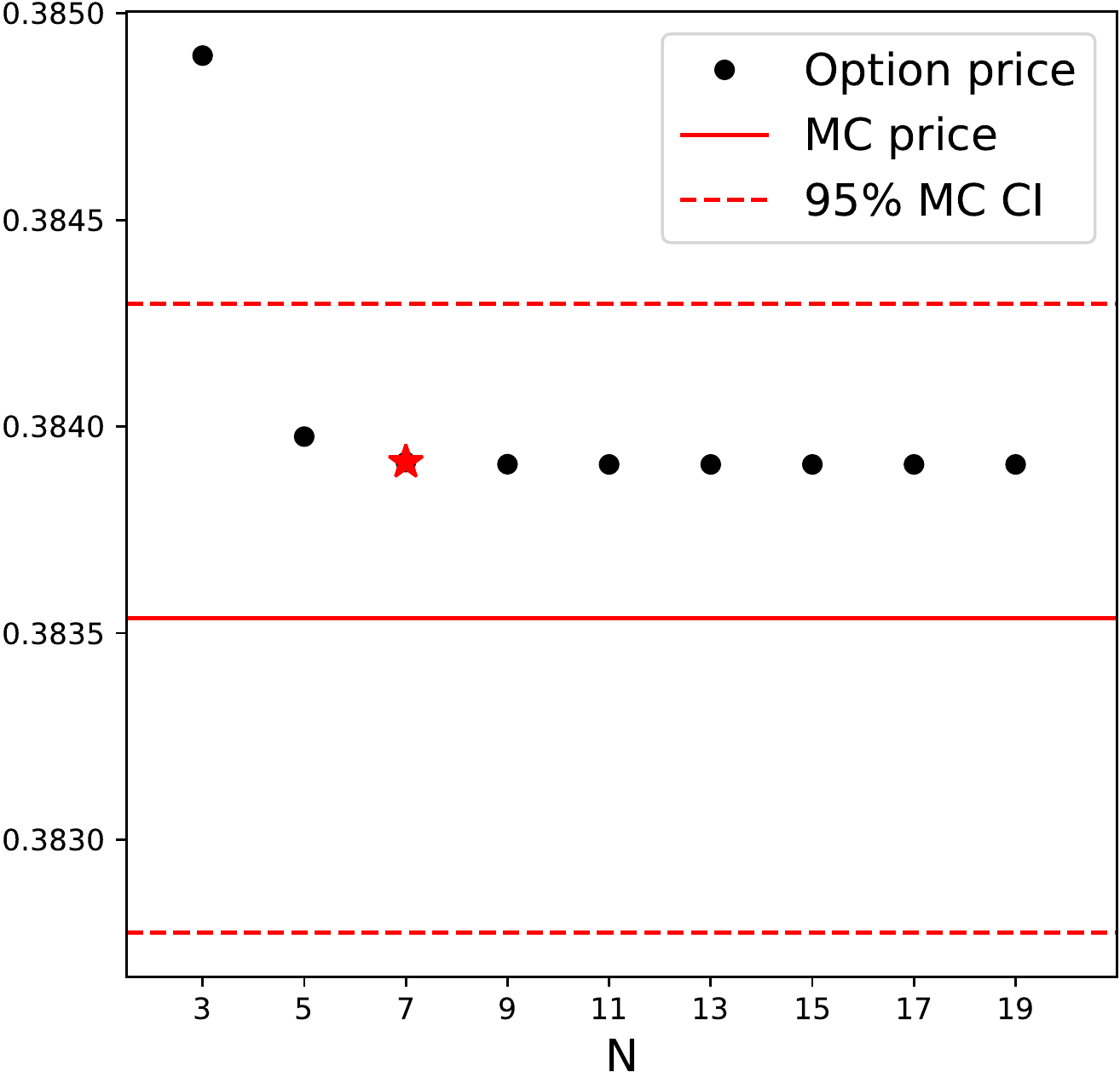}\\
			\begin{turn}{90}$\boldsymbol{m=2}$\end{turn}&
			\includegraphics[width=0.23\textwidth]{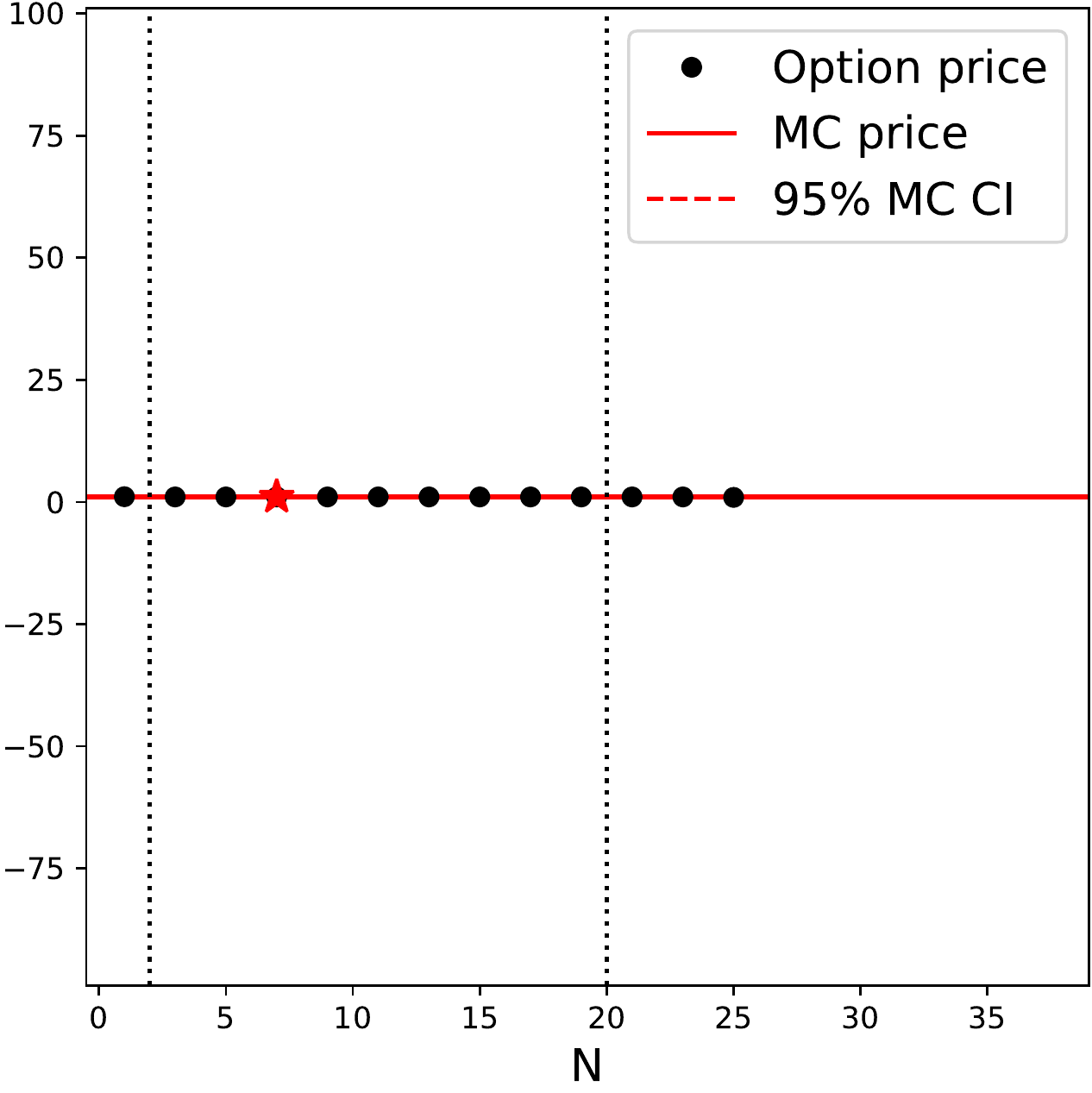}&
			\includegraphics[width=0.23\textwidth]{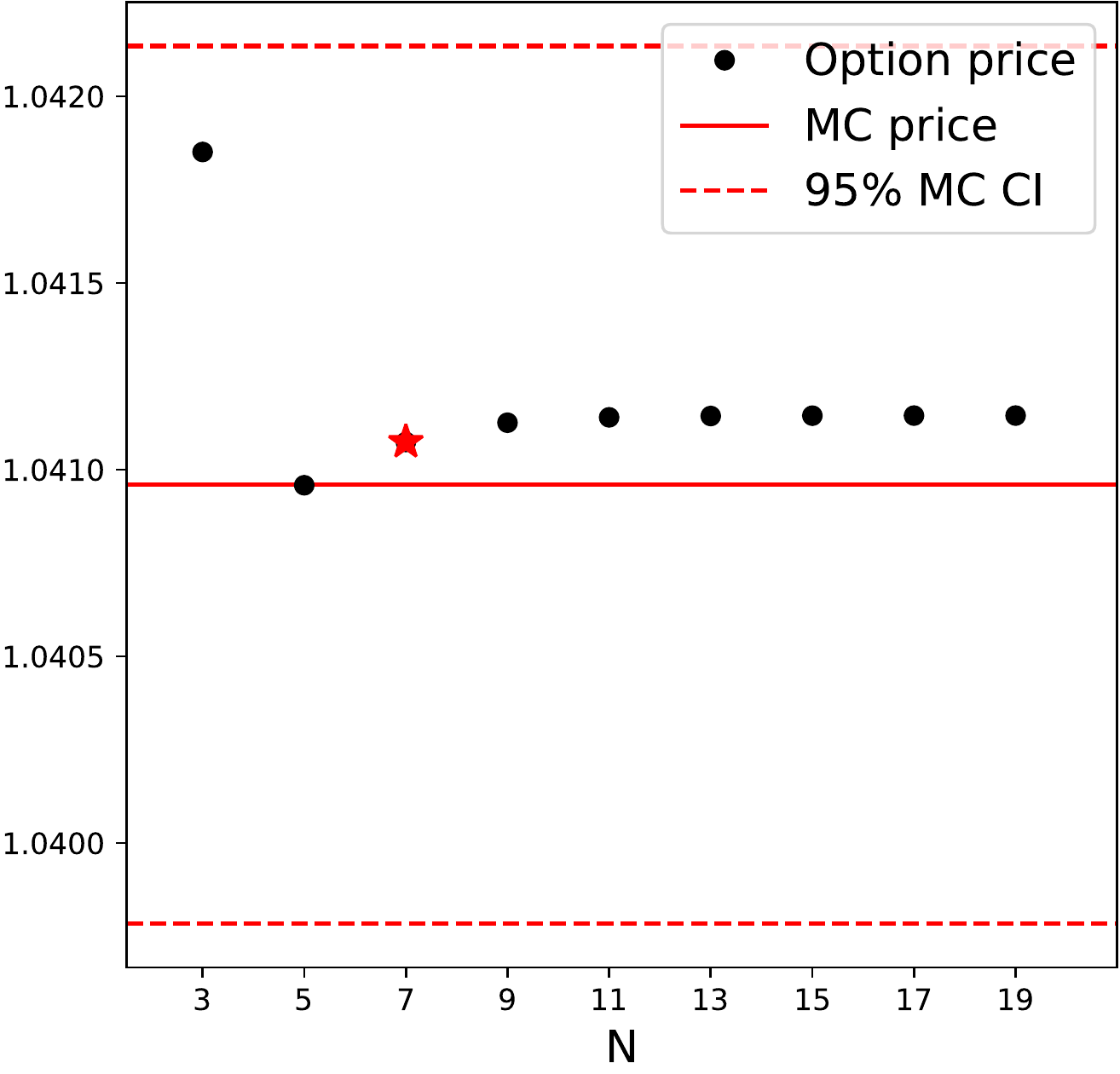} &	\includegraphics[width=0.23\textwidth]{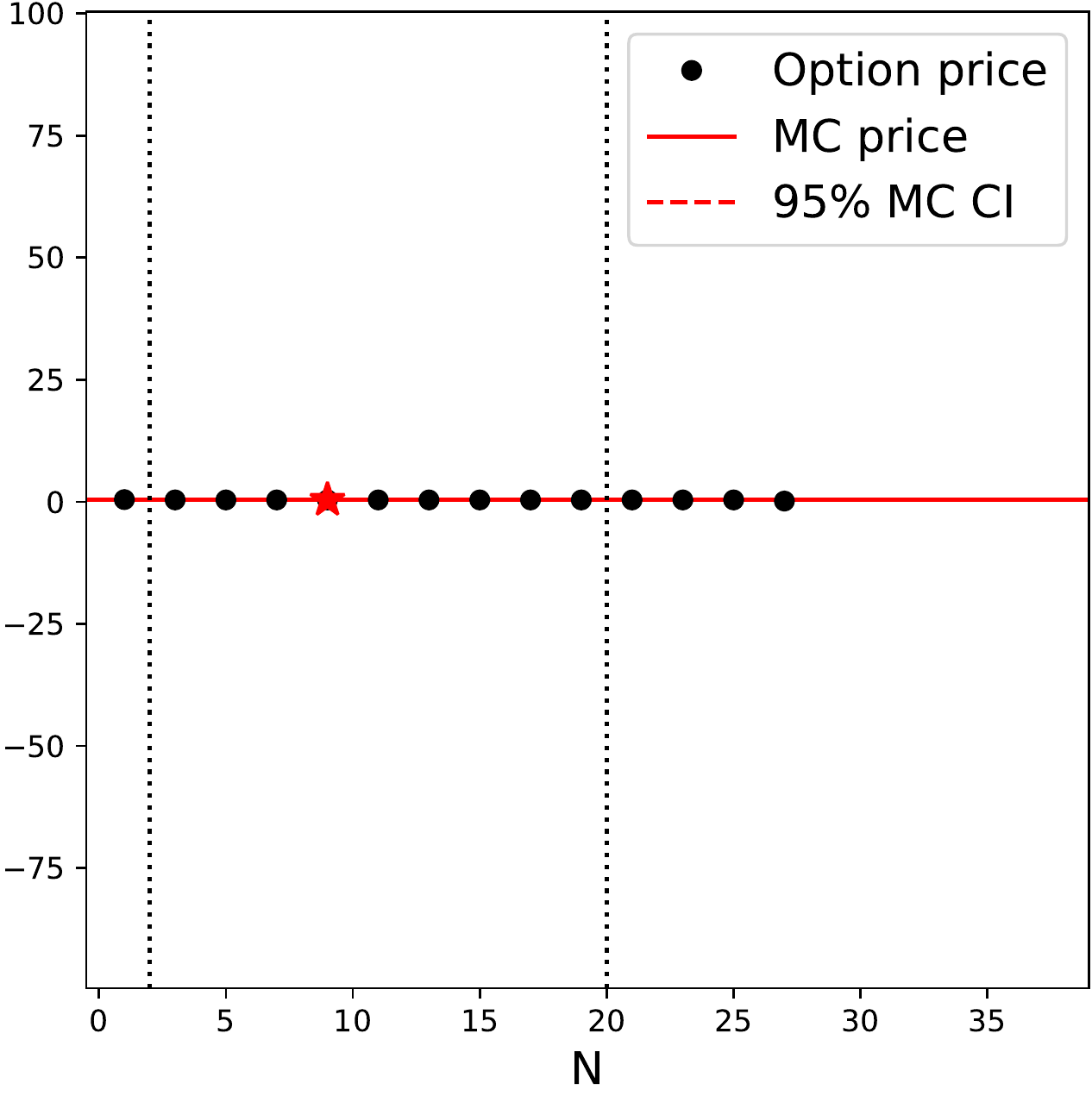} &
			\includegraphics[width=0.23\textwidth]{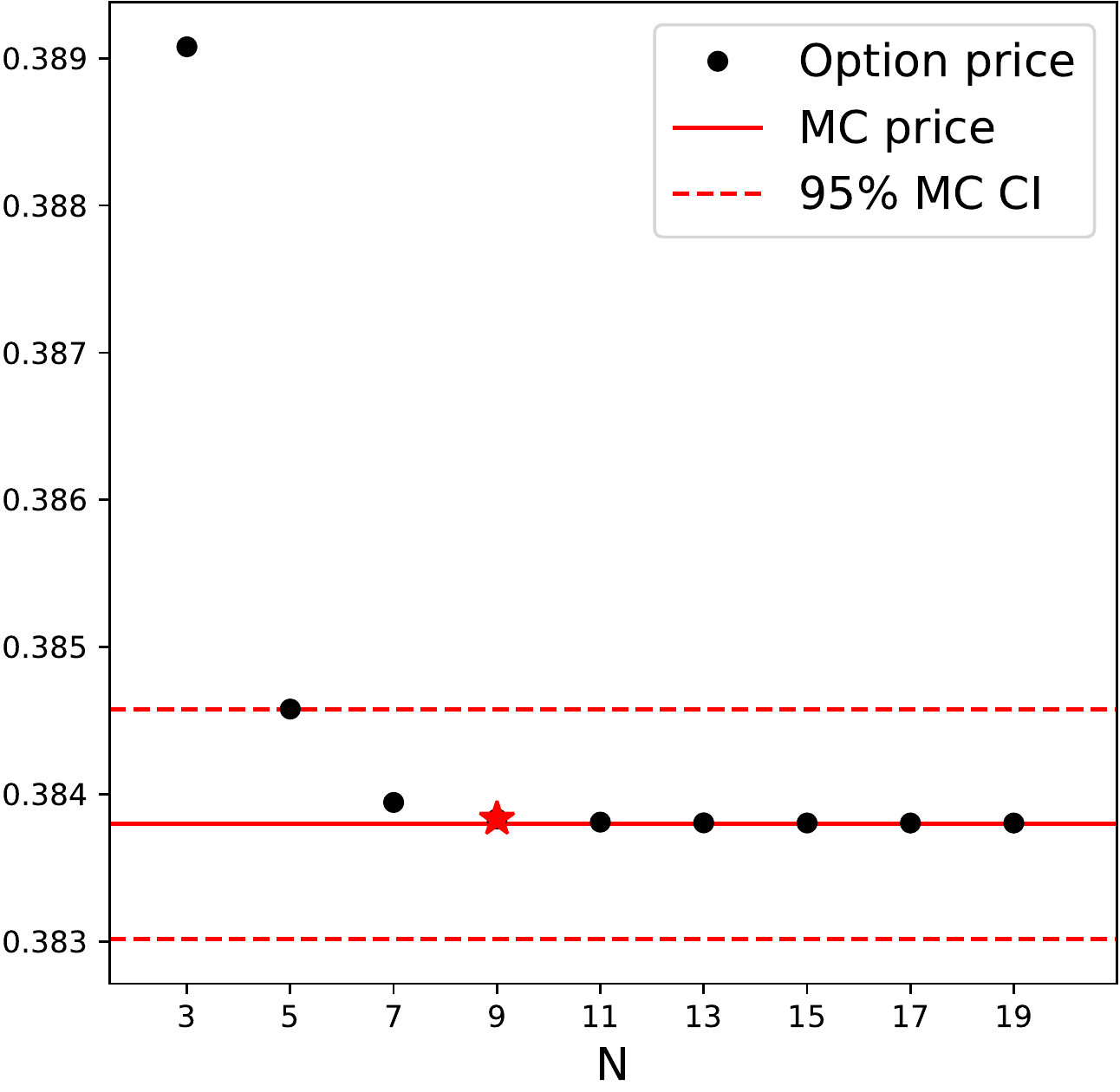}\\
		\end{tabular}
	}
	\caption{$\gamma_{a,b}^N$ as a function of $N$, when the underlying process is an $(m+1)$-point weighted jump-diffusion with NIG measure and $\sigma_X(T;t)=1.39$. The drift is $a= \mu_X(T;t)=2.0$ and the scale is $b = 1.2 \underline{b}_{\sigma}$. The black dots represent the option price, the solid red line the MC price, the two red dashed lines are the $95\%$ confidence interval for MC. The red stars denote the prices chosen with the stopping criterion \eqref{criterio}. The plots in the second and fourth columns are a zoomed subplot of the plots in the first and third columns. \label{JDplot_corr}}
\end{figure}

\section{Conclusions}
\label{conclusions}
We derive explicit pricing formulas for discrete-average arithmetic Asian options in the context of polynomial jump-diffusion processes. This can be easily extended to continuous-average Asian options by approximating the integral with a discrete sum and appropriate small time step. 

The approach is based on the approximation of the call-payoff function with generalized Hermite polynomials and can be easily extended to any other payoff function that is well approximated with polynomials. In particular, the generalized Hermite polynomials are defined in relation to a weight which is basically the density of a Gaussian random variable with mean $a$ (the drift) and standard deviation $b$ (the scale). Hence we get a family of approximations depending on the parameters $a$ and $b$.

By considering the underlying spot price process from the family of jump-diffusion polynomial processes, we then obtain a fully explicit expression for the price functional thanks to the well-known moment formula for polynomial processes and to the correlator formula derived in \cite{benth2019}. This has the advantage to allow for sensitivity analysis, since Greeks are within reach, as we show in Section \ref{greeks}.

From the numerical point of view, the most time consuming part is the computation of moments and correlators. However, since these do not depend, for example, on the strike price of the option, one can compute moments and correlators and use these values to evaluate, e.g., different options with different strike prices.

We can summarize the following findings:
\vspace{-1em}

\begin{itemize}[leftmargin=*]
	\setlength\itemsep{-0.2em}
	\item We provide a lower bound for the scale $b$ which is proportional to the standard deviation $\sigma$ of the underlying process. Values for $b$ smaller than this threshold do not guarantee convergence. On the other hand, big values for $b$ slow down the convergence rate. Moreover, despite the lower bound is proved for the case of a Gaussian underlying process, numerical experiments show similar behaviours also for non-Gaussian processes.
	
	\item We find analytically a relation between the behaviour of the series and the ratio $\frac{b}{\sigma}$ which is also confirmed by numerical results. For example, when doubling the value of $\sigma$, also $b$ must be doubled to get the same behaviour for the approximation. Indeed, in the experiments with the Brownian motion and the Gaussian Ornstein--Uhlenbeck process, by choosing the parameters in such a way that the two processes have the same standard deviation, we obtain two series with very similar behaviours. Despite this being proved for the Gaussian case, the numerical experiments for the non-Gaussian case are in line with this theory.
	
	\item Working with polynomial functions require the evaluation of high-order powers. In our context, this means the evaluation of high-order moments or correlators in the underlying process. The bigger is the initial value of the process, the higher is the value of its moments and correlators. High initial values coupled with high-order powers create numerical instabilities due to rounding errors. This is the main reason which makes the series to diverge after convergence. If the initial value of the process is too high, then the numerical instabilities might start even before reaching convergence.
	
	\item Numerical experiments with Gaussian underlying processes, for which closed price formulas are available for benchmarking, show that the Hermite price approximation can reach much higher accuracies than the MC approach, namely $10^{-10}$ against $10^{-3}$ for MC. However, for big values of the scale $b$, the convergence might be so slow that to reach such an accuracy level one needs many more terms than what numerically feasible. Despite we do not have closed price formulas for the jump-diffusion case, the plots obtained show that the Hermite series converges to a value very closed to the price value approximated via MC. Based on the experiments in the Gaussian case, we actually believe the Hermite series to be more accurate than the MC approximation.
\end{itemize}
\vspace{-1em}

The generalized Hermite polynomials can be replaced with any other family of orthogonal polynomials. In particular, Theorem \ref{error} sets a very strong condition on the tails behaviour for the distribution of the underlying process, which is required to vanish faster than a Gaussian density function. Despite this is usually not true for jump processes, such a condition is only sufficient, and not necessary, for convergence. The numerical results on the jump-diffusion process reveal indeed that convergence is reached even with a NIG jump measure. We also point out that condition \eqref{psicond} is obtained from a  weight function being the density of a Gaussian random variable. Thus, by considering a different weight function, one may obtain theoretical convergence results also for other families of distributions (other than the Gaussian one). For example, the Laguerre polynomials are orthogonal polynomials with respect to the weight function $w(x) = e^{-x}$, which would allow to prove convergence, e.g., for the Gamma distribution. Other examples can be found in \cite[Appendix B]{schoutens2012}. Similarly, \cite{willems2019} obtained convergence for a log-normal distribution by considering a log-normal density as the weight function.

We finally point out that working with the class of polynomial jump-diffusion processes is the key for getting fully explicit price formulas. However, the price approximation with Hermite polynomials might still be applied to other kind of processes. In these cases, one must then rely on Monte Carlo simulations for calculating moments and correlators. Similarly, remaining in the class of polynomial jump-diffusions, it is numerically possible to avoid the use of the correlator formula. Indeed, as pointed out in \cite{benth2019}, this can be numerically replaced by the iterative application of the moment formula combined with the tower rule. However, in both these cases, one looses the advantages of an explicit price functional. 

\appendix
\section{Proofs}
\label{proof}
This section contains the proofs of the main results together with an additional auxiliary result that is needed for the proofs.
\begin{lemma}
	\label{lem2}
	We introduce the map $t_{a,b}: \R \to \R, \; x \mapsto \frac{x-a}{b}$. Then:
	\begin{enumerate}
		\item The weight function $w_{a,b}$ is obtained by composing $w$ with $t_{a,b}$, namely $w_{a,b}(x) = w\left(\frac{x-a}{b}\right).$
		\item The generalized Hermite polynomial $q_n^{a,b}$ is obtained by composing the Hermite polynomial $q_n$ with $t_{a,b}$ and by scaling with the inverse of $b^n$, namely $q_n^{a,b}(x)= \frac{1}{b^n}q_n\left(\frac{x-a}{b}\right),$ $n\ge 0$.
	\end{enumerate}
\end{lemma}
\begin{proof}
	The first part of the lemma is easily verified. For the second part, we proceed by induction on the order $n\ge0$. Since for $n=0$, $q_0^{a,b}(x) = q_0(x) =1$ is trivial, we start from $n=1$.
	\begin{itemize}[leftmargin=*]
		\item $n=1$: one finds that $q_1(x) = x$ and $q_1^{a,b}(x) = \frac{x-a}{b^2} = \frac{1}{b^1}q_1\left(\frac{x-a}{b}\right)$, so the base case holds;
		\item $n \to n+1$: we assume the statement holds for $n$. This means that
		\begin{equation*}
			 (-1)^ne^{\frac{(x-a)^2}{2b^2}}\frac{d^n}{dx^n}e^{-\frac{(x-a)^2}{2b^2}}= \frac{1}{b^n}q_n\left(\frac{x-a}{b}\right) \quad \mbox{hence} \quad  (-1)^n\frac{d^n}{dx^n}e^{-\frac{(x-a)^2}{2b^2}}= \frac{1}{b^n}e^{-\frac{(x-a)^2}{2b^2}}q_n\left(\frac{x-a}{b}\right).
		\end{equation*}
		We now focus on $q_{n+1}^{a,b}$: by induction hypothesis
		\begin{align*}
			q_{n+1}^{a,b}(x)&= (-1)^{n+1}e^{\frac{(x-a)^2}{2b^2}}\frac{d^{n+1}}{dx^{n+1}}e^{-\frac{(x-a)^2}{2b^2}}
			 =   (-1)e^{\frac{(x-a)^2}{2b^2}}\frac{d}{dx}\left(\frac{1}{b^n}e^{-\frac{(x-a)^2}{2b^2}}q_n\left(\frac{x-a}{b}\right)\right)\\
			 & =\frac{1}{b^{n+1}}\left(\left(\frac{x-a}{b}\right)q_n\left(\frac{x-a}{b}\right)-q^{'}_n\left(\frac{x-a}{b}\right)\right) = \frac{1}{b^{n+1}}q_{n+1}\left(\frac{x-a}{b}\right),
		\end{align*}
		where the last equality is due to the recurrence relation for Hermite polynomials, namely $q_{n+1}(y) = yq_n(y)-q'_n(y)$, see \cite{djordjevic1996}. This concludes the proof.
	\end{itemize}
\end{proof}

\subsubsection*{Proof of Lemma \ref{prop}}
\begin{proof}
	From the definition of norm in $L^2_{a,b}$ and Lemma \ref{lem2}, we get
	\begin{equation*}
		\left \| q_n^{a,b}\right \| ^2_{L^2_{a,b}} = \int_{-\infty}^{\infty} \left(q_n^{a,b}(x)\right)^2 w_{a,b}(x)dx
		= \frac{b}{b^{2n}}\int_{-\infty}^{\infty} \left(q_n\left(y\right)\right)^2 w\left(y\right)dy = \frac{\sqrt{2\pi}n!}{b^{2n-1}},
	\end{equation*}
	where we have used the change of variables $y= \frac{x-a}{b}$ and equation \eqref{norm}.
\end{proof}

\subsubsection*{Proof of Proposition \ref{propphi}}
\begin{proof}
We start by proving that the value of the integrals in the series \eqref{phiK} is
\begin{equation}
\label{cases}
\int_{-\infty}^{\infty}\varphi_K(y)q_n^{a,b}(y)w_{a,b}(y)dy =\begin{cases}
b\sqrt{2\pi}\left(b\,\phi\left(\frac{K-a}{b}\right)+\left(a-K\right)\left(1-\Phi\left(\frac{K-a}{b}\right)\right)\right) & \mbox{ for } n=0\\
b\sqrt{2\pi}\left(1-\Phi\left(\frac{K-a}{b}\right)\right) & \mbox{ for } n=1\\
\sqrt{2\pi}\phi\left(\frac{K-a}{b}\right)q_{n-2}^{a,b}\left(K\right) & \mbox{ for } n\ge 2 
\end{cases}.
\end{equation}
The first two equalities are easily verified with $q_0^{a,b}(y)=1$ and $q_1^{a,b}(y)=\frac{y-a}{b^2}$ by, possibly, the change of variables $z = \frac{y-a}{b}$ and by integrating (twice) by parts. For $n\ge2$, we integrate by parts twice:
\begin{multline*}
	\int_{-\infty}^{\infty}\varphi_K(y)q_n^{a,b}(y)w_{a,b}(y)dy =(-1)^n \int_{K}^{\infty}(y-K)\frac{d^n}{dy^n}e^{-\frac{(y-a)^2}{2b^2}} dy = \left.(-1)^{n-2} \frac{d^{n-2}}{dy^{n-2}}e^{-\frac{(y-a)^2}{2b^2}}\right|_{y=K}  \\=e^{-\frac{(K-a)^2}{2b^2}}\left(\left.(-1)^{n-2} e^{\frac{(y-a)^2}{2b^2}}\frac{d^{n-2}}{dy^{n-2}}e^{-\frac{(y-a)^2}{2b^2}}\right)\right|_{y=K} =  \sqrt{2\pi}\phi\left(\frac{K-a}{b}\right)q_{n-2}^{a,b}\left(K\right).
\end{multline*}
From equations \eqref{phiK} and \eqref{cases}, the function $\varphi_K$ is then expressed in terms of the GHPs by
\begin{equation*}
	\varphi^{a,b}_K(x) = \sum_{n=0}^{\infty}
	\alpha_n^{a,b} q_n^{a,b}(x)\qquad  \mbox{ with } \alpha_n^{a,b} :=\begin{cases}
		b\,\phi\left(\frac{K-a}{b}\right)+\left(a-K\right)\left(1-\Phi\left(\frac{K-a}{b}\right)\right) & \mbox{ for } n = 0\\
		b^2\left(1-\Phi\left(\frac{K-a}{b}\right)\right) & \mbox{ for } n = 1\\
		\frac{b^{2n-1}}{n!} \phi\left(\frac{K-a}{b}\right)q_{n-2}^{a,b}\left(K\right)& \mbox{ for } n\ge 2
	\end{cases}.
\end{equation*}
The result then follows by Lemma \ref{lem2}.
\end{proof}

\subsubsection*{Proof of Theorem \ref{price}}
\begin{proof}
	Starting from equation \eqref{phiNK}, for every $N\ge 0$ we can express $\Pi^{a,b}_{K,N}(t)$ in matrix form by
	\begin{equation}
		\label{Piab}
		\Pi^{a,b}_{K,N}(t) 
		=\boldsymbol{\beta}^{a,b\,\top}_NM_N\boldsymbol{E}^X_N(T;t) \qquad \mbox{ where } \qquad  \boldsymbol{E}^X_N(T;t) := \E\left[  \left.H_N\left(\frac{X(T)-a}{b}\right)\right|\F_t\right].
	\end{equation}
	Thus we need to compute the entries of $\boldsymbol{E}^X_N(T;t)$. By the binomial theorem, its $k$-th component, $k=1,\dots, N+1$, is of the form
	\begin{equation*}
		\boldsymbol{E}^X_N(T;t)_k = 
		\frac{1}{b^{k-1}}\E\left[\left.\left(X(T)-a\right)^{k-1}\right|\F_t\right] = \frac{1}{b^{k-1}}\sum_{i=0}^{k-1} \binom{k-1}{i} (-a)^{k-1-i}\, \E\left[\left.X(T)^{i}\right|\F_t\right].
	\end{equation*}
	%
	Then, for $\hat{\boldsymbol{\beta}}_N^{a,b} := \boldsymbol{\beta}^{a,b \,\top}_N M_N$, we rewrite $\Pi_{K,N}^{a,b}(t)$ in equation \eqref{Piab} by expanding the matrix multiplication into the following sum
	\begin{equation*}
		\Pi_{K,N}^{a,b}(t) = \sum_{j=1}^{N+1}\hat{\boldsymbol{\beta}}_{N,j}^{a,b} \boldsymbol{E}^X_N(T;t)_j=\sum_{k=0}^N\hat{\boldsymbol{\beta}}_{N,k+1}^{a,b}\frac{1}{b^k}\sum_{i=0}^k \binom{k}{i} (-a)^{k-i}\, \E\left[\left.X(T)^{i}\right|\F_t\right],
	\end{equation*}
	where the last equality we used the change $k=j-1$. This concludes the proof.
\end{proof}

\subsubsection*{Proof of Theorem \ref{error}}
\begin{proof}
	By definition of $\Pi_K(t)$ and $\Pi^{a,b}_{K,N}(t)$, we can write that
	\begin{align*}
		&\left|\Pi_K(t)- \Pi^{a,b}_{K,N}(t)\right|  = \left|\E\left[ \left.\varphi_K(X(T))\right|\F_t\right] - \E\left[\left. \varphi_{K,N}^{a,b}(X(T))\right|\F_t\right]\right| \\
		&= \left|\E\left[\left. \varphi_K(X(T))- \varphi_{K,N}^{a,b}(X(T))\right|\F_t\right]\right| \le \E\left[\left. \left|\varphi_K(X(T))- \varphi_{K,N}^{a,b}(X(T))\right|\right|\F_t\right],
	\end{align*}
	were the last inequality is due to Jensen's inequality. By definition of conditional expectation, and being $\psi_{X(T)}$ the density function of $X(T)$, by the Cauchy–Schwarz inequality, this becomes
	\begin{multline*}
		\E\left[\left. \left|\varphi_K(X(T))- \varphi_{K,N}^{a,b}(X(T))\right|\right|\F_t\right] =\int_{\R}    \left|\varphi_K(x)- \varphi_{K,N}^{a,b}(x)\right|\psi_{X(T)}(x)dx \\
		\le \left( \int_{\R}    \left|\varphi_K(x)- \varphi_{K,N}^{a,b}(x)\right|^2\omega_{a,b}(x)dx\right)^{\frac{1}{2}} \left( \int_{\R}    \psi_{X(T)}^2(x)\omega_{a,b}^{-1}(x)dx\right)^{\frac{1}{2}},
	\end{multline*}
	which concludes the proof.
\end{proof}

\subsubsection*{Proof of Proposition \ref{condb}}
\begin{proof}
	From equation \eqref{Cab}, we see that  for $C_{a,b}$ to make sense, condition \eqref{psicondG} must hold. In particular, it is obvious that it cannot be $b < \underline{b}_{\sigma}$ because of the squared root, while the expression becomes singular for $b = \underline{b}_{\sigma}$, hence, in this case, we can expect instabilities in the approximation.
\end{proof}

\subsubsection*{Proof of Corollary \ref{cor}}
\begin{proof}
	By taking $a= \mu$ in equation \eqref{Cab}, we obtain that
	\begin{equation*}
		C_{a,b}^2 =  \frac{1}{\sqrt{2\pi\sigma^2}}\frac{b}{\sqrt{2b^2-\sigma^2}} \qquad \mbox{ and } \qquad \frac{\partial C_{a,b}^2}{\partial b} = - \frac{1}{\sqrt{2\pi\sigma^2}}\frac{\sigma^2}{\left(2b^2-\sigma^2\right)^{\frac{3}{2}}}<0,
	\end{equation*}
	hence $C_{a,b}^2$ is decreasing in $b$. The limit is easily obtained. 
\end{proof}

\subsubsection*{Proof of Proposition \ref{distribution}}
	\begin{proof}
	One can easily observe that 
	\begin{equation*}
		\int_{t}^{s_j}e^{b_1(s_j-v)}dB(v) = \sum_{k=0}^{j} \int_{s_{k-1}}^{s_k}e^{b_1(s_j-v)}dB(v), \quad \mbox{ for }  \quad j=0, \dots, m,
	\end{equation*}
	hence, by rearranging the terms in the two summations, we write that
	\begin{equation*}
		\sum_{j=0}^m \int_{t}^{s_j}e^{b_1(s_j-v)}dB(v) = \sum_{j=0}^m\sum_{k=0}^{j} \int_{s_{k-1}}^{s_k}e^{b_1(s_j-v)}dB(v)
		= \sum_{k=0}^m \int_{s_{k-1}}^{s_k}\left(\sum_{j=k}^{m}e^{b_1(s_j-v)}\right)dB(v).
	\end{equation*}
	By switching the role of $k$ and $j$, we get the result.
\end{proof}

\section{Some details on the correlator formula}
\label{appA}
In this appendix, we briefly report the definitions needed to understand the correlator formulas in Theorem \ref{theorem}. For more details and for the idea behind this construction, we refer the reader to \cite{benth2019}. 

\begin{definition}[Vectorization]
	\label{vec}
	Given a matrix $A \in \R^{n\times m}$ whose $j$-th column we denote by $A_{:j}$, we define the \emph{vectorization} of $A$ as the operator $vec: \R^{n\times m} \to \R^{nm\times 1}$ that associates to $A$ the $nm$-column vector $vec(A) = \begin{pmatrix}
		A_{:1}^{\top} & A_{:2}^{\top}&\cdots & A_{:m}^{\top}
	\end{pmatrix}^{\top}.$
\end{definition}

\begin{definition}[Inverse-vectorization]
	\label{ivec}
	Given a vector $v \in \R^{nm}$, we define the \emph{inverse-vectorization} of $v$ as the operator $vec^{-1}: \R^{nm} \to \R^{n\times  m}$ that associates to $v$ the $n\times m$ matrix
	$A = vec^{-1}(v)$ such that $[A]_{i,j} = v_{n(j-1)+i}$, for $i=1, \dots, n$ and $j=1, \dots, m$.
\end{definition}

\begin{definition}[L-vectorization]
	\label{vecL}
	Given a matrix $A \in \R^{n\times m}$ with elements $[A]_{i,j} = a_{i,j}$, for $1\le i \le n$ and $1 \le j \le m$, we define the \emph{L-vectorization} of $A$ as the operator $vecL: \R^{n\times m} \to \R^{n+m-1}$ that associates to $A$ the $(n+m-1)$-column vector obtained by selecting the first column and the last row of $A$, namely $\NiceMatrixOptions{transparent}
		vecL(A) = \begin{pmatrix}
			a_{1,1}&a_{2,1}&\cdots&a_{n,1}&a_{n,2}&\cdots&a_{n,m}
		\end{pmatrix}^{\top}.$
\end{definition}

\begin{definition}[Hankel matrix]
	\label{spaceA}
	We define $\mathcal{A}_{n,m}\subset \R^{n\times m}$ as the space of matrices whose elements on the same skew-diagonal coincide.  We call $A \in \mathcal{A}_{n,m}$ an  \emph{Hankel matrix} and write $\mathcal{A}_{n} := \mathcal{A}_{n,n}$.
\end{definition}

\begin{definition}[Kronecker product]
	\label{kron}
	The \emph{Kronecker product} of a matrix $A \in \R^{n\times m}$ with elements $[A]_{i,j} = a_{i,j}$, for $1\le i \le n$ and $1 \le j \le m$, and a matrix $B \in \R^{r\times s}$, is defined by
	\begin{equation*}
		\NiceMatrixOptions{transparent}
		A \otimes B = \begin{pmatrix}
			a_{1,1}B& \cdots &a_{1,m}B\\
			\vdots & \ddots& \vdots \\
			a_{n,1}B & \cdots & a_{n,m}B
		\end{pmatrix}\in \R^{nr \times ms}.
	\end{equation*}
\end{definition}

\begin{definition}[d-Kronecker product]
	\label{mkrondef}
	We define the \emph{d-Kronecker product} between $A \in \R^{n\times m}$ and $B \in \R^{r\times s}$, as the $d$-th Kronecker power of $A$ multiplied in the Kronecker sense with $B$, for $d\ge1$, or equal to $B$, for $d=0$, namely
	\begin{equation*}
		\begin{cases}
			A\otimes ^{d} B = A^{\otimes d} \otimes B & d\ge 1\\
			A\otimes ^{0} B = B & d = 0\\
		\end{cases}.
	\end{equation*}
\end{definition}

\begin{definition}[L-eliminating matrix]
	\label{Enmth}
	For $n,m\ge 1$ and $A \in \R^{n\times m}$, we define the L-eliminating matrix as the matrix $E_{n,m}\in \R^{(n+m-1)\times nm}$ such that $E_{n,m}vec(A) = vecL(A)$. We write $E_{n} := E_{n,n}$.
\end{definition}

\begin{definition}[L-duplicating matrix]
	\label{Dnth}
	For $n,m\ge1$ and $A \in \mathcal{A}_{n,m}$, we define the L-duplicating matrix as the matrix $D_{n,m}\in\R^{nm\times (n+m-1)}$ such that
	$D_{n,m}vecL(A) = vec(A).$ We write $D_{n} := D_{n,n}$.
\end{definition}

\noindent Let now $X_n^{(m)}(x):=H_n(x)^{\top} \otimes^mH_n(x)$. 

\begin{definition}[m-th L-eliminating matrix]
	\label{propE}
	For $n, m\ge 1$, we define the \emph{$m$-th L-eliminating matrix} as the matrix  $E_{n+1}^{(m)}\in \R^{(n(m+1)+1)\times(n+1)^{m+1}}$ such that $E_{n+1}^{(m)}vec(X_n^{(m)}(x)) = H_{n(m+1)}(x)$. In particular
	\begin{equation*}
		\begin{cases}
			E_{n+1}^{(1)} = E_{n+1}& m=1\\
			E_{n+1}^{(m)} = E_{nm+1, n+1}\left(I_{n+1} \otimes E_{n+1}^{(m-1)}\right) & m\ge 2
		\end{cases}.
	\end{equation*}
\end{definition}

\begin{definition}[m-th L-duplicating matrix]
	\label{propD}
	For $n,m\ge 1$, we define the \emph{$m$-th L-duplicating matrix} as the matrix $D_{n+1}^{(m)}\in \R^{(n+1)^{m+1}\times (n(m+1)+1)}$ such that $D_{n+1}^{(m)} H_{n(m+1)}(x)=vec(X_n^{(m)}(x))$. In particular
	\begin{equation*}
		\begin{cases}
			D_{n+1}^{(1)} = D_{n+1}& m=1\\
			D_{n+1}^{(m)} = \left(I_{n+1} \otimes D_{n+1}^{(m-1)}\right)D_{nm+1, n+1} & m\ge 2
		\end{cases}.
	\end{equation*}
\end{definition}

\end{document}